\documentclass[a4paper,12pt]{book}
\usepackage[utf8]{inputenc}
\usepackage[width=15cm,top=3cm,bottom=3cm,bindingoffset=2cm]{geometry}
\usepackage{setspace} 

\usepackage[pdftex]{graphicx,color}
\usepackage{tikz}
\usepackage[font=footnotesize,labelfont=bf]{caption}
\usepackage{wrapfig}

\newenvironment{compactwrapfigure}
{\setlength{\intextsep}{.5\intextsep}%
\wrapfloat{figure}}
{\endwrapfloat}

\usepackage{comment}
\usepackage[british]{babel} 

\usepackage[numbers, square, comma, sort&compress]{natbib}
\definecolor{vdarkgreen}{RGB}{0,51,0}
\definecolor{vdarkblue}{RGB}{0,0,51}
\definecolor{vdarkred}{RGB}{0,100,0}
\usepackage[hidelinks]{hyperref}
\usepackage{bookmark}
\hypersetup{colorlinks=false}

\usepackage{makeidx}\makeindex

\usepackage{amssymb}
\usepackage{amsmath}
\usepackage{amsthm}
\usepackage{latexsym}
\usepackage{revsymb}

\newcommand{\vertiii}[1]{{\left\vert\kern-0.25ex\left\vert\kern-0.25ex\left\vert #1 
    \right\vert\kern-0.25ex\right\vert\kern-0.25ex\right\vert}}
\newcommand{\abs}[1]{\left| #1 \right|} 
 
\newcommand{\ket}[1]{| #1 \rangle} 
\newcommand{\bra}[1]{\langle #1 |} 
\newcommand{\braket}[2]{\left< #1 \vphantom{#2} \right|
 \left. #2 \vphantom{#1} \right>} 
\newcommand{\f}[2]{\textstyle{\frac{#1}{#2}}}
\newcommand{\g}[1]{\mathbf{#1}}

\newcommand{\proj}[1]{| #1 \rangle\langle #1 |}

\newcommand{\m}[1]{\mathcal{#1}}
\newcommand{\tr}[1]{\mathrm{tr}\! \left[#1\right]}

\newcommand{\norm}[1]{\|#1\|}
\newcommand{\av}[1]{\left\langle #1 \right\rangle}
\newcommand{\eff}{\rm eff}

\newcommand{\mM}{\mathcal{M}}

\newcommand{\mH}{\mathcal{H}}

\newcommand{\cM}{\mathcal{M}}

\newtheorem{theorem}{Theorem}[chapter] 
\newtheorem{lemma}[theorem]{Lemma}

\newtheorem{corollary}[theorem]{Corollary}
\newtheorem{definition}{Definition}[chapter] 
\newtheorem{Definition}{Definition}

\newcommand{\Qed}{\nobreak \ifvmode \relax \else
      \ifdim\lastskip<1.5em \hskip-\lastskip
      \hskip1.5em plus0em minus0.5em \fi \nobreak
      \vrule height0.75em width0.5em depth0.25em\fi}

\begin{document}

\begin{titlepage}
    \begin{center}
        \vspace*{1cm}
        \Huge
          \textbf{Insights from Quantum Information into Fundamental Physics}
          
        \vspace{0.5cm}
        \LARGE        
        \vspace{1.5cm}

        \textbf{Terence C.\ Farrelly}
        \vfill
        
        A thesis submitted for the degree of\\
        {\em Doctor of Philosophy}
        
        \vspace{0.8cm}
        
        \includegraphics[width=0.2\textwidth]{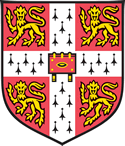}
        
        \Large
        
        Department of Applied Mathematics and Theoretical Physics\\
        St John's College, University of Cambridge\\
        United Kingdom\\
        
        $16^{\mathrm{th}}$ June 2015
    \end{center}
\end{titlepage}

\cleardoublepage
\phantomsection
\addcontentsline{toc}{chapter}{Abstract}
\printindex

\onehalfspacing 

\begin{center}
    \Large
    \textbf{Insights from Quantum Information into Fundamental Physics}    

    \vspace{0.4cm}
    \large
    
    \vspace{0.4cm}
    \textbf{Terence C.\ Farrelly}
    
    \vspace{0.9cm}
    \textbf{Abstract}
\end{center}
This thesis is split into two parts, which are united in the sense that they involve applying ideas from quantum information to fundamental physics.

The first part is focused on examining discrete-time models in quantum computation (discrete-time quantum walks and quantum cellular automata) as discretized models of relativistic systems.  One of the results here is a theorem demonstrating that a large class of discrete-time quantum walks have relativistic dynamics in the continuum limit.  Additionally, the problem of fermion doubling for these models is investigated, and it is seen that the problem can be circumvented in two dimensional space.  This was already known for one dimensional systems.

Another result involves taking the limits of causal free field theories in discrete spacetime to recover continuum field theories, something that is not straightforward because of the nontrivial nature of the vacuum in quantum field theory.  Additionally, it is shown that general systems of fermions evolving causally in discrete spacetime can be represented by quantum cellular automata, which makes them efficiently simulable by quantum computers.  A related result is that quantum cellular automata composed of fermions are equivalent to regular quantum cellular automata.

In the second part of this thesis, the focus is on the foundations of statistical physics.  The main result of this part is a general bound on the time it takes a quantum system to effectively reach equilibrium.  The discussion also includes a practical definition of equilibration that takes our measurement capabilities into account.  Finally, the nature of the equilibrium state is also discussed, with a focus on initial state independence, which relates to the important question of when the equilibrium state is a Gibbs state.



\cleardoublepage
\phantomsection
\addcontentsline{toc}{chapter}{Declaration}
\printindex

\begin{center}
    
    \vspace{9cm}
    \Large
    \textbf{Declaration}
\end{center}
This dissertation is the result of work that I have done as part of my Ph.D.\ studies in the Department of Applied Mathematics and Theoretical Physics at the University of Cambridge.  Any work done in collaboration is declared in the \hyperlink{prol}{prologue} and also specified in the text.

This dissertation is not substantially the same as any that I have submitted, or, is being concurrently submitted for a degree or diploma or other qualification at the University of Cambridge or any other University or similar institution. I further state that no substantial part of my dissertation has already been submitted, or, is being concurrently submitted for any such degree, diploma or other qualification at the University of Cambridge or any other University or similar institution.

\vfill

\noindent\begin{tabular}{ll}
Signature: \makebox[2.2in]{\hrulefill}\ \  & Date: \makebox[1.8in]{\hrulefill}
\end{tabular}



\cleardoublepage
\phantomsection
\addcontentsline{toc}{chapter}{Contents}
\tableofcontents

\cleardoublepage
\phantomsection
\addcontentsline{toc}{chapter}{List of Figures}
\listoffigures

\markboth{}{}  
\cleardoublepage
\phantomsection
\thispagestyle{plain}
\addcontentsline{toc}{chapter}{Prologue}
\printindex

\begin{center}
    \vspace{9cm}
    \Large
    \textbf{\hypertarget{prol}{Prologue}}
\end{center}
\label{Pub}
The unifying theme throughout this thesis is the application of ideas and tools from quantum information to understand fundamental physics.  This is not a new idea, and we will discuss some previous work along these lines, but there remain plenty of unexplored possibilities.  In this thesis, we will look at applying ideas from quantum information to relativistic and statistical physics.  Because these subjects are quite disparate, this thesis is divided into two parts.  The first part looks at models from quantum computation, namely quantum walks and quantum cellular automata, and uses them as discrete models of relativistic quantum systems.  The second part uses tools from quantum information theory to study the problem of equilibration in quantum systems.

The first part (specifically section \ref{sec:Two Dimensional Coins are Special} in chapter \ref{chap:Quantum Walks and Relativistic Particles} and sections \ref{sec:Preparing the Majorana state} and \ref{sec:Causal Discrete-Time Models on a Lattice} in chapter \ref{chap:Quantum Cellular Automata and Field Theory}) of this thesis includes material from the following two papers.
\begin{itemize}
 \item \cite{FS14} T.~C. Farrelly and A.~J. Short. Discrete spacetime and relativistic quantum particles. {\em Phys. Rev. A}, 89:062109, 2014.
 
 \item \cite{FS13} T.~C. Farrelly and A.~J. Short. Causal fermions in discrete space-time. {\em Phys. Rev. A}, 89:012302, 2014.
\end{itemize}
The main results from \cite{FS14} are theorem \ref{th:13}, corollary \ref{cor:12}, theorem \ref{th:14} and the counterexample in section \ref{sec:Counterexample with a Three Dimensional Coin}.  The main results from \cite{FS13} are theorem \ref{th:1}, corollary \ref{cor:1} and theorem \ref{th:CS1}.

The second part (specifically sections \ref{sec:Timescales for Equilibration of Expectation Values}, \ref{sec:Finite Time Equilibration for Systems and Subsystems} and \ref{sec:Slow Equilibration} in chapter \ref{chap:Equilibration}) of this thesis includes material from the two papers below.
\begin{itemize}
 \item \cite{SF12} A.~J. Short and T.~C. Farrelly. Quantum equilibration in finite time. {\em New Journal of Physics}, 14(1):013063, 2012.
 \item \cite{MGLFS14} A.~S.~L. Malabarba, L.~P. Garc\'{i}a-Pintos, N.~Linden, T.~C. Farrelly, and
  A.~J. Short. Quantum systems equilibrate rapidly for most observables.  {\em Phys. Rev. E}, 90:012121, 2014.
\end{itemize}
The main results from \cite{SF12} are theorems \ref{th:exptval}, \ref{th:syseqmeas} and \ref{th:subeq}.  As my role in the \cite{MGLFS14} was not as significant as in the others, most of its results have been omitted, except for the slow equilibration result in section \ref{sec:Slow Equilibration}.  The result is relevant to the discussion of equilibration, though it is not my own work.

Naturally, while writing this, some of the results were expanded and may be a little different from the originals in the papers above.  An example is theorem \ref{cor:2}, which is partly an extension of a result from \cite{FS13}.  Some proofs have also been improved.

Additionally, many of the results presented in this thesis are unpublished.  In chapter \ref{chap:Quantum Walks and Relativistic Particles}, sections \ref{sec:Faster Converging Quantum Walks}, \ref{sec:Symmetries on the Lattice}, \ref{sec:Fermion Doubling in Two Dimensions}, \ref{sec:Fermion Doubling in Three Dimensions} and \ref{sec:Quantum Walks on a Line: Abstract Theory} are unpublished.  In chapter \ref{chap:Quantum Cellular Automata and Field Theory} sections \ref{sec:Discrete Fermion Fields and the Vacuum}, \ref{sec:Continuum Limits2}, \ref{sec:U(1) Lattice Gauge theory as a spin model} and \ref{sec:Structure of the evolution} are unpublished.  In chapter \ref{chap:Equilibration}, section \ref{sec:The Equilibrium state} is unpublished.

Finally, although section \ref{sec:Quantum Cellular Automata as Quantum Field Simulators} is unpublished, I believe that the main idea is probably known to some in quantum information.  That said, I am unaware of any specific references to it.  Sections \ref{sec:A General Recipe} and \ref{sec:Equilibration Times with Physical Spectra} are also unpublished, though the ideas, particularly in the latter, are probably well known.

\part{Quantum Computation and Relativistic Physics}

\chapter{Introduction}
\label{chap:Introduction1}
\begin{center}
{\em A child of five would understand this. Send someone to fetch a child of five.}\\
 - Groucho Marx
 \end{center}
There is certainly a precedent for the idea that quantum information can tell us something new about relativistic physics.  Tools like tensor network states, which allow a concise description of particular classes of states, have been applied to lattice quantum field theory to perform numerical simulations \cite{BHVVV14,TCL14}.  Practical knowledge arising from the experimental side of quantum information has inspired proposals for analogue simulations of quantum field theories on systems of trapped ions and cold atoms  \cite{CMP10,CLEGRGS11,CMLS12}.  So why stop there?  There are other tools in quantum information that could have promising applications to the study of relativistic physics.  Our goal in part one of this thesis is to look at some of these tools to see what they can tell us about relativistic physical systems.  The tools we will look at come from quantum computation.

It is quite natural that we should look to quantum computation in particular for insights into relativistic physics.  Great advances were made in the seventies by considering the problem of simulating quantum field theories on a computer by discretizing space \cite{Wilson93}.  This was made necessary by the realization that the strong force could not be treated perturbatively at low energies.  The end result was lattice quantum field theory, which has not only allowed calculations of physical quantities numerically, but has also provided a concrete mathematical structure underpinning quantum field theory itself.  In fact, interacting quantum field theories are typically defined as the continuum limits of the lattice theories \cite{Creutz83}.  Another development that arose from the study of lattice quantum field theory was the modern viewpoint of renormalization, which gave it a more physically tangible status than that of a mere trick used to remove infinities in perturbative calculations.  This all lends a lot of credence to the idea that {\it quantum} computation could also help us to better understand relativistic physics.

In retrospect, the idea that computation can tell us something about nature is not so unreasonable.  Asking what is possible with a computer can tell us about nature itself.  An example that is especially relevant for us is the question of what physical systems we can efficiently simulate on a computer.\footnote{We will define efficient simulation and some of the other concepts mentioned here in chapter \ref{chap:Background1}.}  If we could show that some physical model was not efficiently simulable, this would have very interesting consequences for our understanding of it.  It would also suggest the possibility that that model could offer an inequivalent framework for computation.

In the opposite direction, our best theory of physics provides us with limits on what kind of computations are possible:\ in reality, we must actually build computers, and physics tells us what we can and cannot build.  As far as we know, nature is quantum, which opens up the possibility of constructing a computer that is based on quantum theory.

One of the key factors for the genesis of quantum computation was the belief that a quantum computer should be better equipped to simulate nature.  This point highlights the possibility that quantum computation should also be useful for simulations of relativistic physics, particularly quantum field theories.  Some progress has been made along these lines in \cite{JLP12,JLP14}, where two particular quantum field theories were shown to be efficiently simulable by quantum computers.

Our goal will be to use models from quantum computation as discretizations of relativistic systems.  The hope is that this may facilitate simulations of physics or even help us to understand relativistic physics better.  It is also compelling to speculate about whether in reality spacetime is continuous or discrete at some small scale.

There are two models in quantum computation that we will look at.  They are discrete-time quantum walks and quantum cellular automata.  Both of these have some properties in common.  One is that they live in discrete spacetime.  The second is that they are causal, which means that there is a maximum speed of propagation of information.  We will go into the finer details in chapter \ref{chap:Background1}, but now let us ask why causal systems in discrete spacetime might be a sensible option for discretizing relativistic physics.

Take causality first.  It seems like a very natural property for discrete models to have, especially if we wish to use them to approximate relativistic systems in the continuum, which are themselves causal.  Perhaps surprisingly, if we want to discretize space and retain causality, then we must also discretize time.  A consequence is that lattice quantum field theories with local Hamiltonians in continuous time are not truly causal.  Let us convince ourselves that local Hamiltonians in discrete space lead to unbounded propagation speeds.

Take a particle on a line that evolves according to the time independent Hamiltonian $H$.  Suppose that at $t=0$ the particle is at position $0$.  If this model were strictly causal, then there would be a $T$ such that for all $t<T$ there is zero probability of finding the particle outside of some finite region $R$ that contains the point $0$.

It follows from causality that, given a position $n\notin R$ and all $t<T$, we have $\bra{n}e^{-iHt}\ket{0}=0$.  Now, we can Taylor expand $e^{-iHt}$ to get
\begin{equation}
\begin{split}
 \bra{n}e^{-iHt}\ket{0} & =\bra{n}-iHt+O(t^2)\ket{0}=0\\
 \Rightarrow\ &\bra{n}iH\ket{0}-\bra{n}O(t)\ket{0}=0.
\end{split}
\end{equation}
By taking $t$ to be small, the second term can be made arbitrarily small.  Therefore, the first term must be zero.  And, by considering higher order terms, we see that $\bra{n}H^l\ket{0}=0$ for all $l$.  But this tells us that $\bra{n}e^{-iHt}\ket{0}=0$ for \emph{any} $t$.  This means that the particle cannot be found outside $R$ at any time.  If the particle is ever going to propagate to the point $n$, it happens instantaneously, though probably with a very small amplitude.

So causality on a lattice requires a completely discrete spacetime.  There are some more arguments for (but also against) discretizing spacetime.  Let us go through some of them, starting with arguments in support of discrete spacetime.  If simulating physics is our goal, an obvious answer is that whatever we do will require truncation of some degrees of freedom.  After all, our computers, whether classical or quantum, only have a finite memory.  So discretizing space and time is a natural course to take.\footnote{In some cases it may be useful to discretize in other ways.  For example, for quantum particles in a box, it makes more sense to introduce a momentum cutoff, which would also render the number of degrees of freedom finite.}

Another particularly good reason for discretizing spacetime is that it allows us to avoid the infinities that plague quantum field theory calculations in continuous spacetime.  This is because discretizing space introduces a momentum cutoff.  But there are hints from more exotic areas in physics too.  The Bekenstein bound, for example, says that the number of degrees of freedom in a finite region of space is finite \cite{Bekenstein81}.

Another interesting point is that discrete spacetime provides a natural scale:\ the lattice spacing.  This is encouraging because nature does not appear to have scale invariance as a symmetry.  There is also a natural maximum speed of propagation of information, which is one lattice spacing per timestep.  Incidentally, it is good to keep in mind that, while evolution in continuous time is determined by the Schr\"{o}dinger equation, in discrete time, there is no Schr\"{o}dinger equation.  Instead, in discrete time, it is more natural to work directly with the unitary operator that updates the state every timestep.

A potential use for causal models in discrete spacetime parallels that of lattice quantum field theory.  Lattice quantum field theory is often used to make mathematical sense of the corresponding continuum quantum field theory.  So maybe we could do the same with causal quantum models in discrete spacetime.  This would provide a different viewpoint, so it is conceivable that such models could further our understanding of quantum field theory.

A final point favouring discrete spacetime is that the discrete spacetime models we will analyze are often particularly simple.  For example, when written as a quantum circuit, the evolution in time of discretized Dirac fermions on a line is remarkably simple.  This could be useful for quantum simulations of nature, but it is also a useful picture to have in mind for intuition about these systems.  Similarly, the discrete-time quantum walk that becomes a particle obeying the Dirac equation on a line, which we will see in chapter \ref{chap:Quantum Walks and Relativistic Particles}, provides an intuitive picture for how the mass in the Dirac equation effectively slows the particle down.

Of course, there are good arguments against discrete spacetime too.  The first, and a very convincing one, is the issue of spacetime symmetries.  Discrete spacetime models do not have the continuous symmetries associated with continuous spacetime.  The situation is actually worse for the models we will look at.  In general, they do not even have the rotational symmetry of the lattice.\footnote{It is actually possible to construct models that do have the symmetries of the lattice, but these are far from natural.}

Another point is the naturalness of the models.  Because discrete spacetime models, at least the type we will look at, live on cubic lattices, they do not seem as natural as continuous spacetime models.  To be fair, we could argue that this and the previous point are purely aesthetic.  After all, if the model gets the job done, meaning it accurately describes nature, and it has some degree of simplicity to it, then that should be good enough.  Still, it is hard to shake the belief that theories of nature should also have some beauty to them.

It almost goes without saying that the biggest reason to favour continuous spacetime is how well physical theories in the continuum have worked so far.

There is one final item that must be addressed when discretizing spacetime, though so far it is not clear how it affects the models we will look at here.  It is the issue of fermion doubling \cite{Tong12}.  This is a problem that occurs in lattice quantum field theory:\ there is a famous theorem that says that local Hamiltonians on the lattice, if they describe massless fermions, have extra low energy modes \cite{NN81,DGDT06}.  These modes are referred to as doublers.  Actually, there is some cause to be optimistic that this fate may not befall the models we will study here.  For one thing, these models do not evolve via local Hamiltonians, rather they evolve in a causal way via a unitary operator that acts over each timestep.  And so these models do not satisfy the hypotheses of the fermion doubling theorem.

If these discrete models are successful, then it is natural to ask how fundamental they are.  This question can be asked about any approximation of a successful physical theory.  If the approximation can be made good enough, then we cannot tell whether it is merely an approximation or how nature operates.  So it is hard to resist speculating about the ontological status of discrete spacetime models when they are good approximations of physical models in continuous spacetime.  Furthermore, we should keep in mind that there are other ideas proposing that space and time are fundamentally discrete.  One example is causal set theory \cite{BLMS87}.

Nevertheless, it is reassuring that, regardless of whether spacetime is discrete or not, the models we will study ought to be useful for performing quantum simulations of physics.  In fact, we will take this as our primary motivation.

The remainder of the first part of this thesis is broken up into four chapters.  Chapter \ref{chap:Background1} provides some background material, including the basics of discrete-time quantum walks and quantum cellular automata.  We will also look at the idea behind simulating quantum dynamics on a quantum computer.  Then we will briefly discuss the relevant physical models:\ relativistic particles and their equations of motion, as well as fermion fields.  At the end, we will look at the fermion doubling problem, which affects local fermion Hamiltonians on lattices.

Chapter \ref{chap:Quantum Walks and Relativistic Particles} considers discrete-time quantum walks as discretized relativistic particles.  We start by looking at how to take continuum limits.  This allows us to see that a large class of these discrete-time quantum walks become relativistic particles in the continuum limit.  This requires the construction of a general scheme for taking the continuum limits of discrete-time quantum walks.  These are results from \cite{FS14}.

Another new result involves quantum walks that converge to particles obeying relativistic equations of motion:\ we modify the standard examples, so that they converge significantly faster.  We also look at the issue of fermion doubling.  Preliminary results indicate that the problem can be overcome for discrete-time quantum walks in two dimensional space, and can be reduced significantly in three dimensional space.

This chapter ends with a look at the abstract theory of discrete-time quantum walks.  We ask whether the evolution of discrete-time quantum walks can be decomposed into products of coin operations and shifts, a result known to hold in one dimension \cite{Vogts09}.

Chapter \ref{chap:Quantum Cellular Automata and Field Theory} contains a study of discretized fields and quantum cellular automata, as well as more general causal models in discrete spacetime.  First, we construct causal fermion fields in discrete spacetime and take their continuum limits to recover relativistic fermion fields in continuous spacetime.  A complicating factor is that the vacuum state of fermion fields in continuous spacetime is not a simple state.  This means that the discrete system must have a nontrivial vacuum state that converges in some sense to the continuum vacuum state.

Our next task involves a digression to Hamiltonian models in continuous time, which will allow us to introduce some useful tools.  We look at a remarkable technique from \cite{Ball05,VC05} that allows us to represent local fermion Hamiltonians by local qubit Hamiltonians.  That this is not straightforward follows because fermion operators anticommute regardless of the separation between them, which implies that there is some form of nonlocality.\footnote{Of course, this nonlocality is not directly observable.}

A question we ask in this chapter is whether nature could be described by a quantum cellular automaton or similar causal system in discrete spacetime.  There are good reasons to think that the answer is yes, though with the proviso that gravity is not considered.  The reasoning goes via approximating lattice quantum field theory by causal discrete spacetime models.  While encouraging, it also highlights the fact that these causal discrete spacetime models are more general in a sense than local Hamiltonian models in continuous time.

With that in mind, we look at the properties of general causal systems of fermions and bosons in discrete spacetime.  We see that, by adding ancillary particles, the dynamics can be viewed as a constant-depth circuit, an analogue of the results of \cite{ANW11,GNVW12}, which held for quantum cellular automata.  A consequence of this is that quantum cellular automata and their analogues with fermionic modes at each site are equivalent.  Most of these results are from \cite{FS13}.

We close this chapter with something more abstract:\ a decomposition of the dynamics of quantum cellular automata on a line into a product of on-site unitaries and shift operations.

Chapter \ref{chap:Conclusions and Open Problems 1} brings part one to an end with a discussion of some open problems, as well as new questions prompted by the preceding chapters.

\chapter{Background}
\label{chap:Background1}
\section{Quantum Computing}
Over the last thirty years, quantum computing has grown to become a large field of research.  One of the driving forces behind this is the idea that quantum computers should be able to simulate physical systems much better than classical computers \cite{Feynman82}.  This probably sounds like a tautology.  After all, nature is quantum, not classical.  Proving it, however, is not at all trivial, but some progress has been made.  One major contribution was \cite{Lloyd96}, where it was shown that the dynamics of quantum systems evolving via local Hamiltonians can be efficiently (we will be precise about what efficient means in the next section) simulated by quantum computers.  It is widely believed that this is not true for classical computers \cite{Feynman82,Lloyd96}.

It is not possible to give a full introduction to quantum computing here.  Luckily, for our purposes we only need a few tools from the subject as well as a basic understanding of the theory of simulating physics with quantum computers.  For a good introduction to the subject, including details of the interesting algorithms that ignited interest in the field, the classic reference is \cite{NC00}.  Let us give a brief introduction to the main ideas behind quantum computation and how it differs from classical computation.

In classical computation information is stored as bits, which can be in one of two states, labelled $0$ or $1$.  In contrast with this, in quantum computing, information is represented by qubits, which are quantum systems with a two dimensional Hilbert space, with orthonormal basis usually denoted by $\ket{0}$ and $\ket{1}$.  One major difference with classical computing is that a qubit can have many more states than in the classical case:\ the entirety of the two dimensional Hilbert space spanned by the states $\ket{0}$ and $\ket{1}$ is available.\footnote{To be more precise, the qubit could also be in a mixed state}  Another difference between quantum and classical computation is that the former has a larger class of correlations between different systems, arising from entanglement.  These correlations lead to entirely non-classical phenomena that can also be taken advantage of for quantum information processing and computation.

In both the classical and quantum case, computers process information by applying logical operations to the bits or qubits.  These logical operations are the building blocks of the algorithms that are actually implemented on a computer.  This is the circuit model of computation.  So, if we have some problem we want to solve on our computer, like factoring an integer, we must come up with an algorithm to solve it and break that algorithm up into a sequence of logical operations.  

After implementing the algorithm, we read out the result.  In the quantum case, generally such a measurement will disturb the state of the qubits.  Still, quantum computing seems to offer more than classical computing since it includes classical computing as a subset,\footnote{To see this we simply restrict each qubit to only be in the states $\ket{0}$ and $\ket{1}$ and only allow logical operations that keep each qubit in one of these states.} and there are quantum algorithms for some problems that offer significant speedup over the best known classical algorithms \cite{NC00}.  (Factoring an integer and simulation of quantum systems are two examples.)

There are different frameworks for quantum computation, analogous to the different frameworks for classical computation.  The most accessible is the circuit model, but two others will be of more relevance to us here:\ quantum walks and quantum cellular automata.  Both of these are universal for quantum computation \cite{LCETK10,Watrous95}, which means they can efficiently simulate other models of computation.  Our interest, however, will be in their use as discretized models of relativistic systems, with one of the goals being simulation of physics.  This endeavour is well justified because one of the key properties of these models is causality, which is the property that information propagates a finite distance over a finite time.  Causality in this sense is, of course, also one of the key principles in relativity.  It is important to point out that what we mean by causality is closer in spirit to what is often referred to as relativistic causality, which is the requirement that causal influences can propagate no faster than the speed of light \cite{Butterfield07}.  This definition of causality in discrete systems was employed in \cite{GNVW12}, for example.

In section \ref{sec:Computational Complexity and Efficiency} we will discuss what efficient means in the context of computation.  Following this, in section \ref{sec:Simulating Physics on a Quantum Computer} we will see the basic idea behind simulating quantum evolution on a quantum computer.  In section \ref{sec:Discrete-Time Quantum Walks} we will encounter discrete-time quantum walks and their basic properties.  Then, before turning to physical models, section \ref{sec:Quantum Cellular Automata} introduces quantum cellular automata.

\subsection{Computational Complexity and Efficiency}
\label{sec:Computational Complexity and Efficiency}
The field of computational complexity is concerned with how difficult it is to solve computational problems, like the problem of simulating a particular physical model.  This is based on how the resources (time or memory) needed to solve a problem scale with the size of the input.  In the classical case, the size of the input is measured by the number of bits needed to express it.  For example, an integer that is written as $n$ digits in binary is an input of size $n$.  Analogously, in the quantum case, the input size is the number of qubits in the input.  Our principle concern will be the gate or time complexity of an algorithm.  The gate complexity is a measure of how the number of logical operations needed to perform the computation scales with the size of the input.  The time complexity of an algorithm measures how the number of time steps needed to implement an algorithm scales with the size of the input.  This may be smaller than the gate complexity if we can apply multiple gates at the same time.

There are various different complexity classes that capture how difficult a problem is to solve.  And there are different sets of complexity classes for quantum and classical computation.  The precise details of these will not be of importance to us here.  Of more relevance to us will be the question of whether there is an efficient quantum algorithm, meaning the problem can be solved efficiently by a quantum computer.  A problem can be solved efficiently if a polynomial time algorithm exists.  In other words, the number of logical operations for each input of size $n$ is bounded above by a polynomial in $n$.  If this is the case for a quantum computer, the problem is said to be in the complexity class BQP, which stands for bounded error quantum polynomial time.  This is the quantum version of the classical complexity class BPP, which stands for bounded error probabilistic polynomial time.  Both of these classes allow for some error in the computation, provided the probability of success is at least $2/3$ {\it for any size input} \cite{NC00}.

It is useful here to define big $O$ notation.  This notation will be helpful for discussing convergence rates of systems in discrete space to continuum systems as we shrink the spatial lattice spacing, amongst other things.  The idea is to quantify how a function behaves as its argument tends towards some limit, usually zero or infinity.  Suppose we are concerned with how a real valued function $f(x)$, where $x\in\mathbb{R}$, behaves as $x\rightarrow 0$.  We write
\begin{equation}
 f(x)=O(g(x)),
\end{equation}
if there is a $\delta>0$ such that for all $0<x<\delta$ we have $|f(x)|\leq cg(x)$, where $c\geq 0$.  Or, if we are concerned with how $f(x)$ behaves as $x\rightarrow \infty$, again we write
\begin{equation}
 f(x)=O(g(x)),
\end{equation}
but now if there is a $x_0>0$ such that for all $x>x_0$ we have $|f(x)|\leq cg(x)$, where $c\geq 0$.

As an example, suppose $f(x)=x+3x^2$.  If we are concerned with the behaviour of $f(x)$ as $x\rightarrow \infty$, then it is $O(x^2)$.  On the other hand, if we look at how it behaves as $x\rightarrow 0$, then it is $O(x)$.

\subsection{Simulating Physics on a Quantum Computer}
\label{sec:Simulating Physics on a Quantum Computer}
A concrete algorithm for simulating physics on a quantum computer was first presented in \cite{Lloyd96}, where it was shown that quantum computers can efficiently simulate a quantum system on a lattice evolving via a local Hamiltonian.\footnote{It is important to emphasize that the lattice spacing is fixed, so that this does not directly prove that quantum systems in continuous spacetime are efficiently simulable.}  It is important to reiterate that, while it is not known for certain, it is strongly believed that classical computers cannot efficiently simulate such systems.  The intuition behind this comes from the fact that the state space of a quantum system grows exponentially with the number of subsystems.  So the number of classical bits needed to describe the state of a quantum system generally grows very quickly with the system size.

The idea behind simulating evolution via a local Hamiltonian $H$ is to write it as a sum of $L$ local terms $H=\sum_lH_l$ that act on only a fixed number of sites.  For example, these local terms could describe interactions between neighbouring subsystems.
The next step is to approximate the evolution over a short time interval of length $\varepsilon$ by using the identity
\begin{equation}
 \prod_le^{-iH_l\varepsilon}=e^{-iH\varepsilon}+O(\varepsilon^2).
\end{equation}
So, to simulate evolution via $H$ over a fixed time $t$, we divide $t$ into $N$ intervals of length $\varepsilon$ and take the limit as $\varepsilon\rightarrow 0$.  Then by the Lie-Trotter product formula \cite{Trotter59,NC00}
\begin{equation}
\label{eq:TrottSim}
 (\prod_le^{-iH_l t/N})^{N}\rightarrow e^{-iHt},
\end{equation}
as $N=t/\varepsilon$ goes to infinity.  This follows because there are $N$ timesteps and the error in the approximation of the evolution over one timestep is $O(\f{1}{N^2})$, so the total error is $O(\f{1}{N})$.  We look at the Lie-Trotter formula in more detail in section \ref{sec:The Lie-Trotter Product Formula}.

The problem is now reduced to implementing the unitaries $e^{-iH_l t/N}$ on a quantum computer.  So these have to be decomposed in terms of the basic logical operations we can do.  These logical operations arise from the finite set of gates we can implement in practice.
It is crucial that each of the unitaries $e^{-iH_l t/N}$ acts on a fixed number of sites.  Because of this, the number of logical operations needed to implement each unitary to within error $\delta$ is bounded above by $O(\log^c(\f{1}{\delta}))$, where $c$ is a constant.  This follows from the Solovay-Kitaev theorem \cite{NC00}.  As a result, the total number of gates we need to apply to approximate evolution via $e^{-iHt}$ to within any fixed error can be shown to be polynomial in $N$ and the number of sites on the lattice.  See \cite{NC00} for more details.  Therefore, the algorithm is efficient.

Actually, the Hamiltonian need not be local for this algorithm to work, provided it contains interaction terms between at most $k$ sites for some fixed $k$.  Furthermore, by using higher order decompositions, one can improve the scaling further \cite{BACS07}.

There is a subtlety glossed over above.  We assumed that the system has a tensor product structure, meaning that the state space of the whole system is described by a tensor product of the Hilbert spaces of each individual site, which is the case for a spin lattice, for example.  As the simulation method above involves representing the state of the system and implementing the dynamics on qubits, it is conceivable that a local {\em fermion} Hamiltonian will not fit into this scheme.  This is because local fermion operators, when represented in terms of qubit operators, are not necessarily local.  We will go into this in more detail in section \ref{sec:The Jordan-Wigner Transformation}.  For now, it will suffice to note that there are several methods to simulate the dynamics of fermions with local Hamiltonians on a system of qubits \cite{AL97,OGKL01,SOGKL02,BK02,PBE10}.

It is also sensible to consider the option of constructing quantum simulators out of fermionic modes, which would avoid these problems entirely.  In fact, in \cite{BK02} fermionic quantum computers are compared to conventional quantum computers (those using qubits).  The main result of the paper is that each can efficiently simulate the other.  In particular, the authors show that the slowdown for simulating a conventional quantum computer by a fermionic one is constant, while in the opposite direction the slowdown is logarithmic.  This means that, when representing the fermionic system in the qubit picture, for every fermionic operation we wish to implement we must perform an additional $O(\log(n))$ operations, where $n$ is the number of fermionic modes.

It is probably a little dishonest not to mention the sobering fact that building a quantum computer capable of outperforming a classical one is, with current technology, a long way off.  Luckily, however, to show a speedup over classical computers in simulating physics a much smaller quantum computer would suffice than is needed for other algorithms \cite{Lloyd96}.  It is possible that, even if very large quantum computers useful for solving other problems (like factoring) never become a reality, simulating physics using smaller quantum computers may still be accomplished.

\subsubsection{An Aside:\ Analogue Quantum Simulators}
Analogue simulators offer another possibility for simulating physical systems.  While these are not directly relevant to the ideas we will discuss, it is interesting (and a little reassuring) to note that there are analogue quantum simulators of relativistic quantum physics that are currently feasible.  The idea is to set up a quantum system in a lab that evolves over time via some Hamiltonian that is similar to that of the physical system we want to study.  One of the main drawbacks of this is that analogue simulators are not really computers since there is no error correction.  But the upside is that they are currently closer to being realized in the lab.

There is a nice proposal for an experiment that would simulate a relativistic model of interacting electrons and positrons in \cite{CMP10}.  Although the model is fairly simple, it still exhibits interesting phenomena that occur in the standard model.  The setup for the experiment involves cold atoms trapped on an optical lattice. By tuning the laser beams, one can change the height of the potential barrier between sites.  This affects the likelihood that atoms will hop from one site to the next.  The state of the field is represented by these atoms such that, if there is an atom present at a site, there is an excitation in the field at that site.  The atoms then mimic electrons and positrons interacting via the Thirring model.  There are some other interesting proposals for analogue simulators in \cite{CLEGRGS11,CMLS12} using trapped ions.

\subsubsection{Another Aside:\ The Computational Power of Nature}
In fact, it is already known that some models in particle physics are efficiently simulable on a quantum computer.  In \cite{JLP12}, it is shown that $\phi^4$ theory is one such model.  Even though $\phi^4$ theory is conceptually simple (as quantum field theories go), showing that this model could be efficiently simulated was a complicated task.  Still, this suggests the possibility that, if nature is described by a quantum field theory, then it is something that we can simulate efficiently with a quantum computer.  It is not obvious that this should be possible.  In some ways it is not clear how quantum field theory fits into the scope of conventional quantum theory.  This is because quantum field theories involve uncountably infinite degrees of freedom, and in most interacting models this makes it unclear what the state space and Hamiltonian are:\ most interacting models are only defined in terms of perturbation theory or as the continuum limits of lattice theories \cite{Creutz83}.

Of course, whatever theory describes nature also tells us what kind of computer we can build in the first place.  If, for example, nature were classical and not quantum, then the best we could do would be to construct classical computers.  So, as the authors of \cite{JLP12} point out, this prompts the question of whether one could do more with a a quantum {\it field} computer than with a conventional quantum computer.  Showing that quantum field theories can be efficiently simulated by quantum computers would show that this is not the case.

\subsection{Discrete-Time Quantum Walks}
\label{sec:Discrete-Time Quantum Walks}
Quantum walks are often described as the quantum analogues of classical random walks, though this is not completely fair:\ random walks are inherently irreversible, whereas quantum walks evolve unitarily and so are reversible.\footnote{To be completely accurate, performing a measurement does introduce irreversibility, but this is always the case in quantum theory.}  Quantum walks come in two forms; they either evolve unitarily over discrete timesteps or via a Hamiltonian in continuous time.  See \cite{Kempe03} for a good review.  Our focus will be on discrete-time quantum walks, which from here on we will simply call quantum walks.  

Much like a classical random walk, we can think of a quantum walk as a particle that lives on a discrete set of points.  For the classical random walk, the standard example is that of a particle that lives on a discrete line.  Then every timestep we toss a coin.  If we get heads, the particle moves one step to the right; if we get tails, it moves one step to the left.  Similarly, the standard example of a quantum walk is a particle that lives on a discrete line and has a two dimensional extra degree of freedom with Hilbert space $\mathcal{H}_C$.  This plays a role analogous to the coin in the classical case and, for this reason, is often referred to as the coin.  Let $\ket{r}$ and $\ket{l}$ be an orthonormal basis for $\mathcal{H}_C$.  Then the Hilbert space of the particle is spanned by the states $\ket{l}\ket{n}$ and $\ket{r}\ket{n}$ with $n\in \mathbb{Z}$ labeling position.  One possible evolution operator is
\begin{equation}
U=S\ket{r}\bra{r}+S^{\dagger}\ket{l}\bra{l},
\end{equation}
where $S$ is the (unitary) shift operator defined by $S\ket{n}=\ket{n+1}$ for all $n$.  This evolution is not terribly exciting, however.  Over each timestep $U$ just shifts all $\ket{l}$ states to the left and all $\ket{r}$ states to the right.  Instead, we can make this more interesting by modifying the evolution operator:
\begin{equation}
 U=W\left(S\ket{r}\bra{r}+S^{\dagger}\ket{l}\bra{l}\right),
\end{equation}
where $W$ is a unitary operator on $\mathcal{H}_C$.  A typical choice is the Hadamard coin $W=\f{1}{\sqrt{2}}\left(\begin{smallmatrix} 1 & 1\\ 1 & -1 \end{smallmatrix}\right)$.  In a way, applying $W$ is like tossing the coin in the classical random walk, but, of course, $W$ is unitary and hence reversible.

We can get some understanding of how quantum walks evolve over time by looking at the evolution over a few timesteps.  Starting with the state $\ket{r}\ket{0}$, the evolution operator $U$ first shifts the position of the particle from $\ket{0}$ to $\ket{1}$.  Then $W$ maps $\ket{r}$ to a superposition of $\ket{r}$ and $\ket{l}$.  Because the state now has overlap with both $\ket{l}$ and $\ket{r}$, over the following timestep the particle spreads out.  This changing of direction caused by the coin leads to an effective slowing down of the particle.  This is a simple discrete analogue of how mass mixes chiralities in the Dirac equation, which we will discuss in section \ref{sec:Dirac and Weyl particles}.

\begin{compactwrapfigure}{r}{0.53\textwidth}
\centering
\begin{minipage}[r]{0.51\columnwidth}%
\centering
    \resizebox{6.0cm}{!}{\includegraphics{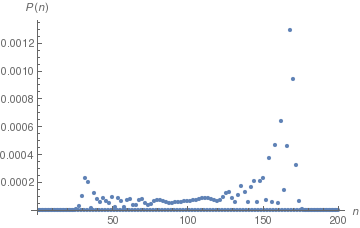}}
    \footnotesize{\caption[Position distribution of a quantum walk with the Hadamard coin]{Position probability distribution of a quantum walk with the Hadamard coin after $100$ timesteps, with $200$ sites and an initial state localized at site $n=100$.  The x-axis corresponds to the position of the particle $n$.  The y-axis shows the probability $P(n)$ of finding the particle at point $n$.\ \\ \ }}
    \label{fig:77}
\end{minipage}
\end{compactwrapfigure}

More generally, quantum walks are defined on a (possibly directed) graph, but the quantum walks we will focus on will always live on cubic lattices.  To this end, we label spatial coordinates by $d$ dimensional vectors $\vec{n}$, where each of the $d$ components of $\vec{n}$ takes integer values.  Then the orthonormal basis $\ket{\vec{n}}$ of the Hilbert space $\mathcal{H}_P$ describes the particle's position.  The particle also has an extra degree of freedom with Hilbert space $\mathcal{H}_C$.  The state space of the quantum walk is $\mathcal{H}_C\otimes\mathcal{H}_P$.  In chapter \ref{chap:Quantum Walks and Relativistic Particles}, we will see that the coin often corresponds to spin or chirality when we take the continuum limit.

We will assume that the dynamics are translationally invariant.  In fact, with this requirement, an extra degree of freedom is necessary for these particles to have nontrivial evolution, where trivial evolution means $U$ is just proportional to a shift operator \cite{Meyer96a}.  We also assume time translation invariance, so the evolution operator $U$ is the same for every timestep.  As mentioned earlier, another defining feature of quantum walks is causality.  Informally, this means that the evolution has the property that, if the particle is initially localized in a finite region of the lattice, then after applying $U$ the particle is still localized in some finite region.  Formally, causality implies that there is an $L\geq 0$ such that
\begin{equation}
 \bra{\vec{m}}U\ket{\vec{n}}=0\ \textrm{if}\ |\vec{n}-\vec{m}|>L.
\end{equation}
We end this section with a precise definition of a quantum walk.
\begin{definition}
 A quantum walk consists of a quantum particle with the properties below.
 \begin{enumerate}
  \item The particle lives on a $d$ dimensional lattice with state space $\mathcal{H}_P$ that has orthonormal basis states $\ket{\vec{n}}$, where $\vec{n}\in\mathbb{Z}^d$.
  \item It has an extra finite dimensional degree of freedom with Hilbert space $\mathcal{H}_C$, so that the total state space is $\mathcal{H}_C\otimes\mathcal{H}_P$.
  \item It evolves over discrete timesteps via a causal unitary operator that is translationally invariant in both space and time.
 \end{enumerate}
\end{definition}

\subsection{Quantum Cellular Automata}
\label{sec:Quantum Cellular Automata}
Quantum cellular automata (QCAs) are the quantum analogues of classical cellular automata \cite{SW04,Wiesner09}.  QCAs consist of a lattice with finite dimensional quantum systems at each site (for example, qubits) evolving in a causal way over discrete timesteps.  Again, causal means that there is a bound on the speed of propagation of information.  The precise definition of causal in this context is as follows.  There is an $L\geq 0$ such that, for any $\vec{n}$ and any operator $A$ that is localized\footnote{We say an operator is localized on a region if it acts like the identity on all systems outside that region.  This definition allows the possibility that the operator acts like the identity on some systems inside the region too.} on site $\vec{n}$, after one timestep the evolved operator acts like the identity on all sites $\vec{m}$ with $|\vec{n}-\vec{m}|>L$.  Something to note here is that it is convenient to work in the Heisenberg picture when dealing with QCAs.

There are some subtleties that arise when defining QCAs on infinite lattices, which are discussed in \cite{SW04,ANW11,GNVW12}, but the concepts and results here all make sense if we take the lattice to be finite.  In that case the state space is the tensor product of the individual systems' state spaces, and the evolution is a unitary operator acting on this Hilbert space every timestep.

We will assume that QCAs are translationally invariant.  Then the systems at each site are identical, and the unitary implementing the dynamics commutes with translations.  This is a good point to give the formal definition of a QCA.
\begin{definition}
 A quantum cellular automaton consists of a discrete lattice, which may have periodic boundary conditions or be $\mathbb{Z}^d$, with the properties below.
 \begin{enumerate}
  \item Each lattice site has an associated $d$ dimensional quantum system.
  \item Evolution takes place over discrete timesteps via a causal unitary that is translationally invariant in space and time.  (For infinite lattices, evolution is via an automorphism of the algebra of observables.)
  \end{enumerate}
\end{definition}

To get some feeling for QCAs, it will help to look at an example.  Suppose we have a line of qubits labelled by an integer $n$.  One way of constructing an evolution that is causal is to apply partitioned layers of unitaries.  Let $V_{n,n+1}$ be a unitary acting nontrivially only on qubits $n$ and $n+1$, then we could take
\begin{equation}
 U=\prod_{n\ \textrm{even}}V_{n,n+1}\prod_{m\ \textrm{odd}}V_{m,m+1}.
\end{equation}
Notice that the ordering within both products does not matter because $V_{n,n+1}$ and $V_{m,m+1}$ commute if they act on pairs of sites that do not overlap.  To see that this evolution is causal, suppose $A$ is an operator that is localized on the qubit at site $k$, where $k$ is even.  Then every unitary $V_{n,n+1}$ with $k\notin\{n,n+1\}$ commutes with $A$.  It follows that $U^{\dagger}AU$ is localized on the qubits at sites $k+1$, $k$, $k-1$ and $k-2$.
\begin{figure}[!ht]
{\centering
\resizebox{9cm}{!}{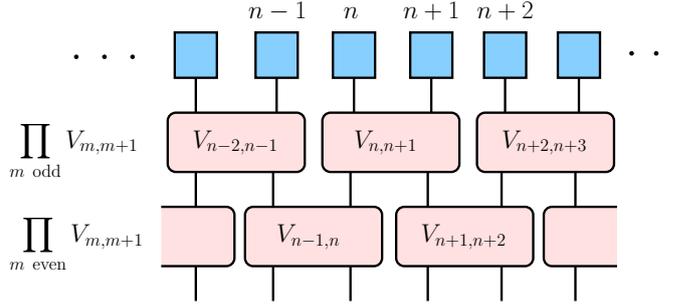} \caption[Example evolution of a quantum cellular automaton]{An example of QCA evolution over one timestep.  Blue boxes represent qubits, and pink boxes represent the unitaries being applied to pairs of qubits. \label{fig:QCA}}
}
\end{figure}

As an example, we could take $V_{n,n+1}=\textrm{CNOT}_{n,n+1}$, which is the controlled-NOT unitary, where $n$ is the control qubit and $n+1$ is the target qubit:
\begin{equation}
 \textrm{CNOT}_{n,n+1}=\ket{0}_n\bra{0}\otimes\openone_{n+1}+\ket{1}_n\bra{1}\otimes X_{n+1},
\end{equation}
where $X_{n+1}$ is the Pauli $X$ operator acting on site $n+1$.  This and the other Pauli operators are given by
\begin{equation}
\begin{split}
 X & =\ket{0}\bra{1} + \ket{1}\bra{0}\\
 Y & =i(\ket{1}\bra{0} - \ket{0}\bra{1})\\
 Z & =\proj{0}- \proj{1}.
 \end{split}
\end{equation}

Not every QCA can be implemented in this way as a constant depth circuit, meaning that there is a constant number of layers of unitaries.  A simple counterexample to this is the evolution that shifts systems one step to the right every timestep \cite{GNVW12}.

It is useful to note that by grouping sites into single sites, where the systems at each new site will have larger dimension, we can view any QCA as a nearest neighbour QCA.  So, if necessary, we can just assume all QCAs are nearest neighbour without loss of generality.

Our interest in QCAs stems from simulation of physics.  But even if our interest is merely discretizing physics, QCAs seem like natural candidates for discretized models of relativistic systems as they are causal by definition.  Additionally, QCAs are universal for quantum computing.  This means that they can efficiently simulate the circuit model and vice versa \cite{Watrous95}.  From a purely practical point of view, if a QCA can efficiently simulate a physical system, then, of course, so can any other model of quantum computation.  And it is not obvious which model will provide the most efficient or practical simulators of quantum physics, but there are indications that QCA simulators can be realised by simple operations:\ for example, free Dirac fermions on a line can be simulated by a QCA that is just a simple product of nearest neighbour and on-site unitaries \cite{BDT12,FS13}.  We will see details of this in chapter \ref{chap:Quantum Cellular Automata and Field Theory}.  Perhaps this simplicity will mean that QCAs will have some practical use as simulators of physics.  The fact that QCAs are causal could also mean they are a more intuitive framework for constructing discrete models of quantum field theories.  

\section{Physical Models}
The physical models we will be interested in are relativistic quantum systems.  Combining relativity and quantum theory is a tricky business.  For free particles there are relativistic wave equations that fit nicely into the framework of quantum theory:\ there is a separable Hilbert space\footnote{Separable means that there is a countable basis.} and a Hamiltonian that acts on it.  When interactions are introduced, however, things become more complicated, and it is not so straightforward to simply write down an equation of motion and solve it.  This is typically the case when interactions are introduced in physics, but there are additional consequences for relativistic systems.  These include antiparticles, vacuum entanglement, and the appearance of infinities in calculations.  Of course, the problem of infinities is solved by renormalization.  One method for performing calculations using renormalization is perturbation theory, which has found widespread success when the interaction is weak \cite{Peskin95}.  When the interaction is strong, recourse is often made to lattice quantum field theory.  This involves discretizing space (but not always time) and studying the properties of the resulting systems as the lattice spacing goes to zero.  This has allowed the calculation of the mass of the proton \cite{DFFHHKKKLLSV08}.  As far as numerical simulation goes, lattice quantum field theory is used for simulation on classical computers.  Note that the approach we take towards simulation of relativistic systems is different.  It is geared towards constructing quantum simulators that are inherently causal.

The notation and style that is most convenient for us differs from that normally encountered in quantum field theory books and papers.  Here we will always work with Hamiltonians as opposed to Lagrangians and path integrals.  As we will see in the next section, the Dirac Hamiltonian will be written in terms of $\alpha_i$ and $\beta$ matrices, instead of the more familiar gamma matrices.  A disadvantage of this is that Lorentz symmetry is not manifest.  We use relativistic units with $\hbar=c=1$.  Aside from the exceptions mentioned, the notation is chosen to coincide as closely as possible with standard conventions in field theory \cite{Peskin95}.

First, in section \ref{sec:Dirac and Weyl particles} we will discuss free relativistic particles, specifically spin-$1/2$ particles.  Then in sections \ref{sec:Fermions} and \ref{sec:The Jordan-Wigner Transformation} we will introduce fermions and the Jordan-Wigner transformation.  The latter is a method for representing fermionic systems on qubits.  Before proceeding to field theory, we will introduce second quantization\footnote{Second quantization is a confusing name for this procedure:\ quantization transforms a classical system into a quantum one, whereas second quantization is really just an elegant way to go from a theory of a single quantum particle to a theory with many such particles, obeying bosonic or fermionic statistics.} in section \ref{sec:Second Quantization}, which will allow us to go from a single particle theory to a multiparticle theory.  We will present the quantum field theory of free fermions in section \ref{sec:Field Theory}.  Finally, in section \ref{sec:Lattice Quantum Field Theory and Fermion Doubling}, we will look at the problem of fermion doubling that occurs when fermions are discretized on a lattice.

\subsection{Dirac and Weyl particles}
\label{sec:Dirac and Weyl particles}
Here we will look at free relativistic particles with spin-$1/2$.  It is interesting that the only known fundamental fermions that exist in nature have spin-$1/2$.  And it is these, together with bosonic fields, that are the basic constituents of nature.  Their free evolution is described by either the Dirac or the Weyl equation.  These are just Schr\"{o}dinger equations with the Dirac or the Weyl Hamiltonian, as we will see.

A particle obeying the Dirac equation is not only described by a spatial wavefunction, it also has extra degrees of freedom.  So the Hilbert space of the particle is a tensor product of the position Hilbert space and the Hilbert space corresponding to the extra degrees of freedom.  In other words, the wavefunction will have multiple components and is written as $\psi_{i}(\vec{x})$, where $\vec{x}\in\mathbb{R}^d$ denotes position, and $i$ labels the components of the wavefunction.  The number of extra degrees of freedom depends on the spatial dimension.  In one and two spatial dimensions, the number of extra degrees of freedom is two, while in three spatial dimensions the number is four.

Let us define the matrices $\alpha_i$ and $\beta$, with $i\in\{1,...,d\}$.  The form of the matrices $\alpha_i$ and $\beta$ also depends on the spatial dimension.  These matrices act on the extra degrees of freedom and obey the relations
\begin{equation}
\begin{split}
 \{\alpha_i,\alpha_j\} & =2\delta_{ij}\openone\\
  \{\alpha_i,\beta\} & =0\\
   \beta^2 & =\openone.
 \end{split}
\end{equation}
Then the Dirac Hamiltonian is given by
\begin{equation}
 h=\vec{\alpha}.\vec{P}+m\beta,
\end{equation}
where $\vec{P}$ is the momentum vector operator, meaning $P_i\equiv -i\partial_{x_i}$, and $m$ is the mass.  Note that this Hamiltonian is {\it linear} in the momentum operator. 

The Hamiltonian for the one dimensional case is particularly simple.  There is no spin in one dimensional space, and the extra degree of freedom is two dimensional and is called chirality.  We can choose a representation with $\alpha=\sigma_z$ and $\beta=\sigma_x$.  Then the Dirac Hamiltonian in one dimension is
\begin{equation}
 h=P\sigma_z+m\sigma_x\equiv\begin{pmatrix}
P & m\\
m & -P\\
\end{pmatrix}.
\end{equation}
Notice that, if the mass were zero, then the two components of the particle's wavefunction would move in opposite directions at the speed of light, which explains why the two components are labelled by $l$ and $r$ for left and right.  Roughly speaking, when $m$ is not zero, its effect is to mix the two components of the wavefunction, which effectively slows the particle down.  We will see this more concretely for discretized particles that behave like Dirac particles in the continuum limit in chapter \ref{chap:Quantum Walks and Relativistic Particles}.

In two dimensional space, there is also no spin, and the extra degree of freedom is two dimensional again.  One choice \cite{Thaller92} for $\alpha_i$ and $\beta$ is $\alpha_1=\sigma_x$, $\alpha_2=\sigma_y$ and $\beta=\sigma_z$.  Then the Hamiltonian is
\begin{equation}
 h=P_1\sigma_x+P_2\sigma_y+m\sigma_z.
\end{equation}

The Hamiltonian in three dimensional space is more complicated.  In three dimensions there is spin, and so there are four extra degrees of freedom, corresponding to two for chirality tensored with two for spin.  We can choose the following representation
\begin{equation}
\begin{split}
 \alpha_i & =\sigma_i\otimes\sigma_z\\
 \beta & = \openone\otimes \sigma_x,
 \end{split}
\end{equation}
where the first tensor factor corresponds to spin and the second corresponds to chirality.  Or, equivalently,
\begin{equation}
\alpha_i=
\begin{pmatrix}
\sigma_i & 0\\
0 & -\sigma_i\\
\end{pmatrix}\ \ 
\beta=
\begin{pmatrix}
0 & \openone\\
\openone & 0\\
\end{pmatrix}.
\end{equation}
Then the Dirac Hamiltonian in three dimensional space is \cite{Messiah99}
\begin{equation}
 h=\begin{pmatrix}
\vec{P}.\vec{\sigma} & m\openone\\
m\openone & -\vec{P}.\vec{\sigma}
\end{pmatrix}.
\end{equation}
As a general rule, we call massive spin-$1/2$ particles {\it Dirac} particles.

\subsubsection{Chirality}
Let us be more clear about what chirality is.  Notice that in one dimensional space, if we set $m=0$, then there are two disconnected components in the Dirac Hamiltonian,
\begin{equation}
 h=\begin{pmatrix}
P & 0\\
0 & -P\\
\end{pmatrix}.
\end{equation}
These two sectors correspond to different chiralities.  We can consider particles that evolve by either of the Hamiltonians,
\begin{equation}
 h=\pm P.
\end{equation}
We say a particle is right-handed if it evolves via $+P$ and left-handed if it evolves via $-P$.  We call such particles {\it Weyl} particles.

This separation of the Hamiltonian into two components does not occur in two dimensional space.  One way to see this is to verify that there is no operator other than $\openone$ that commutes with both $\sigma_x$ and $\sigma_y$.  This means that there exist no rank one projectors on the extra degree of freedom that commute with the massless Hamiltonian.  So there is no notion of chirality.  On the other hand, in one space dimension, $\f{1}{2}(\openone\pm\sigma_z)$ are projectors that commute with the massless Hamiltonian, $h=P\sigma_z$.

In three dimensional space, such projectors also exist.  In our representation, these are given by $\f{1}{2}(\openone\pm i\alpha_1\alpha_2\alpha_3)$.  The massless Dirac Hamiltonian, with $m=0$, in this case is
\begin{equation}
 h=\begin{pmatrix}
\vec{P}.\vec{\sigma} & 0\\
0 & -\vec{P}.\vec{\sigma}
\end{pmatrix}.
\end{equation}
One could consider a particle evolving via either of Hamiltonians
\begin{equation}
 h=\pm \vec{P}.\vec{\sigma}.
\end{equation}
Again, we call such particles {\it Weyl} particles.  The plus sign corresponds to a right-handed particle and the minus sign corresponds to a left-handed particle.

Chiral symmetry is a symmetry that occurs when $m=0$.  It is a symmetry under the application of a phase that depends on chirality.  For example, in the one dimensional case, this is the symmetry under the unitary $e^{-i\sigma_z \theta}$.

Finally, it is interesting to note that, according to the standard model, fermions are massless and only acquire mass via their interaction with the Higgs field.  Of course, this is not to say that the Dirac equation is not useful, just that it may not be as fundamental as the Weyl equation in some sense.

\subsection{Fermions}
\label{sec:Fermions}
Let us now look at systems of fermions.  Suppose that we have a lattice where each site can be occupied by fermions.  If there is more than one type of fermion (electrons with different spin, for example), then there can be at most one fermion of each type present at a site.  We denote the state with all modes empty, meaning there are no particles present, by $\ket{0}$.  Then we define creation and annihilation operators $a^{\dagger}_{\vec{n},i}$ and $a_{\vec{n},i}$.  Here $\vec{n}$ labels the position on the lattice and $i$ labels the fermion type at a site.  These operators satisfy the canonical anticommutation relations:
\begin{equation}
\label{eq:anticomm}
 \begin{split}
 \{a^{\dagger}_{\vec{n},i},a_{\vec{m},j}\} & =\delta_{ij}\delta_{\vec{n}\vec{m}}\\
 \{a_{\vec{n},i},a_{\vec{m},j}\} & =0,
\end{split}
\end{equation}
where $\delta_{\vec{n}\vec{m}}=1$ if $\vec{n}=\vec{m}$ and is zero otherwise.  Annihilation operators are so named because they annihilate a particle of that type when acting on a state.  If there is no particle of that type present, then these operators map the state to zero.  In particular, $a_{\vec{n},i}\ket{0}=0$.  Similarly, creation operators $a^{\dagger}_{\vec{n},i}$ create particles.  For example, the state $a^{\dagger}_{\vec{n},i}a^{\dagger}_{\vec{m},j}\ket{0}$ has a fermion at $\vec{n}$ and a fermion at $\vec{m}$.  Note that it follows from the anticommutation relations that creation operators square to zero, which enforces the exclusion principle, meaning there can be at most one fermion in each mode.  The state space is spanned by states having every different combination of products of creation operators $a^{\dagger}_{\vec{n},i}$ acting on $\ket{0}$.  

A requirement that we make of systems of fermions is that, as well as being self-adjoint, physical observables are linear combinations of products of {\it even} numbers of creation and annihilation operators.  In particular, physical Hamiltonians satisfy this constraint.  As an example, the Hubbard Hamiltonian is
\begin{equation}
H=-  \alpha\sum_{\langle\vec{n}\,\vec{m}\rangle} \sum_{i} (a^{\dagger}_{\vec{n},i}a^{\ }_{\vec{m}i}+a^{\dagger}_{\vec{m},i}a^{\ }_{\vec{n},i})+
 U\sum_{\vec{n}}(a^{\dagger}_{\vec{n}\uparrow}a^{\ }_{\vec{n}\uparrow})(a^{\dagger}_{\vec{n}\downarrow}a^{\ }_{\vec{n}\downarrow}),
\end{equation}
where $i\in\{\uparrow,\downarrow\}$ labels spin in this case, $\langle \vec{n}\,\vec{m}\rangle$ denotes nearest neighbour pairs, and $\alpha,U\geq 0$ are real valued parameters. The first term describes fermions of the same type hopping to nearest neighbour sites, and the second term is an on-site Coulomb repulsion.

\subsection{The Jordan-Wigner Transformation}
\label{sec:The Jordan-Wigner Transformation}
Dealing with fermions involves nonlocality in a sense.  This is because creation and annihilation operators {\it anticommute} regardless of the spatial separation.  This does not lead to any physical nonlocality because observables are always sums of even products of creation and annihilation operators, which means that observables on two separated regions of space do commute.  Still, to mathematically represent fermions by a system of qubits (or spins), we have to account for the nonlocality of the creation and annihilation operators.  To do this, we employ the Jordan-Wigner transformation \cite{JW28}.

As an example, suppose we have a line of $N$ fermionic modes with no internal degree of freedom.  Label positions by $n\in\{0,1,...,N-1\}$, and let us represent these $N$ fermionic modes by $N$ qubits.  It is natural to take the state $\ket{00...0}$ to represent the state with no fermions present.  We can start by representing the fermionic creation operator $a^{\dagger}_0$ on the qubits by\footnote{It is a little disconcerting that in this formula $a^{\dagger}_0$ is represented by $\sigma^-_0$ as opposed to $\sigma^+_0$.  This is because it is conventional to represent the qubit state $\ket{0}$ by the spin state $\ket{\uparrow}$.  And so it is $\sigma^-_0$ that takes the state $\ket{0}$ to the state $\ket{1}$.}
\begin{equation}
\label{eq:JWT0}
\sigma^-_0=\textstyle\frac{1}{2}(X_0-iY_0),
\end{equation}
where $X$ and $Y$ are Pauli operators.  The subscript $0$ implies that the operators in equation (\ref{eq:JWT0}) only act nontrivially on the qubit at site zero, meaning they act on all other spins like the identity.  For example, we could write $X_0=X\otimes I\otimes I...\otimes I$.

Because of the anticommutation relations, we cannot simply represent $a^{\dagger}_1$ by $\sigma^-_1$.  But we can satisfy the anticommutation relations if we choose
\begin{equation}
a^{\dagger}_n\equiv \sigma^-_n\prod_{m<n}Z_m. \label{eq:ordering} 
\end{equation}
The string of $Z$s allows us to preserve the anticommutation relations.  And it is because of these strings that fermionic creation operators are manifestly nonlocal in the qubit picture.
\begin{figure}[!ht]
{\centering
\resizebox{9cm}{!}{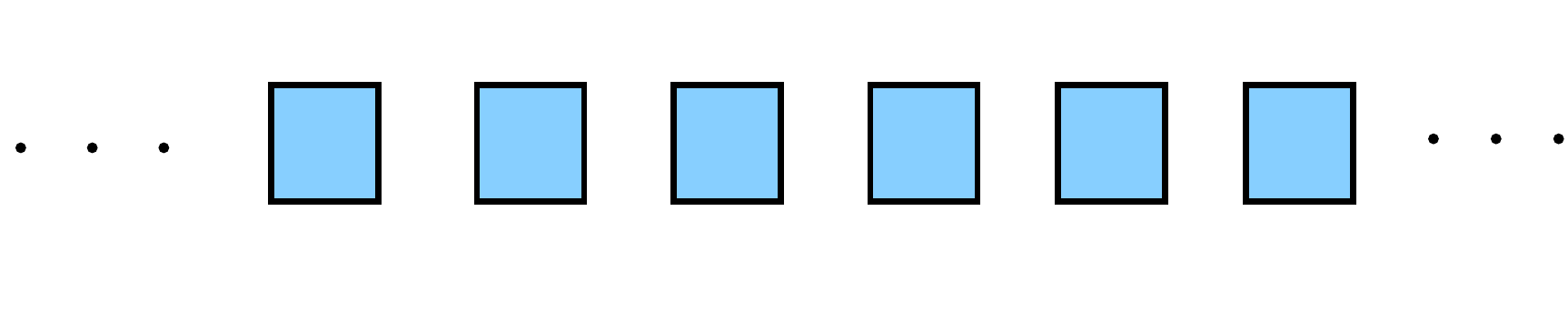} \caption[The Jordan-Wigner transformation]{The infamous strings of $Z$ operators.  With the ordering in equation (\ref{eq:ordering}), the representation of $a^{\dagger}_n$ in the qubit picture has $Z$ operators at all sites $m<n$ and acts like the identity on all sites with $m>n$. \label{fig:JWT}}
}
\end{figure}

Actually, there is a lot of freedom in how we choose to order the strings of $Z$s in the Jordan-Wigner transformation.  To this end, let us write down a more general Jordan-Wigner transformation, which is useful for higher dimensional lattices.  Now we can take $\vec{n}$ and $\vec{m}$ to be vectors.  Then, given $N$ modes, we assign a qubit to each mode.  And we associate a unique label to each site, $\pi(\vec{n})\in\{0,...,N-1\}$.  Then we define
\begin{equation}
\label{eq:general ordering}
a^{\dagger}_{\vec{n}}\equiv \sigma^-_{\vec{n}}\hspace{-1em}\prod_{\pi(\vec{m})<\pi(\vec{n})}\hspace{-1em}Z_{\vec{m}},
\end{equation}
which also satisfies the anticommutation relations.  A line of fermions with $\pi(n)=n$ is a special case, equivalent to the example we saw in equation (\ref{eq:ordering}).

This all generalizes to situations with multiple fermionic modes at each lattice site (for example, fermions with different spins).  Then in the qubit representation of this we have a separate qubit at that lattice site for each mode.  So we define an ordering $\pi(\vec{n},i)$, where $i$ labels different modes at site $\vec{n}$, to assign a unique number to each mode.  It is often convenient if we choose the ordering so that $\pi(\vec{n}, i)$ for the set of modes at each site are consecutive.  This is helpful because products of even numbers of creation and annihilation operators at the same site will be local in the qubit picture.

In the example we gave, the ordering in (\ref{eq:ordering}) meant that local even fermionic operators on a line are mapped to local operators on the line of qubits.  To see this, look at $a^{\dagger}_n a_{n-1}$.  The strings of $Z$s at all sites except $n$ and $n-1$ cancel.  Unfortunately, this does not extend to higher dimensions in general.  In fact, even for a line of fermions with periodic boundary conditions this is not generally true.\footnote{One way to see this is to look at the hopping term $a^{\dagger}_{N-1} a_{0}$, where points $0$ and $N-1$ are beside each other because of the periodic boundary conditions.  Using the ordering in (\ref{eq:ordering}) will mean that this is nonlocal in the qubit picture.}

\subsection{Second Quantization}
\label{sec:Second Quantization}
Second quantization \cite{BR97} allows us to take a single particle system and map it to a multiparticle system, where the particles have the same properties as in the single particle case, and they obey fermionic or bosonic statistics.  Here we will only be interested in fermions.

Suppose we have a single quantum particle that has a state space spanned by the orthonormal basis $\ket{\vec{n}}\ket{i}$, where $\vec{n}$ labels position and $i$ labels an extra degree of freedom.  Then via second quantization we construct the state space of the corresponding fermionic system by starting first with an empty state $\ket{0}$.  Then we map single particle states to creation operators in the following way:
\begin{equation}
 \ket{\vec{n}}\ket{i}\rightarrow a^{\dagger}_{\vec{n},i},
\end{equation}
where the creation operators obey the canonical anticommutation relations, which we saw in equation (\ref{eq:anticomm}).  We extend this by linearity to get
\begin{equation}
 \sum_{\vec{n},i} c_{\vec{n},i}\ket{\vec{n}}\ket{i}\rightarrow \sum_{\vec{n},i} c_{\vec{n},i}\,a^{\dagger}_{\vec{n},i}.
\end{equation}
The way to interpret this map is that it assigns a creation operator to each state of the single particle.  This defines the fermionic multiparticle version of the single particle theory.

This also gives rise to a map on operators.  Dropping the $i$ index to make the notation simpler, this is
\begin{equation}
 \sum_{\vec{n}, \vec{m}} A_{\vec{n},\vec{m}}\ket{\vec{n}}\bra{\vec{m}}\rightarrow \sum_{\vec{n}, \vec{m}} A_{\vec{n},\vec{m}}\,a^{\dagger}_{\vec{n}}a_{\vec{m}}.
\end{equation}
An important example is the Hamiltonian of the particle in question.  Suppose the single particle Hamiltonian commutes with the momentum operator, then it takes the form
\begin{equation}
 h=\sum_{p,i,j} h^{ij}_p \ket{p,i}\bra{p,j},
\end{equation}
where $\ket{p}=\f{1}{\sqrt{N}}\sum_n e^{ipn}\ket{n}$ are momentum states.  Then, via second quantization, this is mapped to
\begin{equation}
 H=\sum_p a_{p,i}^{\dagger}h^{ij}_p a_{p,j},
\end{equation}
where $a^{\dagger}_{p,i}=\f{1}{\sqrt{N}}\sum_n e^{ipn}a^{\dagger}_{n,i}$, creates a particle with momentum $p$ and extra degree of freedom state $i$.

More generally, the recipe for second quantization for fermions is to take a state $\ket{\psi}$ in the Hilbert space $\mathcal{H}$ to a creation operator according to the linear map
\begin{equation}
\label{eq:secquant}
 \ket{\psi}\rightarrow a^{\dagger}(\psi)
\end{equation}
such that the relations below are satisfied.
\begin{equation}
 \begin{split}
 \{a(\phi),a^{\dagger}(\psi)\} & =\langle \phi|\psi\rangle\\
 \{a(\psi),a(\phi)\} & =0.
\end{split}
\end{equation}
Above, $\langle \phi|\psi\rangle$ is the inner product on the Hilbert space $\mathcal{H}$, and the induced map $\ket{\phi}\rightarrow a(\phi)$ that follows by taking the adjoint after applying the map in equation (\ref{eq:secquant}) is antilinear.

\subsection{Field Theory}
\label{sec:Field Theory}
Looking back at the Dirac and Weyl particles of section \ref{sec:Dirac and Weyl particles}, it is straightforward, via second quantization, to upgrade these to multiparticle theories of fermions.  Each particle would evolve over time as in the single particle case.  But this is not the end of the story:\ even the quantum field theory of free fermions involves more than that.  This is because the physical vacuum is not the state annihilated by all the annihilation operators.  Instead it is something more complicated.

To illustrate this, we will go through how, starting with Dirac particles in one dimension, one arrives at the field theory of free Dirac fermions.

The process of second quantization takes us from the single particle to a multiparticle setting.  Here we started with a single particle that has continuous degrees of freedom.  The creation and annihilation operators corresponding to the state $\ket{x}$ also have a continuous label.  In an analogous sense to how $\ket{x}$ is not actually a state, the corresponding annihilation operator $\psi(x)$ is not really an operator.  Instead, it is an {\em operator valued distribution}.  Only by integrating these weighted by suitable test functions do we get operators.\footnote{We will be a bit loose with terminology:\ for example, we will refer to $\psi_{\alpha}^{\dagger}(\vec{x})\ket{0}$ as a state.}  We now have the annihilation operators $\psi_\alpha(x)$, where $\alpha\in\{l,r\}$ labels the extra degree of freedom.  These satisfy the anticommutation relations
\begin{equation}
\begin{split}
 \{\psi_{\alpha}(x),\psi_{\beta}(y)\} & =0\\
 \{\psi_{\alpha}(x),\psi^{\dagger}_{\beta}(y)\} & =\delta_{\alpha\beta}\delta(x-y).
\end{split}
\end{equation}

The state with no particles $\ket{0}$ is still annihilated by $\psi_\alpha(x)$.  As before, the state $\psi_{\alpha}^{\dagger}(x)\ket{0}$ corresponds to a state with a single particle at $x$.  Similarly, $\psi_{\alpha}^{\dagger}(x)\psi_{\beta}^{\dagger}(y)\ket{0}$ is a two particle state.

Recall that in one dimensional space the Dirac Hamiltonian for a single particle is
\begin{equation}
 h=P\sigma_z +m\sigma_x.
\end{equation}
Acting on momentum states, $h$ has eigenvalues $\pm E_p$, where $E_p=\sqrt{p^2+m^2}$.  The corresponding eigenvectors are $\vec{w}^{+}$ and $\vec{w}^-$ respectively.  Note that $\vec{w}^-$ corresponds to negative energy.

We can rewrite the fermion operators in terms of positive and negative energy particles by taking the Fourier transform.  We get
\begin{equation}
\begin{split}
\psi_\alpha (x) & = \int\!\frac{\textrm{d}p}{2\pi}\psi_{\alpha} (p)e^{ipx},\\
& = \int\!\frac{\textrm{d}p}{2\pi}\left(a_{p}w^{+}_\alpha (p)+c_{p}w^{-}_{\alpha} (p)\right)e^{ipx},
\end{split}
\end{equation}
where $a^{\dagger}_{p}$ and $c^{\dagger}_{p}$ create positive and negative energy fermions respectively with momentum $p$.

The presence of negative energy particles in the theory means that the Hamiltonian is unbounded from below.  This is because we can lower the energy by creating particles with negative energy.  But this is problematic from a statistical physics point of view:\ for one thing, there is no Gibbs state for Hamiltonians unbounded from below.

The way to get around this is to define the physical ground state to be the state with all the negative energy modes filled, which we will call $\ket{\Omega}$.  Then, because these are fermions, this means we cannot create any more negative energy particles.  So the new vacuum $\ket{\Omega}$ has no $a_p$ particles, but all the $c_p$ modes are filled up, meaning $c^{\dagger}_p$ annihilates $\ket{\Omega}$.  Viewed in this way in terms of a particle sea where $c_p$ creates holes, the state $\ket{\Omega}$ is sometimes called the Dirac sea.  Note that these holes have positive energy because annihilating a negative energy particle increases the energy.

The modern perspective is to view the hole itself as a type of particle (an antiparticle), created by applying $c_p$, which has 
momentum $-p$.  So we define $b^{\dagger}_{p}=c_{-p}$.  Then the field operator becomes
\begin{equation}
\psi_\alpha (x)= \int\!\frac{\textrm{d}p}{2\pi}\left(a_{p}u_\alpha (p)e^{ipx}+b^{\dagger}_{p}v_{\alpha} (p)e^{-ipx}\right),
\end{equation}
where we have made made the substitution $p\rightarrow -p$ in the second term, and defined $u_\alpha (p)=w^+_\alpha (p)$ and
$v_\alpha (p)=w^-_\alpha (-p)$.  This is the familiar form of the Dirac field operator in one dimension, though we have chosen 
to normalize $u_{\alpha}(p)$ and $v_{\alpha}(p)$ to one, as opposed to the usual normalization of $\sqrt{2E_p}$.

The situation in two and three spatial dimensions is similar, though the notation becomes more cluttered.  The three dimensional field operator is
\begin{equation}
\psi_\alpha (\vec{x})= \sum_{s=1,2}\int\!\frac{\textrm{d}^3p}{(2\pi)^3}\left(a^s_{\vec{p}}\, u^s_\alpha (\vec{p})e^{i\vec{p}.\vec{x}}+b^{s\dagger}_{\vec{p}}v^s_{\alpha} (\vec{p})e^{-i\vec{p}.\vec{x}}\right),
\end{equation}
where $\alpha$ now labels four degrees of freedom and $s$ labels the spin state.

\subsubsection{An Aside:\ Vacuum Entanglement}
\label{sec:Vacuum Entanglement}
An interesting property of the vacuum state $\ket{\Omega}$ is that it is entangled.  This is a general property in non interacting relativistic quantum field theory \cite{SW85}.  In fact, protocols have even been introduced to locally extract this entanglement, though it is unknown whether this leads to observations of Bell inequality violations \cite{SR07}.

It is inconvenient from the point of view of simulation that the vacuum state is not simple, as this state must be prepared.  The problem becomes even worse in interacting models where the vacuum state becomes even more complicated.  One method to prepare the interacting vacuum, employed in \cite{JLP12}, was to turn on the interaction adiabatically after preparing the free vacuum.

\subsection{Lattice Fermions and the Fermion Doubling Problem}
\label{sec:Lattice Quantum Field Theory and Fermion Doubling}
Essentially, a lattice QFT model suffers from fermion doubling if there are high momentum modes of fermions that have low energy.

To avoid clutter in notation, we will omit the $i$ indices so that $\psi_n$ is a column vector of operators, meaning $\psi_{n}\equiv \psi_{ni}$.  So, for example,
\begin{equation}
 \psi^{\dagger}_n\sigma_x\psi_n=\sum_{ij} \psi^{\dagger}_{n,i}\sigma_x^{ij}\psi_{n,j},
\end{equation}
where $\sigma_x^{ij}$ are the matrix elements of $\sigma_x$.

To understand the fermion doubling problem, let us look at an example.  The simplest lattice fermions that correspond to Dirac fermions in the continuum limit are called naive fermions \cite{DGDT06}.  In one dimension the naive fermion Hamiltonian is
\begin{equation}
\label{eq:naiveHam}
 H=\frac{-i}{2a}\sum_{n=0}^{N-1} \psi^{\dagger}_n\sigma_z(\psi_{n+1}-\psi_{n-1})+m\psi^{\dagger}_n\sigma_x\psi_n,
\end{equation}
where $a$ is the lattice spacing, and $m$ is the mass.  This Hamiltonian gives rise to Dirac fermions in the continuum limit, which have the Hamiltonian
\begin{equation}
 H=\int\! \textrm{d}x\, \psi^{\dagger}(x)\big({-i}\partial_x\sigma_z+m\sigma_x\big)\psi(x).
\end{equation}
So the recipe in going from the continuum Hamiltonian above to the discrete Hamiltonian in equation (\ref{eq:naiveHam}) was to replace the spatial derivative by a discretized derivative, $(\psi_{n+1}-\psi_{n-1})/2a$, where $a$ is the lattice spacing.

These naive fermions suffer from the fermion doubling problem \cite{NN81,Kaplan09}.  To see what this means, we switch to momentum space, with $\psi_n=\f{1}{\sqrt{N}}\sum_ne^{ikn}\psi_k$, where $k=2\pi l/N$ and $l\in\{0,...,N-1\}$.  Then
\begin{equation}
 H=\sum_{k} \psi^{\dagger}_k\left(\frac{\sin(k)}{a}\sigma_z+m\sigma_x\right)\psi_k.
\end{equation}
For small $k$, $\sin(k)\simeq k$, and $H$ takes the same form as the continuum Hamiltonian.  Note that the continuum momenta correspond to $p=k/a=\f{2\pi l}{L}$, where $L=Na$ is the length of the line, and $l\in\mathbb{Z}$.  Also, the energies are given by the dispersion relation
\begin{equation}
\label{eq;11276}
E_k=\pm\sqrt{\sin^2(k)/a^2+m^2}.
\end{equation}
\begin{compactwrapfigure}{r}{0.47\textwidth}
\centering
\begin{minipage}[r]{0.43\columnwidth}%
\centering
    \resizebox{6.0cm}{!}{\includegraphics{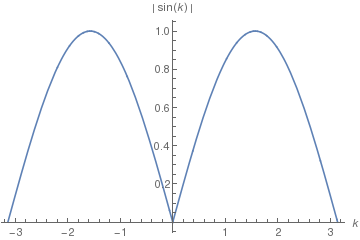}}
    \footnotesize{\caption[Fermion doubling for naive fermions]{Dispersion relation for the positive energy part of the spectrum with $m=0$ and the lattice spacing $a=1$. Here $E_k=|\sin(k)|$.  Note the low energy values for $k$ close to $\pm\pi$.}}
    \label{fig:Doubling}
\end{minipage}
\end{compactwrapfigure}
Then for small $k$, as expected, we have $\sin^2(k)/a^2\simeq p^2$, and we recover the continuum dispersion relation.  But these are not the only low energy modes.  If $k=\pi+k^{\prime}$, then for small $k^{\prime}$ we also have low energies because $\sin(\pi+k^{\prime})=-\sin(k^{\prime})\simeq -k^{\prime}$.  Therefore, there are two sets of low energy modes.  Modes with momentum close to zero and modes with momentum close to $\pi$ have low energy.

For free fermions, none of this is a problem because, if the initial state only has fermions with low momentum, the final state will only have fermions with low momentum.  It is when we include interactions that problems occur.  The interaction term may cause both sets of fermions to interact so that the theory looks like there are two types of fermions as opposed to one.

There is a general theorem stating that local fermion Hamiltonians on a lattice with chiral symmetry suffer from the fermion doubling problem \cite{NN81,DGDT06,Kaplan09}.  Recall that in the example above this is the symmetry under $\psi_n\rightarrow e^{i\sigma_z\theta}\psi_n$.  This symmetry is only present if $m=0$.\footnote{Still, doubling is seen in the dispersion relation in equation (\ref{eq;11276}) even though $m\neq 0$.}

There are several ways to alleviate the fermion doubling problem, each of which violates one of the hypotheses of the fermion doubling theorem.  Some of these involve breaking chiral symmetry.  Wilson fermions, for example, do this by adding a momentum dependent mass term, which is large for high momenta and small for low momenta.  Another example, staggered fermions, has fewer doublers than naive fermions \cite{Susskind77}.  Other approaches include modifying the definition of chiral symmetry at the discrete scale in a way that becomes equivalent to the usual definition in the continuum limit \cite{DGDT06}.

\chapter{Quantum Walks and Relativistic Particles}
\label{chap:Quantum Walks and Relativistic Particles}
\section{Introduction}
Quantum walks certainly seem like a natural model for particles in discrete spacetime, and for most of this chapter our aim will be to study their continuum limits.  So it is compelling that one can construct simple quantum walks that become relativistic particles in the continuum limit.  Of course, on a lattice there is no continuous spacetime symmetry.  And on top of this, the quantum walks we look at do not even have the symmetries of the lattice, yet their continuum limits still exhibit Lorentz symmetry.  At the end of the chapter, we will briefly discuss the more abstract side of quantum walks.

Previously, the connection between quantum walks and relativistic particles was studied in \cite{FH65,JS84,Bial94,Meyer96,Meyer97,Strauch06,BES07,Strauch07,Kurz08}.  One paper that will be of particular significance for us is \cite{Bial94}, where quantum walks\footnote{In this paper the systems studied are referred to as unitary cellular automata, but according to today's nomenclature these are indeed quantum walks.} that obey the three dimensional Weyl or Dirac equations in the continuum limit were presented.  Additionally, with some assumptions about how the evolution operator transforms under rotations, it was shown in \cite{Bial94,DP13} that quantum walks with a body centred cubic neighbourhood and two dimensional extra degrees of freedom must obey the Weyl equation in the continuum limit.  The goal in \cite{DP13} was to derive such quantum walks from a reasonable set of physical principles.

One of our concerns will be taking the continuum limit in a concrete way so that there is a well-defined notion of convergence of a quantum walk to a continuum model.  In \cite{AFN13} a rigorous treatment was given of the convergence of the quantum walks of \cite{Bial94} to their continuum limits.

It is good to keep in mind any practical or conceptual drawbacks of quantum walks as discretized relativistic particles.  One we will touch on is that of symmetries:\ if one studies quantum walks with the motivation that nature might actually be composed of such particles, then it is a little disheartening that the quantum walks we analyze here do not even have the symmetries of the lattice.  Another potential problem is that of fermion doubling, which arises for some of the quantum walks we study but can be mitigated and even removed completely in some cases.  Actually, it may be possible to completely circumvent the problem of fermion doubling with quantum walks, so if true, this would be a great boost.

The more abstract theory of quantum walks is not our primary focus, but we will look into it a little.  The idea is to see what can be said about quantum walks starting from as few assumptions as possible, as done in \cite{Meyer96a,Vogts09,GNVW12}.  Any general structure results for the evolution operator could be very useful.

The breakdown of this chapter is as follows.  First, in section \ref{sec:Continuum Limits}, we discuss how to make sense of continuum limits.  There we also look at quantum walks that converge more quickly than the standard examples to physical particles in the continuum, as well as symmetries of quantum walks.  In section \ref{sec:Two Dimensional Coins are Special}, we see that quantum walks that satisfy a masslessness condition with a two dimensional extra degree of freedom have Lorentz symmetry in the continuum limit.  Next, in section \ref{sec:Fermion Doubling in Quantum Walks}, we look at whether fermion doubling occurs in these quantum walks and what we can do about it.  We conclude the chapter by discussing the more abstract nature of quantum walks in section \ref{sec:Quantum Walks on a Line: Abstract Theory}.  Here we look at the decomposition for one dimensional quantum walk unitaries in terms of coin and shift operators and ask whether such a decomposition exists in higher spatial dimensions.

\section{Continuum Limits}
\label{sec:Continuum Limits}
Before going into the rigorous details involved in taking a continuum limit, let us first look at how the continuum limit works for a simple example.  A good place to start is with the quantum walk on a line that we saw in section \ref{sec:Discrete-Time Quantum Walks}.  Its limit was studied previously in \cite{FH65,Meyer96,Meyer97,Strauch06}.

Recall that the quantum walk can be thought of as a single particle that lives on a discrete line of points labelled by an integer $n$, so position states are $\ket{n}$.  And the particle has a two dimensional extra degree of freedom with orthonormal basis states $\ket{l}$ and $\ket{r}$.  The evolution operator acting over every timestep is
\begin{equation}
 U=W\left(S\ket{r}\bra{r}+S^{\dagger}\ket{l}\bra{l}\right),
\end{equation}
where $S\ket{n}=\ket{n+1}$ is a unitary shift of the particle's position, its inverse $S^{\dagger}$ is a shift in the opposite direction, and $W$ is a unitary that acts on the extra degree of freedom.  The term in brackets in the evolution operator is something we refer to as a conditional shift.  In this case, it shifts to the right if the particle is in the $\ket{r}$ state and to the left if the particle is in the $\ket{l}$ state.  The concept of a shift conditioned on the state of the extra degree of freedom is ubiquitous in the study of quantum walks.

A shift can be written as $S=e^{-iPa}$, where $a$ is the lattice spacing and $P$ is the discrete momentum operator.  Generally, one does not factor out the lattice spacing in the definition of the discrete momentum operator.  Here, however, it is the $P$ we have just defined that becomes the continuum momentum operator in the continuum limit.

Now define $\sigma_z=\ket{r}\bra{r}-\ket{l}\bra{l}$ and $\sigma_x=\ket{r}\bra{l}+\ket{l}\bra{r}$.  The extra degree of freedom is not spin and does not correspond to it in the continuum limit, but it is conventional to use $\sigma_z$ and $\sigma_x$ anyway.  If we choose $W=e^{-im\sigma_xa}$, then
\begin{equation}
\label{eq:1dQW}
\begin{split}
 U & =e^{-im\sigma_xa}e^{-iP\sigma_za}\\
 & = e^{-i(P\sigma_z +m\sigma_x)a}+O(a^2).
\end{split}
\end{equation}
The exponent in the second line is $-iha$, where $h=P\sigma_z +m\sigma_x$, which, as we saw in section \ref{sec:Dirac and Weyl particles}, is the Dirac Hamiltonian in one dimension.  So, as an approximation to evolution via the Dirac Hamiltonian over a short time $a$, the error in the approximation is $O(a^2)$.

Then, if we look at this over $N$ timesteps, we have
\begin{equation}
 U^N=\left(e^{-im\sigma_xa}e^{-iP\sigma_za}\right)^N\rightarrow e^{-i(P\sigma_z +m\sigma_x)t},
\end{equation}
by the Lie-Trotter product formula.  Here we let $Na=t$, which is constant, so $N=O(\f{1}{a})$.  Note that $a$ is the lattice spacing, while $t$ is time, so there is a constant $c$, which we have set to one, accounting for the apparently different dimensions in $t=Na$.  We could also consider having a timestep length $\delta t$, with $c=\f{a}{\delta t}$.  But it will be simpler to have $a=\delta t$.

We should pause for a second to note that our choice for the coin operator, $W=e^{-im\sigma_xa}$, tends to the identity as the lattice spacing goes to zero.  In fact, this is necessary to get a continuum limit.  Because we are considering the limit of infinitely many timesteps, we need the effect of a {\em single} application of the evolution operator to be small.  So, as $a\rightarrow 0$, we need $U\rightarrow \openone$.  On the other hand, something like $X^N=\left(\begin{smallmatrix} 0 & 1\\ 1 & 0 \end{smallmatrix}\right)^N$ does not even have a limit as $N$ tends to infinity.  

The next step is to make these ideas rigorous.  We start in section \ref{sec:Approximating Continuum Models by Discrete Ones} by separating the problem of taking continuum limits into two parts:\ approximating the state of the continuum system and approximating the dynamics.  Next, in section \ref{sec:The Lie-Trotter Product Formula}, we discuss the Lie-Trotter product formula applied to quantum walks in more detail.  In section \ref{sec:A General Recipe}, we look at a general recipe for constructing quantum walks that have a desired continuum limit.  Here we reproduce the quantum walks that give rise to the Dirac and Weyl equation in two and three dimensional space, proposed in \cite{Bial94}.  Following this, in section \ref{sec:Faster Converging Quantum Walks}, we introduce quantum walks that converge significantly faster to their continuum limits than the previous examples.  Next, in section \ref{sec:Symmetries on the Lattice} we look at the issue of discrete symmetries of quantum walks that have relativistic dynamics as their continuum limit.

\subsection{Approximating Continuum Models by Discrete Ones}
\label{sec:Approximating Continuum Models by Discrete Ones}
Now that we have the rough idea of what it means to take a continuum limit, we ought to be more precise.  We should view the discrete particle's dynamics as an approximation to the dynamics of a particle in continuous spacetime.  That the continuum limit of the discrete dynamics is the same as the dynamics of a continuum particle means that we can approximate that continuum particle's evolution arbitrarily well by the discrete one.  There is still the question of how quickly the discrete dynamics converges to the continuum dynamics, which is particularly important if we want to use these discrete models to simulate particles in continuous spacetime.

Our motivation may not be simulation.  We may be interested in the possibility that spacetime could fundamentally be discrete.  Then we are asking whether, on scales large compared to the lattice spacing, the discrete system looks like a continuous system.  We will mostly frame our discussion in a way geared towards simulation, but it is not hard to recast the arguments to look at these quantum walks from a more fundamental point of view.

When looking at continuum limits from a simulation point of view, it is best to think in terms of a family of quantum walks with different lattice spacings that become better approximations to the continuum theory as the lattice spacing goes to zero.

In order to be clear about how good the approximation is, we need to be able to compare the discrete and continuous systems' states in a meaningful way.  In other words, we need a notion of distance between discrete and continuum states.  One way to get this is to have some prescription for mapping the discrete particle's state space into the continuum particle's state space $\mathcal{H}$.  Then we can use the norm on the continuum particle's state space as a distance measure between discrete and continuum states.

So we have a continuum particle in state $\ket{\psi}$ with Hamiltonian $H$ and a discrete particle in some state\footnote{The subscript $d$ will denote discrete objects when there is any possibility of confusion.} $\ket{\psi_d}$ evolving via $U$.  And we need to show that $U^N\ket{\psi_d}$ approximates $e^{-iHt}\ket{\psi}$ well.  In other words, supposing we have mapped $U^N\ket{\psi_d}$ into $\mathcal{H}$, we need to look at
\begin{equation}
 \|e^{-iHt}\ket{\psi}-U^N\ket{\psi_d}\|_2,
\end{equation}
where $\|\cdot\|_2$ is the norm on $\mathcal{H}$.  By applying the triangle inequality, we arrive at
\begin{equation}
\begin{split}
 \|e^{-iHt}\ket{\psi}-U^N\ket{\psi_d}\|_2 & \leq \|e^{-iHt}\ket{\psi}-U^N\ket{\psi}\|_2+\|U^N\ket{\psi}-U^N\ket{\psi_d}\|_2\\
& = \|(e^{-iHt}-U^N)\ket{\psi}\|_2+\|\ket{\psi}-\ket{\psi_d}\|_2,
\end{split}
\end{equation}
where we used the invariance of $\|\cdot\|_2$ under unitary transformations to get the second line.  So the problem has been separated into two parts.  The fist involves showing that the quantum walk evolution operator converges to the continuum evolution operator, and is quantified by the first term above.  The second involves showing that the discrete initial state converges to the continuum initial state, and is quantified by the second term above.  The importance of the second term is mostly for simulation.

We will look at two ways to map discrete states into the continuum here.  The first option is physically satisfying if a little cumbersome.  It entails mapping the discrete position states to continuum wavefunctions localized on cubes with length of side given by the lattice spacing.  We call this the block mapping.  An advantage of this is that localized particles in the discrete picture are still localized in the continuum.  The second option is to map discrete momentum states to continuum momentum states with the same value of momentum.  We will call this the momentum mapping.  It is sometimes more convenient mathematically, but has the conceptual disadvantage that a particle localized on a single site in the discrete picture is no longer localized after being mapped to the continuum.  

\subsubsection{The Block Mapping}
For simplicity, let us start with one spatial dimension.  We map the discrete position state $\ket{n}$ into the continuum via
\begin{equation}
 \ket{n}\ \rightarrow\ \frac{1}{\sqrt{a}}\int_{na}^{(n+1)a}\!\textrm{d}x\,\ket{x}.
\end{equation}
\begin{compactwrapfigure}{r}{0.44\textwidth}
\centering
\begin{minipage}[r]{0.42\columnwidth}%
\centering
    \resizebox{5.0cm}{!}{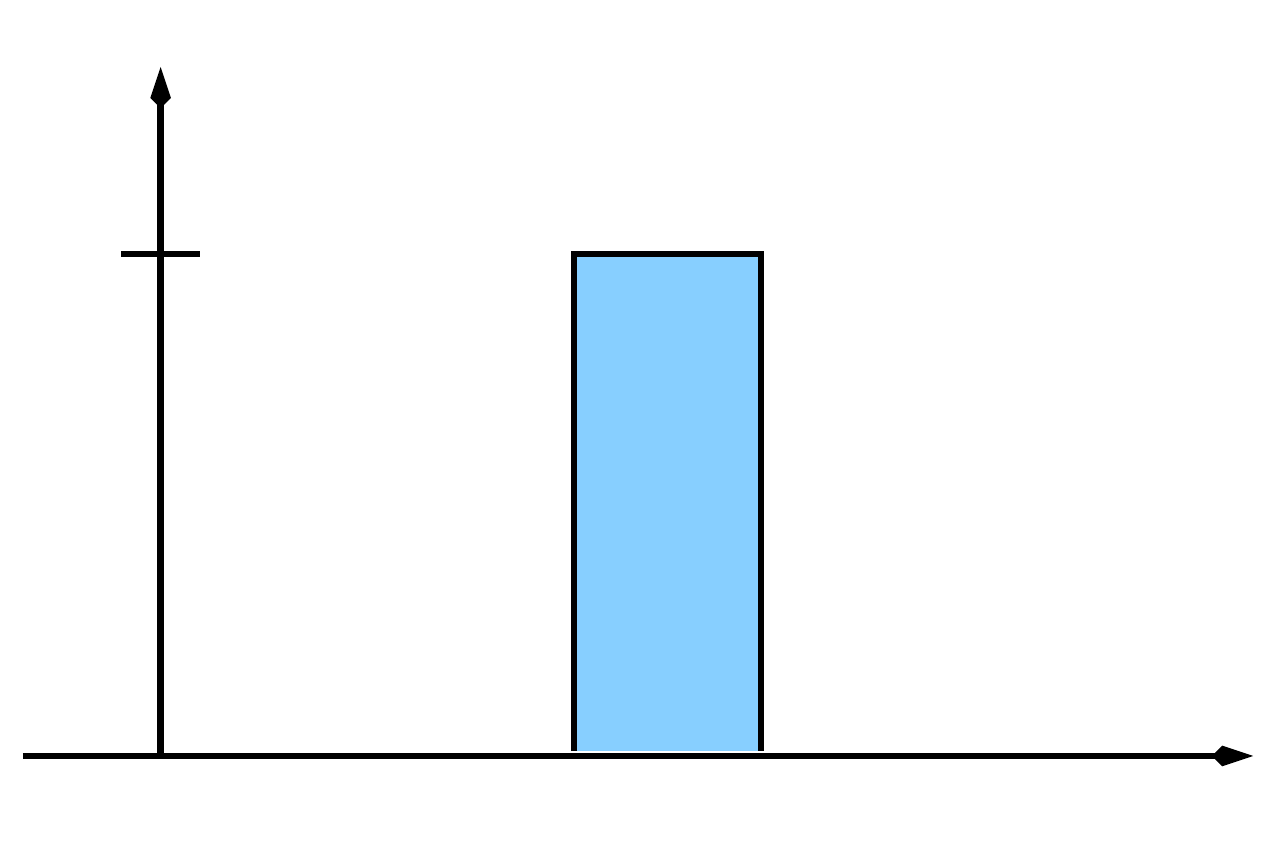}
    \footnotesize{\caption[A discrete position state embedded into the continuum]{A discrete position state embedded into the continuum.}\vspace{0.0cm}}
    \label{fig:70}
\end{minipage}
\end{compactwrapfigure}
These states do not actually converge to anything as $a$ tends to zero.  So it will be convenient for us to work with discrete momentum states.  These are
\begin{equation}
 \ket{p_d}=\sqrt{a}\displaystyle\sum_{n}e^{ip_d na}\ket{n},
\end{equation}
where $p_d\in(-\f{\pi}{a},\f{\pi}{a}]$.  The inner product between two discrete momentum states is\\
\begin{equation}
\label{eq:Tupac}
 \langle p_d|q_d\rangle=(2\pi)\delta(p_d-q_d).
\end{equation}
To see this, use the identity \cite{Peres95}
\begin{equation}
\label{eq:Nas}
\sum_{n}e^{in(x-y)}=2\pi\delta(x-y),
\end{equation}
which holds when $x,y\in(-\pi,\pi]$. 

Continuum momentum states are
\begin{equation}
 \ket{p}=\int_{-\infty}^{\infty}\!\textrm{d}x\, e^{ipx}\ket{x},
\end{equation}
where now $p\in\mathbb{R}$.

The inner product of a discrete momentum state and a continuum momentum state, when its momentum is restricted to $(-\f{\pi}{a},\f{\pi}{a}]$, is
\begin{equation}
\label{eq:elvis}
 \braket{p}{q_d}=2\pi\delta(p-q_d)\left[\frac{1-e^{-ipa}}{ipa}\right]=2\pi\delta(p-q_d)\left[1+O(pa)\right].
\end{equation}
To see this, again use the identity in equation (\ref{eq:Nas}).  Looking at equation (\ref{eq:elvis}), once we restrict to momenta less than some cutoff $\Lambda$ that grows slower than $1/a$, the inner product between discrete momentum states and continuum momentum states tends towards $2\pi\delta(p-q_d)$ as $a$ goes to zero.

Now, our goal is to approximate continuum states by discrete ones, and, since we have mapped the discrete states into the continuum, we are able to quantify how good the approximation is.  To see that we can approximate continuum states arbitrarily well, first expand the continuum state we want to approximate, $\ket{\psi}$, in the momentum basis, meaning
\begin{align}
\ket{\psi} =\int_{-\infty}^{\infty}\!\frac{\textrm{d}p}{2\pi} \psi(p)\ket{p}.
\end{align}
It will help to decompose $\ket{\psi}$ into low and high momentum components:\ 
\begin{equation}
 \begin{split}
 \ket{\psi} &= \int_{-\Lambda}^{\Lambda}\!\frac{\textrm{d}p}{2\pi} \psi(p)\ket{p}+\int_{|p|>\Lambda}\!\frac{\textrm{d}p}{2\pi} \psi(p)\ket{p}\\
 &= \ket{\psi_{\Lambda}}+\ket{\psi_{\Lambda}^{\perp}},
 \end{split}
\end{equation}
where the states in the second line are not normalized.
Now we can take as our discrete approximation
\begin{equation}
\ket{\psi_d} =\alpha\int_{-\Lambda}^{\Lambda}\!\frac{\textrm{d}p_d}{2\pi} \psi(p_d)\ket{p_d},
\end{equation}
where $\alpha$ is chosen to normalize the state.  As $\Lambda$ tends to infinity, $\alpha$ tends to one.  This is because
\begin{equation}
 \alpha^2=\int_{-\Lambda}^{\Lambda}\!\frac{\textrm{d}p}{2\pi} |\psi(p)|^2,
\end{equation}
which tends to one as $a$ goes to zero.  Now, by repeated application of the triangle inequality,
\begin{equation}
\label{eq:jimi}
 \begin{split}
\|\ket{\psi}-\ket{\psi_d}\|_2 & \leq  \|\ket{\psi_{\Lambda}}-\ket{\psi_d}\|_2 + \|\ket{\psi_{\Lambda}^{\perp}}\|_2\\
& \leq \|\ket{\psi_{\Lambda}}-\frac{1}{\alpha}\ket{\psi_d}\|_2 + \|(1-\frac{1}{\alpha})\ket{\psi_d}\|_2+ \|\ket{\psi_{\Lambda}^{\perp}}\|_2,
\end{split}
\end{equation}
Now, the third term tends to zero as $a\rightarrow 0$ since $\Lambda$ tends to infinity as $a\rightarrow 0$.  Furthermore, the second term goes to zero since $\alpha\rightarrow 1$.  To see that the first term in the second line of equation (\ref{eq:jimi}) goes to zero, look at
\begin{align}
\label{eq:cont1}
\frac{1}{\alpha}\langle \psi_{\Lambda} | \psi_d \rangle=\int_{-\Lambda}^{\Lambda}\!\frac{\textrm{d}p}{2\pi} |\psi(p)|^2 +O(\Lambda a),
\end{align}
where we used equation (\ref{eq:elvis}).  Therefore, (\ref{eq:cont1}) tends to one as $a\rightarrow 0$ provided $\Lambda$ grows slowly enough.  For example, taking $\Lambda$ to be $\f{\pi}{a^{1/2}}$ would do.  This means that $\|\ket{\psi_{\Lambda}}-\frac{1}{\alpha}\ket{\psi_d}\|_2$ and hence $\|\ket{\psi}-\ket{\psi_d}\|_2$ tend to zero as $a$ tends to zero.

So we can choose the lattice spacing $a$ small enough that we can approximate a continuum state arbitrarily well.  How accurate the approximation is for a given value of $a$ depends on how smooth the continuum state $\ket{\psi}$ is, which is determined by the coefficients $\psi(p)$ corresponding to $|p|\leq\Lambda$.

Since momentum states in $d$ dimensions are simply tensor products of the one dimensional momentum states, this carries over simply to the $d$ dimensional case.  The same procedure works when the particle has an extra degree of freedom by writing
\begin{equation}
\ket{\psi}=\sum_i\ket{\psi_i}\ket{i},
\end{equation}
where $\ket{\psi_i}$ are position states and $\ket{i}$ are an orthonormal basis for the coin degree of freedom.

\subsubsection{The Momentum Mapping}
The momentum mapping is sensible from a simulation point of view, as it is easy to see how to approximate a continuum state by a discrete one.  The basic idea behind this mapping is to identify each discrete momentum state with the continuum momentum state corresponding to the same value of momentum.  First, discrete momentum states are
\begin{equation}
 \ket{\vec{p}_d}=a^{d/2}\displaystyle\sum_{\vec{n}}e^{i\vec{p}_d.\vec{n}a}\ket{\vec{n}},
\end{equation}
where the components of $\vec{p}_d$ take values in $(-\textstyle{\frac{\pi}{a}},\textstyle{\frac{\pi}{a}}]$.  Again, the inner product between two discrete momentum states is
\begin{equation}
 \langle \vec{p}_d|\vec{q}_d\rangle=(2\pi)^d\delta(\vec{p}_d-\vec{q}_d).
\end{equation}
Continuum momentum states are
\begin{equation}
 \ket{\vec{p}\,	}=\int_{-\infty}^{\infty}\!\textrm{d}^dx\, e^{i\vec{p}.\vec{x}}\ket{\vec{x}},
\end{equation}
where the components of $\vec{p}$ take values in $\mathbb{R}$.

The mapping between discrete and continuum state spaces follows by identifying discrete and continuous momentum states, meaning we map $\ket{\vec{p}_d}$ to $\ket{\vec{p}\,}$ when $\vec{p}_d=\vec{p}$.

Then to see how to approximate a continuum state by a discrete state, let us look at the one dimensional case.  We approximate the continuum state
\begin{align}
 \ket{\psi} &=\int_{-\infty}^{\infty}\!\frac{\textrm{d}p}{2\pi} \psi(p)\ket{p}
\end{align}
by
\begin{equation}
 \ket{\psi_d}=\alpha\int_{-\f{\pi}{a}}^{\f{\pi}{a}}\!\frac{\textrm{d}p_d}{2\pi} \psi(p_d)\ket{p_d}\equiv\alpha\int_{-\f{\pi}{a}}^{\f{\pi}{a}}\!\frac{\textrm{d}p}{2\pi} \psi(p)\ket{p},
\end{equation}
where $\alpha$ normalizes the state.  It is straightforward to see that this converges to the continuum state.  So $\ket{\psi_d}$ tends to $\ket{\psi}$ as $a\rightarrow 0$.

This method is convenient since convergence of discrete to continuum states is clear, but physically it is a little unsatisfying.

\subsubsection{The Embedded Evolution Operator}
It is convenient that, when mapped to an operator on $\mathcal{H}$, the quantum walk evolution operator can be chosen to be the same operator for both mappings.  Let us see why.  It helps that any quantum walk unitary can be written as $\textstyle\sum_{\vec{q}}A_{\vec{q}}S_{\vec{q}}$, where $A_{\vec{q}}$ are operators on $\mathcal{H}_C$ and $S_{\vec{q}}$ is a shift by lattice vector $\vec{q}$.  We will take this for granted for now, postponing a proof until section \ref{sec:Form of an arbitrary Quantum Walk Unitary}.

In the block mapping, an operator that implements a shift by a lattice vector $\vec{q}$ is $e^{-i\vec{q}.\vec{P}a}$, where $\vec{P}$ is the continuum momentum vector operator.

In the momentum mapping, a shift is naturally written in terms of the discrete momentum operator:\ a shift by $\vec{q}$ is given by $e^{-i\vec{q}.\vec{P}_d a}$, where $\vec{P}_d$ is the discrete momentum vector operator.  But when we embed the discrete momentum states into the continuum we do it by mapping them to continuum momentum states.  This means that applying $e^{-i\vec{q}.\vec{P}a}$ has the same effect as $e^{-i\vec{q}.\vec{P}_d a}$.

\subsection{The Lie-Trotter Product Formula}
\label{sec:The Lie-Trotter Product Formula}
Let us look again at the quantum walk from section \ref{sec:Continuum Limits}, which evolves via
\begin{equation}
 U=e^{-im\sigma_xa}e^{-iP\sigma_za}.
\end{equation}
The error in approximating the evolution of a Dirac particle over a time $t=Na$ is given by
\begin{equation}
\label{eq:LTform1}
 \|(e^{-iHt}-U^N)\ket{\psi}\|_2\leq  \|e^{-iHt}-U^N\|,
\end{equation}
where $\|\cdot\|$ is the operator norm.\footnote{Equation (\ref{eq:LTform1}) follows from the definition of the operator norm, $\|A\|^2=\displaystyle\max_{\langle\psi|\psi\rangle=1} \bra{\psi}A^{\dagger}A\ket{\psi}$.}  Now, the operator norm of the momentum operator is $\f{\pi}{a}$, which grows too quickly with $a$ to make this formula useful without more care.  So we introduce a cutoff $\Lambda$ such that the states $\ket{\psi}$ we consider have momentum restricted to values of $p$ with $|p|\leq \Lambda$.  This means it is enough to bound $\|e^{-iHt}-U^N\|_{\Lambda}$, where $\|\cdot\|_{\Lambda}$ is the operator norm restricted to the subspace with momentum cutoff $\Lambda$.  Next, we use the following useful formula for unitaries \cite{NC00}
\begin{equation}
 \|U_1...U_N-V_1...V_N\|\leq N\max_i\|U_i-V_i\|.
\end{equation}
What this says is that when approximating one product of unitaries by another the errors add.  The final thing we need is
\begin{equation}
 \|e^{-im\sigma_xa}e^{-iP\sigma_za}-e^{-i(P\sigma_z+m\sigma_x)a}\|_{\Lambda}=O(\Lambda^2a^2),
\end{equation}
which follows by Taylor expanding both terms.  Putting all these together leads to
\begin{equation}
 \|e^{-iHt}-U^N\|_{\Lambda}= O(\Lambda^2 a)
\end{equation}
since $t=Na$ is a constant.  Then, provided we allow $\Lambda$ to grow sufficiently slowly as $a\rightarrow 0$, this proves the result we need.

This is very similar to the proof of the Lie product formula \cite{Hall03,NC00}.
\begin{theorem}
\label{th:Lie}
 Given two bounded operators $A$ and $B$ on a Hilbert space.
 \begin{equation}
 \label{eq:Lie}
 \|(e^{iA/N}e^{iB/N})^N-e^{i(A+B)}\|=O(K^2/N),
\end{equation}
where $K=\max\{\|A\|,\|B\|\}$.
\end{theorem}

Although equation (\ref{eq:Lie}) is very useful, it is not applicable to some of the limits we will encounter later in section \ref{sec:The Continuum Limit}.  There we will have to use different methods to take the continuum limit.

This theorem is also valid in a sense for unbounded operators on infinite dimensional Hilbert spaces.  This is the content of Trotter's formula \cite{Trotter59}.  (See \cite{Teschl09} for a nice discussion.)
\begin{theorem}
 Let $A$, $B$ and $A+B$ be self-adjoint operators on a Hilbert space.  Then, for any state $\ket{\psi}\in \mathcal{D}(A)\cap \mathcal{D}(B)$, where $\mathcal{D}(A)$ and $\mathcal{D}(B)$ are the domains of $A$ and $B$ respectively.  
 \begin{equation}
 \|\left((e^{iA/N}e^{iB/N})^N-e^{i(A+B)}\right)\ket{\psi}\|_2\rightarrow 0
\end{equation}
as $N\rightarrow \infty$.
\end{theorem}
The notion of convergence is different here:\ it is not convergence in the operator norm, but convergence of the dynamics for fixed states.  Although we will not need this formula in the main results here, it seems that it should be useful for studying convergence of quantum walks in external fields.  We will discuss this problem in more detail in chapter \ref{chap:Conclusions and Open Problems 1}.

\subsection{A General Recipe}
\label{sec:A General Recipe}
Here we will cook up a scheme for constructing quantum walks that have a desired continuum limit.  For example, suppose we want a quantum walk that has the two dimensional Dirac equation as its continuum limit.  The continuum Hamiltonian, as we saw in section \ref{sec:Dirac and Weyl particles}, is
\begin{equation}
 h=P_x\sigma_x+P_y\sigma_y+m\sigma_z,
\end{equation}
where $P_x$ and $P_y$ are momentum operators along the $x$ and $y$ axes respectively.  The quantum walk given by
\begin{equation}
\label{eq:2dQW}
\begin{split}
 U & =e^{-im\sigma_za}e^{-iP_x\sigma_xa}e^{-iP_y\sigma_ya}\\
 & = e^{-i(P_x\sigma_x+P_y\sigma_y+m\sigma_z)a}+O(a^2)
 \end{split}
\end{equation}
will do the trick.  The second two terms in the first line are conditional shifts.  For example, 
\begin{equation}
e^{-iP_x\sigma_xa}=S_{x}\ket{\uparrow_{x}}\bra{\uparrow_{x}\!}+S^{\dagger}_{x}\ket{\downarrow_{x}}\bra{\downarrow_{x}\!}=T_x,
\end{equation}
where $S_{x}=e^{-iP_xa}$ is a shift by one lattice site along the $x$ direction.  So the evolution in equation (\ref{eq:2dQW}) is indeed a quantum walk, and it is given by a product of two conditional shifts and a coin operator $e^{-im\sigma_za}$.

\begin{figure}[ht!]
\centering
    \resizebox{8.5cm}{!}{\includegraphics{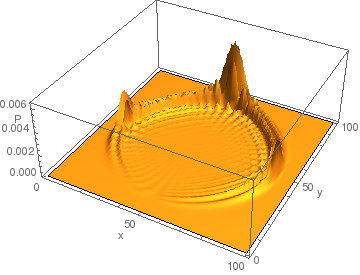}}
    \footnotesize{\caption[Position distribution of the massless Dirac quantum walk]{Position probability distribution of the massless Dirac quantum walk in two dimensional space after $40$ timesteps, with $100\times 100$ sites and an initial Gaussian state centred at site $(50,50)$.  The x- and y-axes correspond to the position of the particle.  The z-axis shows the probability $P$ of finding the particle at that point.}}
    \label{fig:7}
\end{figure}

In a remarkable paper, \cite{Bial94}, quantum walks giving rise to Weyl and Dirac particles in three dimensional space were presented.  They can be constructed in the same way.  For the case of a Weyl particle, the evolution operator is a product of conditional shifts in each direction:
\begin{align}
\label{eq:GR1}
U_R & =T_xT_yT_z\\
 & =e^{-iP_x\sigma_xa}e^{-iP_y\sigma_ya}e^{-iP_z\sigma_za}\\
 & =e^{-i\vec{P}.\vec{\sigma}a}+O(a^2).
\end{align}
Again the conditional shifts are
\begin{equation}
 T_{b}=S_{b}\ket{\uparrow_{b}}\bra{\uparrow_{b}\!}+S^{\dagger}_{b}\ket{\downarrow_{b}}\bra{\downarrow_{b}\!}
\end{equation}
where $b\in\{x,y,z\}$, $S_{b}$ shifts one lattice site in the $b$ direction and $\ket{\uparrow_{b}}$ and $\ket{\downarrow_{b}}$ are spin up and spin down along the $b$ axis.  So, for example, $T_{z}$ shifts a particle in the state $\ket{\vec{n}}\ket{\uparrow_{z}}$ one step in the $+\hat{z}$ direction.

The evolution operator in equation (\ref{eq:GR1}) approximates evolution via the Hamiltonian for a right handed Weyl fermion, $H=\vec{\sigma}.\vec{P}$.  To get a left handed fermion, we could consider the same evolution but with the signs in the exponents flipped, which we will call $U_L$.  In the continuum limit, this would result in the Hamiltonian $H=-\vec{\sigma}.\vec{P}$.

It is interesting that this discrete evolution essentially uses a body centred cubic neighbourhood.  In fact, a seemingly more natural choice, the cubic neighbourhood, cannot give the three dimensional Weyl equation in the continuum limit \cite{Bial94}.

To get a quantum walk that behaves like a Dirac particle in three dimensional space in the continuum limit, we need a particle with a four dimensional internal degree of freedom.  Suppose that the extra degree of freedom has Hilbert space $\mathcal{H}_S\otimes\mathcal{H}_H$, where $\mathcal{H}_H$ is spanned by the orthonormal states $\ket{l}$ and $\ket{r}$.  We take the evolution operator to be
\begin{align}
 U & =e^{-im\beta a}e^{-iP_x\alpha_xa}e^{-iP_y\alpha_ya}e^{-iP_z\alpha_za}\\
 & =e^{-i(\vec{P}.\vec{\alpha}+m\beta)a}+O(a^2). 
\end{align}
where
\begin{equation}
 \begin{split}
  \beta & =\openone_S\otimes(\ket{l}\bra{r}+\ket{r}\bra{l})\\
  \alpha_i & =\sigma_i\otimes(\ket{l}\bra{l}-\ket{r}\bra{r}),
 \end{split}
\end{equation}
with $\openone_S$ the identity on $\mathcal{H}_S$.  A little work shows that $U =W(U_R\ket{r}\bra{r}+U_L\ket{l}\bra{l})$, so it is really just a direct sum of both the left and right handed Weyl quantum walk, with the mass term $W=e^{-im\beta a}$.  The continuum limit of this quantum walk was treated rigorously in \cite{AFN13}.

We will often refer to these quantum walks as Dirac or Weyl quantum walks if they give the Dirac or Weyl equation in the continuum limit.  Interestingly, at the discrete level, these quantum walks do not generally have the rotational symmetry of the lattice, something we will return to in section \ref{sec:Symmetries on the Lattice}.

Clearly, this recipe will work for other Hamiltonians that are linear in the momentum operator, though we have already seen the most important examples here.

\subsection{Faster Converging Quantum Walks}
\label{sec:Faster Converging Quantum Walks}
In fact, we can do better than the examples in the previous section by modifying the dynamics.  For example, to get a discrete-time quantum walk that better approximates a particle obeying the Dirac equation in one dimension, we apply different operators over even and odd timesteps.\footnote{We could just think of two of these timesteps as one single timestep, which would allow us to keep time translation invariance of the dynamics.}  We can take $U$ to be
\begin{equation}
\begin{split}
 & e^{-i\sigma_z Pa}e^{-im\sigma_xa}\ \textrm{for}\ n\ \textrm{odd,}\\
 & e^{-im\sigma_xa}e^{-i\sigma_z Pa}\ \textrm{for}\ n\ \textrm{even.}
 \end{split}
\end{equation}
Then we have
\begin{equation}
 (e^{-im\sigma_xa}e^{-i\sigma_z Pa})(e^{-i\sigma_z Pa}e^{-im\sigma_xa})=e^{-i(P\sigma_z-m\sigma_x)2a}+O(a^3).
\end{equation}
So the error is $O(a^3)$ now, compared to $O(a^2)$ before.  This is basically just using Strang splitting to approximate $e^{i(A+B)2a}$ by $e^{iAa}e^{iB2a}e^{iAa}$, which has error $O(a^3)$.  Normally, when used as a tool for simulation, the operators in the exponents would be local.  This is not the case for us:\ $P_i$ is a nonlocal operator.

This trick works more generally, as we can see from the following theorem.
\begin{theorem}
 Suppose we have a continuum quantum particle with the Hamiltonian $H=\textstyle\sum_l H_l$,
 where each term $H_l$ is either linear in a component of the momentum vector operator ($P\sigma_z$, for example) or constant, like $m\sigma_x$.  Then the evolution operator given by
 \begin{equation}
 V=(e^{-iH_1a}...e^{-iH_na})(e^{-iH_na}...e^{-iH_1a})
 \end{equation}
converges to evolution via $H$ with error $O(\Lambda^3a^2)$, in the sense that
\begin{equation}
 \|V^{N/2}-e^{-iHt}\|_{\Lambda}= O(\Lambda^3a^2),
\end{equation}
where $t=Na$, and $\Lambda$ is the momentum cutoff.
\end{theorem}
\begin{proof}
We restrict to states in a subspace with momentum components $|p_b|\leq \Lambda$.  The next step is to use the result, which follows from Taylor expanding, that
\begin{equation}
\label{eq:fast1}
 \|e^{iAa}e^{iB2a}e^{iAa}-e^{i(A+B)2a}\|_{\Lambda}= O(K^3a^3),
\end{equation}
where now $K=\max\{\|A\|_{\Lambda},\|B\|_{\Lambda}\}$.  By using the triangle inequality and repeated application of this formula (first using equation (\ref{eq:fast1}) with $B=H_n$ and $A=H_{n-1}$, then using it again with $B=H_n+H_{n-1}$ and $A=H_{n-2}$ and so on), we get
\begin{equation}
\|V-e^{-iH(2a)}\|_{\Lambda}=O(K^3a^3),
\end{equation}
where now $K=\max\|H_l\|_{\Lambda}$.  Then, by choosing $t=Nq	a$, we have
\begin{equation}
 \|V^{N/2}-e^{-iHt}\|_{\Lambda}\leq \frac{N}{2}\|V-e^{-iH(2a)}\|_{\Lambda}= O(K^3a^2).
\end{equation}
As before, we want to let $\Lambda$ tend to infinity as $a\rightarrow 0$.  As $\Lambda$ goes to infinity, $K=O(\Lambda)$ since $H_l$ are either constant or linear in a component of the momentum operator.
\end{proof}
It may be possible to construct quantum walks that are better approximations to physical particles using higher order Suzuki-Trotter formulas \cite{Suzuki90}, though one would need to ensure that all terms like $(e^{iPa})^c$ have integer $c$, otherwise the evolution may not be causal.

\subsection{Symmetries on the Lattice}
\label{sec:Symmetries on the Lattice}
It is a little disconcerting that the quantum walks we have seen do not generally have the rotational symmetry of the lattice \cite{Short13}, despite having Lorentz symmetry in the continuum limit.  Here is an argument to see why this is the case.  Take the quantum walk
\begin{align}
U & =e^{-iP_x\sigma_xa}e^{-iP_y\sigma_ya}\\
 & =T_xT_y,
\end{align}
where again
\begin{equation}
 T_{b}=S_{b}\ket{\uparrow_{b}}\bra{\uparrow_{b}\!}+S^{\dagger}_{b}\ket{\downarrow_{b}}\bra{\downarrow_{b}\!}.
\end{equation}
Already this seems suspect.  The conditional shifts are applied in a specific order.  To check whether there is the rotational symmetry of the lattice, we can see if a rotation by $90$ degrees commutes with the evolution operator.  Whatever the unitary is it will take $S_x$ to $S_y$ and $S_y$ to $S_x^{\dagger}$.  It may also have some effect on the coin degree of freedom, in which case a unitary $R$ acts on the coin.  Then under this rotation
\begin{equation}
\begin{split}
 U=e^{-iP_x\sigma_xa}e^{-iP_y\sigma_ya} & \rightarrow R^{\dagger}e^{-iP_y\sigma_xa}e^{iP_x\sigma_ya}R\\
 & =e^{-iP_y(R^{\dagger}\sigma_xR)a}e^{iP_x(R^{\dagger}\sigma_yR)a}=V.
 \end{split}
\end{equation}
Regardless of what the effect of the unitary $R$ is on $\mathcal{H}_C$ is, this transformed operator first does a conditional shift in the $x$ direction.  On the other hand, $U$ first does a conditional shift in the $y$ direction, so $U$ is not invariant under this transformation.  To see this, consider the initial state $\ket{0,0}\ket{\uparrow_y}$.  After applying $U$, the state only has support on positions $(1,1)$ and $({-1},1)$.  In particular, the state has no support on $(1,{-1})$ and $({-1},{-1})$.  On the other hand, $V$ first does a conditional shift in the $x$ direction.  Whether the particle moves in the ${+x}$ or ${-x}$ direction depends on whether its extra degree of freedom is in a ${-1}$ or ${+1}$ eigenstate of $R^{\dagger}\sigma_yR$.  Because each eigenvector of $R^{\dagger}\sigma_xR$ has nonzero inner product with each eigenvector of $R^{\dagger}\sigma_yR$, after applying $V$ to $\ket{0,0}\ket{\uparrow_y}$, the particle will have some amplitude to move down in the $y$ direction.  So the particle's state {\em will} have support on $(1,{-1})$ and $({-1},{-1})$.  Therefore, $V$ is not the same as $U$.

There is a cheap way to construct a quantum walk that has Lorentz symmetry in the continuum, while also having lattice rotational symmetry at a discrete scale.  This works by taking a quantum walk with a four dimensional coin degree of freedom $\mathcal{H}_C=\mathcal{H}_S\otimes\mathcal{H}_H$, where $\mathcal{H}_H$ has orthonormal basis $\ket{l}$ and $\ket{r}$.  Let the evolution operator be
\begin{equation}
 U=e^{-iP_x\sigma_xa}e^{-iP_y\sigma_ya}\ket{l}\bra{l}+e^{-iP_y\sigma_ya}e^{-iP_x\sigma_xa}\ket{r}\bra{r}.
\end{equation}
Note the different order of the conditional shifts in both terms.  Under a rotation by $90$ degrees, the momentum operators will transform so that $U$ becomes
\begin{equation}
 e^{-iP_y\sigma_xa}e^{iP_x\sigma_ya}\ket{l}\bra{l}+e^{iP_x\sigma_ya}e^{-iP_y\sigma_xa}\ket{r}\bra{r}.
\end{equation}
Now, there is rotational symmetry if there is a unitary $R$ on $\mathcal{H}_C$ that undoes this transformation.  To see that such an $R$ exists, we construct it in two steps.  First, apply $\sqrt{\sigma_z}=\left(\begin{smallmatrix} 1 & 0\\ 0 & i \end{smallmatrix}\right)$, which takes $\sigma_x$ to $\sigma_y$ and $\sigma_y$ to $-\sigma_x$.  So we get
\begin{equation}
 e^{-iP_y\sigma_ya}e^{-iP_x\sigma_xa}\ket{l}\bra{l}+e^{-iP_x\sigma_xa}e^{-iP_y\sigma_ya}\ket{r}\bra{r}.
\end{equation}
Second, we apply $\ket{l}\bra{r}+\ket{r}\bra{l}$ to get back $U$.

So this quantum walk has the rotational symmetry of the lattice,\footnote{All lattice rotations here are products of the rotation by $90$ degrees, so it was enough to show that the evolution is invariant under a rotation by $90$ degrees.  Also, $R$ generates a representation of this discrete rotation group.} and, as it is just a direct sum of two quantum walks that have Lorentz symmetry in the continuum, the continuum limit has Lorentz symmetry.

One of the reasons this is a little unsatisfying is that the degrees of freedom that allow us to have rotational symmetry on the lattice are not the same as those allowing rotational symmetry in the continuum.  Also, the evolution is essentially a direct sum of two independent quantum walks, which is a little unnatural.  Still, similar tricks work for quantum walks leading to relativistic dynamics in higher spatial dimensions.

\section{Two Dimensional Coins are Special}
\label{sec:Two Dimensional Coins are Special}
In this section, we will mostly focus on quantum walks with a two dimensional coin, following \cite{FS14}.  Our first task, in section \ref{sec:Form of an arbitrary Quantum Walk Unitary} will be to come up with a simple expression for the quantum walk evolution operator.  Following this, in section \ref{sec:Mass for Quantum Walks}, we will introduce the concept of mass for general quantum walks.  Then we give a general formula for the continuum limit Hamiltonian corresponding to a quantum walk in section \ref{sec:The Continuum Limit}.  After taking the continuum limit, in section \ref{sec:The Continuum Hamiltonian}, we will see that for a massless quantum walk with a two dimensional coin these continuum limits correspond to evolution via a massless relativistic equation of motion.  This is provided that we do a simple change of coordinates or just demand rotational symmetry in the continuum limit.  It is somewhat surprising that this occurs so generally:\ it would be reasonable to assume that quantum walks that give rise to relativistic dynamics in the continuum limit would have to be very fine-tuned.  Finally, in section \ref{sec:Counterexample with a Three Dimensional Coin}, we see that this result does not hold for quantum walks with higher dimensional coins.  We do this by looking at a counterexample with a three dimensional coin that does not have relativistic symmetry in the continuum limit.

\subsection{General Form of Quantum Walks}
\label{sec:Form of an arbitrary Quantum Walk Unitary}
It is not known whether all quantum walks have a decomposition into products of shift operators and coins, like the examples we saw in section \ref{sec:A General Recipe}.  So here we will write the evolution operator in as concise a form as possible.  First, we can write the evolution operator as
\begin{equation}
 U=
\displaystyle\sum_{\vec{n},
\vec{q}}A^{\vec{n}}_{\vec{q}}\ket{\vec{n}+\vec{q}}\bra{\vec{n}},
\end{equation}
where $A^{\vec{n}}_{\vec{q}}=\bra{\vec{n}+\vec{q}}U\ket{\vec{n}}$ is an operator on $\mathcal{H}_C$.  Now we impose translational invariance.  This means that $A^{\vec{n}}_{\vec{q}}$ does not actually depend on $\vec{n}$.  Let us write $A_{\vec{q}}=A^{\vec{n}}_{\vec{q}}$ and define $S_{\vec{q}}$ to be the operator that shifts a position state by $\vec{q}$.  Then we have that
\begin{equation}
 U=\displaystyle\sum_{\vec{q}}A_{\vec{q}}S_{\vec{q}}.
\end{equation}
We also require quantum walks to be causal, so that $A_{\vec{q}}$ is only nonzero for a finite set of vectors $\vec{q}$.  We denote this set by $Q$ and refer to it as the neighbourhood.

As mentioned before, it is interesting that an extra degree of freedom is necessary for these quantum walks to have nontrivial evolution.  Trivial in this context means that $U$ is just a shift operator times a phase \cite{Meyer96a}.  Very little is known about the general properties of quantum walk unitaries.  We revisit this in section \ref{sec:Quantum Walks on a Line: Abstract Theory}.

\subsection{Mass for Quantum Walks}
\label{sec:Mass for Quantum Walks}
Recall how we took the continuum limit of the quantum walk with unitary
\begin{equation}
\label{eq:1dmass}
 U=e^{-im\sigma_xa}e^{-iP\sigma_za}
\end{equation}
in section \ref{sec:Continuum Limits}.  Although the conditional shift $e^{-iP\sigma_za}$ is written in terms of the lattice spacing, the operator does not depend on $a$.  Having $a$ here is an artifact arising from our method of taking the continuum limit, as it is $P$ that is identified with the continuum momentum operator.  And the continuum limit really requires us to look at states that are smooth compared to the lattice.  On the other hand, the coin operator $e^{-im\sigma_xa}$ does depend on $a$ and tends to the identity as $a\rightarrow 0$.  With a trivial coin operator, meaning the identity operator, this artificial looking prescription is unnecessary.  If the coin is not the identity, however, we need a general scheme for ensuring that we can take the continuum limit.  To this end, we will define mass for quantum walks.

First, unitarity implies that
\begin{equation}
 U^{\dagger}U=\displaystyle\sum_{\vec{q},\vec{p}\, \in Q}A^{\dagger}_{\vec{q}}S^{\dagger}_{\vec{q}}A_{\vec{p}}S_{\vec{p}}=\openone.
\end{equation}
For this to be possible, terms like $S^{\dagger}_{\vec{q}}S_{\vec{p}}$ with $\vec{q}\neq\vec{p}$ must be zero.  This means we need
\begin{equation}
 \displaystyle\sum_{\substack{\vec{q},\vec{p}\, \in Q \\ \vec{q}\neq\vec{p}}}A^{\dagger}_{\vec{q}}A_{\vec{p}}=0\ \textrm{and}\ 
 \displaystyle\sum_{\vec{q}\in Q}A^{\dagger}_{\vec{q}}A_{\vec{q}}=\openone,
\end{equation}
which implies that $\sum_{\vec{q}}A_{\vec{q}}$ is a unitary operator on $\mathcal{H}_C$.  So we can write
\begin{equation}
 U=W\displaystyle\sum_{\vec{q}\in Q}A^{\prime}_{\vec{q}}S_{\vec{q}},
\end{equation}
where $W=\sum_{\vec{q}}A_{\vec{q}}$ is a unitary on $\mathcal{H}_C$ and $A^{\prime}_{\vec{q}} = W^{\dagger} A_{\vec{q}}$.  This means that $\sum_{\vec{q}}A^{\prime}_{\vec{q}}=\openone$.

Then, if $W=\openone$, we say that the particle is massless.  This is reminiscent of the example above in equation (\ref{eq:1dmass}).  There, if $m=0$, then $e^{-im\sigma_x a}=\openone$.  In general, to take a continuum limit in the massive case, we will let $W$ tend to $\openone$. 

In a way, massless quantum walks are more natural from the point of view of continuum limits.  This is because only the lattice spacing and the length of the timestep goes to zero, while the unitary on the lattice does not change.  This is not the case, however, for massive quantum walks, where the unitary changes as the lattice spacing goes to zero.\footnote{Nevertheless, doing this is necessary so that we recover the Dirac equation in the continuum limit.  That said, it is reassuring that fermions are fundamentally massless in the standard model, where they only acquire mass because of their interaction with the Higgs field.}

\subsection{The Continuum Limit}
\label{sec:The Continuum Limit}
Now we will take the continuum limit of these quantum walks.  First, it will be convenient to start with the case where the quantum walk is massless.  We will prove that in this case the continuum Hamiltonian is straightforward to write down.  This is one of the main results of \cite{FS14}.
\begin{theorem}
\label{th:13}
Given a massless quantum walk with evolution operator
\begin{equation}
 U=\displaystyle\sum_{\vec{q}\in Q}A_{\vec{q}}S_{\vec{q}},
\end{equation}
with $\textstyle{\sum_{\vec{q}}A_{\vec{q}}=\openone}$, then in the continuum limit this corresponds to evolution via the Hamiltonian
\begin{equation}
  H=\displaystyle \sum_{\vec{q}\in Q}\,\vec{q}.\vec{P}A_{\vec{q}}.
\end{equation}
\end{theorem}
\begin{proof}
We must evaluate
\begin{equation}
 \|e^{-iHt}\ket{\psi_{\Lambda}}-U^{N}\ket{\psi_{\Lambda}}\|_2 \leq \|e^{-iHt}-U^{N}\|_{\Lambda},
\end{equation}
where $\ket{\psi_{\Lambda}}\in\mathcal{H}_{\Lambda}$ and $\|\cdot\|_{\Lambda}$ is the operator norm on $\mathcal{H}_{\Lambda}$.  Here we will drop the momentum cutoff subscript $\Lambda$ to make the notation less cluttered.

Now we use the inequality for unitaries, $U$ and $V$, $\|U^N-V^N\|\leq N\|U-V\|$ from section \ref{sec:The Lie-Trotter Product Formula}.  It follows that
\begin{equation} 
 \|e^{-iHt}-U^N\|\leq N\|e^{-iHa}-U\|.
\end{equation}
To bound the right hand side, first write
\begin{equation}
  U=\displaystyle\sum_{\vec{q}\in Q}A_{\vec{q}}S_{\vec{q}}=
 \displaystyle\sum_{\vec{q}\in Q}A_{\vec{q}}\,e^{-i(\vec{q}.\vec{P})a},
\end{equation}
where $\vec{P}$ is the momentum vector operator.  By Taylor expanding both $e^{-iHa}$ and $U$, we see that for sufficiently small values of $\Lambda a $
\begin{equation}
\label{eq:otherbound}
 \|e^{-iHa}-U\|\leq C(\Lambda a)^2,
\end{equation}
where $C$ is a constant.  Let us go through this in more detail.  After Taylor expanding, $\|e^{-i H a}-U\|$ becomes
\begin{align}
 & \|\displaystyle\sum_{m\geq 2}\frac{(-i Ha)^m}{m!}-\displaystyle\sum_{\vec{q}\in Q}A_{\vec{q}}\displaystyle\sum_{l\geq 2}\frac{(-i\vec{q}.\vec{P}a)^l}{l!}\|\\
&\leq \displaystyle\sum_{m\geq 2}\frac{1}{m!}\|(-i H a)^m-\displaystyle\sum_{\vec{q}\in Q}A_{\vec{q}}(-i\vec{q}.\vec{P}a)^m\|\\
\label{eq:ab} &\leq \displaystyle\sum_{m\geq 2}\frac{a^m}{m!}\left(\|\displaystyle\sum_{\vec{q}\in Q}A_{\vec{q}}\,\vec{q}.\vec{P}\|^m+\|\displaystyle\sum_{\vec{q}\in Q}A_{\vec{q}}(\vec{q}.\vec{P})^m \|\right)\\
&\leq \displaystyle\sum_{m\geq 2}\frac{a^m}{m!}\left((\displaystyle\sum_{\vec{q}\in Q}\|A_{\vec{q}}\|\|\vec{q}.\vec{P}\|)^m+\displaystyle\sum_{\vec{q}\in Q}\|A_{\vec{q}}\|\|\vec{q}.\vec{P}\|^m\right)\\
\label{eq:cd} &\leq \displaystyle\sum_{m\geq 2}\frac{a^m}{m!}\left((Kq\Lambda)^m+K(q\Lambda)^m\right)\\
&\leq 2 \displaystyle\sum_{m\geq 2}\frac{(Kq\Lambda a)^m}{m!}\\
&\leq C(\Lambda a)^2, 
\end{align}
where $K$ is the number of $A_{\vec{q}}\neq 0$ and $q$ is the largest value of $|\vec{q}|$ for which $A_{\vec{q}}\neq 0$.  The fifth line follows from $\|A_{\vec{q}}\|\leq 1$, which in turn follows from $\sum_{\vec{q}} A_{\vec{q}}^{\dagger}A_{\vec{q}}=\openone$.  The last line applies when $Kq\Lambda a \leq 1$ and relies on the fact that, when $\alpha \leq 1$,
\begin{equation}
 \sum_{m\geq 2} \frac{\alpha^m}{m!} \leq \alpha^2\sum_{m\geq 2} \frac{1}{m!} =(e-2) \alpha^2=C^{\prime}\alpha^2.
\end{equation}

Later, it will be useful in the proof of corollary \ref{cor:12} that in going from line (\ref{eq:ab}) to line (\ref{eq:cd}) we bounded $\|H\|$ from above by $Kq\Lambda$.

The end result is that
\begin{equation}
\label{eq:bound}
 \|e^{-iHt}\ket{\psi_{\Lambda}}-U^{N}\ket{\psi_{\Lambda}}\|_2\leq Ct\Lambda^2 a.
\end{equation}
To get the continuum limit, we let $a\rightarrow 0$.  It makes sense to take $\Lambda \rightarrow \infty$ at a slower rate than $a \rightarrow 0$ so that $\Lambda^2 a \rightarrow 0$. Then the right-hand side of equation (\ref{eq:bound}) tends to zero and the momentum cut-off tends to infinity, which tells us that the quantum walk converges to evolution via the continuum Hamiltonian $H$.
\end{proof}

Let us turn to quantum walks with mass.  Recall that the quantum walk's unitary is
\begin{equation}
 U=W\displaystyle\sum_{\vec{q}\in Q}A^{\prime}_{\vec{q}}S_{\vec{q}},
\end{equation}
where $W$ is a unitary acting on $\mathcal{H}_C$ and $\textstyle\sum_{\vec{q}}A^{\prime}_{\vec{q}}=\openone$.  As we mentioned in the previous section, one way to get a continuum limit is to let $W$ tend to the identity as $a\rightarrow 0$ with
\begin{equation}
 W=e^{-iMa},
\end{equation}
where $M$ is a fixed self-adjoint operator on $\mathcal{H}_C$.  This leads to the corollary below.
\begin{corollary}
\label{cor:12}
 The continuum Hamiltonian for a massive quantum walk, where $W=e^{-iMa}$, is
\begin{equation}
 H=\displaystyle\sum_{\vec{q}}A_{\vec{q}}^{\prime}\,(\vec{q}.\vec{P})+M.
\end{equation}
\end{corollary}
\begin{proof}
It follows from the triangle inequality that
\begin{equation}
\|e^{-iHa}-U\| \leq \|e^{-iHa}-e^{-iMa}e^{-iH_0a}\|
+ \|e^{-iMa}e^{-iH_0 a}-e^{-iM a}U_0\|,
\end{equation}
where $U_0=W^{-1}U$ is a massless quantum walk operator and $H_0=H-M$.  Again, we have dropped the subscript $\Lambda$ on the operator norm, which is restricted to $\mathcal{H}_{\Lambda}$.  So $U_0$ is a massless quantum walk unitary with corresponding continuum Hamiltonian $H_0=H-M$

The second term above is equal to $\|e^{-iH_0 a}-U_0\|$ because the operator norm is unitarily invariant.  This expression was already bounded by $C(\Lambda a)^2$ in the proof of theorem \ref{th:13}.  Bounding the first term is also straightforward because we know that $H=H_0+M$.  Then we can use 
\begin{equation}
 \|e^{-iHa}-e^{-iMa}e^{-iH_0a}\|\leq C^{\prime}B^2a^2,
\end{equation}
where $C^{\prime}$ is a constant and $B=\max\{\|M\|,\|H_0\|\}$. Next we can use the result, noted in the proof of theorem \ref{th:13}, that $\|H_0\|\leq Kq\Lambda$ to get
\begin{equation}
\|e^{-iHa}-e^{-iMa}e^{-iH_0a}\|\leq C^{\prime\prime}(\Lambda a)^2
\end{equation}
since $\|M\|\leq\|H_0\|$ as $\Lambda\rightarrow \infty$.
Finally, putting this all together, we have
\begin{equation}
\|e^{-iHa}-U\| \leq (C+C^{\prime\prime})(\Lambda a)^2,
\end{equation}
which means that
\begin{equation}
 \|e^{-iHt}-U^{N}\|\leq  N\|e^{-iHa}-U\|\leq (C+C^{\prime\prime})t\Lambda^2 a.
\end{equation}
Therefore we see convergence of the quantum walk to evolution via the Hamiltonian $H$ as the lattice spacing goes to zero.
\end{proof}

\subsection{The Continuum Hamiltonian}
\label{sec:The Continuum Hamiltonian}
Now that we have a formula for the continuum Hamiltonian, let us see what it means for the dynamics.  Again, we can start with the massless case following \cite{FS14}.

\begin{theorem}
\label{th:14}
 After a change of coordinates (possibly rescaling and rotating the coordinate axes and removing a constant velocity shift), the continuum Hamiltonian of a massless quantum walk with a two dimensional coin is massless and relativistic.
\end{theorem}
\begin{proof}
We start by focusing on the case where the quantum walk takes place in three dimensional space.  At the end, we will look at other possibilities.

The continuum Hamiltonian is
\begin{equation}
 H=\sum_{\vec{q}\in Q}A_{\vec{q}}\,\vec{q}.\vec{P}.
\end{equation}
Because $A_{\vec{q}}\,q_i$ is an operator on a two dimensional Hilbert space, we can expand these terms out in the basis $\{\openone,\sigma_i\}$ to get
\begin{equation}
 H=\sigma_1\tilde{P}_{1}+\sigma_2\tilde{P}_{2}+\sigma_3 \tilde{P}_{3}+\openone \tilde{P}_4,
\end{equation}
where $\tilde{P}_{i}$ are real linear combinations of components of the momentum vector operator $\vec{P}$.

First, the overall shift term $\openone \tilde{P}_4$ is meaningless from a physical point of view.  By changing to coordinates that are moving with a constant velocity, we can remove this term.

It is important to notice that the remaining $\tilde{P}_i$ may not be momentum operators in orthogonal directions.  To deal with this, we can think of 
$\sigma_1 \tilde{P}_{1}+\sigma_2 \tilde{P}_{2}+\sigma_3 \tilde{P}_{3}$ 
as a sum of tensor products of vectors in real vector spaces.  Indeed, $\sigma_i$ are an orthonormal basis for the vector space of $2\times 2$ traceless self-adjoint matrices with the inner product $(A,B)=\f{1}{2}\tr{A^{\dagger}B}$.  And $\tilde{P}_j$ span a subspace of the real vector space that has $P_i$ as its orthonormal basis.  This has the inner product $(A,B)=\textstyle\sum_ia_ib_i$, where $A=\textstyle\sum_ia_iP_i$ and $B=\textstyle\sum_ib_iP_i$.

This allows us to use singular value decomposition to rewrite $H$.  Let us go through how this works.
If $A$ is a $n\times n$ real matrix, then it may be written in the form \cite{Horn85}
\begin{equation}
 A=VBW^{\dagger},
\end{equation}
where $V$ and $W$ are $n\times n$ rotation matrices, and $B$ is a diagonal matrix with non-negative entries.  Next, suppose we have the real vector in a tensor product space $\sum_{ij}A_{ij}\vec{v}_i\otimes\vec{w}_j$, where $\vec{v}_i$ and $\vec{w}_j$ are orthonormal bases for $\mathcal{H}_1$ and $\mathcal{H}_2$ respectively.  By applying singular value decomposition, we can rewrite this as $\sum_{ij}B_{ij}\vec{e}_i\otimes\vec{f}_j$, where $\vec{e}_i$ and $\vec{f}_j$ are orthonormal bases for $\mathcal{H}_1$ and $\mathcal{H}_2$, and $B_{ij}$ is a diagonal matrix with positive entries.\footnote{This is essentially equivalent to the proof of the Schmidt decomposition, except we need to use the fact that singular value decomposition can be restricted to real matrices.  
This is essential so that the resulting operators have the right physical interpretation.}

So, applying this to the Hamiltonian, we get
\begin{equation}
 H=\gamma_1\sigma_{1}^{\prime}P_{1}^{\prime}+\gamma_2\sigma_{2}^{\prime}P_{2}^{\prime}+\gamma_3\sigma_{3}^{\prime}P_{3}^{\prime},
\end{equation}
where $\gamma_i$ are real numbers, $P_{i}^{\prime}$ are momentum operators along orthogonal spatial axes and $\sigma_{i}^{\prime}$ are an orthonormal basis of the vector space of $2\times 2$ traceless self-adjoint matrices.  Note that such a basis is unitarily equivalent to either $\sigma_i$ or $-\sigma_i$.

When all the $\gamma_i$ are non-zero, the spatial axes can be rescaled so that $\gamma_i P^{\prime}_i\to P^{\prime}_i$, and we can drop the primes.  Then we apply a unitary change of basis to the coin to arrive at either $\sigma_i$ or $-\sigma_i$.  So, finally, we get
\begin{equation}
H=\pm\left(\sigma_{1}P_{1}+\sigma_{2}P_{2}+\sigma_{3}P_{3}\right)\equiv\pm\vec{\sigma}.\vec{P}.
\end{equation}

When some of the $\gamma_i$ are zero, the Hamiltonian corresponds to either $P\sigma_z$ or $P_x\sigma_x+P_y\sigma_y$.  This means that all massless quantum walks with a two dimensional coin obey massless relativistic equations of motion in the continuum limit.

For quantum walks in fewer than three spatial dimensions, the results still apply, but the particle evolves according to a lower dimensional equation of motion.  For quantum walks in more than three spatial dimensions, the particle still evolves according to a massless relativistic equation restricted to at most three dimensions.  This means that it does not move along some spatial directions.
\end{proof}

\begin{compactwrapfigure}{r}{0.43\textwidth}
\centering
\begin{minipage}[r]{0.40\columnwidth}%
\centering
    \resizebox{6.0cm}{!}{\begin{picture}(0,0)%
\includegraphics{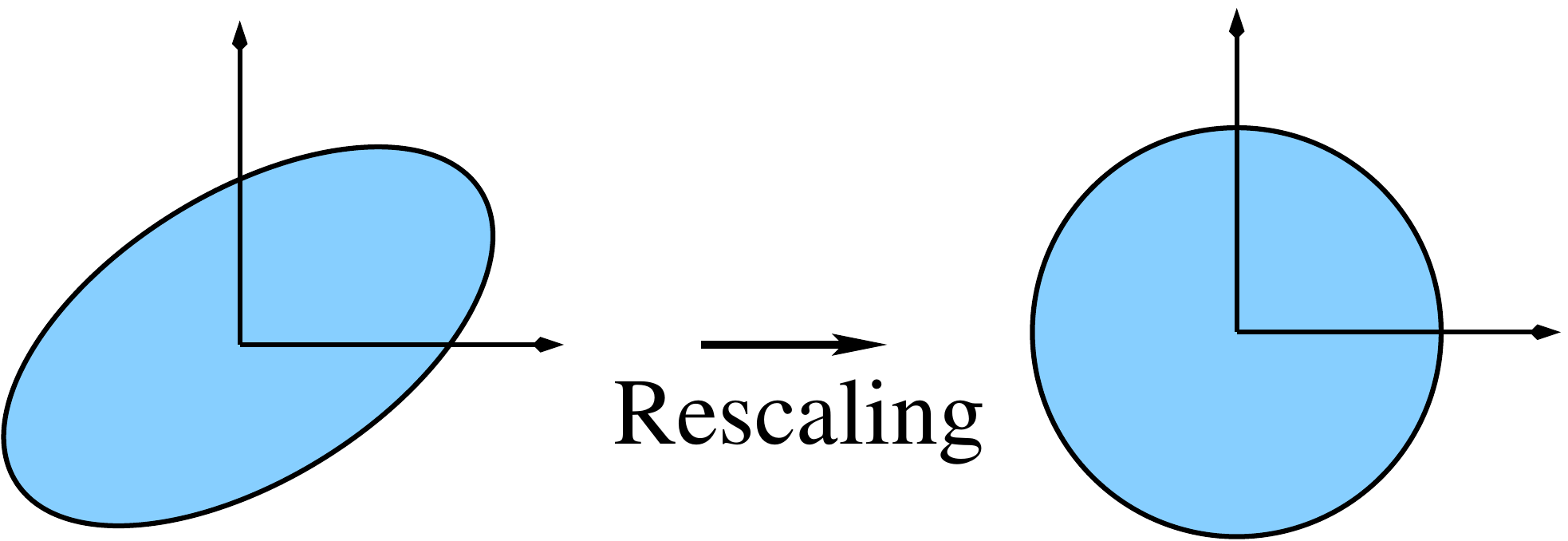}%
\end{picture}%
\setlength{\unitlength}{3947sp}%
\begingroup\makeatletter\ifx\SetFigFont\undefined%
\gdef\SetFigFont#1#2#3#4#5{%
  \reset@font\fontsize{#1}{#2pt}%
  \fontfamily{#3}\fontseries{#4}\fontshape{#5}%
  \selectfont}%
\fi\endgroup%
\begin{picture}(9427,3235)(582,-4313)
\end{picture}%
}
    \footnotesize{\caption[Rescaling coordinates to recover rotational symmetry]{Rescaling coordinates to recover rotational symmetry.}}
    \label{fig:73}
\end{minipage}
\end{compactwrapfigure}

In this proof we had to do a change of coordinates.  The only way this transformation would be physically nontrivial would be if there were different particles present with continuum limits that could not be made into the same form by the {\it same} changes of coordinates.

There is also a simple argument that, if a massless quantum walk with a two dimensional coin has the rotational symmetries of the lattice in the continuum limit, then it also has Lorentz symmetry.

Let us suppose that the continuum Hamiltonian has the rotational symmetry of the lattice.  The Hamiltonian is
\begin{equation}
 H=\vec{B}.\vec{P},
\end{equation}
where $\vec{B}=\sum_{\vec{q}}A_{\vec{q}}\,\vec{q}$.  Since $B_i$ are self-adjoint, we can write $B_i=c_i\openone+\vec{n}_i.\vec{\sigma}$, where $c_i$ and $\vec{n}_i$ are real.  Assuming that this Hamiltonian has the rotational symmetries of the lattice means there is a subgroup $G$ of $SU(2)$ whose action on $\{B_i:i=1,2,3\}$ forms a representation of these rotations.  Then, for a three dimensional lattice and some given $i$ and $j \neq i$ there must be a $V\in G$ satisfying $VB_iV^{\dagger}=-B_i$ and $VB_jV^{\dagger}=B_j$.  This means that $c_i=0$, and further that $\textrm{tr}[B_i^{\dagger}B_j]=0$.  This implies that $\vec{n_i}.\vec{\sigma}$ are an orthogonal set.  Also, for any $i$ and $j$ there must be a $V\in G$ with $VB_iV^{\dagger}=B_j$, meaning $|\vec{n_i}|=|\vec{n_j}|$.  Therefore, $B_i$ are proportional to a representation of $\sigma_i$ or $-\sigma_i$.  So the Hamiltonian will be proportional to either the left or right-handed Weyl Hamiltonian.

\subsubsection{Continuum Hamiltonians with Mass}
Let us see what changes if the quantum walks have mass.  We can transform coordinates as in the massless case so that the Hamiltonian becomes
\begin{equation}
 H=\vec{\sigma}.\vec{P}+M,
\end{equation}
where $M$ is a self-adjoint operator on $\mathcal{H}_C$.

In one dimensional space, for example, with $M=m\sigma_x$, the result is the Dirac Hamiltonian in one dimension:
\begin{equation}
 H=\sigma_z P_z+m\sigma_x.
\end{equation}

It is not the case, however, that we get relativistic behaviour in the continuum in general.  As an example, we could have $M=m_1\sigma_z+m_2\sigma_x$, which does not correspond to relativistic evolution, as observed in \cite{Kurz08}.  Nevertheless, by requiring that the continuum Hamiltonian should have symmetry under a parity transformation, this Hamiltonian would be ruled out.

\subsection{Higher Dimensional Coins are not so Special}
\label{sec:Counterexample with a Three Dimensional Coin}
Results often hold in two dimensional Hilbert spaces that do not generalize to higher dimensions, which is sometimes due to the natural orthonormal operator basis, the Pauli matrices and the identity, having such nice properties.  Here we will see that the results of the previous section do not generalize to higher dimensional coins.  To see this, we will reproduce the simple counterexample from \cite{FS14} with a three dimensional coin that has rotational symmetry in the continuum limit but not Lorentz symmetry.

We will construct a quantum walk that corresponds to the Hamiltonian below in the continuum limit.
\begin{equation}
 H=\vec{J}.\vec{P},
\end{equation}
where $J_i=-i \sum_{jk} \varepsilon_{ijk}\ket{j}\bra{k}$ form a three dimensional representation of the generators of the lie algebra of SO(3) acting on $\mathcal{H}_C$, which has orthonormal basis $\ket{1}$, $\ket{2}$ and $\ket{3}$.  This Hamilonian does have rotational symmetry, but it does not have Lorentz symmetry 

By using the recipe from section \ref{sec:A General Recipe}, a quantum walk that has the Hamiltonian above in the continuum limit can be written as a product of conditional shifts:
\begin{equation}
U=T_xT_yT_z,
\end{equation}
but now with
\begin{equation}
T_b=\exp(-iaP_bJ_b),
\end{equation}
where $P_b$ is the momentum operator in the $b$ direction, with $b\in\{x,y,z\}$.  Here we have relabelled $J_i$ by $x$, $y$ and $z$ according to the usual convention:\ $J_1=J_x$, $J_2=J_y$, and $J_3=J_z$.

Written in terms of shift operators, we can rewrite $T_b$ as
\begin{equation}
 T_{b}=S_{b}\ket{\!+\!1_b}\bra{+1_b}+\ket{0_b}\bra{0_b}+S^{\dagger}_{b}\ket{\!-\!1_b}\bra{-1_b},
\end{equation}
where $\ket{\lambda_b}$ is the eigenvector of $J_b$ with eigenvalue $\lambda$ and $S_{b}$ is a shift by one lattice site in the $b$ direction.

Let us see why the Hamiltonian above is not relativistic.  If this were the case, $H^2-\vec{P}^2$ should be a Lorentz invariant quantity.  To check, we can look at how it transforms under a boost.  First, it follows from the commutation relations for $J_i$ that
\begin{equation}
 H^2-\vec{P}^2= -\displaystyle\sum_{ij}P_iP_j\ket{j}\bra{i}.
\end{equation}
Then,
\begin{equation}
\displaystyle\sum_{i}\bra{i}H^2-\vec{P}^2\ket{i}=-\vec{P}^2.
\end{equation}
It is necessary for Lorentz invariance that
\begin{equation}
 \displaystyle\sum_{i}\bra{i}U_{\Lambda}(H^2-\vec{P}^2)U^{\dagger}_{\Lambda}\ket{i}=\displaystyle\sum_{i}\bra{i}H^2-\vec{P}^2\ket{i},
\end{equation}
where $U_{\Lambda}$ is the unitary operator implementing a Lorentz transformation $\Lambda$.  As we are talking about free particles, the effect of a Lorentz transformation is \cite{Weinberg95}
\begin{equation}
 U_{\Lambda}\ket{\vec{p}\,}\ket{k}=\sqrt{\frac{E_{\Lambda\vec{p}}}{E_{\vec{p}}}}\ket{\vec{\Lambda p}}D(\Lambda,\vec{p}\,)\ket{k},
\end{equation}
where $D(\Lambda,\vec{p}\,)$ is a unitary on the internal degree of freedom that depends on $\Lambda$ and $\vec{p}$.  Using the invariance of the trace under unitary transformations, which in this case are implemented by $D(\Lambda,\vec{p}\,)$, we get
\begin{equation}
 \sum_{i}\bra{i}U_{\Lambda}(H^2-\vec{P}^2)U^{\dagger}_{\Lambda}\ket{i}=U_{\Lambda}(-\vec{P}^2)U^{\dagger}_{\Lambda},
\end{equation}
which is not equal to $\textstyle\sum_{i}\bra{i}H^2-\vec{P}^2\ket{i}=-\vec{P}^2$ if $\Lambda$ is a boost.  Therefore, $H^2-\vec{P}^2$ is not Lorentz invariant, so the dynamics are not relativistic.

Note that we cannot add a term like $\openone \tilde{P}_4$ to $H$ as we did in section \ref{sec:The Continuum Hamiltonian} because this would break rotational symmetry, as would rescaling coordinate axes.  Additionally, any rotation of the coordinate axes or a unitary change of basis applied to the extra degrees of freedom would not change the fact that $H^2-\vec{P}^2$ is not Lorentz invariant.
\vfill

\section{Fermion Doubling in Quantum Walks}
\label{sec:Fermion Doubling in Quantum Walks}
In this section, we study the question of whether quantum simulators offer additional advantages over classical simulators for simulating quantum fields on a lattice.  In particular, we look at the problem of fermion doubling, which is guaranteed to plague lattice models with local chiral fermion Hamiltonians \cite{NN81,Kaplan09}.  We see that discrete quantum simulators may provide one way around this problem.  This is possible because the systems we discuss evolve in discrete spacetime via a unitary operator $U$ that is causal.  So a Hamiltonian $H$, defined by $e^{-iH}=U$, would be nonlocal.  In spite of this, in the continuum limit we recover relativistic fermions.  Furthermore, the systems we will discuss can be implemented by very simple quantum circuits.  

Of course, there are many proposed ways of dealing with the fermion doubling problem:\ Wilson fermions, which break chiral symmetry by introducing a momentum dependent mass; staggered fermions, which reduce the number of doubler fermions, again breaking chiral symmetry; as well as methods that redefine what chiral symmetry means on the lattice in a way that corresponds to the usual definition in the continuum limit (domain wall fermions, for example) \cite{DGDT06}.  The potential solution we discuss is conceptually simple and practical from the point of view of quantum computing.  Still, the ideas need further study if they are to be of any use.

Fermion doubling is only a problem in interacting models:\ if there are no interactions, then low energy low momentum states do not evolve to low energy high momentum states (the doublers).

With this in mind, we begin by introducing gauge interactions for quantum walks in section \ref{sec:Gauge Interactions for Quantum Walks}.  Then we define what fermion doubling means in the context of quantum walks in section \ref{sec:Fermion Doubling in Discrete Time}, where we also see that it is not a problem for the one dimensional case.  In sections \ref{sec:Fermion Doubling in Two Dimensions} and \ref{sec:Fermion Doubling in Three Dimensions} we look at fermion doubling in the cases of two and three dimensional space respectively.

We look at fermion doubling in the context of quantum walks, which are single particle systems.  These can be upgraded to many particle systems and discrete field theories, as we will see in the next chapter.  Because of this, results about fermion doubling naturally apply to these, but it is easier to work with spectra in this single particle setting.  This is mainly to reduce clutter in notation.

\subsection{Gauge Interactions for Quantum Walks}
\label{sec:Gauge Interactions for Quantum Walks}
It is important to realize that the existence of doubler modes is only a problem in the presence of interactions.  So let us introduce gauge interactions into this picture.  For simplicity, we will look at $U(1)$ gauge fields.  The approach we take here is equivalent to one suggested in \cite{Bial94}.  The idea is that the gauge interaction is manifested as a phase picked up by the particle as it moves from one point to the next.

How we introduce these gauge interactions is essentially equivalent to working in the temporal gauge, the choice usually employed in the Hamiltonian formulation of lattice gauge theory \cite{KS75}.  There the gauge is chosen so that the $\mu=0$ component of the electromagnetic field $A_\mu(x)$ is zero.  A further constraint is that one needs to restrict the Hilbert space to those states satisfying Gauss' law.

Now, suppose that we have a quantum walk with no extra degree of freedom.  And every timestep it evolves by moving one step to the right via $U=e^{-iPa}$.  This is not a gauge invariant evolution because a shift does not commute with applying a site dependent phase (a gauge transformation) $\ket{n}\rightarrow e^{-i\lambda_n}\ket{n}$.

We can fix this by introducing the operator $A=\sum_nA_n\ket{n}\bra{n}$ that under a gauge transformation transforms like $A_n a\rightarrow A_n a +(\lambda_{n+1}-\lambda_{n})$.  Because of this, it is natural to associate $A_n$ with the link between sites $n$ and $n+1$.  Then the evolution operator
\begin{equation}
 U_A=e^{-iPa}e^{-iAa}
\end{equation}
commutes with gauge transformations.  For small $a$, we have $e^{-iPa}e^{-iAa}= e^{-i(P+A)a}+O(a^2)$.  So it is $A_n$ that should become the gauge field $A(x)$ in the continuum limit.

\begin{compactwrapfigure}{r}{0.43\textwidth}
\centering
\begin{minipage}[r]{0.40\columnwidth}%
\centering
    \resizebox{6.0cm}{!}{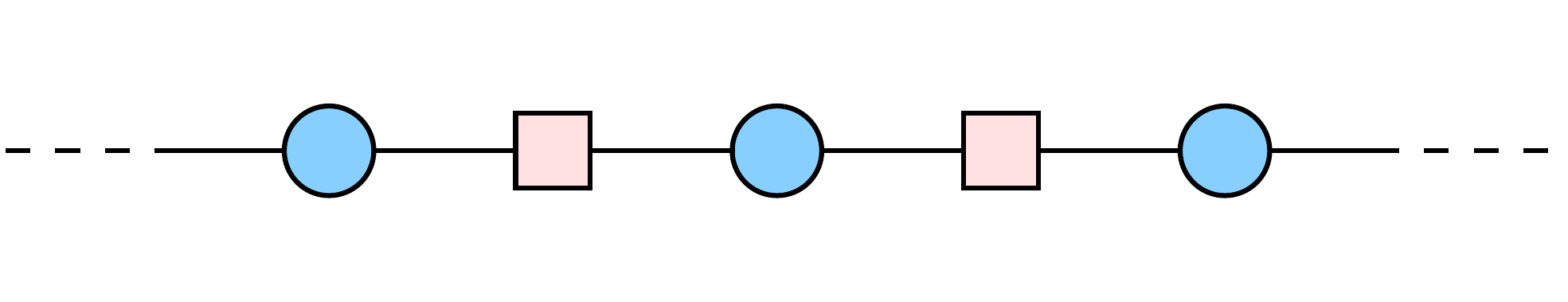}
    \footnotesize{\caption[Gauge fields for quantum walks on a line]{Gauge fields are naturally associated with links between sites.}}
    \label{fig:89}
\end{minipage}
\end{compactwrapfigure}

A quantum walk evolving via $e^{iPa}=(e^{-iPa})^{\dagger}$ is also not gauge invariant, but can be made so in a similar way, with the gauge invariant evolution $U_A^{\dagger}=e^{iAa}e^{iPa}$.  The gauge invariant evolution for a quantum walk giving the massless Dirac equation in one dimension in the continuum limit is
\begin{equation}
\begin{split}
 U & =U_A\ket{r}\bra{r}+U_A^{\dagger}\ket{l}\bra{l}\\
 & = e^{-iPa}e^{-iAa}\ket{r}\bra{r}+e^{iAa}e^{iPa}\ket{l}\bra{l}\\
 & = \openone -iP\sigma_za -iA\sigma_z a +O(a^2)\\
 & = e^{-i(P+A)\sigma_z a} +O(a^2).
\end{split}
\end{equation}
So in the continuum limit we expect that we recover a massless Dirac particle interacting with the electromagnetic field.  Of course, we have not been rigorous here.  Some smoothness assumptions are probably necessary in taking such limits.  We will return to this point in chapter \ref{chap:Conclusions and Open Problems 1}.

This works similarly in higher dimensions.  As before, to make a shift $S_b=e^{-iP_ba}$ gauge invariant, we introduce a gauge field, which now takes the form $A_b=\sum_{\vec{n}}A_b(\vec{n})\ket{\vec{n}}\bra{\vec{n}}$, where $b\in\{x,y,z\}$.  Under a gauge transformation, this transforms as
\begin{equation}
 A_b(\vec{n})a\rightarrow A_b(\vec{n})a+(\lambda(\vec{n}+\vec{e}_b)-\lambda(\vec{n})),
\end{equation}
where $\vec{e}_b$ is the lattice vector pointing in the positive $b$ direction.  Then the operator
\begin{equation}
 S_be^{-iA_ba}=e^{-iP_ba}e^{-iA_ba}
\end{equation}
\begin{compactwrapfigure}{r}{0.43\textwidth}
\centering
\begin{minipage}[r]{0.40\columnwidth}%
\centering
    \resizebox{3.6cm}{!}{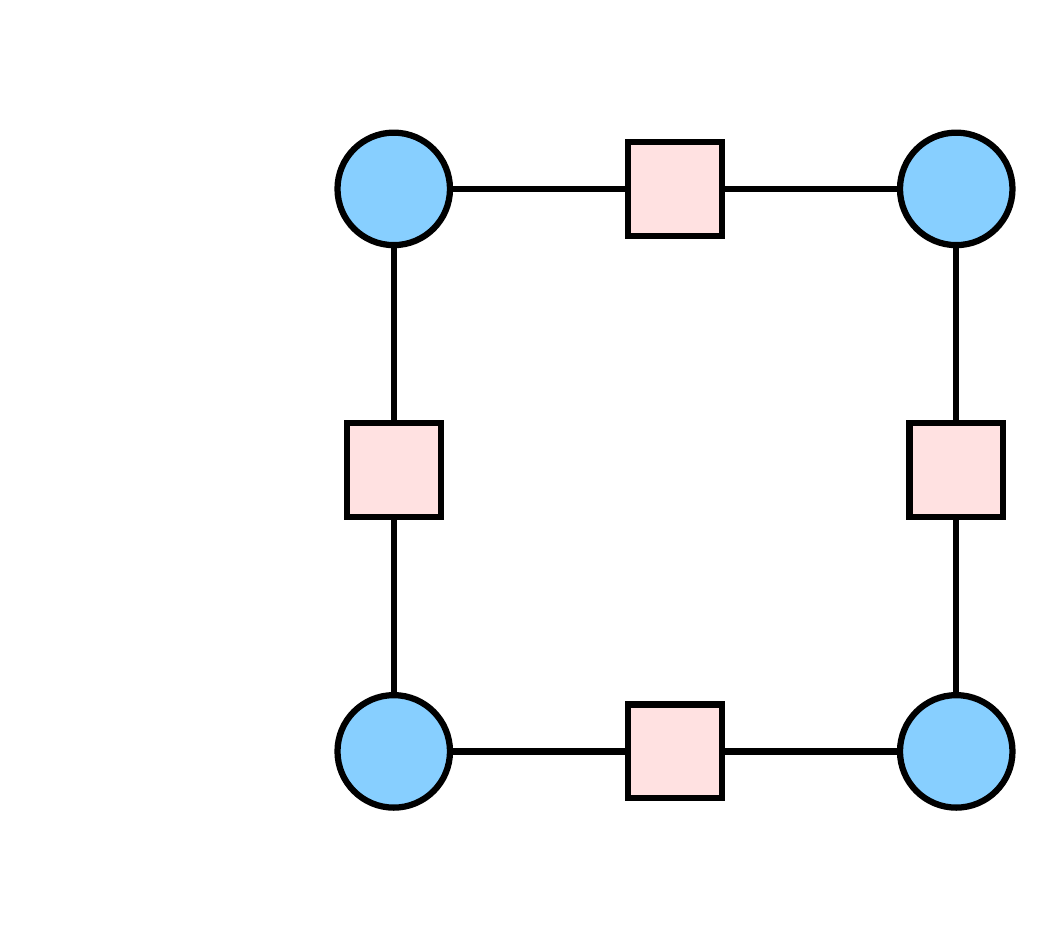}
    \footnotesize{\caption[Gauge fields for quantum walks in two dimensions]{Different components of the gauge field are naturally associated with links between sites.}}
    \label{fig:72}
\end{minipage}
\end{compactwrapfigure}
commutes with gauge transformations.  So, for example, the unitary $T^{\prime}_xT^{\prime}_yT^{\prime}_z$ is gauge invariant, where
\begin{equation}
 T^{\prime}_b=e^{-iP_ba}e^{-iA_ba}\ket{\uparrow_b}\bra{\uparrow_b}+e^{iA_ba}e^{iP_ba}\ket{\downarrow_b}\bra{\downarrow_b},
\end{equation}
commutes with gauge transformations.  Again, the effect of the gauge interaction is that, whenever a particle travels along a link, it picks up a phase associated to that link.

\subsection{Fermion Doubling in Discrete Time}
\label{sec:Fermion Doubling in Discrete Time}
To understand how fermion doubling can affect discrete-time models, it helps to think of them as approximations to continuum theories.  In other words, we need to keep in mind that we intend to take the continuum limit.  As we shrink the lattice spacing $a$ to zero, the discrete-time unitary evolution should approach the continuum Hamiltonian evolution.  We say that doubling occurs if modes other than those with low momentum survive the continuum limit.  We will see that some discrete-time models that do not satisfy the hypotheses of the fermion doubling theorem still have doubler fermions.  With a little work, we can alter these models to alleviate this problem.

First, take the familiar one dimensional quantum walk, with evolution operator
\begin{equation}
 U=e^{-im\sigma_xa}e^{-iP\sigma_za}.
\end{equation}
Now, as mentioned in \cite{DdV87}, this avoids the fermion doubling problem.  It evades the consequences of the fermion doubling theorem because $H$, defined by $U=e^{-iHa}$ is nonlocal.

It is instructive to go into the details a little.  For this model, when $m=0$, there is chiral symmetry, and the evolution operator becomes just
\begin{equation}
U=e^{-iP\sigma_za}.
\end{equation}
\begin{compactwrapfigure}{r}{0.49\textwidth}
\centering
\begin{minipage}[r]{0.46\columnwidth}%
\centering
    \resizebox{6.0cm}{!}{\includegraphics{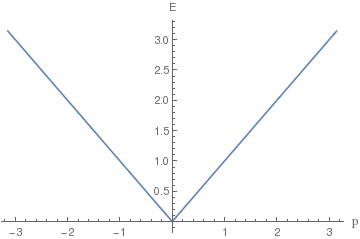}}
    \footnotesize{\caption[Energies of the Dirac quantum walk in one dimension]{Positive energies of the Dirac quantum walk in one dimension.  The $x$- and $y$-axes correspond to the momentum and energy respectively.  (After looking at the corresponding field theory, negative energy particle states are filled up, and annihilating a negative energy particle is interpreted as creating an antiparticle.)}}
    \label{fig:9}
\end{minipage}
\end{compactwrapfigure}	
So a natural candidate for the corresponding Hamiltonian is $H=P\sigma_z$.  There is some ambiguity in choosing a Hamiltonian because we are in a discrete-time picture.  That said, because high momentum ($p\simeq\f{\pi}{a}$) states are eigenvectors of $U$ with a large phase, whatever way we define a Hamiltonian, these states have high energy.  Furthermore, because these high momenta pick up large phases when we apply $U$, they do not survive the continuum limit.  So there is no doubling.

\subsection{Fermion Doubling in Two Dimensions}
\label{sec:Fermion Doubling in Two Dimensions}
Recall that the quantum walk that gives rise to the massless Dirac equation in two dimensional space is
\begin{equation}
 U=e^{-i\sigma_x P_x a}e^{-i\sigma_y P_y a}.
\end{equation}
Let us see whether there are doublers in the sense we defined in section \ref{sec:Fermion Doubling in Discrete Time}.  In other words, are there high momentum particles that remain in the continuum limit?  The answer is actually yes, even though the Hamiltonian corresponding to the quantum walk is nonlocal \cite{Short13}.  Recall also that there is no chiral symmetry in two dimensions, but that does not rule out fermion doubling.

To see that there is a problem, let us work in momentum space.  Naturally, the momentum state $\ket{0,0}$ is an eigenstate with eigenvalue one.  But there is another.  The state with momentum $\ket{\f{\pi}{a},\f{\pi}{a}}$ is also an eigenvector with eigenvalue one.  Note that the state of the coin is irrelevant, meaning, for any coin state $\ket{\phi}$, both $\ket{0,0}\ket{\phi}$ and $\ket{\f{\pi}{a},\f{\pi}{a}}\ket{\phi}$ are plus one eigenstates of $U$.

Now let $\vec{p}_1=(\f{\pi}{a},\f{\pi}{a})$ and $\vec{p}$ be a momentum vector with components $|p_b|\ll \f{\pi}{a}$.  Then we have
\begin{equation}
\begin{split}
 U\ket{\vec{p}+\vec{p}_1} & =e^{-i\sigma_x P_xa}e^{-i\sigma_y P_ya}\ket{\vec{p}+\vec{p}_1}\\
 & =e^{-i\sigma_x p_xa}e^{-i\sigma_y p_ya}\ket{\vec{p}+\vec{p}_1}\\
 & =(1-i(\sigma_x p_x+\sigma_y p_y)a+O(a^2))\ket{\vec{p}+\vec{p}_1}.
\end{split}
\end{equation}
So for small $\vec{p}$ the state $\ket{\vec{p}+\vec{p}_1}$ behaves like a massless Dirac particle too.  Therefore, states in the space spanned by $\ket{\vec{p}+\vec{p}_1}$ with $|\vec{p}|\ll\f{\pi}{a}$ act like massless Dirac particles.
\begin{compactwrapfigure}{r}{0.47\textwidth}
\centering
\begin{minipage}[r]{0.43\columnwidth}%
\centering
    \resizebox{6.0cm}{!}{\includegraphics{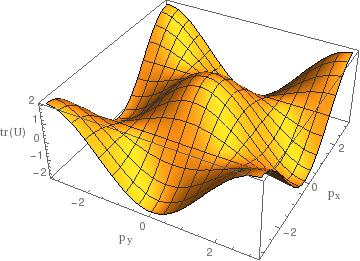}}
    \footnotesize{\caption[Fermion doubling for a quantum walk in two dimensional space]{Trace of the massless Dirac quantum walk in two dimensional space.  The $x$- and $y$-axes correspond to the $x$ and $y$ components of the particle's momentum.  The peaks at $p_x=p_y=0$ and $p_x=p_y=\pm\f{\pi}{a}$ correspond to low energies.}}
    \label{fig:8}
\end{minipage}
\end{compactwrapfigure}
Let us check that there is only one set of doublers.  In other words, we want to see when $U$ is close to the identity.  If $U=\openone$, then its trace is two.  Using
\begin{equation}
 e^{-i\theta\sigma_b}=\cos(\theta)\openone-i\sin(\theta)\sigma_b,
\end{equation}
we see that the trace of $U$ is $2\cos(p_xa)\cos(p_ya)$.  Then, given that $p_b\in(-\f{\pi}{a},\f{\pi}{a}]$, there are only two points where the trace is two.  The first is when $p_x=p_y=0$ and the second is when $p_x=p_y=\f{\pi}{a}$.

\subsubsection{Removing the Doubling Modes}
Fortunately, there is a way to remove these doubler fermions.  This works by restricting ourselves to initial states where the particle only occupies sites on a body centred cubic (b.c.c.) sublattice of the original cubic lattice.\footnote{Of course, a b.c.c.\ lattice in two dimensions is equivalent to a cubic lattice.  And the arguments in this section could be made easier by using this fact, but it will help to understand the following section to proceed in the way we do.  Part of the reason is that, in the next section, we look at a b.c.c.\ sublattice of a cubic lattice in three spatial dimensions, which is {\em not} equivalent to a cubic lattice.}  We can choose the sublattice that includes the origin $(0,0)$.  Any other site on the sublattice can be reached by adding multiples of $(e_1,e_2)$, where $e_i\in\{-1,1\}$.  Therefore, we know that the non-interacting evolution operator keeps us in this subspace because the unitary is a product of conditional shifts in the $x$ and $y$ direction.  So, if initially localized on this sublattice, a free particle does not move off it.

\begin{compactwrapfigure}{r}{0.37\textwidth}
\centering
\begin{minipage}[r]{0.34\columnwidth}%
\centering
    \resizebox{4.0cm}{!}{\begin{picture}(0,0)%
\includegraphics{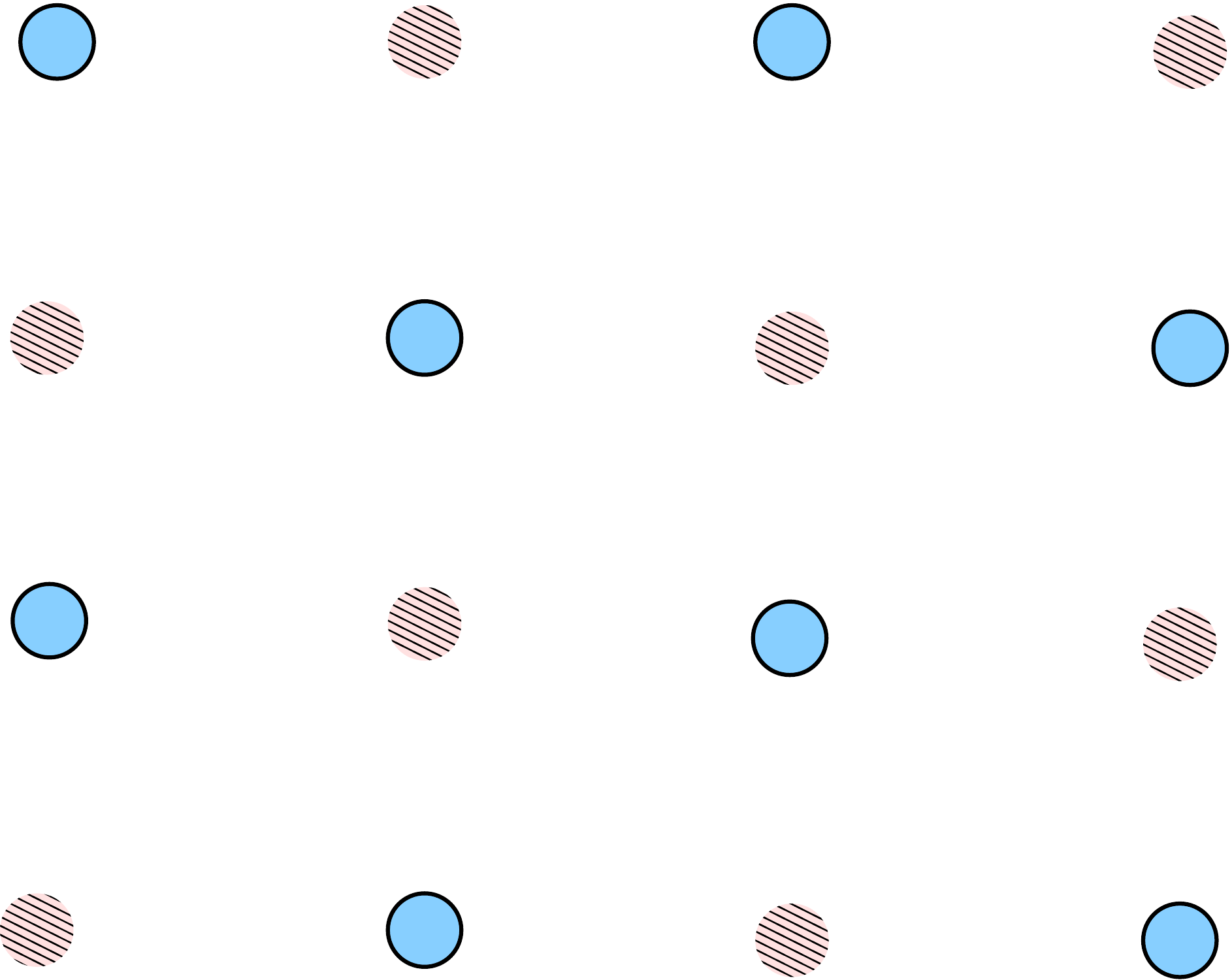}%
\end{picture}%
\setlength{\unitlength}{3947sp}%
\begingroup\makeatletter\ifx\SetFigFont\undefined%
\gdef\SetFigFont#1#2#3#4#5{%
  \reset@font\fontsize{#1}{#2pt}%
  \fontfamily{#3}\fontseries{#4}\fontshape{#5}%
  \selectfont}%
\fi\endgroup%
\begin{picture}(9038,7184)(856,-8303)
\end{picture}%
}
    \footnotesize{\caption[Body centred cubic (b.c.c.) sublattice]{Points on a b.c.c.\ sublattice of a cubic lattice, coloured blue.}}
    \label{fig:74}
\end{minipage}
\end{compactwrapfigure}
The crucial point is that a particle remains on this sublattice in the presence of gauge or on-site interactions with an external potential.  A natural way to introduce on-site interactions is by applying a position dependent phase every timestep.  For example, we could have the quantum walk with evolution operator
\begin{equation}
 U=e^{-iVa}e^{-i\sigma_x P_x a}e^{-i\sigma_y P_y a},
\end{equation}
where $V=\sum_{\vec{n}}V_{\vec{n}}\ket{\vec{n}}\bra{\vec{n}}$.  Roughly, in the continuum limit, we would expect to get evolution via the Hamiltonian
\begin{equation}
 P_x\sigma_x+P_y\sigma_y +V(\vec{x}),
\end{equation}
by choosing $V_{\vec{n}}=V(\vec{n}a)$.  So interactions of this type would not move the particle off the b.c.c.\ sublattice.  Also, a gauge interaction would not move the particle off the b.c.c.\ sublattice because the particle just picks up a phase after every conditional shift.

\begin{compactwrapfigure}{r}{0.43\textwidth}
\centering
\begin{minipage}[r]{0.40\columnwidth}%
\centering
    \resizebox{5.4cm}{!}{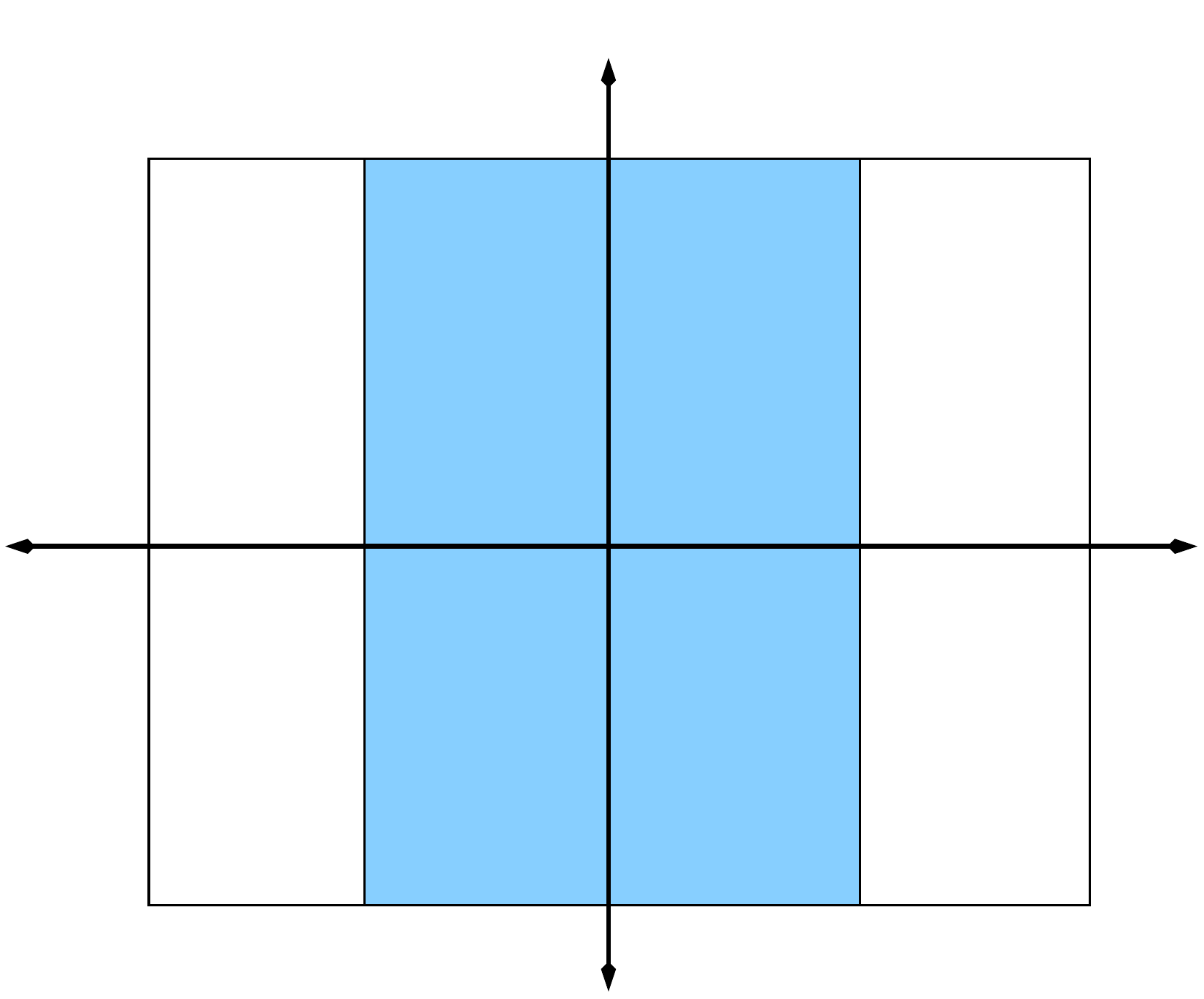}
    \footnotesize{\caption[Momenta on the b.c.c.\ sublattice]{The subset $\mathcal{S}$ in blue.}}
    \label{fig:71}
\end{minipage}
\end{compactwrapfigure}

Let us switch to momentum space.  Position states are written as
\begin{equation}
 \ket{\vec{n}}=\int \frac{\textrm{d}^2p}{(2\pi)^2}e^{-i\vec{p}.\vec{n}a}\ket{\vec{p}\,},
\end{equation}
where the integral is over $(-\f{\pi}{a},\f{\pi}{a}]\times(-\f{\pi}{a},\f{\pi}{a}]$.  Now, define the set $\mathcal{S}$ by
\begin{equation}
 \mathcal{S}=(-\f{\pi}{a},\f{\pi}{a}]\times(-\f{\pi}{2a},\f{\pi}{2a}].
\end{equation}

So we can split the integral up into
\begin{equation}
\label{eq:doub1}
 \ket{\vec{n}}=\int_{\mathcal{S}} \frac{\textrm{d}^2p}{(2\pi)^2}e^{-i\vec{p}.\vec{n}a}\ket{\vec{p}\,}+\int_{\mathcal{S}^{\prime}} \frac{\textrm{d}^2p}{(2\pi)^2}e^{-i\vec{p}.\vec{n}a}\ket{\vec{p}\,},
\end{equation}
where $\mathcal{S}^{\prime}$ is the complement of $\mathcal{S}$.  Next, recall that $\vec{p}_1=(\f{\pi}{a},\f{\pi}{a})$.  And define $\vec{p}+\vec{p}_1$ such that each component is restricted to $(-\f{\pi}{a},\f{\pi}{a}]$ by subtracting $\f{2\pi}{a}$ if necessary.\footnote{Had we defined momenta originally to have components in $[0,\f{2\pi}{a})$, then $\vec{p}+\vec{p}_1$ could just be defined modulo $\f{2\pi}{a}$.  Unfortunately, this would have been awkward from the point of view of taking continuum limits.}  Now we need the fact that $\mathcal{S}^{\prime}$ can be written as
\begin{equation}
 \mathcal{S^\prime}=\{\vec{p}+\vec{p}_1|\vec{p}\in\mathcal{S}\}.
\end{equation}
This allows us to rewrite equation (\ref{eq:doub1}) as
\begin{equation}
\begin{split}
 \ket{\vec{n}} & =\int_{\mathcal{S}} \frac{\textrm{d}^2p}{(2\pi)^2}e^{-i\vec{p}.\vec{n}a}\ket{\vec{p}\,}+\int_{\mathcal{S}} \frac{\textrm{d}^2p}{(2\pi)^2}e^{-i(\vec{p}+\vec{p}_1).\vec{n}a}\ket{\vec{p}+\vec{p}_1}\\
  & =\int_{\mathcal{S}} \frac{\textrm{d}^2p}{(2\pi)^2}\left(e^{-i\vec{p}.\vec{n}a}\ket{\vec{p}\,}+e^{-i(\vec{p}+\vec{p}_1).\vec{n}a}\ket{\vec{p}+\vec{p}_1}\right).
 \end{split}
\end{equation}
If we restrict to $\vec{n}$ being a point on the b.c.c.\ sublattice containing the origin, $(\vec{p}+\vec{p}_1).\vec{n}a=\vec{p}.\vec{n}a$ modulo $2\pi$.  Then
\begin{equation}
\begin{split}
 \ket{\vec{n}} & =\int_{\mathcal{S}} \frac{\textrm{d}^2p}{(2\pi)^2}e^{-i\vec{p}.\vec{n}a}\big(\ket{\vec{p}\,}+\ket{\vec{p}+\vec{p}_1}\big)\\ & = \int_{\mathcal{S}} \frac{\textrm{d}^2p}{(2\pi)^2}e^{-i\vec{p}.\vec{n}a}\ket{\vec{\mathbf{p}}},
 \end{split}
\end{equation}
where we defined
\begin{equation}
\ket{\vec{\mathbf{p}}}=\ket{\vec{p}\,}+\ket{\vec{p}+\vec{p}_1}.
\end{equation}

Now we will see that we should treat $\ket{\vec{\mathbf{p}}}$ as the physical state of the particle with momentum $\vec{p}$.  These behave like they have momentum $\vec{p}$ under the action of the evolution operator:
\begin{equation}
  e^{-i\sigma_x P_xa}e^{-i\sigma_y P_ya}\ket{\vec{\mathbf{p}}}
  =e^{-i\sigma_x p_xa}e^{-i\sigma_y p_ya}\ket{\vec{\mathbf{p}}}.
 \end{equation}
 \begin{compactwrapfigure}{r}{0.52\textwidth}
 \centering
\begin{minipage}[r]{0.50\columnwidth}%
\centering
    \resizebox{6.5cm}{!}{\includegraphics{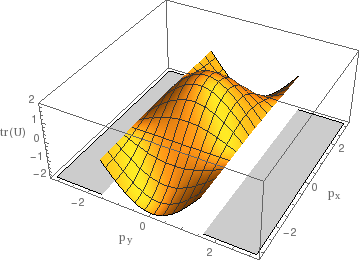}}
    \footnotesize{\caption[Trace of a quantum walk operator on the b.c.c.\ sublattice]{Trace of the massless Dirac quantum walk in two dimensional space again.  The $x$- and $y$-axes correspond to the $x$ and $y$ components of the particle's momentum.  Now momenta are restricted to $\mathcal{S}$, and only momenta close to zero have low energy.\ \\ \ }}
    \label{fig:11}
\end{minipage}
\end{compactwrapfigure}
Note that in the first line capital $P_b$ denotes the components of the momentum operator, whereas lower case $p_b$ in the second line are the components of the momentum $\vec{p}$.  So by restricting to these states, there are no doublers.  This is because the formula for the trace of $U$ remains the same, but now values of $\vec{p}$ with both components close to $\f{\pi}{a}$ never appear because $\vec{p}\in\mathcal{S}$.

Of course, there are two distinct b.c.c.\ sublattices on the cubic lattice.  Because of this, we can consider particles living on different sublattices as different particle types.

\subsection{Fermion Doubling in Three Dimensions}
\label{sec:Fermion Doubling in Three Dimensions}
The situation in three dimensional space is more complicated.  We will see that there are more doublers, though it is not clear whether they can be entirely removed in the same way.

Take the quantum walk that gives rise to the right-handed Weyl equation in the continuum limit, with evolution operator
\begin{equation}
 U_R=e^{-i\sigma_x P_xa}e^{-i\sigma_y P_ya}e^{-i\sigma_z P_za},
\end{equation}
where $P_b$ is the $b$th component of the momentum operator. 

It is claimed in \cite{Bial94} that these models evade the fermion doubling problem because one can always add $2\pi/a$ to the energies.  This does not change the evolution (since $e^{i(p+2\pi/a)a}=e^{ipa}$), and so is meaningless in a sense.  What is meaningful is to look at whether there are doublers in the sense we defined in section \ref{sec:Fermion Doubling in Discrete Time}.  Again, there are \cite{Short13}.  Plots of the dispersion relation for such quantum walks were done in \cite{DP13}, where doubling was noticed in the spectrum.

To see that there is a problem, let us again work in momentum space.  The momentum state $\ket{0,0,0}$ is an eigenstate with eigenvalue one.  Again, there are more.  One example is $\ket{0,\f{\pi}{a},\f{\pi}{a}}$.  Similarly, any momentum state with two components $\f{\pi}{a}$ and one zero will be an eigenvector with eigenvalue one.

And these doublers obey the Weyl equation in the continuum limit.  To see this, suppose that $\vec{p_i}$ is one such momentum vector and let $\vec{p}$ be a momentum vector close to zero.  Then we have
\begin{equation}
\begin{split}
 U_R\ket{\vec{p}+\vec{p}_i} & =e^{-i\sigma_x P_xa}e^{-i\sigma_y P_ya}e^{-i\sigma_z P_za}\ket{\vec{p}+\vec{p}_i}\\
 & =e^{-i\sigma_x p_xa}e^{-i\sigma_y p_ya}e^{-i\sigma_z p_za}\ket{\vec{p}+\vec{p}_i}\\
 & =(1-i\vec{p}.\vec{\sigma}a+O(a^2))\ket{\vec{p}+\vec{p}_i},
\end{split}
\end{equation}
so $\ket{\vec{p}+\vec{p}_i}$ for small $\vec{p}$ behaves like a right handed Weyl particle too.  Therefore, states in the subspace spanned by such states with $|\vec{p}|\ll\f{\pi}{a}$ behave like Weyl particles.  So, since there are four such $\vec{p}_i$ vectors (including $\vec{p}_i=\vec{0}$), there appear to be four types of fermions here.  We call these $\f{\pi}{a}$ doublers.

Now, there is a way to remove these doublers, which we will get to, but first we should ask whether we have found all of them.  In fact, the answer is no!  One way to see this is to again use $e^{-i\theta\sigma_b}=\cos(\theta)-i\sin(\theta)\sigma_b$.  Acting on a state with momentum $(\f{-\pi}{2a},\f{\pi}{2a},\f{\pi}{2a})$, the evolution operator becomes
\begin{equation}
 U_R =(i\sigma_x)(-i\sigma_y)(-i\sigma_z)=\openone.
\end{equation}
And similarly, there are three more doublers like this, with momenta given by $(\f{\pi}{2a},\f{-\pi}{2a},\f{\pi}{2a})$, $(\f{\pi}{2a},\f{\pi}{2a},\f{-\pi}{2a})$ and $(\f{-\pi}{2a},\f{-\pi}{2a},\f{-\pi}{2a})$.  We call these $\f{\pi}{2a}$ doublers.

We can use a similar argument to that used for the $\f{\pi}{a}$ doublers, by looking at the behaviour of $\vec{p}+\vec{k}_i$, where $\vec{k}_i$ is one of the momenta for which $U_R$ is the identity and $\vec{p}$ is small.  Then we see that for small $\vec{p}$ the particle behaves like a Weyl particle.

At this point, it seems like there are eight fermion species where we only wanted one.  We really should pause to verify that there are no more.  This is complicated enough to merit a lemma.
\begin{lemma}
The quantum walk with evolution operator
 \begin{equation}
 U_R=e^{-i\sigma_x P_xa}e^{-i\sigma_y P_ya}e^{-i\sigma_z P_za},
\end{equation}
has eight modes that obey the Weyl equation in the continuum limit.
\end{lemma}
\begin{proof}
To see this, note that the trace of $U_R$ is equal to $2(c_xc_yc_z-s_xs_ys_z)$, where $c_b=\cos(p_ba)$ and $s_b=\sin(p_ba)$.  If $U_R=\openone$, then its trace is two.  So the only way both eigenvalues of $U_R$ for a given $\vec{p}$ can be one is if $c_xc_yc_z-s_xs_ys_z=1$, but this is only possible if $c_xc_yc_z=1$ or $s_xs_ys_z=-1$.  To see this, treat this expression as an inner product between $(c_x,-s_x)$, a normalized vector, and $(c_yc_z,s_ys_z)$, a vector with norm less than or equal to one.  Now, for this inner product to be one, the two vectors must have norm one (they must also be parallel).  Then it follows that we need $c_y^2c_z^2+s_y^2s_z^2=1$.  Now treat  $c_y^2c_z^2+s_y^2s_z^2$ as an inner product between $(c^2_y,s_y^2)$ and $(c^2_z,s^2_z)$.  Both vectors' norms are less than or equal to one, so we need $c_y^4+s_y^4=1$, which is only possible if $|c_y|=1$ or $|s_y|=1$.  So the only way to satisfy $c_xc_yc_z-s_xs_ys_z=1$ is if $c_xc_yc_z=1$ or $s_xs_ys_z=-1$.  And, given that $p_ba\in(-\pi,\pi]$, we have already found the eight possible values of $\vec{p}$ for which either of these two conditions can be satisfied.
\end{proof}

The situation for $U_L=e^{i\sigma_x P_xa}e^{i\sigma_y P_ya}e^{i\sigma_z P_za}$, which gives rise to the left-handed Weyl equation in the continuum limit, is almost identical.  The only difference is that the trace is equal to $2(c_xc_yc_z+s_xs_ys_z)$.

\subsubsection{Mitigating the Doubling problem}
We can at least remove the $\f{\pi}{a}$ doublers.  Again, this works by restricting to states that live on a b.c.c.\ sublattice of the original cubic lattice.  Position states are
\begin{equation}
 \ket{\vec{n}}=\int \frac{\textrm{d}^3p}{(2\pi)^3}e^{-i\vec{p}.\vec{n}a}\ket{\vec{p}\,}.
\end{equation}
Define the set $\mathcal{S}$ by
\begin{equation}
 \mathcal{S}_1=(-\f{\pi}{a},\f{\pi}{a}]\times(-\f{\pi}{2a},\f{\pi}{2a}]\times(-\f{\pi}{2a},\f{\pi}{2a}].
\end{equation}
Let us label the momentum vectors by
\begin{equation}
 \begin{split}
  \vec{p}_1 & = (0,0,0)\\
  \vec{p}_2 & = (\f{\pi}{a},\f{\pi}{a},0)\\
   \vec{p}_3 & = (\f{\pi}{a},0,\f{\pi}{a})\\
   \vec{p}_4 & =(0,\f{\pi}{a},\f{\pi}{a}).
 \end{split}
\end{equation}
We define $\vec{p}+\vec{p_i}$ such that each component is restricted to $(-\f{\pi}{a},\f{\pi}{a}]$ by subtracting $\f{2\pi}{a}$ if necessary.  By this definition, one can construct the disjoint sets
\begin{equation}
 \mathcal{S}_i=\{\vec{p}+\vec{p_i}|\vec{p}\in\mathcal{S}_1\}.
\end{equation}
The union of these sets contains all momentum vectors.  This allows us to write position states as
\begin{equation}
\begin{split}
 \ket{\vec{n}} & =\sum_i\int_{\mathcal{S}_i} \frac{\textrm{d}^3p}{(2\pi)^3}e^{-i\vec{p}.\vec{n}a}\ket{\vec{p}\,}\\
 & =\sum_i\int_{\mathcal{S}_1} \frac{\textrm{d}^3p}{(2\pi)^3}e^{-i(\vec{p}+\vec{p}_i).\vec{n}a}\ket{\vec{p}+\vec{p}_i}
 \end{split}
\end{equation}
If we restrict to $\vec{n}$ being a point on the b.c.c.\ sublattice containing the origin, $(\vec{p}+\vec{p}_i).\vec{n}a=\vec{p}.\vec{n}a$ modulo $2\pi$.  Then
\begin{equation}
 \ket{\vec{n}} =\int_{\mathcal{S}_1} \frac{\textrm{d}^3p}{(2\pi)^3}e^{-i\vec{p}.\vec{n}a}\sum_i\ket{\vec{p}+\vec{p}_i}
 =\int_{\mathcal{S}_1} \frac{\textrm{d}^3p}{(2\pi)^3}e^{-i\vec{p}.\vec{n}a}\ket{\vec{\bf{p}}},
\end{equation}
where we defined
\begin{equation}
\ket{\vec{\bf{p}}}=\sum_i\ket{\vec{p}+\vec{p}_i}.
\end{equation}

These behave like they have momentum $\vec{p}$ under the action of the evolution operator:
\begin{equation}
  e^{-i\sigma_x P_xa}e^{-i\sigma_y P_ya}e^{-i\sigma_z P_za}\ket{\vec{\bf{p}}}
  =e^{-i\sigma_x p_xa}e^{-i\sigma_y p_ya}e^{-i\sigma_z p_za}\ket{\vec{\bf{p}}}.
 \end{equation}
In the first line, capital $P_b$ denotes the components of the momentum operator, whereas lower case $p_b$ in the second line are the components of the momentum $\vec{p}$.  So by restricting to these states, we have none of the $\f{\pi}{a}$ type of doublers.  Again, a particle remains in this subspace even in the presence of gauge or on-site interactions.  As before, this is because on-site or gauge interactions would not move the particle off the b.c.c.\ sublattice.  

In short, the $\f{\pi}{a}$ doublers have been removed.  Furthermore, there are four distinct b.c.c.\ sublattices that are independent.  In a multiparticle picture, we could consider particles living on different sublattices as different particle types.

Because we are restricted to the space spanned by $\ket{\vec{\bf{p}}}$ with $\vec{p}\in\mathcal{S}$, there is only one doubler left.  Label each of the four $\f{\pi}{2a}$ doubler momentum vectors by $\vec{k}_i$.  Then, for a given $j$,
\begin{equation}
 \sum_i\ket{\vec{k}_j+\vec{p}_i}=\sum_i\ket{\vec{k}_i}.
\end{equation}
So after the restriction to the b.c.c.\ sublattice, there is only one doubler left.

Unfortunately, it is not immediately clear whether this $\f{\pi}{2a}$ doubler can be removed.  But here is one possible strategy.  Suppose that every second timestep, instead of implementing the evolution $T_xT_yT_z$, we implement $T_zT_yT_x$.  Acting on a doubler momentum state, $(\f{-\pi}{2a},\f{\pi}{2a},\f{\pi}{2a})$, for example, the latter operator is equal to
\begin{equation}
 (-i\sigma_z)(-i\sigma_y)(i\sigma_x)=-\openone,
\end{equation}
which should correspond to high energy.  Similarly, acting on the other $\f{\pi}{2a}$ doublers $T_zT_yT_x$ will equal $-\openone$.  Of course, $T_zT_yT_x$ also has $+1$ eigenvectors with momentum components $\pm\f{\pi}{2a}$.  But when $T_xT_yT_z$ acts on these it has eigenvalue $-1$, so again these ought to correspond to high energy.

Understanding completely how this would play out would require further work.  It is not obvious what the strategy above means for continuum limits.  Perhaps formulating a concrete procedure for simulation that takes all aspects of the process into account, like state preparation, for example, would provide the answer.  Nevertheless, it is very interesting that at this level we have managed to remove the doubler modes completely in two dimensional space and reduce the number to one in three dimensional space.  And we can take encouragement from the fact that all this was done using only simple quantum walks.  So it is entirely possible that other quantum walks are even better from this point of view.

\section{Quantum Walks on a Line:\ Abstract Theory}
\label{sec:Quantum Walks on a Line: Abstract Theory}
One reason that we could not just use the Lie-Trotter product formula to take the continuum limit in section \ref{sec:Two Dimensional Coins are Special} was that it is not known whether all quantum walk unitaries can be written as a product of conditional shifts and coin operators.  In one dimension, however, this is known to be true.

Here is a simple proof of this.  This was proved previously by different methods in \cite{Vogts09}.  It is interesting that there is an analogue of this result for quantum cellular automata, which we will prove in chapter \ref{chap:Quantum Cellular Automata and Field Theory}.

\begin{theorem}
\label{th:qwlocaldecomp13}
 Given a translationally invariant quantum walk on a line, with evolution operator $U$, there exist pairs of projectors on $\mathcal{H}_C$ denoted by $\{P_i,P_i^{\perp}\}$ with the property that $P_i+P_i^{\perp}=\openone$ such that
\begin{equation}
U=W\prod_{i=1}^{N}\left(P_i S_{q_i}+P_i^{\perp}\right),
\end{equation}
where $S_{q_i}$ are shift operators that shift by $q_i$ steps along the line and $W$ is a unitary operator on $\mathcal{H}_C$.
\end{theorem}
\begin{proof}
We saw in section \ref{sec:Form of an arbitrary Quantum Walk Unitary} that $U$ always takes the form
\begin{equation}
 U=\sum_{q\in Q}A_{q}S_{q},
\end{equation}
where $A_{q}$ are operators on $\mathcal{H}_C$ and, because $U$ is causal, $A_{q}$ is only non-zero for some finite set of vectors $q$, denoted $Q$, called the neighbourhood.
 
By imposing unitarity and expanding $UU^{\dagger}=\openone$, we get 
\begin{equation}
\displaystyle\sum_{q,p\in Q} A_qA^{\dagger}_p S_qS_p^{\dagger}=\openone.
\end{equation}
So all terms with $S_qS_p^{\dagger}$ where $q\neq p$ must vanish.  Now, there is a maximum and a minimum element in $Q$, which we will denote by $q_{\max}$ and $q_{\min}$ respectively.  And there is only one term in the expansion of $UU^{\dagger}$ that shifts by $q_{max}-q_{\min}$.  Since this must vanish, we know that
\begin{equation}
  A_{q_{\max}}A^{\dagger}_{q_{\min}}=0.
\end{equation}
This means that the supports of these two operators are orthogonal.  One way to see this is to use singular value decomposition to write each operator as $a_1\ket{v_1}\bra{w_1}+a_2\ket{v_2}\bra{w_2}$, where $a_i$ are real numbers, which may be zero, and both $\ket{v_i}$ and $\ket{w_j}$ are two orthonormal bases.  The decomposition will be different in general for $A_{q_{\max}}$ and $A_{q_{\min}}$.  The support of an operator is spanned by the $\ket{w_i}$ corresponding to the nonzero $a_i$.  Then it follows that $A_{q_{\max}}A^{\dagger}_{q_{\min}}=0$ implies that the supports of both operators are orthogonal.

Let us denote the projector onto the support of $A_{q_{\max}}$ by $P_1$ and define $P_1^{\perp}=\openone-P_1$.  Then we define the operator $V$ by
\begin{equation}
 V=U\left(P_1 S_{-1}+P_1^{\perp}\right),
\end{equation}
where $S_{-1}$ shifts by $-1$.  Because $A_{q_{\min}}P_1=0$, it follows that $V$ has a neighbourhood that has a smaller range of values than that of $U$.

We can continue this process of multiplying $U$ by operators of the form $(P_i S_{-1}+P_i^{\perp})$ until the neighbourhood of the resulting operator contains only one point.  The resulting operator may be proportional to a shift, but this can be undone by multiplication by its inverse, which is a special case of an operator of the form $(P_iS_{q_i}+P_i^{\perp})$.  It follows that the resulting operator is a unitary on $\mathcal{H}_C$, which we will denote by $W$.  Inverting this expression gives
\begin{equation}
U=W\prod_{i=1}^{N}\left(P_i S_{q_i}+P_i^{\perp}\right).
\end{equation}
\end{proof}

\begin{compactwrapfigure}{r}{0.40\textwidth}
\centering
\begin{minipage}[r]{0.38\columnwidth}%
\centering
    \resizebox{5.0cm}{!}{\begin{picture}(0,0)%
\includegraphics{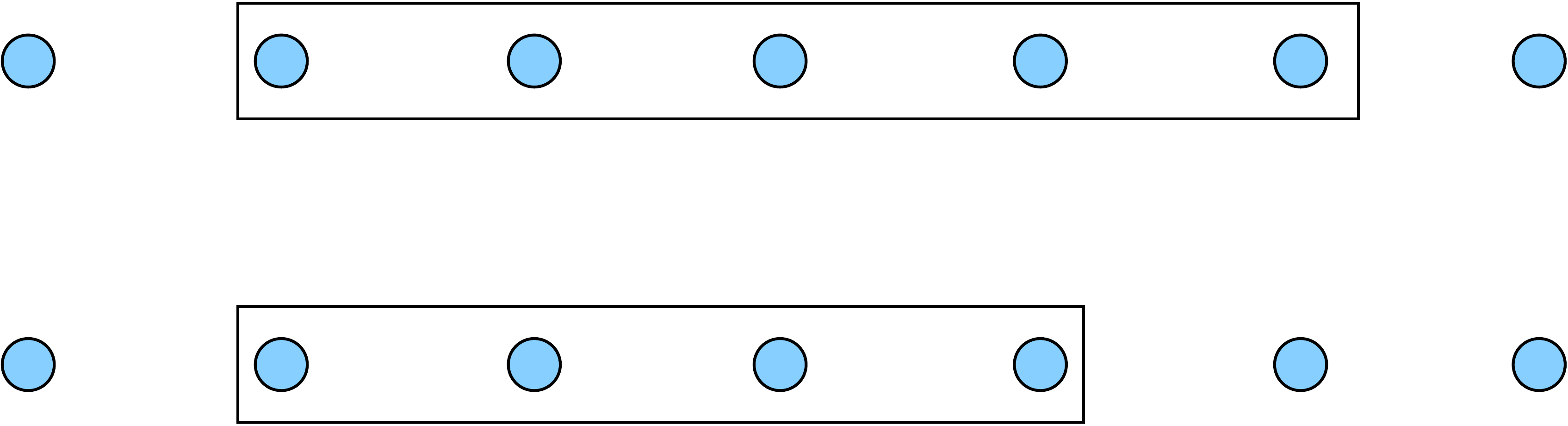}%
\end{picture}%
\setlength{\unitlength}{3947sp}%
\begingroup\makeatletter\ifx\SetFigFont\undefined%
\gdef\SetFigFont#1#2#3#4#5{%
  \reset@font\fontsize{#1}{#2pt}%
  \fontfamily{#3}\fontseries{#4}\fontshape{#5}%
  \selectfont}%
\fi\endgroup%
\begin{picture}(16261,4416)(-6817,-5194)
\end{picture}%
}
    \footnotesize{\caption[Shrinking the neighbourhood of a quantum walk]{Shrinking the neighbourhood of a quantum walk.  The top neighbourhood corresponds to $U$ and the bottom to $V$, mentioned in the proof of theorem \ref{th:qwlocaldecomp13}.}}
    \label{fig:75}
\end{minipage}
\end{compactwrapfigure}

The obvious question now is whether or not such a decomposition exists in higher dimensional spaces.  At the moment, the answer is not known.

Let us attempt to construct a counterexample using the quantum walks that we have already seen.  Of course, they have all been explicitly presented as a product of conditional shifts and coin operators.  But the two and three dimensional examples giving rise to the Dirac or Weyl equations in the continuum limit really only use a b.c.c.\ sublattice of the original cubic lattice.  So it is not inconceivable that we could construct a counterexample by {\em restriction} to these sublattices.

\begin{compactwrapfigure}{r}{0.35\textwidth}
\centering
\begin{minipage}[r]{0.34\columnwidth}%
\centering
    \resizebox{4.0cm}{!}{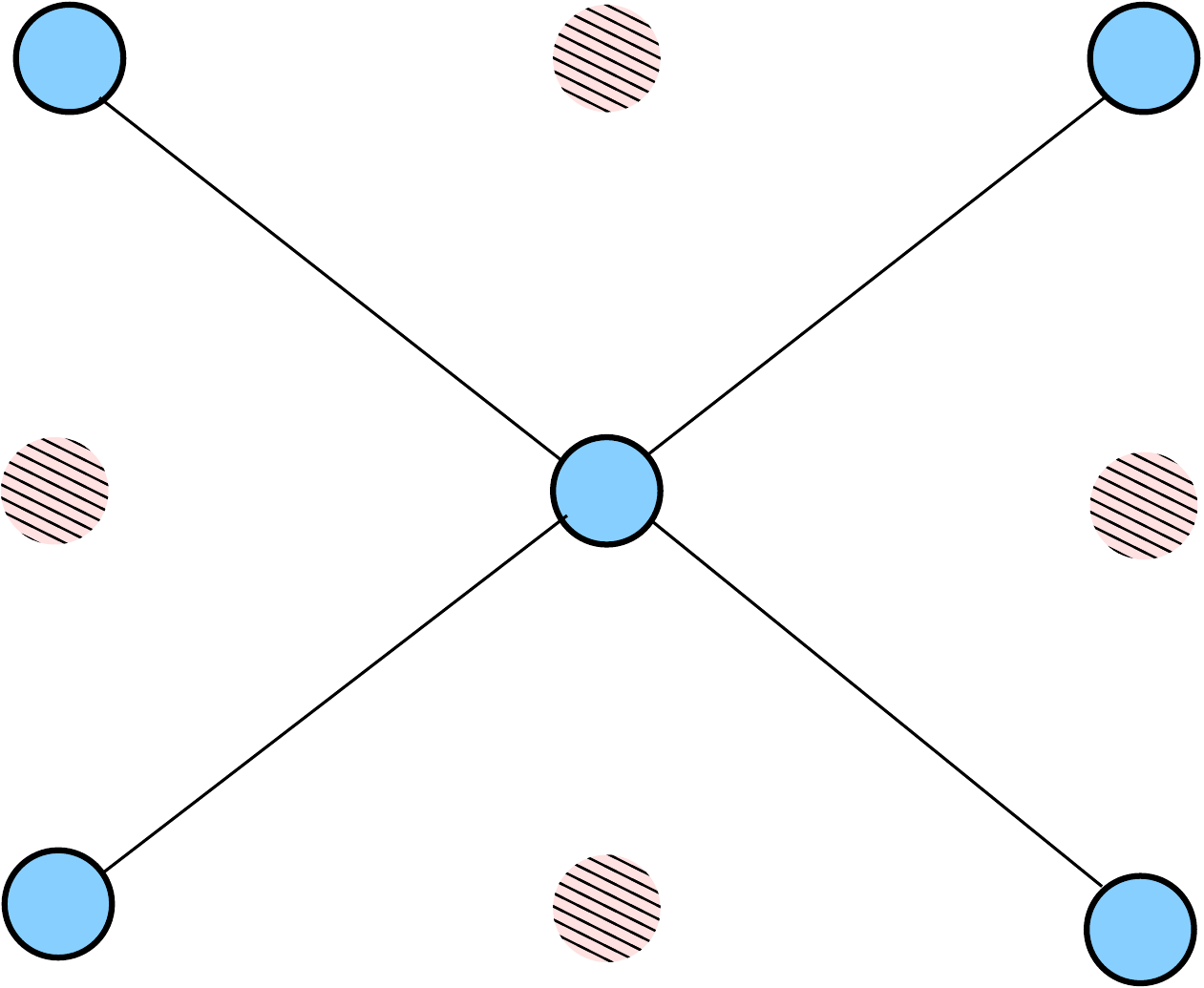}
    \footnotesize{\caption[Shifts restricted to a b.c.c.\ sublattice]{Shifts by lattice vectors on the b.c.c.\ sublattice.}}
    \label{fig:90}\ \\ \ 
\end{minipage}
\end{compactwrapfigure}

Let us try the now familiar quantum walk
\begin{equation}
 U=e^{-iP_x\sigma_xa}e^{-iP_y\sigma_ya}=T_xT_y,
\end{equation}
where $T_b$, with $b\in\{x,y\}$, is a conditional shift along the $b$ direction depending on the spin along that direction, meaning
\begin{equation}
 T_b=S_b\ket{\uparrow_b}\bra{\uparrow_b}+S_b^{\dagger}\ket{\downarrow_b}\bra{\downarrow_b},
\end{equation}
where $S_b$ are shifts by one step along the $b$ direction.  If we restrict to the even b.c.c.\ sublattice, which includes the point $(0,0)$, then we still have a quantum walk with a causal unitary evolution, but now the evolution operator cannot be written as $T_xT_y$ since the walk is only defined on a lattice with sites $(n,m)$, where $n+m$ is even.  It does not make sense to talk about applying $T_x$ as there are no points $(n,m)$ with $n+m$ odd.

Still, with a little thought we can write down an evolution operator defined only on this even lattice that is a product of conditional shifts:
\begin{equation}
 \left(S_xS_y\ket{\uparrow_x}\bra{\uparrow_x}+S_x^{\dagger}S_y\ket{\downarrow_x}\bra{\downarrow_x}\right)\left(\ket{\uparrow_y}\bra{\uparrow_y}+S_y^{\dagger 2}\ket{\downarrow_y}\bra{\downarrow_y}\right).
\end{equation}
The point now is that $S_xS_y$, $S_x^{\dagger}S_y$ and $S_y^{\dagger 2}=(S_xS_y^{\dagger})(S_x^{\dagger}S_y^{\dagger})$ are all shifts by lattice vectors on the b.c.c.\ lattice.  So the evolution has indeed been recast as a product of conditional shifts.

A similar result can be obtained for the three dimensional analogue of this quantum walk, with evolution operator
\begin{equation}
 U=e^{-iP_x\sigma_xa}e^{-iP_y\sigma_ya}e^{-iP_z\sigma_za}=T_xT_yT_z.
\end{equation}
Restricted to the b.c.c.\ sublattice containing $(0,0,0)$, the evolution operator can be written as
\begin{equation}
\begin{split}
 \Big(S_1\ket{\uparrow_x}\bra{\uparrow_x}+S_2\ket{\downarrow_x}\bra{\downarrow_x}\Big) & \Big(\ket{\uparrow_y}\bra{\uparrow_y}+S_2^{\dagger}S_3\ket{\downarrow_y}\bra{\downarrow_y}\Big)\times\\
 & \left(\ket{\uparrow_z}\bra{\uparrow_z}+S_1^{\dagger}S_3^{\dagger}\ket{\downarrow_z}\bra{\downarrow_z}\right),
 \end{split}
\end{equation}
where
\begin{equation}
 \begin{split}
  S_1 & =S_xS_yS_z\\
  S_2 & =S_x^{\dagger}S_yS_z\\
  S_3 & =S_x^{\dagger}S_y^{\dagger}S_z
 \end{split}
\end{equation}
are all shifts by b.c.c.\ lattice vectors.  So again, even only defined on a b.c.c.\ sublattice, the quantum walk still has a decomposition into conditional shifts.

If there is some small moral to this, it is that finding counterexamples (if they exist) is probably going to take something more complicated than simply restricting these quantum walks to sublattices.

\chapter{Quantum Cellular Automata and Fields}
\label{chap:Quantum Cellular Automata and Field Theory}
\section{Introduction}
In the last chapter we focused on quantum walks as discretized models of relativistic particles.  In this chapter, however, we will be a bit more vague at first about how the systems fit into the category of quantum cellular automata.  What we consider first are systems that share most of the properties of quantum cellular automata, including causality and discrete spatial structure, but with fermionic modes at each site, as opposed to qubits or higher dimensional quantum systems.  Nevertheless, as we progress, we will see that these are in some sense two sides of the same coin:\ these fermionic quantum cellular automata are equivalent to regular quantum cellular automata.

Before seeing this equivalence, as it is quite abstract, we will look at using causal fermionic systems in discrete spacetime as quantum fields in discrete spacetime, something previously explored in \cite{D'Ariano12a,D'Ariano12b,BDT12,DP13}.   The systems studied in these papers were essentially free fermions that obey the Dirac or Weyl equation in the continuum limit.  The fact that no mention was made of the vacuum or that it was explicitly taken to be the state annihilated by all annihilation operators $\psi_{\vec{n}\alpha}$ means that these were not fermion fields.  The key point, which we touched on in section \ref{sec:Field Theory}, is that the vacuum is a complicated entangled state.  So this is something else that must be approximated by the discrete system.  We will see how to do this in the following section.  The fact that we can show that there is a discrete vacuum state that converges to the continuum vacuum state is very interesting.

Generally, interacting quantum field theories are defined by the continuum limits of lattice models with local Hamiltonians \cite{Creutz83}.  But we will see in section \ref{sec:Quantum Cellular Automata as Quantum Field Simulators} that we can always approximate the dynamics of these lattice models by quantum cellular automata.  So the idea that one can generally define quantum field theories by continuum limits of causal models is feasible.  Approximating quantum fields by causal discrete spacetime models is natural as quantum field theory is intrinsically causal.  Furthermore, discretizing spacetime allows us to regulate the infinities that appear in quantum field theory calculations, though there are often other ways to do this.  In some cases, these discrete models are easily simulable on quantum computers and so may be useful for simulation.  As a bonus, at least in the free field case, the evolution is very simple

In this chapter, we also explore an idea from \cite{Ball05,VC05} that says that the nonlocality due to anticommutation of fermion operators is not necessary for a description of physical systems.  A drawback is that the applications and extensions of this that we will look at do not have the same simplicity as the original description in terms of fermions.  Still, they could be useful from a simulation point of view, as the nonlocality of fermion operators increases the computational cost of simulations.  Furthermore, it is not inconceivable that, with more work, representations of fermions without anticommutation could be made simpler.  That said, our primary goal will be to use the idea from \cite{Ball05,VC05} as a tool.  Nevertheless, it is interesting to speculate on the philosophical significance of the idea.

Afterwards, we will turn our attention to the more abstract theory of causal systems in discrete spacetime.  The special case of quantum cellular automata was studied previously in \cite{SW04,ANW11,GNVW12}, for example.  We will look at the connection between causal systems of fermions and bosons on lattices and quantum cellular automata.  This will allow us to see that fermionic quantum cellular automata and regular quantum cellular automata are equivalent.

The final part of the chapter stands alone, in that it involves the abstract theory of quantum cellular automata only and uses very different tools.  The result is a structure theorem for quantum cellular automata in one dimension.

The ordering of this chapter, as with the previous one, starts with concrete ideas that are closer to things like simulation, and then progresses onto more abstract ideas.  We begin in section \ref{sec:Continuum Limits of Discrete Fermion Fields} by upgrading the quantum walks of the previous chapter into fermionic models and ultimately discrete fermion field theories.  Then we take continuum limits of these systems to recover free fermion fields in continuous spacetime.  In section \ref{sec:Lattice QFT without Anticommutation}, we look at and build upon a fascinating idea from \cite{Ball05,VC05} that allows us to represent local fermion Hamiltonians and observables by local qubit Hamiltonians and observables.  This requires the introduction of auxiliary degrees of freedom.  We incorporate these degrees of freedom into the gauge field that mediates interactions between fermions on a lattice in the hope that these seemingly redundant degrees of freedom can be given some physical interpretation.  Following this, in section \ref{sec:Quantum Cellular Automata as Quantum Field Simulators} is an argument justifying why nature could be described by a quantum cellular automaton.  Section \ref{sec:Causal Discrete-Time Models on a Lattice}, which comes next, is more abstract and looks at the general properties of causal systems in discrete spacetime.  One of the results in this section is the equivalence of fermionic and regular quantum cellular automata.  We finish the chapter with section \ref{sec:Information Flow in Quantum Cellular Automata}, which is the most abstract.  Here, analogously to the final section of chapter \ref{chap:Quantum Walks and Relativistic Particles}, we derive a decomposition of the evolution of quantum cellular automata on a line into shifts and on-site unitaries.

\section{Continuum Limits of Discrete Free Fermion Fields}
\label{sec:Continuum Limits of Discrete Fermion Fields}
In this section we look at causal systems of fermions in discrete spacetime that become free fermion fields in the continuum limit.  This construction is nontrivial:\ while going from the single particle to the fermion picture is straightforward, we must take into account that the ground state of the continuum theory is a complicated entangled state, so the discrete model has to approximate this.  Interestingly, the discrete dynamics are surprisingly simple and can be decomposed into products of swaps and local unitaries.  Put another way, the unitary evolution operator takes the form of a simple constant-depth circuit, as discussed in \cite{D'Ariano12b}.

The breakdown of this section is as follows.  In section \ref{sec:Second Quantization and Dynamics}, we take the single particle quantum walks of chapter \ref{chap:Quantum Walks and Relativistic Particles} and upgrade them to systems of fermions evolving in the same way by using second quantization.  A consequence of this is that the evolution actually has a very simple decomposition in terms of local unitaries.  Next, in section \ref{sec:Discrete Fermion Fields and the Vacuum}, we construct the vacuum for the discrete fields.  The choice of this state is justified by showing that it converges to the continuum vacuum in section \ref{sec:Continuum Limits2}, where we also take the continuum limit of the dynamics.  The end result is that in the continuum limit we get continuum fermion fields.

\subsection{Second Quantization of Quantum Walks}
\label{sec:Second Quantization and Dynamics}
Now we are going to move from the single particle quantum walk picture to systems of fermions with the same dynamics.  This works by second quantization, which we saw in section \ref{sec:Second Quantization}.  Let us start with the Dirac quantum walk on a line as an example to get the idea.  Though we will not present them here, all the quantum walks we considered in sections \ref{sec:A General Recipe} and \ref{sec:Faster Converging Quantum Walks} fit into this framework.  Something interesting we will see is that the evolution of all of these fermionic systems can be decomposed into very simple one- and two-site unitaries, much like what was seen in \cite{D'Ariano12b}.

\subsubsection{Discrete Dirac Fermions in One Dimension}
\label{sec:Discrete Dirac Fermions in One Dimension}
Recall from section \ref{sec:Continuum Limits} in chapter \ref{chap:Quantum Walks and Relativistic Particles} that the Dirac quantum walk on a line has evolution operator
\begin{equation}
 U  =e^{-im\sigma_xa}e^{-iP\sigma_za},
\end{equation}
where this evolves a single particle on a line, with position labelled by $n$ and states of the extra degree of freedom labelled by $l$ and $r$.  Now we apply second quantization to get fermions with annihilation operators $\psi_{n,\alpha}$, where $n$ labels position and $\alpha\in\{l,r\}$ labels the extra degree of freedom.  In momentum space, we have
\begin{equation}
 \psi_{p,\alpha}=\sqrt{a}\displaystyle\sum_{n}e^{-ipna}\psi_{n,\alpha},
\end{equation}
where $p\in (-\f{\pi}{a},\f{\pi}{a}]$.  The discrete momentum operator for each degree of freedom is
\begin{equation}
P_{\alpha} =\int_{-\f{\pi}{a}}^{\f{\pi}{a}}\!\frac{\textrm{d}p}{2\pi}\, p\,\psi^{\dagger}_{p,\alpha}\psi^{\ }_{p,\alpha}.
\end{equation}
So after second quantization, the single particle conditional shift $e^{-iP\sigma_z a}$ becomes
\begin{equation}
\exp(-i[P_r-P_l]a)=\exp\left(-i\int_{-\f{\pi}{a}}^{\f{\pi}{a}}\!\frac{\textrm{d}p}{2\pi}\, p\,\psi^{\dagger}_{p}\sigma_z\psi^{\ }_{p}a\right),
\end{equation}
where we use the notation
\begin{equation}
 \psi_{p}=\begin{pmatrix}
           \psi_{p,r}\\
           \psi_{p,l}
          \end{pmatrix}
\ \textrm{and} \ \psi_{n}=\begin{pmatrix}
           \psi_{n,r}\\
           \psi_{n,l}
          \end{pmatrix}.
\end{equation}
Similarly, the coin operator that models mass $e^{-im\sigma_x a}$ becomes
\begin{equation}
W=\exp\left(-im\sum_n\psi^{\dagger}_{n}\sigma_x\psi^{\ }_{n}a\right)=\exp\left(-im\int_{-\f{\pi}{a}}^{\f{\pi}{a}}\!\frac{\textrm{d}p}{2\pi}\,\psi^{\dagger}_{p}\sigma_x\psi^{\ }_{p}a\right).
\end{equation}

Let us pause to look at the continuum limit of this.  Naively using the Lie-Trotter product formula from \ref{sec:The Lie-Trotter Product Formula}, we get
\begin{equation}
\lim_{a \to 0}U^{t/a}=\exp\left(-i\int_{-\infty}^{\infty}\!\frac{\textrm{d}p}{2\pi}\,\psi^{\dagger}_{p}(p\sigma_z+ m\sigma_x)\psi_{p}\, t\right).
\end{equation}
So in the continuum limit, as expected, we recover evolution via the Dirac Hamiltonian in one dimensional space:
\begin{equation}
 H=\int_{-\infty}^{\infty}\!\frac{\textrm{d}p}{2\pi}\,\psi^{\dagger}_{p}(p\sigma_z+ m\sigma_x)\psi_{p}.
\end{equation}
We have not been too careful with the details of convergence here.  We will return to this in section \ref{sec:Continuum Limits2}.

Let us return to the discrete-time evolution.  The conditional shift component of the evolution shifts $\psi_{n,l}$ to the left and $\psi_{n,r}$ to the right.  This is equivalent to the fermionic swaps\footnote{This is actually an application of theorem \ref{th:1} that we will see in section \ref{sec:Local Decomposition}.  It is basically the statement that, if we have a system of fermions on a lattice evolving via a causal unitary $U$ and we take a copy of that system evolving via $U^{-1}$, then the joint evolution can be rewritten as a product of local unitaries.  In the example we have here, the $\psi_{n,l}$ fermions evolve via a shift to the left and the $\psi_{n,r}$ fermions evolve via the inverse unitary.}\\
\begin{equation}
\label{eq:119}
  \psi_{n,l}  \leftrightarrow \psi_{n-1,r},
\end{equation}
at each $n$, followed by
\begin{equation}
\label{eq:200}
   \psi_{n,r}  \leftrightarrow \psi_{n,l},
\end{equation}
at each $n$.  Let us write down an operator that implements a swap.  Given two fermionic modes with annihilation operators $a$ and $b$, we define the swap operator to satisfy
\begin{equation} 
S a S^{\dagger} = b, \qquad S b S^{\dagger} = a.
\end{equation}
\begin{compactwrapfigure}{r}{0.45\textwidth}
\centering
\begin{minipage}[r]{0.40\columnwidth}%
\centering
    \resizebox{6.0cm}{!}{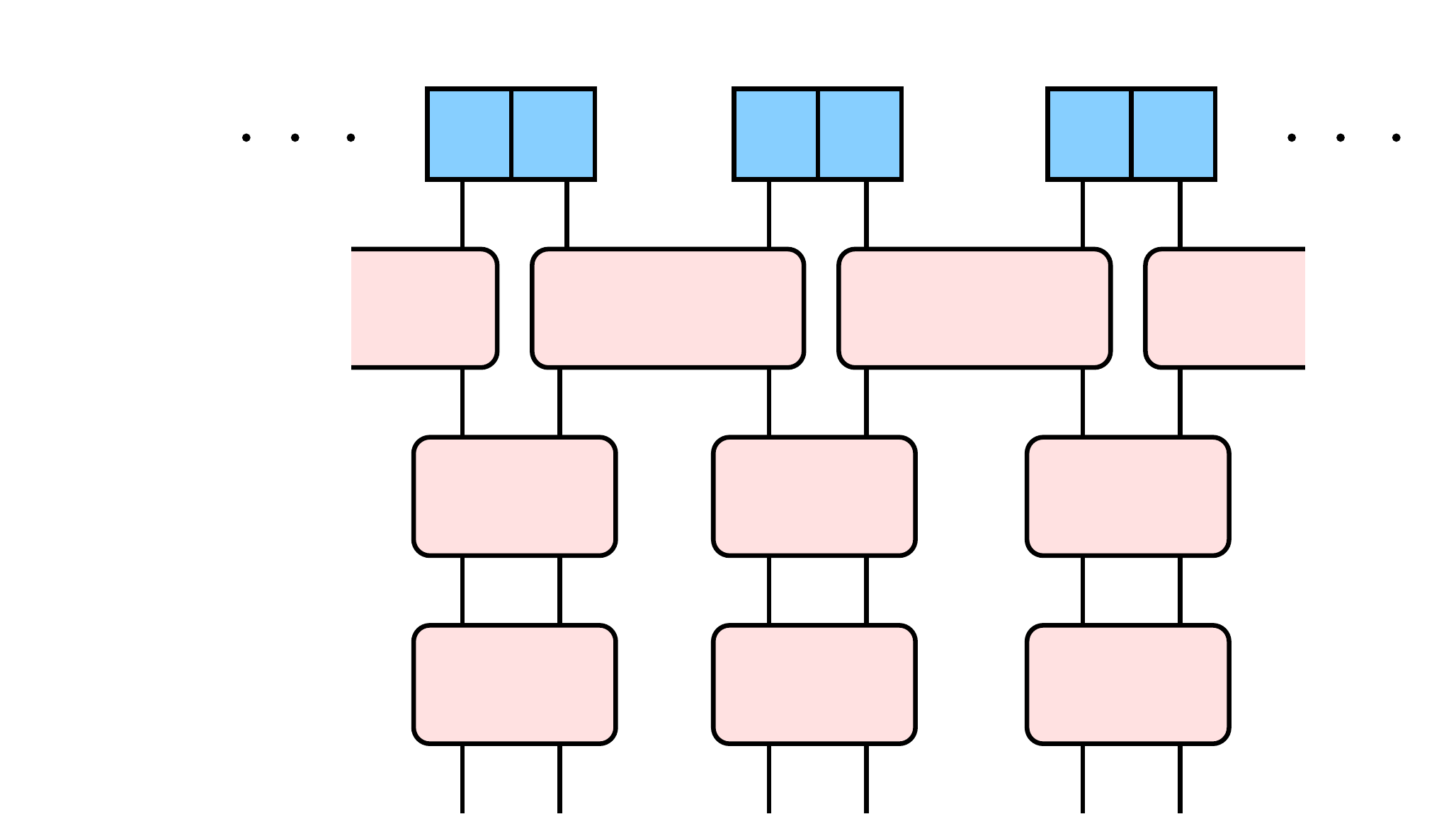}
    \footnotesize{\caption[Fermion circuit for discrete Dirac fermions]{Evolution of discrete Dirac fermions over one timestep.  Blue boxes represent fermion modes, and pink boxes represent the unitaries being applied to modes.}}
    \label{fig:1dimDirac_circuit}
\end{minipage}
\end{compactwrapfigure}
A unitary that implements this is
\begin{equation} 
S= \exp[i\frac{\pi}{2}(b^{\dagger}-a^{\dagger})(b-a)].
\end{equation}
Let us verify that this does the job.  The two operators 
\begin{equation}  
c = \frac{1}{\sqrt{2}} (b - a),  \qquad d =  \frac{1}{\sqrt{2}} (b + a),
\end{equation} 
satisfy the usual canonical anticommutation relations.  Then $S =  \exp[i \pi c^{\dagger} c ]$, and 
\begin{equation}
S a S^{\dagger} =e^{i \pi c^{\dagger} c } \frac{1}{\sqrt{2}} \left(d -c   \right)e^{-i \pi c^{\dagger} c}  \nonumber \\
=  \frac{1}{\sqrt{2}}(d + c)  = b.
\end{equation} 
Here we used $(c^{\dagger} c) c= 0$ and $c (c^{\dagger} c) = c$.  Similarly, it follows that $S b S^{\dagger} = a$.

So applying the local unitaries implementing the swaps in equations \ref{eq:119} and \ref{eq:200} reproduces the conditional shift.  A similar trick will work for any conditional shift where one set of modes moves in one direction and another moves in the opposite direction.

The part of the evolution modeling the effects of mass is also a product of local fermionic unitaries because, in position space, it is
\begin{equation}
\label{eq:90}
\exp(-im\displaystyle\sum_{n}\psi^{\dagger}_n\sigma_x\psi_n\, a)=\prod_{n}\exp(-im\, \psi^{\dagger}_n\sigma_x\psi_n\, a).
\end{equation}
Therefore, the evolution operator $U$ is a product of local fermionic unitaries that form a very simple constant-depth circuit.

These discrete systems of fermions on a line first appeared in \cite{D'Ariano12a}.  The corresponding two dimensional example first appeared in \cite{D'Ariano12b}, where the three dimensional case, which we will look at next, was also briefly mentioned.  The three dimensional case was studied further in \cite{DP13}.

\subsubsection{Discrete Weyl Fermions in Three Dimensions}
\label{sec:Discrete Dirac Fermions in Three Dimensions}
Let us apply this recipe to the Weyl quantum walk, which leads to the Weyl equation in the continuum limit.  The quantum walk evolution operator is
\begin{equation}
 U =e^{-iP_x\sigma_xa}e^{-iP_y\sigma_ya}e^{-iP_z\sigma_za}.
\end{equation}
The particle lives on a three dimensional cubic lattice with integer coordinates $\vec{n}$ and has an extra degree of freedom, which we can think of as spin, with the orthonormal basis $\ket{\uparrow_z}$ and $\ket{\downarrow_z}$.  Now, after second quantization, we have annihilation operators $\psi_{\vec{n},\alpha}$ for each site $\vec{n}$, with $\alpha\in\{\uparrow_z,\downarrow_z\}$.  
In momentum space the annihilation operators become
\begin{equation}
\psi_{\vec{p}} =a^{3/2}\displaystyle\sum_{\vec{n}}e^{-i\vec{p}.\vec{n}a}\psi_{\vec{n}},
\end{equation}
where we defined
\begin{equation}
 \psi_{\vec{p}}=\begin{pmatrix}
           \psi_{\vec{p},\uparrow_z}\\
           \psi_{\vec{p},\downarrow_z}
          \end{pmatrix}
\ \textrm{and}\ \psi_{\vec{n}}=\begin{pmatrix}
           \psi_{\vec{n},\uparrow_z}\\
           \psi_{\vec{n},\downarrow_z}
          \end{pmatrix}.
\end{equation}
The evolution operator becomes a product of conditional shifts $T_xT_yT_z$, where
\begin{equation}
T_b=\exp\left(-i\int_{-\f{\pi}{a}}^{\f{\pi}{a}}\!\frac{\textrm{d}^3p}{(2\pi)^3}\,p_b\psi_{\vec{p}}^{\dagger}\, \sigma_b\psi_{\vec{p}}\, a\right),
\end{equation}
where $b\in\{x,y,z\}$.  Again, by naively applying the Lie-Trotter product formula, these become Weyl fermions in the continuum limit, meaning the continuum Hamiltonian is
\begin{equation}
 H=\int_{-\infty}^{\infty}\!\frac{\textrm{d}^3p}{(2\pi)^3}\,\psi_{\vec{p}}^{\dagger}\, (\vec{p}.\vec{\sigma})\, \psi_{\vec{p}}.
\end{equation}

As was the case in one dimensional space, each conditional shift $T_b$ can be recast as a product of local swap unitaries.  This means that again the total evolution is just a simple product of local operations.  In fact, since there is no mass unitary in this case, the evolution is just a product of swaps.  It is remarkable that something so simple can lead to fermions obeying the Weyl equation in the continuum limit.

As we saw in section \ref{sec:Fermion Doubling in Two Dimensions} in chapter \ref{chap:Quantum Walks and Relativistic Particles}, we may want to restrict our attention to particles that live on a b.c.c.\ sublattice.  In the \hyperlink{QQAppendix}{appendix}, we see that, even if the system just has fermionic modes on the b.c.c.\ sublattice, there is still a local decomposition of the dynamics. 

\subsection{Discrete Fermion Fields and the Vacuum}
\label{sec:Discrete Fermion Fields and the Vacuum}
To truly construct a discretized quantum field theory, we need to construct a discrete version of the vacuum state that becomes equivalent to the continuum vacuum state in the continuum limit.  Let us do this for the one dimensional case.

Again, the single particle evolution operator for discrete Dirac fermions on a line is
\begin{equation}
 U=e^{-im\sigma_xa}e^{-i\sigma_z Pa}.
\end{equation}
Acting on a momentum state, $U$ becomes $U(p)$, which has eigenvalues
\begin{equation}
 \lambda_{\pm}(p)=\cos(ma)\cos(pa)\pm i\sqrt{1-\cos^{2}(ma)\cos^{2}(pa)}.
\end{equation}
It is natural to think of energies as being in the interval $(-\textstyle\frac{\pi}{a},\textstyle\frac{\pi}{a}]$.  We define $\lambda_{-}(p)$ and $\lambda_{+}(p)$ to correspond to positive and negative energy respectively.  This is sensible because, when the imaginary part of $e^{-ix}=\cos(x)-i\sin(x)$ is negative, $x$ is positive for $x\in(-\pi,\pi]$.  We also denote the normalized eigenvectors of $U(p)$ by $w^+(p)$ and $w^-(p)$, which correspond to positive and negative energy respectively.

This allows us to define the creation operators that create particles with positive and negative energies to be
\begin{equation}
\begin{split}
a_{p}^{\dagger} &= \sum_{\alpha} w^{+}_{\alpha}(p)\psi^{\dagger}_{p,\alpha}\\
c_p^{\dagger} &= \sum_{\alpha} w^{-}_{\alpha}(p)\psi^{\dagger}_{p,\alpha}.
\end{split}
\end{equation}
respectively.  Now we can split momentum field operators into positive and negative energy components:
\begin{equation}
\psi_{p,\alpha} = a_{p}w^{+}_{\alpha} (p)+c_{p}w^{-}_{\alpha} (p).
\end{equation}
This used the fact that $w^{+}_{\alpha} (p)$ and $w^{-}_{\alpha}(p)$ are an orthonormal basis.

In position space, the discrete field operators become
\begin{equation}
\psi_{n,\alpha}  =\int_{-\f{\pi}{a}}^{\f{\pi}{a}}\!\frac{\textrm{d}p}{2\pi}\, \left(a_{p}w^{+}_{\alpha} (p)+c_{p}w^{-}_{\alpha} (p)\right)e^{ipna}.
\end{equation}

Now, as we saw in section \ref{sec:Field Theory} in chapter \ref{chap:Background1}, for fermion fields in continuous spacetime, the physical ground state has all the negative energy states filled up.  In analogy with this, we define the discrete ground state $\ket{\Omega_d}$ so that it is annihilated by $c_p^{\dagger}$ and $a_p$.  In other words, it has all the negative energy modes filled.

Acting on $\ket{\Omega_d}$ with the evolution operator gives a phase $e^{i\theta}=\prod_{p}\lambda_+(p)$, so instead we consider evolution via $e^{-i\theta}U$.  This is analogous to adding a constant to the Hamiltonian, as done in the continuous setting.

Finally, the discrete field operator, once rewritten in terms of particles and antiparticles, is
\begin{equation}
\psi_{n,\alpha}=\int_{-\f{\pi}{a}}^{\f{\pi}{a}}\!\frac{\textrm{d}p}{2\pi}\, \left(a_{p}u_{\alpha} (p)e^{ipna}+b^{\dagger}_{p}v_{\alpha} (p)e^{-ipna}\right),
\end{equation}
where $b^{\dagger}_p=c_{-p}$, and we define $u_{\alpha} (p)=w^{+}_{\alpha} (p)$ and  $v_{\alpha} (p)=w^{-}_{\alpha} (-p)$.

\subsection{Continuum Limits}
\label{sec:Continuum Limits2}
Taking a continuum limit now is not going to be as straightforward as it was in the previous chapter.  One reason for this is that it is not so obvious how to map the states of the discrete system into the state space of the continuum system.  After all, the discrete system's ground state only has negative energy particles with momentum components in $(-\f{\pi}{a},\f{\pi}{a}]$, whereas the continuum system's ground state is a Dirac sea of infinite depth.  Luckily, from a physical perspective, it is enough to show that expectation values of observables in the discrete picture converge to their continuum counterparts.  In other words, we need to look at
\begin{equation}
\label{eq:contlim}
 |\bra{\psi}e^{iHt}\m{M}e^{-iHt}\ket{\psi}-\bra{\psi_d}U^{-N}\m{M}_dU^N\ket{\psi_d}|,
\end{equation}
where $\m{M}$ and $\m{M}_d$ are measurement operators in the continuous and discrete case respectively.

\begin{compactwrapfigure}{r}{0.43\textwidth}
\centering
\begin{minipage}[r]{0.41\columnwidth}%
\centering
    \resizebox{5.5cm}{!}{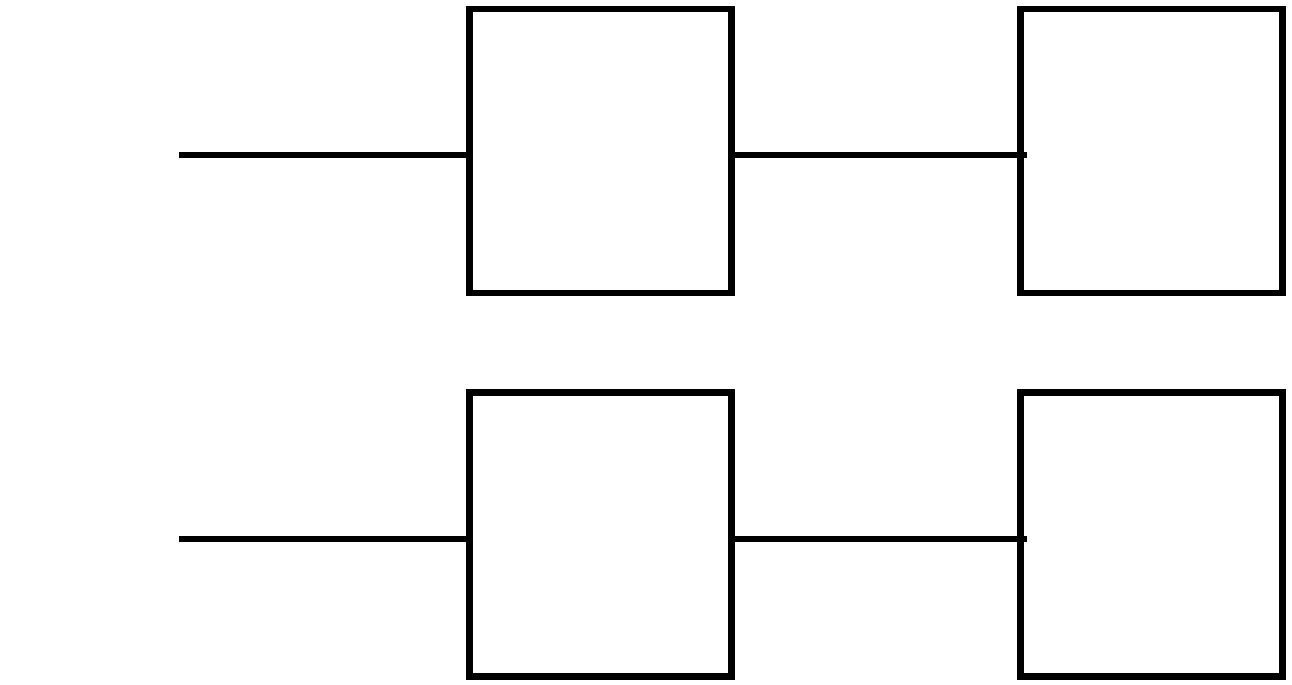}
    \footnotesize{\caption[Simulation of free fields]{To see convergence, we just need the measurement results for the discrete system to approximate those of the continuum system.}}
    \label{fig:FermionFieldsSim}
\end{minipage}
\end{compactwrapfigure}

Now, from a physical point of view, we know that we can only probe energy scales below some threshold.  So let us suppose that there is a momentum cutoff $\Lambda$, or equivalently an inverse length scale, above which we cannot detect anything.  So it makes sense to rewrite the physical observable $\m{M}$ as $\m{M}_{\Lambda}$ to reflect this, and we suppose that states only have particles with momenta below $\Lambda$.  Then, because $\m{M}_{\Lambda}$ will only contain creation and annihilation operators with momentum below $\Lambda$,
\begin{equation}
 \bra{\psi}e^{iHt}\m{M}e^{-iHt}\ket{\psi}=\bra{\psi_{\Lambda}}e^{iH_{\Lambda}t}\m{M}_{\Lambda}e^{-iH_{\Lambda}t}\ket{\psi_{\Lambda}},
\end{equation}
where $\ket{\psi_{\Lambda}}$ is a state with positive and {\em negative} energy particles\footnote{It will be useful to use the positive and negative energy particle picture here, instead of referring to particles and antiparticles.} that only have momentum below $\Lambda$, and $H_{\Lambda}$ is the Hamiltonian that only counts the energy of particles with momentum less than $\Lambda$.  We impose similar restrictions on the discrete system.  Note that, because we are only dealing with states in the subspace with momentum cutoff $\Lambda$, both the operator and Hilbert space norm can be replaced by the operator and Hilbert space norm restricted to this subspace.

Analogously to the momentum mapping in section \ref{sec:Approximating Continuum Models by Discrete Ones}, we associate discrete and continuum momentum creation operators.  This will allow us to compare expectation values in the discrete and continuum cases.   Now, we can always write the continuum state as
\begin{equation}
  \ket{\psi_{\Lambda}}  =V\ket{\Omega_{\Lambda}},
\end{equation}
where $V$ is some unitary, which may create particles.  The truncated continuum ground state $\ket{\Omega_{\Lambda}}$ only has negative energy modes with momentum less than $\Lambda$ filled.  We can take the discrete state to be
\begin{equation}
  \ket{\psi_{d}}  =V\ket{\Omega_{d\Lambda}},
\end{equation}
where $\ket{\Omega_{d\Lambda}}$ is a truncated discrete ground state that only has negative energy modes with momentum less than $\Lambda$ filled.  If we are interested in simulation, $V$ should be a unitary that we can implement efficiently.

We take the discrete observable to be $\m{M}_{d}=\m{M}_{\Lambda}$.  Now, returning to our expression for the difference between discrete and continuum expectation values, 
\begin{equation}
 |\bra{\psi_{\Lambda}}e^{iH_{\Lambda}t}\m{M}_{\Lambda}e^{-iH_{\Lambda}t}\ket{\psi_{\Lambda}}-\bra{\psi_d}U^{-N}\m{M}_{\Lambda}U^N\ket{\psi_d}|,
\end{equation}
we can use the inequality
\begin{equation}
|\bra{\phi_1}A\ket{\phi_{1}} -\bra{\phi_2}A\ket{\phi_2}|\leq  2\|A\|\, \|\ket{\phi_{1}} -\ket{\phi_2}\|_2.
\end{equation}
It will also simplify things to suppose that the measurement is a projector so that $\|\m{M}_{\Lambda}\|=1$.  This leads to
\begin{equation}
  |\bra{\psi_{\Lambda}}e^{iH_{\Lambda}t}\m{M}_{\Lambda}e^{-iH_{\Lambda}t}\ket{\psi_{\Lambda}}  -\bra{\psi_d}U_{\Lambda}^{-N}\m{M}_dU_{\Lambda}^N\ket{\psi_d}|
 \leq  2\|e^{-iH_{\Lambda}t}\ket{\psi_{\Lambda}}  -U_{\Lambda}^N\ket{\psi_d}\|_2
\end{equation}
Then, after applying the triangle inequality and using the fact that $\|\cdot\|_2$ is invariant under unitaries, it follows that we just need to bound
\begin{equation}
\label{eq:4cl1}
 \|e^{-iH_{\Lambda}t}\ket{\psi_{\Lambda}} -U_{\Lambda}^N\ket{\psi_d}\|_2 \leq \|(e^{-iH_{\Lambda}t}-U_{\Lambda}^N)V\ket{\Omega_{\Lambda}}\|_2+\|\ket{\Omega_{\Lambda}} -\ket{\Omega_{d\Lambda}}\|_2.
\end{equation}
We will come back to the first term on the right hand side, which quantifies how close the discrete and continuum dynamics are later.  First, we will deal with the second term on the right hand side.  This quantifies how close the discrete and continuum vacuum states are.  It is not obvious that this term should tend to zero.  Fortunately, we can prove that it does, at least in the one dimensional case.
\begin{theorem}
\label{th:cl47}
In one dimension, the discrete vacua converge, in the sense that
\begin{equation}
\|\ket{\Omega_{\Lambda}} -\ket{\Omega_{d\Lambda}}\|_2=O(\Lambda^3 a), 
\end{equation}
which tends to zero as the lattice spacing tends to zero.
\end{theorem}
\begin{proof}
To make sense of this expression, we are going to work on a finite line so that momenta are discrete.  Otherwise, the creation operators for negative energy particles would not be valid creation operators since they would not create normalizable states.  Let us denote the discrete negative energy annihilation operators by $C_p$ and the continuum negative energy annihilation operator by $c_p$.  Note that these operators are not the same.  We identified discrete and continuum \emph{momentum} creation operators, but that does not mean that $C_p$ is equivalent to $c_p$.  This is because the forms of the positive and negative energy eigenvectors in the single particle setting were not the same in the discrete and continuum cases.

We have that
\begin{equation}
 \begin{split}
  \ket{\Omega_{\Lambda}} & =\prod_{|p|\leq \Lambda}c^{\dagger}_p\ket{0}\\
  \ket{\Omega_{d\Lambda}} & =\prod_{|p|\leq \Lambda}C^{\dagger}_p\ket{0}.
 \end{split}
\end{equation}
Now, the action of $c^{\dagger}_p$ on $\ket{0}$ is the same as the action of the unitary $c^{\dagger}_p+c_p$ on $\ket{0}$.  So, to bound $\|\ket{\Omega_{\Lambda}} -\ket{\Omega_d}\|_2$, we can use the fact that
\begin{equation}
 \|\prod_{i=1}^{N}U_i\ket{0} -\prod_{j=1}^{N}V_j\ket{0}\|_2\leq \|\prod_{i=1}^{N}U_i -\prod_{j=1}^{N}V_j\|
 \leq  N\max_{i}\|U_i-V_i\|,
\end{equation}
when $U_i$ and $V_i$ are unitaries.  The second inequality is something we saw previously in section \ref{sec:The Lie-Trotter Product Formula}.  This means
\begin{equation}
\begin{split}
  \|\prod_{|p|\leq \Lambda}\left(c^{\dagger}_p+c_p\right)\ket{0} -\prod_{|p|\leq \Lambda}\left(C^{\dagger}_p+C_p\right)\ket{0}\|_2&\\
 \leq \, (\alpha\Lambda)\max_{|p|\leq \Lambda}\|(c^{\dagger}_p+c_p) -(C^{\dagger}_p+C_p)\|&,
 \end{split}
\end{equation}
where $\alpha$ is a constant, and the second line follows because the number of negative energy modes filled is proportional to $\Lambda$.  This is because, on a finite line, momentum is quantized.  So it remains to bound $\|(c^{\dagger}_p+c_p) -(C^{\dagger}_p+C_p)\|$.  To do this we just use the definition of the operator norm:
\begin{equation}
\label{eq:4cl2}
\begin{split}
 & \|(c^{\dagger}_p+c_p) -(C^{\dagger}_p+C_p)\|\\
 = & \max_{\ket{\psi}}\sqrt{\bra{\psi}\left((c^{\dagger}_p+c_p) -(C^{\dagger}_p+C_p)\right)\left((c^{\dagger}_p+c_p) -(C^{\dagger}_p+C_p)\right)\ket{\psi}}\\
 = & \sqrt{2-\{c^{\dagger}_p,C_p\}-\{C^{\dagger}_p,c_p\}}.
 \end{split}
\end{equation}
Now, recall that one of the properties of second quantization is that
\begin{equation}
 \{c^{\dagger}_p,C_p\}=\langle p_d | p_c\rangle,
\end{equation}
where $\ket{p_d}$ and $\ket{p_c}$ are the single particle discrete and continuous negative energy states with momentum $p$ created by $C^{\dagger}_p$ and $c^{\dagger}_p$ respectively.  So equation (\ref{eq:4cl2}) becomes
\begin{equation}
\begin{split}
 \|(c^{\dagger}_p+c_p) -(C^{\dagger}_p+C_p)\| & =\sqrt{2-\langle p_d | p_c\rangle-\langle p_c | p_d\rangle}\\
 & = \|\ket{p_c}-\ket{p_d}\|_2.
\end{split}
\end{equation}
We have reduced the problem to showing convergence of eigenvectors.  This remaining step is a little tricky.

First, note that the phase of the discrete eigenvectors is immaterial since our original concern is expectation values, like $\bra{\psi_d}U^{-N}\m{M}_dU^N\ket{\psi_d}$, as in equation (\ref{eq:contlim}).  So we can take it for granted that the discrete momentum states have phases chosen such that $\langle p_d | p_c\rangle$ is real and positive.  Once this is the case, we can use
\begin{equation}
 \|\ket{\phi}-\ket{\psi}\|_2\leq \sqrt{2}\|\big(\ket{\psi}\bra{\psi}- \ket{\phi}\bra{\phi}\big)\|,
\end{equation}
which holds if $\langle \psi|\phi\rangle$ is real and positive.  So now we just need to bound $\|\big(\ket{p_d}\bra{p_d}- \ket{p_c}\bra{p_c}\big)\|$.  To do this we need to use theorem $VII.3.1$ from \cite{Bhatia97}.
\begin{theorem}
\label{th:Bhatia}
 Let $A$ and $B$ be normal operators, and let $S_1$ and $S_2$ be two subsets of $\mathbb{C}$ separated by strip or annulus of width $\delta$.  Denote the projector onto the eigenspace corresponding to eigenvalues of the operator $A$ in the set $S$ by $P_A(S)$.  Then
\begin{equation}
 \vertiii{P_A(S_1)P_B(S_2)}\leq  \frac{1}{\delta}\vertiii{A-B},
\end{equation}
where $\vertiii{\cdot}$ is any norm that satisfies $\vertiii{X}=\vertiii{UXV}$ for unitaries $U$ and $V$.  The operator norm, for example, has this property.
\end{theorem}

For our purposes, we use this theorem with
\begin{equation}
 \begin{split}
  A & =e^{-im\sigma_xa}e^{-iP\sigma_za}\\
  B & =e^{-i(P\sigma_z+m\sigma_x)a}.
 \end{split}
\end{equation}
In our case $S_1$ and $S_2$ are the smallest sets containing the eigenvalues of $A$ and $B$ corresponding to negative energy and positive energy respectively.  Bear in mind that we are looking at a fixed value of $p$.  What we want to bound is
\begin{equation}
 \|\ket{p_c}\bra{p_c}-\ket{p_d}\bra{p_d}\|=\|P_A(S_1)-P_B(S_1)\|.
\end{equation}
Because $P_A(S_1)+P_A(S_2)=\openone$ and $P_B(S_1)+P_B(S_2)=\openone$, it follows that
\begin{equation}
\begin{split}
 \|P_A(S_1)-P_B(S_1)\| & = \|P_A(S_1)P_B(S_2)-P_A(S_2)P_B(S_1)\|\\
 & \leq \|P_A(S_1)P_B(S_2)\|+ \|P_A(S_2)P_B(S_1)\|.
 \end{split}
\end{equation}
So we can apply theorem \ref{th:Bhatia} to the two terms on the right hand side.  We know from section \ref{sec:The Lie-Trotter Product Formula} that
\begin{equation}
 \|A-B\|=\|e^{-im\sigma_xa}e^{-iP\sigma_za}-e^{-i(P\sigma_z+m\sigma_x)a}\|_{\Lambda}=O(\Lambda^2a^2).
\end{equation}
So it remains to bound $\delta$ from below.  To do this, we calculate the distance between the pairs of eigenvalues of both $e^{-im\sigma_xa}e^{-iP\sigma_za}$ and $e^{-i(P\sigma_z+m\sigma_x)a}$ corresponding to positive and negative energies.  Start with $e^{-im\sigma_xa}e^{-iP\sigma_za}$.  In section \ref{sec:Discrete Fermion Fields and the Vacuum} we found its eigenvalues for a given $p$, so we get
\begin{equation}
\begin{split}
|\lambda_+(p)-\lambda_-(p)|= & 2\sqrt{1-\cos^2(ma)\cos^2(pa)}\\
\geq & 2\sqrt{1-\cos^2(ma)}=2\sin(ma)=O(a).
\end{split}
\end{equation}
Next, look at $e^{-i(P\sigma_z+m\sigma_x)a}$.  This time, we get that the distance between the eigenvalues is
\begin{equation}
2\sin(\sqrt{p^2+m^2}a)\geq 2\sin(ma),
\end{equation}
where we assumed that $|p|,m\ll \f{1}{a}$.  The main point\footnote{Note that the distance between the discrete and continuum eigenvalues corresponding to positive energy is $O(a^3)$.  The same is true for the eigenvalues corresponding to negative energy.  So the distance between both sets is $O(a)$.} is that $\delta =O(a)$.

Putting this all together, we find that
\begin{equation}
\|\ket{p_c}-\ket{p_d}\|_2=O(\Lambda^2 a).
\end{equation}
So finally it follows that
\begin{equation}
\|\ket{\Omega_{\Lambda}} -\ket{\Omega_{d\Lambda}}\|_2=O(\Lambda^3 a).
\end{equation}
Provided the cutoff is taken to grow sufficiently slowly as $a$ goes to zero, this proves convergence of the discrete vacuum to the continuum vacuum.
\end{proof}
In this proof, we worked on a finite line so that momenta were discrete.  This is more relevant from a simulation point of view, as we cannot create systems on infinite lines in the laboratory.  Nevertheless, if we wanted to prove something analogous for an infinite line, then looking at the difference between the states in the Hilbert space norm is not a sensible thing to do.  Instead, it would make more sense to look at the difference in expectation values of physical observables constructed as smoothed-out products of the field operators.

As far as showing that our fermion fields in discrete spacetime converge to fermion fields in the continuum, all that remains is to prove convergence of the dynamics.  This means bounding the first term on the right hand side of equation (\ref{eq:4cl1}), which was
\begin{equation}
\|(e^{-iH_{\Lambda}t}-U_{\Lambda}^N)V\ket{\Omega_{\Lambda}}\|_2.
\end{equation}
We could just directly apply Trotter's theorem from section \ref{sec:The Lie-Trotter Product Formula}.  But it may be useful to have a bound on the error introduced by the approximation.  This is possible on a finite line since there is a momentum cutoff $\Lambda$.
\begin{equation}
 \|(e^{-iH_{\Lambda}t}-U_{\Lambda}^N)V\ket{\Omega_{\Lambda}}\|_2\leq \|e^{-iH_{\Lambda}t} -U_{\Lambda}^N\|_{\Lambda}.
\end{equation}
Then we can use theorem \ref{th:Lie} in section \ref{sec:The Lie-Trotter Product Formula} to get
\begin{equation}
 \|e^{-iH_{\Lambda}t} -U_{\Lambda}^N\|_{\Lambda}=O(K^2a),
\end{equation}
where $K=\max\{\|A\|,\|B\|\}$, with
\begin{equation}
\begin{split}
A & =\int_{-\Lambda}^{\Lambda}\!\frac{\textrm{d}p}{2\pi}\,\psi^{\dagger}_{p}p\sigma_z\psi_{p}\\
B & =\int_{-\Lambda}^{\Lambda}\!\frac{\textrm{d}p}{2\pi}\,\psi^{\dagger}_{p}m\sigma_x\psi_{p}.
\end{split}
\end{equation}
Therefore, as the number of occupied modes on a finite line is at most proportional to $\Lambda$, it follows that $K=O(\Lambda^2)$ and
\begin{equation}
 \|e^{-iH_{\Lambda}t} -U_{\Lambda}^N\|_{\Lambda}=O(\Lambda^4a).
\end{equation}
So the dynamics also converge.

In higher dimensions, the same results should still hold.  The only thing that needs more careful consideration is the Dirac quantum walk in three dimensions.  This has a four dimensional extra degree of freedom, so there are two positive and negative energy modes for each momentum.

\section{Lattice QFT without Anticommutation}
\label{sec:Lattice QFT without Anticommutation}
In the previous section, we took the continuum limits of free discrete fermion fields, which allowed us to recover free fermion fields in continuous spacetime.  The natural question to ask now is whether we can do anything similar to recover {\it interacting} quantum field theories in the continuum limit.  In other words, can we construct causal discrete models, perhaps even quantum cellular automata, that become interacting quantum field theories in the continuum limit?  We will discuss this in section \ref{sec:Quantum Cellular Automata as Quantum Field Simulators} and section \ref{sec:Causal Discrete-Time Models on a Lattice}.

Before tackling this, let us digress a little to look at lattice quantum field theory, particularly lattice fermions interacting with gauge fields.  This digression will also allow us to introduce some tools that will be useful in section \ref{sec:Causal Discrete-Time Models on a Lattice}.  It will turn out that in some cases lattice fermions interacting with gauge fields can be viewed as local models without the need for anticommuting fermion fields.  

Actually, this a modification of ideas from \cite{Ball05,VC05}, where it was shown that local fermion Hamiltonians can be viewed as local spin Hamiltonians by adding auxiliary fermions at each lattice site.  Here, instead of introducing redundant additional fermions, we partially represent the gauge degrees of freedom by pairs of fermions, which can be thought of as hardcore bosons.  This allows us to map some simple lattice quantum field Hamiltonians to qubit (or spin) Hamiltonians.  In fact, it is the principle of gauge symmetry that means that this can be done.  The fact that gauge symmetry prevents us from having nonlocal observables in the qubit picture is profound.  Indeed, one of the most interesting messages to take from this section and \cite{Ball05,VC05} is that fermionic models with gauge interactions in discrete space do not need anticommuting fields.

We start in \ref{sec:Representing Local Fermionic Hamiltonians by Local Qubit Hamiltonians} by introducing the protocol of \cite{Ball05,VC05} that allows us to map local fermion Hamiltonians to local spin Hamiltonians.  This relies on introducing additional redundant fermionic modes.  A result of this mapping is that in the qubit picture, there is entanglement in the state of the qubit system.  In \ref{sec:Preparing the Majorana state} we ensure that this entangled state can be efficiently prepared on a quantum computer, something that would be necessary if these ideas are to be used in simulations of physics.  Next, we introduce fermions interacting with a lattice gauge field in section \ref{sec:Lattice Gauge Theories}.  Finally, in section \ref{sec:U(1) Lattice Gauge theory as a spin model}, we show that the gauge degree of freedom on each link can be decomposed into two components, one of which we can represent by a pair of fermions.  This allows us to use ideas similar to those from \ref{sec:Representing Local Fermionic Hamiltonians by Local Qubit Hamiltonians} to map the Hamiltonian to a completely local model with no fermion anticommutation.

\subsection{Representing Fermion Models by Local Qubit Models}
\label{sec:Representing Local Fermionic Hamiltonians by Local Qubit Hamiltonians}
Fermion creation and annihilation operators anticommute regardless of the spatial separation between them.  This means there is inherent nonlocality in the description of fermionic systems.  Remarkably, this nonlocality can be completely removed by introducing redundant fermion modes\footnote{There are some seemingly contradictory notions of locality here.  In the fermion picture, operators are local if they are sums and products of creation and annihilation operators associated to some finite region.  The reason we sometimes refer to such things as nonlocal, is that they can be nonlocal when represented by operators on qubits.  Furthermore, creation and annihilation operators on different sites anticommute no matter how large the spatial separation.  (Nonlocality of operators here has nothing to do with quantum nonlocality resulting from entanglement.)} \cite{Ball05,VC05}.  Here, we will present this procedure in a slightly different way in terms of links between sites on cubic lattices.  Still, the core idea is the same as in \cite{Ball05,VC05}.  This will be useful for section \ref{sec:U(1) Lattice Gauge theory as a spin model}, where our goal will be to include these ideas into lattice gauge theories.

Let us see how this works.  As an example, first we will look at the fermion hopping Hamiltonian.
\begin{equation}
\label{eq:L2}
 H=\sum_{\langle\vec{n}\vec{m}\rangle}(\psi^{\dagger}_{\vec{n}}\psi^{\ }_{\vec{m}}+\psi^{\dagger}_{\vec{m}}\psi^{\ }_{\vec{n}}),
\end{equation}
where $\langle \vec{n}\vec{m}\rangle$ denotes nearest neighbour pairs.  So $\vec{m}=\vec{n}+\vec{e}$, where $\vec{e}$ is a lattice basis vector.

Now take the link between sites $\vec{n}$ and $\vec{m}$.  And look at the term $\psi_{\vec{n}}^{\dagger} \psi_{\vec{m}}$ in the Hamiltonian.  Let us introduce two additional fermionic modes, with annihilation operators $a_{(\vec{n},\vec{m})}$ and $a_{(\vec{m},\vec{n})}$.  The different ordering of the indices is associated to the direction along the link, and the first index denotes the site on which the fermion modes live.  We define the Majorana fermion operators
\begin{equation}
\begin{split}
 c_{(\vec{n},\vec{m})} & =a_{(\vec{n},\vec{m})}+a_{(\vec{n},\vec{m})}^{\dagger}\\
c_{(\vec{m},\vec{n})} & =a_{(\vec{m},\vec{n})}+a_{(\vec{m},\vec{n})}^{\dagger}.
\end{split}
\end{equation}

\begin{figure}
\centering
    \resizebox{10cm}{!}{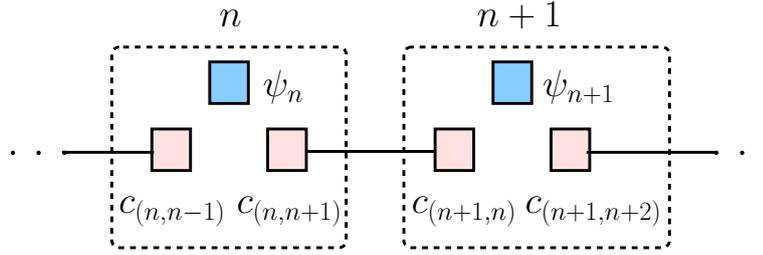}
    \footnotesize{\caption[Illustration of the Majorana pairs on a line]{Illustration of the Majorana pairs on a line.  There is a pair for each link, so on a line we need two extra fermions per site.  In general, however, we only need to add these extra fermions in spatial dimensions higher than one.}}
    \label{fig:MajoranaPairs}
\end{figure}
These are self-adjoint and anticommute with all other creation and annihilation operators.  But $c_{(\vec{n},\vec{m})}$ squares to the identity.  So $c_{(\vec{n},\vec{m})}$ obey different anticommutation relations to those obeyed by creation and annihilation operators.  Now define the operator $M_{(\vec{n},\vec{m})}$ by
\begin{equation}
M_{(\vec{n},\vec{m})}=ic_{(\vec{n},\vec{m})}c_{(\vec{m},\vec{n})}. 
\end{equation}
This operator is self-adjoint and has eigenvalues $\pm1$, which follows from $M_{(\vec{n},\vec{m})}^2=1$ and $M_{(\vec{n},\vec{m})}\neq 1$.

Next, we make the transformation
\begin{equation} 
\psi_{\vec{n}}^{\dagger} \psi_{\vec{m}} \rightarrow \psi_{\vec{n}}^{\dagger} M_{(\vec{n},\vec{m})} \psi_{\vec{m}}=i\left(\psi_{\vec{n}}^{\dagger}\, c_{(\vec{n},\vec{m})}\right)\Big(c_{(\vec{m},\vec{n})} \psi_{\vec{m}}\Big).
\end{equation}
Acting on a $+1$ eigenstate of $M_{(\vec{n},\vec{m})}$, the transformed operator above is equivalent to the original operator.  And because $c_{(\vec{m},\vec{n})}$ anticommutes with other fermionic operators, the operator $c_{(\vec{m},\vec{n})} \psi_{\vec{m}}$ commutes with fermionic operators on all other sites.  Therefore, it commutes with {\it all} operators on other sites, so it can be thought of as local.  Of course, $c_{(\vec{m},\vec{n})} \psi_{\vec{m}}$ is local as a fermionic operator.  Here we really mean that, because it commutes with operators on all other sites, after the Jordan-Wigner transformation, it is local in the qubit picture.  We will see this in the following section.

We can apply the same trick on every link.  And there is a joint $+1$ eigenstate of all $M_{(\vec{n},\vec{m})}$ because they all commute.  In section \ref{sec:Preparing the Majorana state}, we will construct one such state and show that it can be efficiently prepared on a quantum computer.

Furthermore, because $[\psi^{\dagger}_{\vec{k}},M_{(\vec{n},\vec{m})}]=0$ for any $\vec{n}$, $\vec{m}$ and $\vec{k}$, it follows that we can act on a $+1$ eigenstate of $M_{(\vec{n},\vec{m})}$ with physical fermion creation operators and still have a $+1$ eigenstate.  Therefore, the original fermion Hilbert space is mapped to a subspace of the new state space.

It is crucial that self-adjointness is preserved.  This follows because $M_{(\vec{n},\vec{m})}$ are self-adjoint.  Then, for example,
\begin{equation}
 \psi_{\vec{n}}^{\dagger} \psi_{\vec{m}}+\psi_{\vec{m}}^{\dagger} \psi_{\vec{n}} \rightarrow \psi_{\vec{n}}^{\dagger} M_{(\vec{n},\vec{m})} \psi_{\vec{m}}+\psi_{\vec{m}}^{\dagger} M_{(\vec{n},\vec{m})} \psi_{\vec{n}}.
\end{equation}
This means that the Hamiltonian and all other self-adjoint operators are mapped to self-adjoint operators.  And as a result of this, the  dynamics are preserved.

\begin{compactwrapfigure}{r}{0.34\textwidth}
\centering
\begin{minipage}[r]{0.32\columnwidth}%
\centering
    \resizebox{3.2cm}{!}{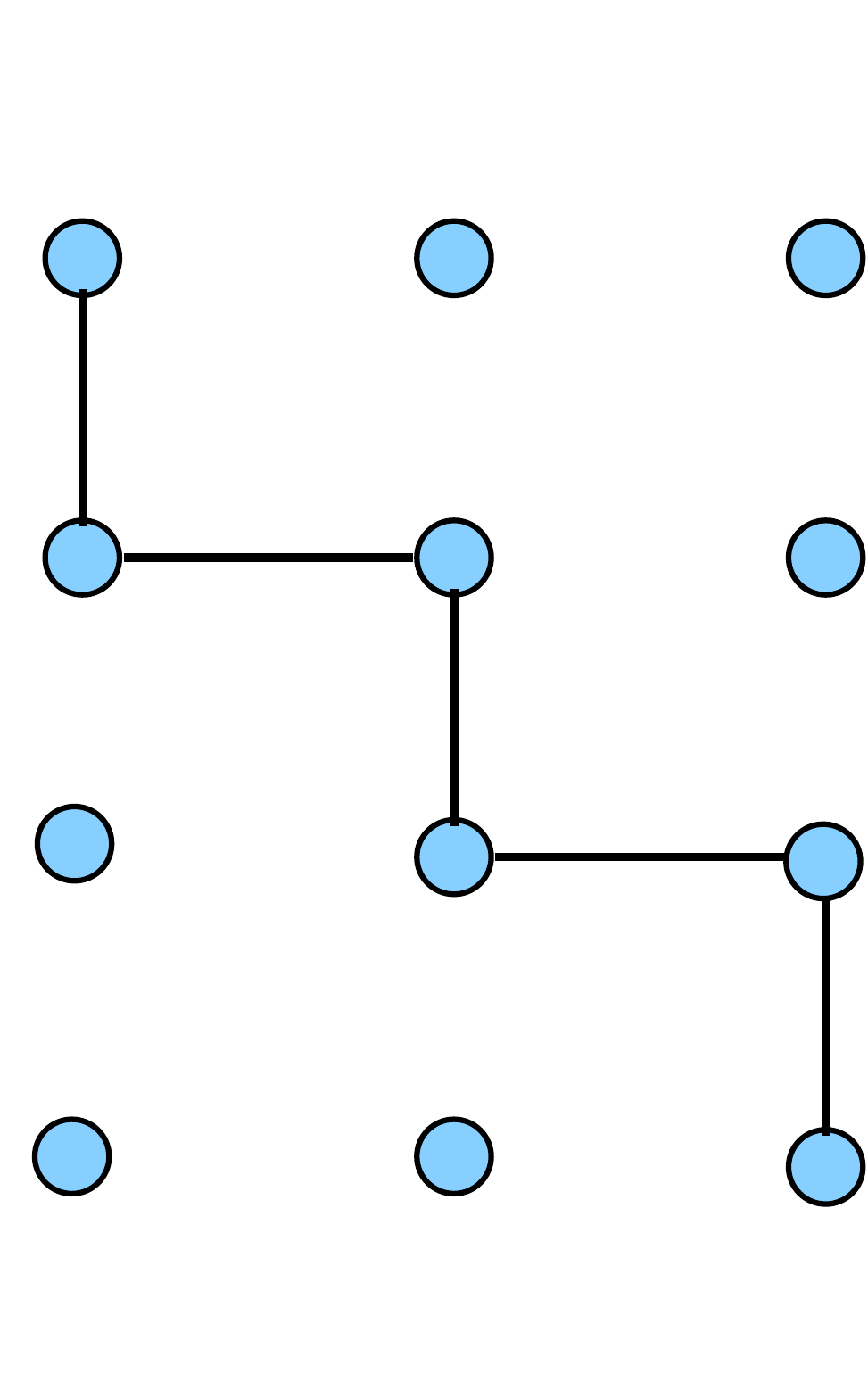}
    \footnotesize{\caption{Path along links on a lattice}}
    \label{fig:MajoranaPath}
\end{minipage}
\end{compactwrapfigure}

Higher order terms can be made local in the same way.  For example, we can replace $\psi^{\dagger}_{\vec{n}}\psi_{\vec{m}}\psi_{\vec{k}}^{\dagger}\psi_{\vec{l}}$ by
\begin{equation}
\psi^{\dagger}_{\vec{n}}M_{(\vec{n},\vec{m})}\psi_{\vec{m}}\psi_{\vec{k}}^{\dagger}M_{(\vec{k},\vec{l})}\psi_{\vec{l}},
\end{equation}
where $\vec{n}$ and $\vec{m}$ are nearest neighbour pairs, as are $\vec{k}$ and $\vec{l}$.  Furthermore, if there is a product of fermion operators that are not neighbours, but are separated by a fixed distance, then we use the prescription
\begin{equation}
 \psi_{\vec{n}}^{\dagger} \psi_{\vec{m}} \rightarrow \psi_{\vec{n}}^{\dagger} \prod_i M_{(\vec{k}_i,\vec{l}_i)} \psi_{\vec{m}},
\end{equation}
where $(\vec{k}_i,\vec{l}_i)$ are links on a path from $\vec{m}$ to $\vec{n}$.  Note the similarity to making operators gauge invariant by inserting gauge field operators on links.  We will see more of this in sections \ref{sec:Lattice Gauge Theories} and \ref{sec:U(1) Lattice Gauge theory as a spin model}.

The requirement that observables are sums of even products of creation and annihilation operators was vital to be able to use this trick.

\subsubsection{Representation by Qubits}
We know from section \ref{sec:The Jordan-Wigner Transformation} that fermions can be mapped to qubits via the Jordan-Wigner transformation.  Generally, in higher than one dimensional space this transformation takes local fermion operators to nonlocal qubit operators.

We can apply the Jordan-Wigner transformation given by equation (\ref{eq:general ordering}) from section \ref{sec:The Jordan-Wigner Transformation}.  This works by choosing an ordering scheme $\pi(\vec{n},i)$, where $\vec{n}$ denotes the site and $i$ denotes the mode at that site.  Then we have the map
\begin{equation}
\psi^{\dagger}_{\vec{n}i}\equiv \sigma^-_{\vec{n},i}\hspace{-1em}\prod_{\pi(\vec{m},j)<\pi(\vec{n},i)}\hspace{-1em}Z_{\vec{m},j}.
\end{equation}
Some local fermion observables will be nonlocal after this mapping.

We know that products of two fermion operators at the same site will not be nonlocal in the qubit picture if they are consecutive in the ordering.  With the auxiliary modes included, it is natural to choose the new ordering such that fermions at the same site are still consecutive in the ordering scheme.  Because of this, $c_{(\vec{n},\vec{m})}\psi_{\vec{n}}$ is local in the qubit picture.  And, in particular, a local fermion Hamiltonian will be mapped to a local qubit Hamiltonian by making this judicious choice of the ordering for the Jordan-Wigner transformation.

Finally, if simulating physics is the goal, we may want to be more economical.  One option is to only add auxiliary modes on links that correspond to nonlocal hopping terms in the corresponding qubit Hamiltonian.  Note that the number of additional fermionic modes we need at each site does not depend on the number of physical fermionic modes; it just depends on the spatial dimension.

\subsection{Preparing the Entangled Initial State}
\label{sec:Preparing the Majorana state}
The result of the last section is remarkable.  From a conceptual point of view, it is surprising that the intrinsic nonlocality of fermionic systems can be removed.  Nevertheless, for this to be of any use for simulations of physics, it is paramount that a state of the qubits that allows this can be efficiently prepared.  Otherwise, any benefits gained by the mapping would be lost.  Let us follow the procedure presented in \cite{FS13}.

We want to prepare a state in the qubit picture that is a $+1$ eigenstate of every $M_{(\vec{n},\vec{m})}=ic_{(\vec{n},\vec{m})}c_{(\vec{m},\vec{n})}$.  To do this, we will use
\begin{equation}
M_{(\vec{n},\vec{m})}\Big(c_{(\vec{n},\vec{m})}-ic_{(\vec{m},\vec{n})}\Big)=\Big(c_{(\vec{n},\vec{m})}-ic_{(\vec{m},\vec{n})}\Big),
\end{equation}
which follows from $c_{(\vec{n},\vec{m})}^2=c_{(\vec{m},\vec{n})}^2=1$.  This means that
\begin{equation}
\label{eq:L5}
 \prod_{\langle \vec{n}\vec{m}\rangle}\frac{1}{\sqrt{2}}(c_{(\vec{n},\vec{m})}-ic_{(\vec{m},\vec{n})})\ket{0}
\end{equation}
is a normalized $+1$ eigenstate of all Majorana pairs $M_{(\vec{n},\vec{m})}$, where $\langle \vec{n}\vec{m}\rangle$ denotes pairs of sites joined by a link.  The order of the product is immaterial since any order will be a $+1$ eigenstate of the $M_{(\vec{n},\vec{m})}$ pairs.  Recall that
\begin{equation}
\label{eq:LQ12}
 c_{(\vec{n},\vec{m})}=a_{(\vec{n},\vec{m})}+a^{\dagger}_{(\vec{n},\vec{m})},
\end{equation}
where $a_{(\vec{n},\vec{m})}$ is an annihilation operator.  If follows that the unitary $c_{(\vec{n},\vec{m})}$ has the same strings of $Z$s as $a_{(\vec{n},\vec{m})}$ in the qubit picture.  In fact, in the qubit picture $c_{(\vec{n},\vec{m})}$ is just $X$ on the qubit corresponding to that mode and $Z$s on some others.

We want to create the invariant state on qubits, which means that we need to apply $c_{(\vec{n},\vec{m})}$.  Unfortunately, in the qubit representation the operators $c_{(\vec{n},\vec{m})}$ have those awkward strings of $Z$s, making it a nonlocal unitary.  Fortunately, we can handle these by using a method from \cite{AL97}.  First, consider all the qubits that the qubit representation of $c_{(\vec{n},\vec{m})}$ acts on with a $Z$.  We map the parity of these qubits to a flag qubit.  So a single $Z$ acting on the flag qubit has the exact same effect as applying the string of $Z$s to the other qubits.  For example, with $r_i\in\{0,1\}$,
\begin{equation}
\begin{split}
 Z_0...Z_n\ket{r_0...r_n} & \ket{\sum_{j=0}^n r_j\bmod 2}=\\
 (-1)^{\sum_{j=0}^n r_j\bmod 2}\ket{r_0...r_n} & \ket{\sum_{j=0}^n r_j\bmod 2}=\\
 \ket{r_0...r_n}Z & \ket{\sum_{j=0}^n r_j\bmod 2},
\end{split}
\end{equation}
where the  last qubit is the flag qubit storing the parity of the other qubits.  So after preparing flag qubits, $c_{(\vec{n},\vec{m})}$ is equivalent to a unitary on two qubits.  After applying $c_{(\vec{n},\vec{m})}$, we need to reverse the operation preparing the flag qubits, but this and the original flag preparation can be done using only $K$ two-qubit unitaries, where $K$ is the number of qubits we count the parity of.  For example, if $K=2$, we need only two steps:
\begin{equation}
 \ket{r_1r_2}\ket{0}_{f}\rightarrow \ket{r_1r_2}\ket{r_1}_{f} \rightarrow
 \ket{r_1r_2}\ket{(r_1+r_2)\bmod 2}_{f},
\end{equation}
where $r_i\in\{0,1\}$ and the subscript $f$ denotes the flag qubit.

Actually, we want to apply $(c_{(\vec{n},\vec{m})}-ic_{(\vec{m},\vec{n})})/\sqrt{2}$ to the state, which is not unitary.  But the desired state can be prepared by using an ancillary qubit first mapped to $(\ket{0}_A-i\ket{1}_A)/\sqrt{2}$.  Then, we apply a unitary controlled on the ancillary qubit.  It works by acting with the unitary $c_{(\vec{n},\vec{m})}$ if the qubit is in state $\ket{0}_A$ and the unitary $c_{(\vec{m},\vec{n})}$ if the qubit is in state $\ket{1}_A$.  Afterwards, we can disentangle the ancilla and return it to the state $\ket{0}_A$ by applying the four-qubit unitary
\begin{equation}
(\openone-a_{(\vec{m},\vec{n})}^{\dagger}a_{(\vec{m},\vec{n})})\openone_A+a_{(\vec{m},\vec{n})}^{\dagger}a_{(\vec{m},\vec{n})}X_A,
\end{equation}
where the subscript $A$ indicates that the operators act on the ancillary qubit.

Each site has a constant number of qubits associated to it because the number of Majorana fermion pairs that we introduce per site does not grow with the number of sites $\mathcal{N}$.  It just depends on the number of nearest neighbours, which depends on the lattice dimension.  So counting the parity takes $O(\mathcal{N})$ two-qubit operations.  Therefore, preparing the $+1$ eigenstate of all Majorana link terms takes $O(\mathcal{N}^2)$ operations.


Similarly, to create initial physical fermion states, we can use the same method to handle the strings of $Z$s that appear if we want to apply unitaries like $\psi^{\dagger}_{\vec{n}}+\psi_{\vec{n}}$ to create particles.

The state in equation (\ref{eq:L5}) in the qubit picture is entangled.  So, in some rough sense, the nonlocality of fermionic systems has been replaced by entanglement in the qubit system.  It is good to emphasize that the nonlocality we are removing is the nonlocality of the fermion operators when represented in the qubit picture.  The possibility of nonlocal correlations due to entanglement between systems remains.

\subsection{Lattice Gauge Theories}
\label{sec:Lattice Gauge Theories}
Earlier, in section \ref{sec:Lattice Quantum Field Theory and Fermion Doubling}, we looked at a specific model of fermions on a lattice, called naive fermions.  Here, we will be more general and look at an arbitrary local free fermion lattice Hamiltonian.  This would look like
\begin{equation}
\label{eq:L1}
H_{F}= \sum_{\left<\vec{n}\vec{m}\right>} (\psi^{\dagger}_{\vec{n}}B_{\vec{n}\vec{m}}\psi^{\ }_{\vec{m}}+\psi^{\dagger}_{\vec{m}}B^{\dagger}_{\vec{n}\vec{m}}\psi^{\ }_{\vec{n}}),
\end{equation}
where $\left<\vec{n}\vec{m}\right>$ denotes pairs of sites connected by the Hamiltonian.  These need not be nearest neighbour and may include $\vec{n}=\vec{m}$ as there may be on-site terms.  We have suppressed any extra labels that $\psi_{\vec{n}}$ may carry.  So $B_{\vec{n}\vec{m}}\equiv (B_{\vec{n}\vec{m}})_{ij}$ is a matrix coupling fermion modes.  Had we explicitly written these indices, we would have
\begin{equation}
 \psi^{\dagger}_{\vec{n}}B_{\vec{n}\vec{m}}\psi_{\vec{m}}=\sum_{ij}\psi^{\dagger}_{\vec{n}i}\left(B_{\vec{n}\vec{m}}\right)_{ij}\psi_{\vec{m}j}.
\end{equation}

Again, one example of a local fermion Hamiltonian is the naive fermion Hamiltonian on a line, which we saw in section \ref{sec:Lattice Quantum Field Theory and Fermion Doubling}.  It is
\begin{equation}
 H_{F}=\frac{i}{2a}\sum_{n}(\psi^{\dagger}_{n}\sigma_z\psi^{\ }_{n-1}-\psi^{\dagger}_{n}\sigma_z\psi^{\ }_{n+1}+m\psi^{\dagger}_{n}\sigma_x\psi^{\ }_{n}).
\end{equation}
So in this case $B_{nn}=m\sigma_x$ and $B_{n,n\pm1}=\mp \f{i}{2a}\sigma_z$.

Now let us introduce abelian gauge fields.  At the end, we will have quantum electrodynamics (QED) on a lattice.  First, under a gauge transformation, fermion fields transform like $\psi_{\vec{n}}\rightarrow V(\vec{n})\psi_{\vec{n}}$, where $V(\vec{n})$ is a phase that depends on position.  The free fermion Hamiltonians we have just seen do not have symmetry under such gauge transformations.  This is because the Hamiltonian in equation (\ref{eq:L1}) contains products of fermion operators at different sites.  So, let us replace this Hamiltonian by
\begin{equation}
H_{F}= \sum_{\left<\vec{n}\vec{m}\right>} (\psi^{\dagger}_{\vec{n}}B_{\vec{n}\vec{m}}U_{\vec{n}\vec{m}}\psi^{\ }_{\vec{m}}+\psi^{\dagger}_{\vec{m}}B^{\dagger}_{\vec{n}\vec{m}}U^{\dagger}_{\vec{n}\vec{m}}\psi^{\ }_{\vec{n}}).
\end{equation}
Then, if we also require that under a gauge transformation, $U_{\vec{n}\vec{m}}\rightarrow V(\vec{n})U_{\vec{n}\vec{m}}V(\vec{m})^{\dagger}$, this is a gauge invariant Hamiltonian.  It is natural to think of $U_{\vec{n}\vec{m}}$ as living on the link joining sites $\vec{n}$ and $\vec{m}$.  Also, we can make the connection with the continuum theory by writing $U_{\vec{n}\vec{m}}=e^{igaA_1(\vec{n}a)}$, where $g$ is the electric charge and $A_1(\vec{n}a)$ is the gauge field at the point $\vec{x}=\vec{n}a$ \cite{DGDT06}.

Currently, in the Hamiltonian above $U_{\vec{n}\vec{m}}$ is not dynamical.  Let us change this by introducing a Hamiltonian for the gauge degrees of freedom.  First, we define the self-adjoint operators $E_{\vec{n}\vec{m}}$ to be the electric field on the same link as $U_{\vec{n}\vec{m}}$.  Any two operators on different links commute, but on the same link,
\begin{align}
 [E_{\vec{n}\vec{m}},U_{\vec{n}\vec{m}}]=U_{\vec{n}\vec{m}}.
\end{align}

The generators of a gauge transformation in the absence of fermions are
\begin{equation}
 G_{\vec{n}}=\sum_{\vec{e}}\left[E_{\vec{n},\vec{n}+\vec{e}}-E_{\vec{n},\vec{n}-\vec{e}}\right],
\end{equation}
where $\vec{e}$ are lattice basis vectors.  So a simple example of a gauge invariant operator is the electric field itself $E_{\vec{n}\vec{m}}$.  Another gauge invariant quantity is a product of link operators on plaquettes:
\begin{equation}
 Z_p=\big(U_{\vec{n},\, \vec{n}+\vec{e}}\big)\big(U_{\vec{n}+\vec{e},\, \vec{n}+\vec{f}+\vec{e}}\big)\big(U_{\vec{n}+\vec{f}+\vec{e},\, \vec{n}+\vec{f}}\big)\big(U_{\vec{n}+\vec{f},\, \vec{n}}\big),
\end{equation}
where $p$ labels the plaquette.

\begin{compactwrapfigure}{r}{0.44\textwidth}
\centering
\begin{minipage}[r]{0.42\columnwidth}%
\centering
    \resizebox{3.4cm}{!}{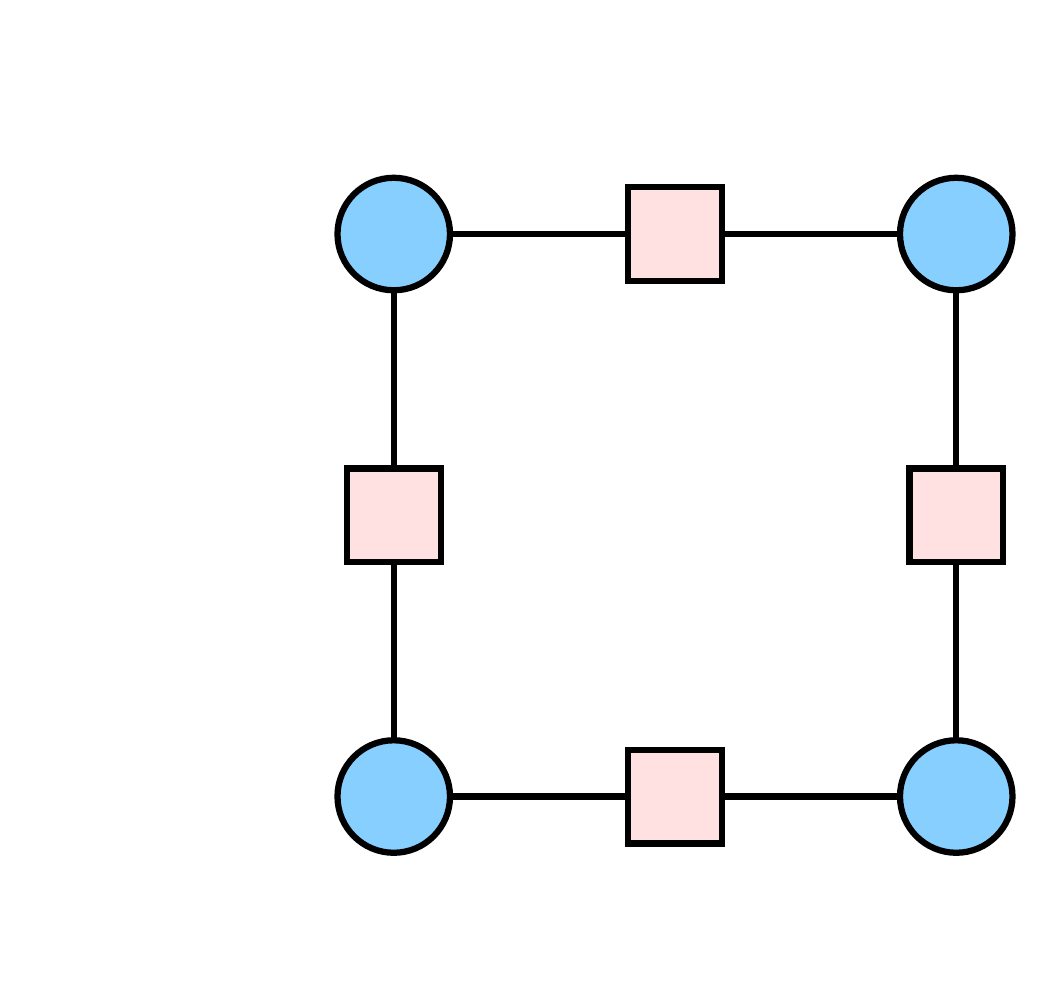}
    \footnotesize{\caption[Gauge operators around a plaquette]{Gauge operators around a plaquette.}}
    \label{fig:2DimGauge_ch4}
\end{minipage}
\end{compactwrapfigure}

A lattice Hamiltonian with free electromagnetism as its continuum limit \cite{KS75} is
\begin{align}
H_{G} = \beta\sum_{\vec{n},\vec{e}}E_{\vec{n},\vec{n}+\vec{e}}^2-\gamma\sum_{p}[Z_p+Z^{\dagger}_p],
\end{align}
where $\beta$ and $\gamma$ are constants that depend on the electric charge and the lattice spacing.  So, finally, the Hamiltonian for quantum electrodynamics on a lattice is
\begin{align}
 H  =\sum_{\left<\vec{n}\vec{m}\right>} (\psi^{\dagger}_{\vec{n}}B_{\vec{n}\vec{m}}U_{\vec{m}\vec{n}}\psi^{\ }_{\vec{m}}+\psi^{\dagger}_{\vec{m}}B^{\dagger}_{\vec{n}\vec{m}}U^{\dagger}_{\vec{n}\vec{m}}\psi^{\ }_{\vec{n}}) +\beta\sum_{\vec{n},\vec{e}}E_{\vec{n},\vec{n}+\vec{e}}^2-\gamma\sum_{p}[Z_p+Z^{\dagger}_p].
\end{align}
The generators of gauge transformations in the presence of fermions are now
\begin{equation}
 G_{\vec{n}}=\psi^{\dagger}_{\vec{n}}\psi_{\vec{n}}+\sum_{\vec{e}}\left[E_{\vec{n},\vec{n}+\vec{e}}-E_{\vec{n},\vec{n}-\vec{e}}\right],
\end{equation}
which commute with the Hamiltonian.  For a fully gauge invariant theory we also need $G_{\vec{n}}$ to vanish on physical states:\ this is Gauss' law.

\subsection{Lattice QED as a Local Qubit Model}
\label{sec:U(1) Lattice Gauge theory as a spin model}
Section \ref{sec:Representing Local Fermionic Hamiltonians by Local Qubit Hamiltonians} tells us that we can always add redundant extra degrees of freedom to map a local fermion model to a local qubit model.  Now we ask the question of whether we can do better:\ can we somehow incorporate these degrees of freedom into the gauge field?  Here, we will see that, to some extent, we can.

As we saw in the previous section, the full Hamiltonian for quantum electrodynamics on a lattice is
\begin{align}
 H  =\sum_{\left<\vec{n}\vec{m}\right>} (\psi^{\dagger}_{\vec{n}}B_{\vec{n}\vec{m}}U_{\vec{m}\vec{n}}\psi^{\ }_{\vec{m}}+\psi^{\dagger}_{\vec{m}}B^{\dagger}_{\vec{n}\vec{m}}U^{\dagger}_{\vec{n}\vec{m}}\psi^{\ }_{\vec{n}}) +\beta\sum_{\vec{n},\vec{e}}E_{\vec{n},\vec{n}+\vec{e}}^2-\gamma\sum_{p}[Z_p+Z^{\dagger}_p].
\end{align}
The Hilbert space of the gauge degrees of freedom can be constructed as follows \cite{KS75,BY06}.  Looking at a single link, and dropping the link label, we have
\begin{equation}
 [E,U]=U.
\end{equation}
Because $U$ is a phase, this system is equivalent to a quantum particle that lives on a circle.  So $E$ is analogous to the angular momentum operator.  Suppose now that $\ket{0}$ is the state of the link's Hilbert space satisfying $E\ket{0}=0$.  Then it follows that
\begin{equation}
 E^2U^l\ket{0}=l^2U^l\ket{0},
\end{equation}
so we can define the orthonormal basis $\ket{l}=U^l\ket{0}$, where $l$ is an integer.  In this basis, we can think of $U$ as a shift operator.\footnote{With this state space for the gauge fields, this is known as compact quantum electrodynamics.}

For our purposes, however, it will be useful to use a different representation of these states.  We will split up the Hilbert space into a tensor product of a qubit and the remainder.  To achieve this, the simplest thing to do is to let the qubit represent the parity of $l$.  So let
\begin{equation}
 \ket{n}\ket{r}_p=\ket{l},
\end{equation}
where $r\in\{0,1\}$ represents the parity and $n$ is an even integer.  Both are chosen such that $l=n+r$.  

We need to see what $U$ looks like in this representation:
\begin{equation}
U=\big(\ket{0}_p\bra{0}S+\ket{1}_p\bra{1}\big)X_p, 
\end{equation}
where $S\ket{n}=\ket{n+2}$.  This has the same effect as $U$ but now in the $\ket{n}\ket{r}_p$ representation.  Coincidentally, in this basis $U$ is a quantum walk operator, where the parity qubit is playing the role of the coin.  In the original representation, $U$ was just a shift, so that was an even simpler quantum walk.

The next step involves representing the parity by two fermion modes with creation operators $a^{\dagger}$ and $b^{\dagger}$.  Denoting the state with no fermions present by $\ket{0_F}$, let us assign
\begin{equation}
 \begin{split}
  \ket{0}_p & = \ket{0_F},\\
  \ket{1}_p & = a^{\dagger}b^{\dagger}\ket{0_F}.
 \end{split}
\end{equation}
In a sense, this corresponds to separating the gauge degree of freedom into a hard-core boson and a single particle living on the $\ket{n}$ states.  Now, with this representation, the operator $X_p$ can be represented by
\begin{equation}
 X_p = a^{\dagger}b^{\dagger}+ba,
\end{equation}
and the projectors $\ket{0}_p\bra{0}$ and $\ket{1}_p\bra{1}$ can be replaced by
\begin{equation}
\begin{split}
 \ket{0}_p\bra{0} & = aa^{\dagger}bb^{\dagger}\\
 \ket{1}_p\bra{1} & = a^{\dagger}ab^{\dagger}b.
\end{split}
\end{equation}
This trick of representing qubits by pairs of fermions will be useful in section \ref{sec:Representation by Qubits}.  Here it allows us to represent $U$ by\footnote{This is no longer unitary since it annihilates the states $a^{\dagger}\ket{0_F}$ and $b^{\dagger}\ket{0_F}$.  We could make it unitary by choosing $X_p=a^{\dagger}b^{\dagger}+ba+b^{\dagger}a+a^{\dagger}b$, which works just as well.  But this is unnecessary and makes the formulas longer.}
\begin{equation}
U=a^{\dagger}b^{\dagger}-Sab.
\end{equation}

Now recall that the fermionic part of the Hamiltonian is
\begin{align}
H_{F} = \alpha\sum_{\left<\vec{n}\vec{m}\right>} (\psi^{\dagger}_{\vec{n}}B_{\vec{n}\vec{m}}U_{\vec{n}\vec{m}}\psi^{\ }_{\vec{m}}+\psi^{\dagger}_{\vec{m}}B^{\dagger}_{\vec{n}\vec{m}}U^{\dagger}_{\vec{n}\vec{m}}\psi^{\ }_{\vec{n}}).
\end{align}
Look at the hopping term between a given pair of neighbouring sites $\vec{n}$ and $\vec{m}$, $\psi^{\dagger}_{\vec{n}}B_{\vec{n}\vec{m}}U_{\vec{n}\vec{m}}\psi^{\ }_{\vec{m}}$.  Inserting our new expression for $U$, this becomes
\begin{equation}
 \psi^{\dagger}_{\vec{n}}B_{\vec{n}\vec{m}}\left(a^{\dagger}_{\vec{n}}b^{\dagger}_{\vec{m}}-S_{\vec{n}\vec{m}}a_{\vec{n}}b_{\vec{m}}\right)\psi^{\ }_{\vec{m}},
\end{equation}
where $a_{\vec{n}}$ is chosen to live at site $\vec{n}$ and $b_{\vec{m}}$ is chosen to live at site $\vec{m}$.  This is very similar to what we saw in section \ref{sec:Representing Local Fermionic Hamiltonians by Local Qubit Hamiltonians}.  This means that this term only has quadratic products of fermion operators on each site, like $b_{\vec{m}}\psi^{\ }_{\vec{m}}$.  The resulting quadratic operators commute with their counterparts on any other site on the lattice.

As for the electric field on the link, it becomes
\begin{equation}
 E= N+a^{\dagger}ab^{\dagger}b,
\end{equation}
where $N\ket{n}=n\ket{n}$.  The second term is also quadratic in fermion operators on a site, so this commutes with operators on other links.  The same is true for the plaquette operators $Z_p$.  And we can repeat the process on each link on the lattice.  Therefore, it follows that the Hamiltonian is composed of local commuting operators.

To make this more concrete, it is useful to map the system to the qubit picture using the Jordan-Wigner transformation.  As in section \ref{sec:Representing Local Fermionic Hamiltonians by Local Qubit Hamiltonians}, we have the convention that fermions at the same site are sequential in the ordering scheme under this mapping.  So, gauge invariant operators, such as $\psi^{\dagger}_{\vec{n}}U_{\vec{n}\vec{m}}\psi_{\vec{m}}$ are all local qubit operators.  It follows that the Hamiltonian is a local qubit Hamiltonian.

How we extracted the parity qubit can be repeated by applying the same procedure to $S$ and $\ket{n}$ that we applied to $U$ and $\ket{l}$.  This, together with a truncation scheme, which would render the number of degrees of freedom finite, like the method employed in \cite{BY06}, would allow us to represent the gauge degrees of freedom by a finite number of qubits.  Then, as long as the original parity qubit was represented by a pair of fermions in the manner we have just seen, the whole system would be equivalent to a system with a finite number of qubits at each site evolving via a local Hamiltonian.

This is the end of our digression.  Now we will return to our main interest:\ quantum cellular automata and other causal models in discrete spacetime.

\section{Nature as a Quantum Cellular Automaton}
\label{sec:Quantum Cellular Automata as Quantum Field Simulators}
The title of this section must appear ambitious, but as we will see the idea is not so far-fetched.  The question we really want to ask is whether it is possible to construct quantum cellular automata that can approximate physical interacting field theories.  After all, the best theory we currently have to describe everything in nature, excluding gravity, is the standard model, an interacting quantum field theory.  And, in general, interacting quantum field theories are defined by the continuum limits of the corresponding lattice quantum field theories \cite{Creutz83}.  Therefore, if we can approximate the lattice quantum field theory arbitrarily well with a QCA, we can approximate the corresponding continuum model arbitrarily well.  We will deal exclusively with the Hamiltonian formulation of lattice quantum field theory \cite{KS75}.  A different approach is the path integral formulation, but typically such models are not unitary, except in the continuum limit \cite{Creutz12}.

A point we will come back to is that lattice quantum field theory comes with its share of problems.  Fermion doubling is one good example.  And the causal discrete spacetime models in this section will inherit these problems.  Furthermore, these causal models may not be the most natural or efficient ways to simulate physical models.  An upshot is that this section illustrates	 that causal discrete spacetime models are more general in a sense than local Hamiltonian models.

Let us see that we can indeed approximate the evolution of quantum fields on a lattice by a quantum cellular automaton.  Provided that the Hamiltonian is local,\footnote{This is not necessarily the case in lattice quantum field theory.  For example, it is believed that, for free fermions on a lattice, having couplings between sites that fall off exponentially quickly with distance is still sufficient to get the right continuum limit \cite{DGDT06}.} we can always simulate the dynamics by using the Lie-Trotter product formula, which we saw in chapter \ref{chap:Background1}.  To reiterate, suppose $H=\sum_lH_l$, where $H_l$ are local terms.  Then, by the Lie-Trotter product formula,
\begin{equation}
\label{eq:N1}
 (\prod_le^{-iH_l t/N})^{N}\rightarrow e^{-iHt},
\end{equation}
as $N=t/\varepsilon$ goes to infinity.  Now the order in the product over $l$ is something we can choose.  As long as the Hamiltonian is local, we can divide the set of all $H_l$ into a finite number of subsets, where all elements in a given subset act only on regions that do not overlap.  When, as is typically the case, $H$ is translationally invariant, we can pick the subsets such that they are translations of each other.  Denote these subsets by $V_i$. Then we can recast equation (\ref{eq:N1}) as
\begin{equation}
 (\prod_{V_i}\prod_{l\in V_i}e^{-iH_l t/N})^{N}\rightarrow e^{-iHt}.
\end{equation}

\begin{compactwrapfigure}{r}{0.52\textwidth}
\centering
\begin{minipage}[r]{0.50\columnwidth}%
\centering
    \resizebox{7.3cm}{!}{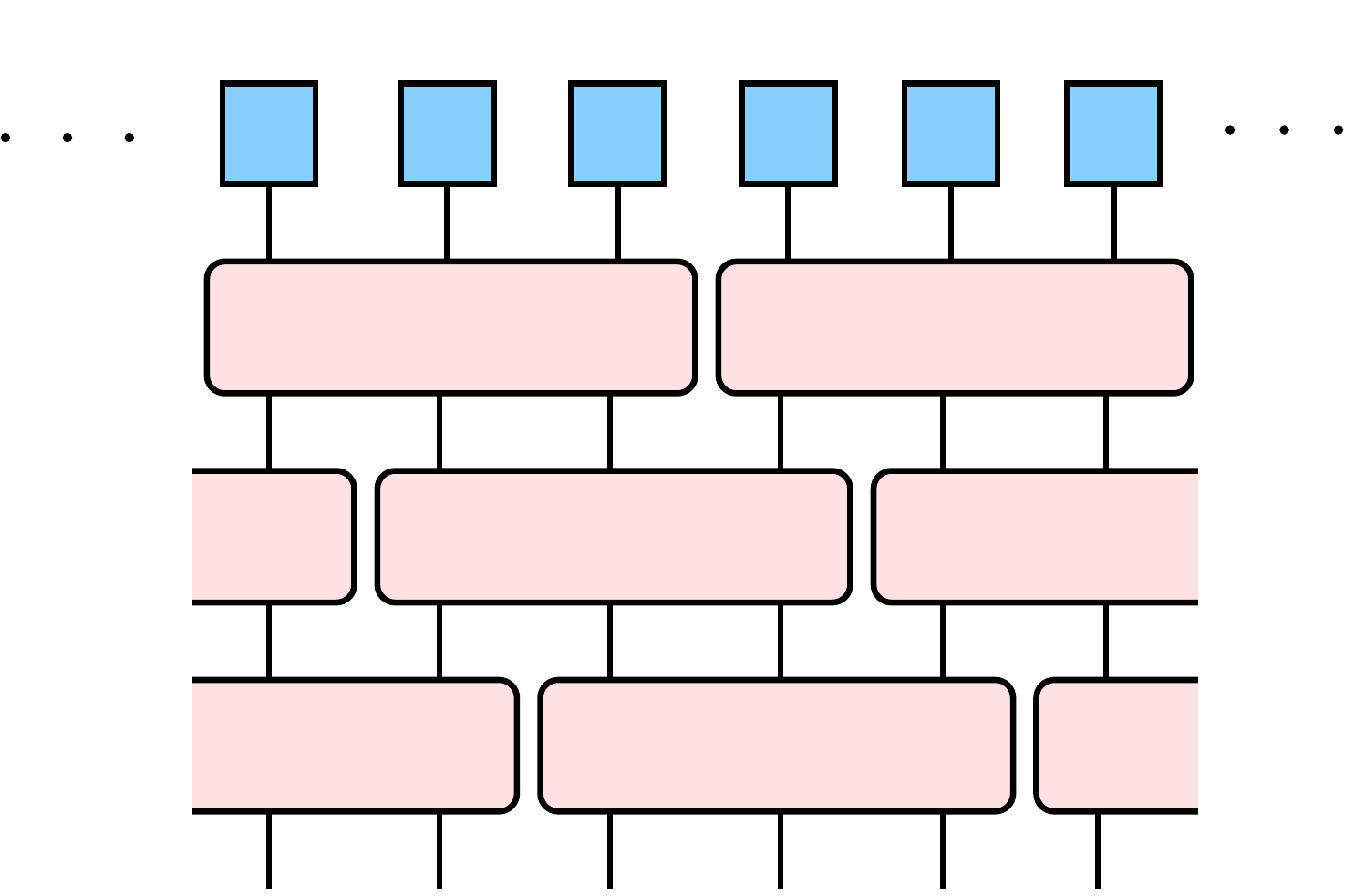}
    \footnotesize{\caption[Local Hamiltonian dynamics as a quantum cellular automaton]{Evolution of the QCA over one timestep.  We have assumed that the Hamiltonian is nearest neighbour.  The term $H_n$ in the Hamiltonian acts on sites $n-1$, $n$ and $n+1$, so the circuit has depth three.}}
    \label{fig:TrotterQCA}
\end{minipage}
\end{compactwrapfigure}

Now over each timestep our lattice system evolves via the unitary
\begin{equation}
 U=\prod_{V_i}\prod_{l\in V_i}e^{-iH_l t/N},
\end{equation}
but this is causal since there is a finite number of subsets $V_i$.  So, a quantum cellular automaton with the above evolution operator approximates the evolution of quantum fields on a lattice.  Note that the QCA needed to approximate evolution via a lattice Hamiltonian may change with the lattice spacing to retain the same level of approximation.  This is because it is the size of the timestep $t/N$ that controls how accurate the approximation is.

And since the lattice quantum field theory approximates the continuum theory, we can approximate the continuum theory with a quantum cellular automaton.  A very naive argument justifying this follows from the triangle inequality:
\begin{equation}
\label{eq:N2}
 \|\left(U^N-e^{-iH_{\textrm{phys}}t}\right)\ket{\psi}\|_2\leq \|\left(U^N-e^{-iHt}\right)\ket{\psi}\|_2+\|\left(e^{-iHt}-e^{-iH_{\textrm{phys}}t}\right)\ket{\psi}\|_2,
\end{equation}
where $H$ is the lattice Hamiltonian, and $H_{\textrm{phys}}$ is the continuum physical Hamiltonian.  Because both terms on the right hand side go to zero as the lattice spacing goes to zero, the QCA dynamics should converge to the physical dynamics.

Of course, there are probably many details that would need to be ironed out, and the formula above is not likely to be the way to do this.  For one thing, we made no mention of what Hilbert space $\ket{\psi}$ lives in.  In fact, it is not clear whether $H_{\textrm{phys}}$ and its Hilbert space even exist.  The approach of equation (\ref{eq:N2}) assumes that they do.  It is probably better to argue that, provided the physical predictions of lattice QFT converge to observed values, and since we can simulate the lattice QFT arbitrarily well with a QCA, then we know that the QCA can approximate nature arbitrarily well.

Now there are some very important points to mention here.  We have made no mention of the efficiency of this approximation by QCAs.  One thing that could affect this is that the coupling constants of lattice Hamiltonians are not independent of the lattice spacing.  And finding out how they scale with $a$ is a very difficult problem \cite{JLP12}.  This would affect state preparation efficiency for one thing.  It could also affect the efficiency of the Trotter decomposition, as the error in the approximation is a function of the operator norm of $H_l$, but if the couplings grow with decreasing $a$, then so will $\|H_l\|$.  How couplings scale cannot affect whether the Trotter (and hence the QCA) approximation are valid since we can pick $t/N$ to be as small as we like, but this could affect how efficient they are as algorithms for simulation.

In fact, this problem of renormalization is one of the key obstacles to seeing whether quantum field theories can be simulated by quantum computers, something we discussed briefly in chapter \ref{chap:Background1}.  This was overcome in \cite{JLP12} for $\phi^4$ theory.\footnote{Another issue is that the lattice Hamiltonian may also have systems that have infinite degrees of freedom, like bosonic modes.  One way to take care of this is to truncate these degrees of freedom, as done in \cite{BY06}.}

To sum up, from a conceptual point of view, it is not inconceivable that nature could be a quantum cellular automaton.  But we should reiterate that this is predicated on the belief that lattice QFT does indeed describe nature, so we have ignored the role played by gravity.

A conclusion we can draw from this is that, to some extent, QCAs are more general than lattice quantum field theories (with local Hamiltonians).  This and the fact that they are causal are two good reasons to study them with applications to quantum field theory in mind.  But it is tempting to speculate about whether other QCAs we can construct may offer solutions to problems in lattice QFT.  For example, the quantum walks of section \ref{sec:Fermion Doubling in Quantum Walks} after second quantization would be good candidates.  So, with this in mind, it is interesting to look at more general quantum cellular automata and even more general causal models in discrete spacetime.  Furthermore, the prescription above does not seem so natural in comparison with the free fermion models we saw in section \ref{sec:Second Quantization and Dynamics}.  The possibility of constructing causal models in discrete spacetime that are natural looking, maybe evolving via simple circuits, that reproduce physical systems in the continuum limit is compelling.  With that in mind, we turn our attention to general causal systems in discrete spacetime in the following section.

Before we do this, let us note that there is already a precedent for the idea that more natural looking models can reproduce interesting physics in the continuum limit.  In \cite{DdV87} a causal discrete spacetime model was studied that was similar to the free fermion fields in section \ref{sec:Continuum Limits of Discrete Fermion Fields} but with the addition of local interactions.  In the continuum limit, this became fermion fields evolving via the massive Thirring Hamiltonian.  The massive Thirring model describes massive fermions that interact via a local quartic interaction with the Hamiltonian below.
\begin{equation}
 H=\int\!\frac{\textrm{d}p}{2\pi}\,\psi^{\dagger}_{p}(p\sigma_z+ m\sigma_x)\psi_{p}+2g\!\int\! \textrm{d}x\,\psi_r^{\dagger}(x)\psi_r^{\ }(x)\psi_l^{\dagger}(x)\psi_l^{\ }(x),
\end{equation}
where $g$ is a constant and the left hand term is the free Dirac Hamiltonian on a line.

\section{Causal Discrete-Time Models on a Lattice}
\label{sec:Causal Discrete-Time Models on a Lattice}
In the previous section, with a sensible choice of the Lie-Trotter decomposition, we saw that we could approximate lattice quantum field theories by quantum cellular automata.  In some sense the procedure was a little clumsy.  In contrast, the discrete models of fermions we saw in section \ref{sec:Second Quantization and Dynamics} were quite elegant and simple.  Perhaps we can construct discrete models that are more natural and still reproduce physical models in the continuum limit.  It is also conceivable that such models could circumvent problems in lattice QFT, such as the fermion doubling problem.  So in this section we will turn things around.  Now we will take causality as our starting point and see what can be said about general causal discrete-time systems on a lattice.  This will be a little more abstract than what we have seen so far in this chapter.  Still, it will lead to some interesting results, which may lead to the construction of physically relevant models.

We start in section \ref{sec:Causal Fermions} by formalizing the idea of fermionic quantum cellular automata.  Aside from a few additional constraints required because of the anticommuting nature of fermions, these are the fermionic analogues of regular quantum cellular automata.  Next, in section \ref{sec:Local Decomposition}, we will see that for causal discrete-time models of fermions and bosons on a lattice causality implies localizability.  This entails adding a copy of the original system so that any causal dynamics of fermions and bosons can be decomposed into a constant-depth circuit of local unitaries.  This is an extension of the results of \cite{ANW11,GNVW12}, which proved this result for regular quantum cellular automata.  In section \ref{sec:Representation by Qubits}, we will see that fermionic quantum cellular automata are equivalent to regular quantum cellular automata, provided we can add a constant number of ancillary systems per site (either fermionic modes or qubits).  Finally, we look at the prospect of simulating these abstract causal systems on a quantum computer in section \ref{sec:Simulating Causal Fermions}.

Most of the results in the following sections are based on \cite{FS13}.

\subsection{Fermionic Quantum Cellular Automata}
\label{sec:Causal Fermions}
Fermionic quantum cellular automata are essentially quantum cellular automata with fermionic modes at each site instead of qubits (or other finite dimensional quantum systems).  Though, to be fair, there are also a couple of extra details that are specific to fermionic systems.  Here the systems we consider have a finite number of spatial points.  See \cite{FS13} for the extension to systems with infinitely many sites.

We say that an operator is localized on a region if it can be written in terms of creation and annihilation operators on sites only from that region.  This allows us to define causality for these models, which, as for regular quantum cellular automata, is naturally defined in the Heisenberg picture.  Causality means that there is an $L\geq 0$ such that, for any annihilation operator $a_{\vec{n}}$, after one timestep the evolved operator is localized on sites $\vec{m}$ with $|\vec{n}-\vec{m}|<L$.  As operators localized on a region are sums of products of creation and annihilation operators from that region, this definition implies that any operator does not spread by more than a distance $L$ over each timestep.

Let us define the \emph{neighbourhood} of a spatial point $\vec{n}$.  This is the smallest set of points $\vec{m}$ on which all evolved creation operators from $\vec{n}$ are localized.  Localized observables on non overlapping regions of space always commute.  This is because they must be sums of even products of creation and annihilation operators.

It is useful to recall that physical fermion Hamiltonians in continuous time systems are also sums of even products of creation and annihilation operators.  An implication of this is that, if we add an ancillary mode with annihilation operator $b$, it is left invariant under evolution via $e^{-iHt}$, meaning $e^{-iHt}$ commutes with $b$.  Inspired by this, we will assume that all evolution operators for discrete-time systems of fermions have this property.  This will be useful because, given a system of fermions evolving via $U$, we can add ancillary fermionic modes that are unaffected by $U$.  We will not just assume that $U=e^{-iH}$ for some even, self-adjoint operator $H$ because this does not generalize to infinite dimensional systems.

To make the notation more compact and closer to that used for infinite systems, given the evolution operator $U$, we write $u(A)$ instead of $U^{\dagger}AU$.

The following lemma will be useful for later proofs.
\begin{lemma}
\label{lem:2}
 Given a fermionic unitary $U$ and the annihilation operator $a$, $u(a)$ is a linear combination of odd products of fermion creation and annihilation operators.
\end{lemma}
\begin{proof}
We write $u(a)=A_{odd}+A_{even}$, where $A_{odd}$ are all the terms that are products of an odd number of creation and annihilation operators and $A_{even}$ are all terms that are products of an even number of creation and annihilation operators.

The extra requirement we made above implies that we can add a fermionic mode with annihilation operator $b$, which anticommutes with all of the original creation and annihilation operators while satisfying $u(b)=b$.  But this implies that $\{b,A_{odd}+A_{even}\}=0$, which is only possible if $A_{even}=0$.
\end{proof}

We can apply this to prove the following important fact about these systems.
\begin{lemma}
\label{lem:1}
The inverse of a causal fermionic unitary is also causal.
\end{lemma}
\begin{proof}
Lemma \ref{lem:2} tells us that $u(a_{\vec{m}})$ must be a linear combination of odd products of creation and annihilation operators.  So, since $u$ is causal,
\begin{equation}
 \begin{split}
&\{u(a_{\vec{m}}),a_{\vec{n}}\} =0,\\
\textrm{and}\ & \{u(a^{\dagger}_{\vec{m}}),a_{\vec{n}}\}= 0
\end{split}
\end{equation}
for all $\vec{m}$ when $\vec{n}$ is not in the neighbourhood of $\vec{m}$.  Applying $u^{-1}$ to the equations above,
\begin{equation}
\begin{split}
&\{a_{\vec{m}},u^{-1}(a_{\vec{n}})\}= 0,\\
\textrm{and}\ & \{a^{\dagger}_{\vec{m}},u^{-1}(a_{\vec{n}})\}= 0
\end{split}
\end{equation}
for all $\vec{m}$ when $\vec{n}$ is not in the neighbourhood of $\vec{m}$.
It follows that $u^{-1}(a_{\vec{n}})$ is localized, so $u^{-1}$ is causal.  Here we used that, for odd $B$ if $\{a_{\vec{n}},B\}=\{a_{\vec{n}}^{\dagger},B\}=0$, then $B$ has no $a_{\vec{n}}$ or $a_{\vec{n}}^{\dagger}$ terms in its expansion in terms of fermion creation and annihilation operators.
\end{proof}

Finally, we define fermionic quantum cellular automata.
\begin{definition}
 A fermionic quantum cellular automaton consists of a discrete lattice, which may have periodic boundary conditions or be $\mathbb{Z}^d$, with the properties below.
 \begin{enumerate}
  \item Each lattice site has an associated finite number of fermionic modes.
  \item Evolution takes place over discrete timesteps via a causal unitary that is translationally invariant in space and time.  (For infinite lattices, evolution is via an automorphism of the algebra of observables.)
  \item Furthermore, if we add additional modes, the evolution leaves them invariant.
  \end{enumerate}
\end{definition}

\subsection{Local Decomposition of Causal Models on a Lattice}
\label{sec:Local Decomposition}
In this section, we show that causal models on lattices can be decomposed into constant-depth circuits of local unitaries, provided we append a copy of each mode to each site.  First, we will give the proof of this for fermions, which is similar in spirit to the analogous proof for quantum cellular automata in \cite{ANW11,GNVW12}.

Constructing the local decomposition requires us to look at the joint evolution of the system of fermions and an identical copy of that system.  We denote the annihilation operators for the original fermions by $a_{\vec{n}}$ and those for the corresponding modes of the copy by $b_{\vec{n}}$.  All of this extends to having multiple modes at each site, but it will make the notation a lot cleaner if we look at the proof with just one per site.

It will be helpful to recall the unitary from section \ref{sec:Second Quantization and Dynamics} that implements a fermionic swap:
\begin{equation} 
S= \exp[i\frac{\pi}{2}(b^{\dagger}-a^{\dagger})(b-a)].
\end{equation}
This swaps the modes $a$ and $b$.  And we denote the swap unitary between the modes $a_{\vec{n}}$ and $b_{\vec{n}}$ by $S_{\vec{n}}$.  Because it is natural for the copy modes to be at the same site as the original modes, this is a local unitary.

This allows us to derive the local decomposition from \cite{FS13}.
\begin{theorem}
\label{th:1}
Take a system of fermions with annihilation operators $a_{\vec{n}}$, that evolve via the causal unitary $U_A$.  And consider the evolution of two copies of this system via $U_AU_B^{\dagger}$, where $U_B$ is equivalent to $U_A$ but acting on the copy fermions, which have annihilation operators $b_{\vec{n}}$.  This can be decomposed into local fermionic unitaries:
\begin{equation}
 U_AU_B^{\dagger}=\prod_{\vec{n}}S_{\vec{n}}\prod_{\vec{m}}[U_B S_{\vec{m}}U_B^{\dagger}],
\end{equation}
where $U_B S_{\vec{m}}U_B^{\dagger}$ are local fermionic unitaries that commute.
\end{theorem}
\begin{proof}
First,
\begin{equation}
 \prod_{\vec{n}}S_{\vec{n}}\prod_{\vec{m}}[U_B S_{\vec{m}}U_B^{\dagger}]=S U_B S U_B^{\dagger},
\end{equation}
where
\begin{equation}
S=\prod_{\vec{n}}S_{\vec{n}}
\end{equation}
swaps all modes.  And, because $S$ swaps all modes, $S U_B S=U_A$.  It follows that
\begin{equation}
 \prod_{\vec{n}} S_{\vec{n}}\prod_{\vec{m}}[U_B S_{\vec{m}}U_B^{\dagger}]=U_AU_B^{\dagger}.
\end{equation}
Furthermore, $U_B S_{\vec{n}}U_B^{\dagger}$ is a local fermionic unitary because $S_{\vec{n}}$ is
\begin{equation}
 \exp[i\frac{\pi}{2}(b_{\vec{n}}^{\dagger}-a_{\vec{n}}^{\dagger})(b_{\vec{n}}-a_{\vec{n}})].
\end{equation}
So
\begin{equation}
 U_B S_{\vec{n}}U_B^{\dagger}=\exp[i\frac{\pi}{2}(b_{\vec{n}}^{\prime\dagger}-a_{\vec{n}}^{\dagger})(b^{\prime}_{\vec{n}}-a_{\vec{n}})],
\end{equation}
where $b^{\prime}_{\vec{n}}=U_Bb_{\vec{n}}U_B^{\dagger}$, which is localized because $U_B$ is causal. Therefore, $ U_B S_{\vec{n}}U_B^{\dagger}$ is also localized.  The unitaries $ U_B S_{\vec{n}}U_B^{\dagger}$ commute because $[S_{\vec{n}}, S_{\vec{m}}]=0$. 
\end{proof}
What this theorem tells us is that the causal evolution of two copies of a system of fermions can be rewritten as a product of local unitaries.  Furthermore, given any state $\ket{\psi}$ of the original system of fermions and its copy, then, for any measurement operator $\mathcal{M}_A$ on the original fermions,
\begin{equation}
 \bra{\psi}U_BU_A^{\dagger}\mathcal{M}_AU_AU_B^{\dagger}\ket{\psi}=\bra{\psi}U_A^{\dagger}\mathcal{M}_AU_A\ket{\psi}.
\end{equation}
This means that the joint evolution of both the original and auxiliary fermions reproduces the original evolution.  For example, $\ket{\psi}$ could be any state of the original fermions with all of the copy modes empty, like $\textstyle{\frac{1}{\sqrt{2}}(a^{\dagger}_{\vec{n}}+a^{\dagger}_{\vec{m}}})\ket{0}$, where $\ket{0}$ is the state annihilated by all annihilation operators.

In some cases there is already a constant depth local unitary decomposition of the evolution without having to add a copy of the system.  This was the case for the free fermion systems in \ref{sec:Discrete Dirac Fermions in One Dimension}, for example.  This is not true in general.  A good counterexample is the evolution where everything is shifted one step to the right every timestep.

Before proceeding, it will be useful to note that the analogous result from \cite{ANW11,GNVW12} for quantum cellular automata can be proved in a very similar fashion.  The only differences are that at each site there are finite dimensional quantum systems, such as qubits, and that the swaps become swaps between these finite dimensional systems and their copies.

We can extend theorem \ref{th:1} to systems of bosons and fermions, which may interact.  To do this, we define $\mathbb{S}_{\vec{n}}$ to be the bosonic swap operator between the mode at $\vec{n}$ with annihilation operator $c_{\vec{n}}$ and its copy with annihilation operator $d_{\vec{n}}$.
\begin{equation}
\mathbb{S}_{\vec{n}}= \exp[i\frac{\pi}{2}(d^{\dagger}_{\vec{n}}-c^{\dagger}_{\vec{n}})(d_{\vec{n}}-c_{\vec{n}})].
\end{equation}
Then we have the following corollary \cite{FS13}.
\begin{corollary}
\label{cor:1}
 Take a finite system of fermions and bosons evolving via the causal unitary $U_A$.  Then, the evolution of two copies of the system via $U_AU_B^{\dagger}$, where $U_B$ is equivalent to $U_A$ but acting on the copy system, can be decomposed into local unitaries:
\begin{equation}
 U_AU_B^{\dagger} =\prod_{\vec{k}}S_{\vec{k}}\prod_{\vec{n}}\mathbb{S}_{\vec{n}}\prod_{\vec{m}}[U_B S_{\vec{m}}U_B^{\dagger}]\prod_{\vec{l}}[U_B\mathbb{S}_{\vec{l}\,}U_B^{\dagger}],
\end{equation}
where $U_B S_{\vec{n}}U_B^{\dagger}$ and $U_B\mathbb{S}_{\vec{n}}U_B^{\dagger}$ are commuting local unitaries.
\end{corollary}

\subsection{Representation by Quantum Cellular Automata}
\label{sec:Representation by Qubits}
So we now have a decomposition of causal evolution into products of local unitaries.  It is interesting to see whether we can represent these systems by quantum cellular automata.  An issue that could cause problems for us is that bosonic systems involve infinite dimensional Hilbert spaces:\ the Hilbert space of even a single bosonic mode is infinite dimensional because we can have an arbitrary number of particles of a given type.  So to represent bosonic modes by finite dimensional systems (such as a few qubits), we would need some form of truncation procedure.  In free models, where the number of bosons does not change over time, this is straightforward, as the dynamics ensures that we only use a finite dimensional subspace of the whole Hilbert space.  In general, however, it is not clear how to do this.  It is probably the case that we would need to look at specific models before anything concrete could be said about truncating the bosonic degrees of freedom.  So from here on we will focus purely on causal systems of fermions.

As we saw in section \ref{sec:The Jordan-Wigner Transformation}, the Jordan-Wigner transformation allows us to represent fermions in the qubit picture.  So, we assign one qubit to the fermionic mode at $\vec{n}$ and one to its copy.  For now, suppose we just have fermions on a line.  With modes on the same site chosen to be consecutive in the Jordan-Wigner transformation ordering, the operators $S_{\vec{n}}$ are local unitaries in the qubit representation.  Next, recall that
\begin{equation}
U_B S_{\vec{n}}U_B^{\dagger}=\exp[i\frac{\pi}{2}(b^{\prime\dagger}_{\vec{n}}-a^{\dagger}_{\vec{n}})(b^{\prime}_{\vec{n}}-a_{\vec{n}})].
\end{equation}
Now note that $b^{\prime}_{\vec{n}}=U_Bb_{\vec{n}}U_B^{\dagger}$ must be a linear combination of {\it odd} powers of creation and annihilation operators on the neighbourhood of $\vec{n}$.  We proved this in lemma \ref{lem:2} in the previous section.  This means that $U_B S_{\vec{n}}U_B^{\dagger}$ has the form $e^{-iH_{\vec{n}}}$, where $H_{\vec{n}}$ is self-adjoint and localized.  Furthermore, it is a sum of even products of creation and annihilation operators.  Then, with the natural ordering for the Jordan-Wigner transformation, $H_{\vec{n}}$ is localized in the qubit representation also.  Therefore, $U_B S_{\vec{n}}U_B^{\dagger}$ is also localized in the qubit representation.  So causally evolving fermions on a line can be represented by causally evolving qubits on a lattice.

We will need to work harder in higher spatial dimensions (or on lines with periodic boundary conditions).  This is because, as we have seen, in these cases the Jordan-Wigner transformation is nonlocal in general.  We still choose the Jordan-Wigner ordering so that fermionic modes and their copies are consecutive.  Then each fermionic swap $S_{\vec{n}}$ is local in the qubit representation.

The $U_B S_{\vec{n}}U_B^{\dagger}$ terms in the decomposition are more problematic because they may not be localized unitaries in the qubit picture.  Fortunately, we saw in section \ref{sec:Representing Local Fermionic Hamiltonians by Local Qubit Hamiltonians} how to circumvent this problem.  The price is that we must add more auxiliary fermions.  And this means that we need more qubits to represent this larger system of fermions.

As a fermionic unitary, $U_B S_{\vec{n}}U_B^{\dagger}$ is local, and it has the form $e^{-iH_{\vec{n}}}$, where $H_{\vec{n}}$ is self-adjoint and localized.  So, by adding pairs of fermions for each link, $H_{\vec{n}}$ is equivalent to a self-adjoint even operator on a larger system that is also local in the qubit picture.  We saw how this worked in section \ref{sec:Representing Local Fermionic Hamiltonians by Local Qubit Hamiltonians}.  Therefore, $U_B S_{\vec{n}}U_B^{\dagger}$ is equivalent to a local fermionic unitary on a larger system that is local in the qubit picture.

The fermionic unitaries $U_B S_{\vec{n}}U_B^{\dagger}$ commute.  After adding additional fermions, the resulting unitaries $V_{\vec{n}}$ implementing $U_B S_{\vec{n}}U_B^{\dagger}$ in the qubit picture may not commute if the neighbourhoods on which they are localized overlap.  Fortunately, the order in which they are applied is irrelevant when acting on a $+1$ eigenstate of the auxiliary fermion operators $ic_{(\vec{n},\vec{m})}c_{(\vec{m},\vec{n})}$.  And we can apply many $V_{\vec{n}}$ simultaneously, provided they act on regions that do not overlap.  For example, on a line, where evolution is nearest neighbour, we can apply all $V_n$ with $n\bmod 3=0$ first, followed by all $V_n$ with $n\bmod 3=1$, followed by all $V_n$ with $n\bmod 3=2$.  In this case, we need only three steps to implement every $V_n$.  So the circuit is constant depth.  The procedure is similar in higher dimensions and for different neighbourhoods.  So the qubit evolution is always causal.

It is important, particularly from a simulation perspective and for the extension to infinite lattices, that the number of additional fermions we add per site is constant, meaning it is independent of $\mathcal{N}$, the number of sites.

The end result is the following theorem \cite{FS13}.
\begin{theorem}
\label{th:CS1}
 Any causal fermionic evolution in discrete spacetime is equivalent to a subsector of the causal evolution of a system of qubits.  This may require the addition of a constant number of qubits per site.
\end{theorem}

Furthermore, this leads to another interesting theorem.
\begin{theorem}
\label{cor:2}
 Quantum cellular automata and fermionic quantum cellular automata are equivalent, in the sense that, by adding a constant number of systems to each site, one can simulate the other.
\end{theorem}
That quantum cellular automata can simulate fermionic quantum cellular automata follows from theorem \ref{th:CS1}.  The other direction follows from a useful fact:\ we can represent a qubit by a pair of fermions without any nonlocality.  We saw how to do this in section \ref{sec:U(1) Lattice Gauge theory as a spin model}.  Here is how it works.  For a single qubit with basis states $\ket{0}_Q$ and $\ket{1}_Q$ we introduce a pair of fermion modes with annihilation operators $a$ and $b$.  Then we represent the qubit states by
\begin{equation}
 \begin{split}
  \ket{0}_Q & =\ket{0_F}\\
  \ket{1}_Q & =a^{\dagger}b^{\dagger}\ket{0_F},
 \end{split}
\end{equation}
where $\ket{0}_F$ is the state with no fermions present.  Qubit operators are represented by\footnote{The choices here differ from those in section \ref{sec:U(1) Lattice Gauge theory as a spin model}.  This is because we want the operators to be unitary.}
\begin{equation}
 \begin{split}
  \openone & =\openone\\
  X & =(a^{\dagger}-a)(b^{\dagger}+b)\\
  Z & =aa^{\dagger}-a^{\dagger}a.
  \end{split}
\end{equation}
The representation of any other qubit operator can be gotten by taking linear combinations and products of these operators.
Now, suppose that we have a lattice of qubits.  To represent this by fermions, we take the same lattice but now with two fermionic modes at each site for every qubit.  Then we can represent the qubit states as above.  Note that all fermion representations of qubit operators are quadratic, and so any two on different sites commute.

We must ensure that the dynamics in the fermion representation is a valid causal fermionic unitary.  One way to see that this is possible is to replace the original QCA by two copies, where the second copy evolves via the inverse unitary.  This allows us to use the analogous result of theorem \ref{th:1} for regular QCAs that we mentioned in the previous section.  It says that the joint evolution of the system and its copy is equivalent to a constant-depth circuit of local unitaries.  Suppose we represent this joint system by fermions in the manner described above.  The local unitaries in the decomposition of the qubit evolution can be written in terms of $X$ and $Z$ operators on a finite region.  And we know how to map these into the fermion picture.  Furthermore, as they are unitary and local in the fermion picture, we get a constant-depth circuit of local unitaries in the fermion picture.  Therefore, the result is a bona fide fermionic quantum cellular automaton.

It was shown in \cite{BK02} that, in the circuit model, conventional quantum computers (those with information stored by qubits) could efficiently simulate quantum computers composed of fermions but with slowdown logarithmic in the number of fermionic modes $\mathcal{N}$.  This means that in the worst case scenario $O(\log(\mathcal{N}))$ qubit gates would be necessary to simulate a single fermionic gate.  In contrast to this, the slowdown is only constant in the converse direction:\ to simulate a qubit gate with the fermionic system, only $O(1)$ gates are needed.  Here we have an analogue of \cite{BK02} but for quantum cellular automata.  It is interesting that in our case the situation is more symmetric.

\subsection{Simulating Causal Systems on a Quantum Computer}
\label{sec:Simulating Causal Fermions}
Let us turn to the problem of simulating causal systems by quantum computers.  Again, the situation for bosons is not clear, which is because it is not obvious when systems with bosons can be approximately represented on qubits.  At least for causal fermions we now know how to simulate the dynamics by local unitaries applied to a lattice of qubits.  This tells us how to simulate the dynamics of causal fermions on a quantum computer.  Additionally, we saw in section \ref{sec:Preparing the Majorana state} how to prepare the initial state of the auxiliary fermions needed to make the mapping to qubits local.

From a complexity point of view, this can be done efficiently.  This is because, for $\mathcal{N}$ fermionic modes, we only have to apply $O(\mathcal{N})$ local unitaries:\ first the $\mathcal{N}$ unitaries $V_{\vec{n}}$ (the unitaries implementing $U_B S_{\vec{n}}U_B^{\dagger}$ on the qubits), followed by the qubit representation of the $\mathcal{N}$ fermionic swap operators.  As we saw in section \ref{sec:Representation by Qubits}, we may require additional qubits to ensure that these operators are still local in the qubit picture.

Many of the unitaries can be implemented in parallel:\ all swap operations can be done simultaneously, and many $V_{\vec{n}}$ operations can be done at the same time, as long as the areas on which they are localized do not overlap.  As each $V_{\vec{n}}$ is localized on the neighbourhood of $\vec{n}$, which contains at most all the points in a hypercube with length of side $2L$ (because the evolution is causal), the time needed for one step of the evolution is $O\big((2L)^d\big)$, where $d$ is the lattice dimension.  Therefore, the time does not depend on $\mathcal{N}$.

So we know how to simulate the evolution on a quantum computer.  On its own, simulating arbitrary causal models seems uninteresting unless there is some physical significance to the models.  Still, as a proof of principle, this is another aspect of the problem of whether quantum computers can efficiently simulate quantum physics, which we discussed in chapter \ref{chap:Background1}.  The answer was shown to be yes in the case of local Hamiltonians on a lattice in \cite{Lloyd96,AL97,NC00,OGKL01,SOGKL02}.  Here we saw that the same holds for {\it causal} dynamics of fermions on lattices.

Aside from the dynamics, there is the question of state preparation.  One thing we already know is that a $+1$ eigenstate of the pairs of Majorana fermions can be prepared efficiently on a quantum computer, which we saw how to do in section \ref{sec:Preparing the Majorana state}.  Initial state preparation for fermion systems is discussed in \cite{AL97}, for example.

\section{Information Flow in QCAs}
\label{sec:Information Flow in Quantum Cellular Automata}
This section stands alone from the general narrative of this chapter since it is a study of the abstract theory of QCAs.  Nevertheless, the structure theorem for QCAs on a line that we will prove here may be useful for future work.  Here, we will look at the evolution operator of a QCA on a line and prove a theorem that describes how information moves along the line under the evolution of the  QCA.  Information flow for QCAs on a line was quantified previously in \cite{GNVW12} by a locally computable number called the index.  Here, we provide a general structure theorem for QCAs on a line that is analogous to the result for quantum walks in section \ref{sec:Quantum Walks on a Line: Abstract Theory}.

First, in section \ref{sec:Spreading of Information}, we reproduce some tools used in \cite{SW04,GNVW12} that allow us to quantify how information spreads locally over one timestep.  This allows us to prove the main result in section \ref{sec:Structure of the evolution}, which is the statement that, possibly after regrouping sites, any QCA on a line can be written as a product of shifts and on-site unitaries.

\subsection{Spreading of Information}
\label{sec:Spreading of Information}
Before getting to the results, it will be useful to have a more concrete understanding of how information spreads locally on a QCA.  This requires the introduction of support algebras.
\begin{Definition}
 Take the algebras $\m{A}$, $\m{B}_1$ and $\m{B}_2$, with $\m{A}\subseteq \m{B}_1\otimes \m{B}_2$.  We define $S(\m{A},\m{B}_1)$ to be the support of $\m{A}$ on $\m{B}_1$.  This is the smallest subalgebra contained in $\m{B}_1$ that is needed to create the elements of $\m{A}$.  In other words, writing $a\in\m{A}$ as $a=\sum_{ij}b^1_i\otimes b^2_j$, with $b_i^1\in \m{B}_1$ and $b_j^2\in \m{B}_2$, the support of $\m{A}$ on $\m{B}_1$ is generated by all such $b_i^1$.
\end{Definition}
The intuition behind this definition is that, in a sense, the support describes how much $\m{A}$ overlaps with $\m{B}_1$.

We say two algebras $\m{A}$ and $\m{B}$ commute, if, for every $a\in \m{A}$ and $b\in \m{B}$, $[a,b]=0$, and we write this as $[\m{A},\m{B}]=0$.  Then, we have the following useful lemma proved in \cite{GNVW12}.
\begin{lemma}
 Take the algebras $\m{B}_1$, $\m{B}_2$ and $\m{B}_3$.  And suppose that $\m{A}_L$ and $\m{A}_R$ are subalgebras with
 \begin{equation}
  \begin{split}
   \m{A}_L & \subseteq \m{B}_1\otimes \m{B}_2\\
   \m{A}_R & \subseteq \m{B}_2\otimes \m{B}_3.
  \end{split}
 \end{equation}
 Then if $[\m{A}_L,\m{A}_R]=0$,
\begin{equation}
 [\g{S}(\m{A}_L,\m{B}_2),\g{S}(\m{A}_R,\m{B}_2)]=0.
\end{equation}
\end{lemma}

\begin{compactwrapfigure}{r}{0.49\textwidth}
\centering
\begin{minipage}[r]{0.47\columnwidth}%
\centering
    \resizebox{6.5cm}{!}{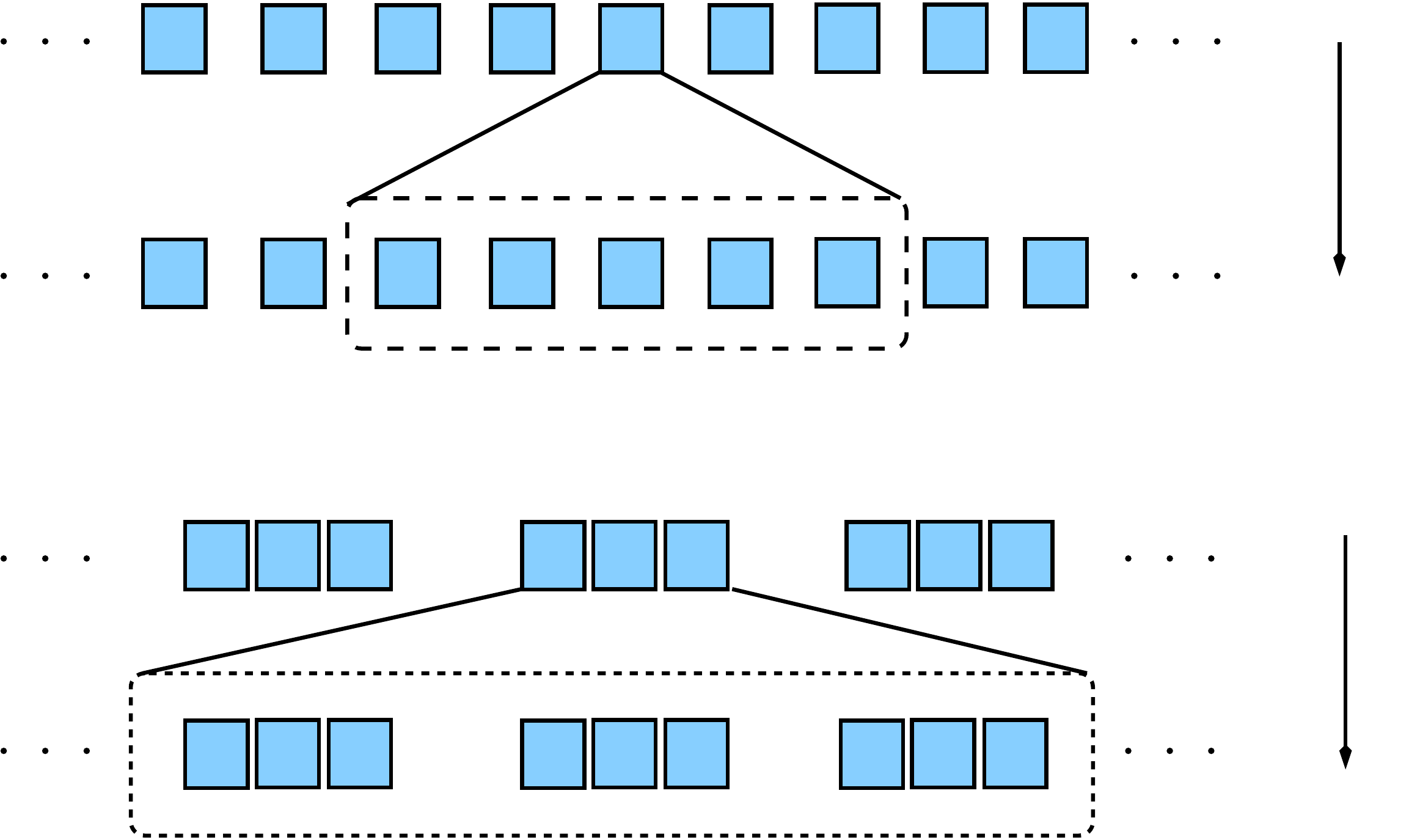}
    \footnotesize{\caption[Regrouping of sites to get a nearest neighbour QCA]{After regrouping of sites, the next nearest neighbour QCA becomes a nearest neighbour QCA.\ \\ \ }}
    \label{fig:Regrouping}
\end{minipage}
\end{compactwrapfigure}

This lemma and the notion of support algebras defined above are particularly useful for understanding the dynamics of QCAs.  The algebra of a QCA is essentially composed of sums of products of elements corresponding to single sites.  Because of this, to analyse the dynamics it is enough to look at how the algebras of one or two sites evolve.  It will help to regroup sites together in such a way that the evolution becomes nearest neighbour.  Then the new sites will have higher dimensional quantum systems.

As in section \ref{sec:Causal Discrete-Time Models on a Lattice}, we will write $u(A)$ to denote $U^{\dagger}AU$.  And we denote the algebra of operators at a site by $\m{A}_{n}$, where $n$ denotes the site.  It will help to visualize the sites grouped into two-site blocks, with algebras $\m{A}_{2n}\otimes\m{A}_{2n+1}$.  We define the even and odd algebras to be
\begin{equation}
\begin{split}
 \m{R}_{2n} & =\g{S}\left(u\left(\m{A}_{2n}\otimes\m{A}_{2n+1}\right),\left(\m{A}_{2n-1}\otimes\m{A}_{2n}\right)\right)\\
 \m{R}_{2n+1} & =\g{S}\left(u\left(\m{A}_{2n}\otimes\m{A}_{2n+1}\right),\left(\m{A}_{2n+1}\otimes\m{A}_{2n+2}\right)\right).
\end{split}
\end{equation}
\begin{figure}[!ht]
{\centering
\resizebox{9cm}{!}{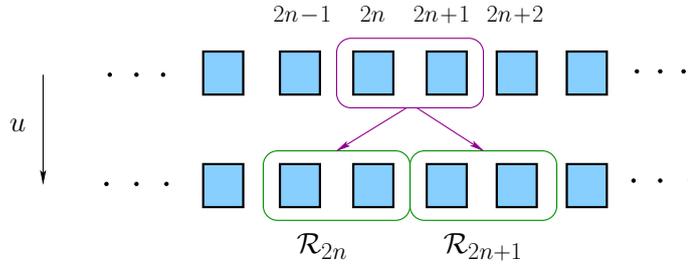} \caption[Illustration of the even and odd algebras]{Illustration of the even and odd algebras.  Note that $\m{R}_{2n}$ does not in general equal $\m{A}_{2n-1}\otimes\m{A}_{2n}$.\label{fig:EvenOdd}}
}
\end{figure}

So the image of the algebra $\m{A}_{2n}\otimes\m{A}_{2n+1}$ under $u$ is contained in $\m{R}_{2n}\otimes\m{R}_{2n+1}$.  The even/odd algebras roughly encapsulate what information goes to the right/left.  This construction appeared, for example, in \cite{GNVW12}.

Now we will see that these even and odd algebras generate the QCA's algebra.  First, the definition of $\m{R}_{2n+1}$ and $\m{R}_{2n}$ ensures that $u(\m{A}_{2n}\otimes\m{A}_{2n+1})\subseteq \m{R}_{2n}\otimes\m{R}_{2n+1}$.  Then because $u$ is unitary, the algebra generated by all $\m{R}_m$ is equal to $\m{A}$.

The entire algebra has trivial center, meaning that the only elements of the algebra that commute with everything are multiples of the identity.  This means that each $\m{R}_m$ must have trivial center too since all $\m{R}_m$ commute with each other.  This means that these $\m{R}_m$ are subsystem algebras \cite{ANW08}.  In other words, each $\m{R}_m$ is isomorphic to the algebra of $n\times n$ complex matrices, and the entire algebra can be thought of as a tensor product of the $\m{R}_m$.

Finally, $\m{R}_{2n}\otimes\m{R}_{2n+1}=u(\m{A}_{2n}\otimes\m{A}_{2n+1})$.  This follows because $u(\m{A}_{2n}\otimes\m{A}_{2n+1})\subseteq \m{R}_{2n}\otimes\m{R}_{2n+1}$.  If this were strict, the evolution would not be invertible.

\subsection{Structure of QCA Evolution on a Line}
\label{sec:Structure of the evolution}
Using the tools introduced in the previous section, we are now able to show that the evolution of a QCA on a line has a relatively simple decomposition in terms of shifts and local unitaries.  This is an analogue for QCAs of the result for quantum walks on a line that we saw in section \ref{sec:Quantum Walks on a Line: Abstract Theory}.  It said that the quantum walk evolution operator could always be decomposed into a product of conditional shifts and coin operations \cite{GNVW12,Vogts09}.

The proof of the following theorem roughly works by reverse engineering the dynamics.  After one timestep, we apply partial shifts in such a way that the net result is that no information has been transferred along the line.  This is then equivalent to a product of local unitaries.  Inverting this means that the evolution operator is a product of local unitaries and partial shifts.

\begin{theorem}
\label{th:32}
 After regrouping sites, the evolution of a one dimensional QCA can be written as a product of on-site unitaries and partial shifts.
\end{theorem}
\begin{proof}
First, let $s_1$ denote a shift to the left by one lattice site.  The effect of $s_1^{-1}\circ u$ is to move the even and odd algebras one step to the left.

Now, because of translational invariance, it makes sense to define a shift of just the even or just the odd algebras.  For example, the unitary $t$, taking $\m{R}_{2n}\rightarrow\m{R}_{2n+2}$ and $\m{R}_{2n+1}\rightarrow\m{R}_{2n+1}$ is a type of shift.  Instead of $t$, however, we are interested in the unitary that takes
\begin{equation}
 \begin{split}
  s_1^{-1}(\m{R}_{2n}) & \rightarrow s_1^{-1}(\m{R}_{2n+2})\\
  s_1^{-1}(\m{R}_{2n+1}) & \rightarrow s_1^{-1}(\m{R}_{2n+1}).
 \end{split}
\end{equation}
This can be written as $s_2^{-1}=s_{1}^{-1}\circ t\circ s_{1}$.  

Then the evolution $s_2^{-1}\circ s_1^{-1}\circ u$ has the property that elements localized on $\m{A}_{2n}\otimes\m{A}_{2n+1}$ are still localized on $\m{A}_{2n}\otimes\m{A}_{2n+1}$ after applying this unitary.  See figure \ref{fig:QCALineDecomp}.
\begin{figure}[!ht]
{\centering
\resizebox{9cm}{!}{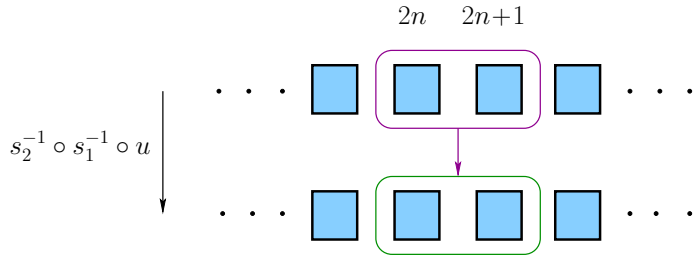} \caption[Deconstructing the evolution of a quantum cellular automaton]{The effect of $s_{2}^{-1}\circ s_{1}^{-1}\circ u$ is to take the algebra $\m{A}_{2n}\otimes\m{A}_{2n+1}$ to itself.  (This does not necessarily mean that this is the identity map.)\label{fig:QCALineDecomp}}
}
\end{figure}
Put another way, this implies that the unitary $s_{2}^{-1}\circ s_{1}^{-1}\circ u$ is just a product of local unitaries.  Therefore,
\begin{equation}
u=s_{1}\circ s_{2}\circ v,
\end{equation}
where $v$ is a product of identical local (two-site) unitaries.

Next, we regroup pairs of sites into single sites, with algebras $\m{B}_n=\m{A}_{2n}\otimes\m{A}_{2n+1}$.  So after this regrouping $v$ is a product of single site unitaries and $s_{2}$ is a partial shift, which we will relabel $s_r$.  Now $s_1$ is equivalent to the product of on-site unitaries, denoted by $w$, swapping $\m{A}_{2n}$ and $\m{A}_{2n+1}$, followed by the partial shift, with $\m{A}_{2n}\rightarrow\m{A}_{2n+2}$ and $\m{A}_{2n+1}\rightarrow\m{A}_{2n+1}$, denoted by $s_l$.  It follows from this that one dimensional QCAs in this regrouped form can be decomposed into on-site unitaries and partial shifts:
\begin{equation}
 u=s_l\circ w\circ s_r\circ v.
\end{equation}
\end{proof}
Incidentally, while this result extends to infinite systems, for the case of a QCA on a finite line with $\mathcal{N}$ sites (and periodic boundary conditions) the evolution can be decomposed into $O(\m{N})$ local unitaries.  This works because a shift on a finite line is equivalent to a product of swaps.  Defining $S_{n,n+1}$ to be the swap between systems at sites $n$ and $n+1$, the unitary that shifts by one step to the left $S_l$ can be decomposed into a product of swaps $S_l=S_{1,2}S_{2,3}...S_{\m{N}\!,\,\m{N}+1}$, where periodic boundary conditions mean $\m{N}+1\equiv 1$.  It is significant that we cannot do these swaps in parallel:\ the order matters.  So the circuit is not constant depth.  This contrasts with the infinite case studied in \cite{ANW11,GNVW12}, where local unitary decompositions that are constant depth can be found but only by appending a copy of the original system.

\chapter{Conclusions and Open Problems}
\label{chap:Conclusions and Open Problems 1}
\begin{center}
{\em Ah, but if less is more, just think how much more more will be.}\\
- Frasier Crane, Frasier
\end{center}
We have seen some compelling arguments for believing that quantum cellular automata and quantum walks could be useful tools for discretizing relativistic physics.  Let us go over these.  First, the generality with which quantum walks with a two dimensional coin become relativistic particles in the continuum limit, which we saw in chapter \ref{chap:Quantum Walks and Relativistic Particles}, was very encouraging.  Part of how we proved this required showing that Hamiltonians for a particle with a two dimensional extra degree of freedom that are {\it linear} in the momentum operator are equivalent to relativistic Hamiltonians.  But it was by taking the continuum limits of quantum walks that we could naturally recover Hamiltonians linear in the momentum operator.

Perhaps more compelling is the simplicity of many of the quantum walks and quantum cellular automata that became Dirac particles and fields in the continuum limit.  Whether this simplicity would be especially useful for simulating physics is unclear.  After all, most of the continuum limits involved Lie-Trotter type arguments, but generally when simulating Hamiltonian evolution, one makes recourse to higher order Suzuki-Trotter decompositions since these converge faster.  Nevertheless, from a foundational point of view, it is an attractive feature of these models that the evolution is so simple.  This prompts the question of whether this simplicity extends to quantum cellular automata (or other causal models in discrete spacetime) that approximate interacting models.

Even if the goal is simulating physics, the possibility that nature itself could fundamentally involve discrete spacetime is hard to ignore.  Could physics ultimately be viewed as a quantum cellular automaton that is a simple circuit of local unitaries?  One thing we learned is that, with a few caveats, it should be possible to approximate interacting quantum field theories arbitrarily well by quantum cellular automata.  This is encouraging, though we did not account for gravity.

We can take further encouragement from the arguments of chapter \ref{chap:Quantum Walks and Relativistic Particles} showing that there may be a simple resolution of the fermion doubling problem for these discrete models.  The question is still open, but the initial results give cause for optimism.  Further results on quantum walks included a discussion of lattice symmetries and the construction of quantum walks that converge to their continuum limits faster than the usual examples.

Another useful result was that we could construct a vacuum state for the discrete Dirac fields that converged to its continuum counterpart.  This allowed us to prove, at least for the one dimensional case, that the causal fermionic models in discrete spacetime from section \ref{sec:Second Quantization and Dynamics} become free Dirac fields in the continuum limit.  Note that in interacting models, the vacuum is almost certainly an even more complicated state than in the free case.  So it is not quite so clear how to proceed on this front.

In chapter \ref{chap:Quantum Cellular Automata and Field Theory}, we also saw some tricks for removing the nonlocality of physical fermion observables when represented in the qubit picture, which resulted in an ancilla in an entangled state.  Although this was not motivated towards saying something fundamental about the nature of fermions, it is very interesting that the inherent nonlocality of fermions is not necessary to describe these systems.  In section \ref{sec:U(1) Lattice Gauge theory as a spin model}, we were able to modify these tricks to view lattice fermions interacting with gauge fields as local qubit models.

As an aside, it is amusing how often adding ancillary fermionic modes or extra qubits helped us to prove results.  The addition of ancillas is something that is often done in computation.  And it is interesting to wonder where else we can apply such tricks to simplify or better understand physical models.  One example we already saw was that we could use this to show that causal dynamics on a lattice can always be rewritten as products of local unitaries, which makes it easier in a sense to understand precisely what causal evolution operators are.

On the more abstract side, we saw that fermionic quantum cellular automata are equivalent to regular quantum cellular automata.  And, in the process, we were able to show that causal models of fermions in discrete spacetime can be efficiently simulated by a quantum computer, which is similar in spirit to the result that local Hamiltonian models are efficiently simulable by quantum computers.  We also proved a structure theorem for quantum cellular automata on a line, analogous to a result for quantum walks on a line, discussed in section \ref{sec:Quantum Walks on a Line: Abstract Theory}.

From a practical point of view, it makes sense to ask whether we can use fermionic systems as the building blocks of quantum computers, instead of qubits.  This is simply because matter is made of fermions, though using photons to store information is certainly a possibility \cite{NC00}.  It is prudent to explore all physical possibilities for computation in case some systems may be superior to others.  That said, fermionic quantum computation is known to be equivalent to quantum computation with qubits \cite{BK02}.

There are some fascinating open questions that remain unanswered.  Let us take a look at some of them.

\section{Quantum Walks}
A natural next step for further study is to look at interacting quantum walks, like the examples in \cite{AAMSWW12}.  In particular, we could study the continuum limit of two-particle quantum walks with local interactions.\footnote{This is work in progress, and the idea was proposed by A.\ H.\ Werner.}  For models with particle number conservation, this may be a simpler way to study interactions than proceeding straight to quantum cellular automata.

Another option is to look at quantum walks in a possibly time dependent potential.  For example, it would be very interesting to take continuum limits of the quantum walks in an external gauge field that we saw in section \ref{sec:Gauge Interactions for Quantum Walks}.  Doing this might be useful from the point of view of taking continuum limits of discrete field models to recover continuum field theories.  This is because gauge interactions preserve particle number in the positive/negative energy particle setting.  While this interpretation of Dirac particles is not believed to be fundamental, it may still be helpful from a simulation point of view.  In fact, we used this old fashioned interpretation to take the continuum limit of the discrete vacuum state in section \ref{sec:Continuum Limits2}.

Directly applying the Lie-Trotter product formula with a momentum cutoff to take continuum limits of interacting models would probably not work because the interaction with the gauge field can change the particle's momentum.  We could try Trotter's theorem for unbounded operators, which we saw, together with the Lie-Trotter theorem, in section \ref{sec:The Lie-Trotter Product Formula}.  With some technical conditions on the self-adjoint operators $B_1$ and $B_2$, the theorem states that
\begin{equation}
 \|\left((e^{iB_1/N}e^{iB_2/N})^N-e^{i(B_1+B_2)}\right)\ket{\psi}\|_2\rightarrow 0.
\end{equation}
But the gauge interaction term that would appear in one of the exponents, $B_1=\sum_nA_n\ket{n}\bra{n}a$, changes with the lattice spacing, as we would expect $A_n=A(na)$, where $A(x)$ is the continuum field.  So, for this approach to have any hope of success, the theorem may need to be modified.

The question of whether there are structure theorems for quantum walks in higher spatial dimensions still remains open.  We know that in one dimension every quantum walk can be written as a product of shifts and coin operations \cite{Vogts09}, which we discussed in section \ref{sec:Quantum Walks on a Line: Abstract Theory}.  It is still unknown whether this or something similar holds in higher dimensions.

\section{Quantum Cellular Automata}
Probably the most interesting open questions regarding quantum cellular automata involve simulating or approximating interacting physical models.  From this point of view, understanding the continuum limit is absolutely crucial.  And, as is generally the case in physics, once interactions are introduced the complexity of the problem increases greatly.  It seems that, to take continuum limits of interacting models, it is probably the case that more sophisticated ideas will be needed, like those in \cite{Osb14}.

Another interesting possibility is to look at quantum cellular automata or even quantum walks in a path integral picture.  This has been done for free models in one and two spatial dimensions in \cite{DMPT14,DMPT14b}.  It may be useful to extend this to interacting models to take continuum limits or to recover perturbative models when the interaction is weak.

We briefly looked at the vacuum state for the discrete fermionic system that reproduces Dirac fields in the continuum limit in section \ref{sec:Discrete Fermion Fields and the Vacuum}.  It would be interesting to study the properties of such discrete vacuum states in more detail \cite{BDFPST15}.  In particular, they should provide simple mathematical test grounds for questions about the entanglement properties of continuum vacuum states.  For example, in \cite{SR07} the authors ask whether Bell inequalities can be violated by doing measurements on two finite spacetime regions of the Dirac field vacuum.  A first step towards answering this question could be to look at simpler vacuum states that have nicer mathematical properties.  And, as the discrete Dirac vacuum we have seen converges to the continuum Dirac field vacuum, it may even be possible to prove results about the continuum vacuum using its discrete counterpart.

We saw in section \ref{sec:Representation by Qubits} that fermionic quantum cellular automata and regular quantum cellular automata are equivalent.  This looked like an analogue for QCAs of the result of \cite{BK02}, where it was shown that fermionic and regular quantum computers in the circuit model are equivalent.  It would be interesting to study this equivalence of fermionic and regular QCAs further in light of the fact that QCAs are themselves universal for quantum computation \cite{Watrous95}.

Many other results of chapter \ref{chap:Quantum Cellular Automata and Field Theory} involved representations of fermionic systems in the qubit picture and vice versa.\footnote{Fermionic systems have also been viewed in terms of operational theories in \cite{DMPT14a}, where the similarities and differences between qubit and fermion theories were studied further.}  Along these lines, we saw that, with a little trickery, the nonlocality that inevitably comes with fermions can be removed \cite{Ball05,VC05}.  But we could do a little better and remove some of the redundancy of the description:\ in section \ref{sec:U(1) Lattice Gauge theory as a spin model}, we replaced a model with fermions interacting via abelian gauge fields by a completely local model.

It would be fascinating if this were generally true, particularly if we could extend the ideas to nonabelian gauge fields.  After all, the standard model involves interactions mediated by nonabelian fields.  To facilitate this, the truncation scheme for gauge fields described in \cite{BY06} ought to be useful. 

Furthermore, from a quantum simulation point of view, this could reduce the complexity of simulating fermions interacting with gauge fields.  That said, as we saw in chapter \ref{chap:Background1}, the nonlocality of fermions can already be dealt with efficiently in discrete quantum simulations of quantum physics \cite{AL97,OGKL01,SOGKL02,PBE10}.

The prescription we used is a little artificial.  But perhaps it is possible to construct something that seems more natural.  It may even be possible to reduce the number of redundant degrees of freedom further.  It is worth noting, however, that gauge theory itself is built upon introducing redundant degrees of freedom.  Still, it would be remarkable if there were a choice of gauge in which we could use this trick to get a more natural and possibly more economical model.  With this in mind, it may help to note that we only need to reproduce a theory with states satisfying Gauss' law.

There does not appear to exist a direct analogue of this result for gauge theories in the continuum.  Still, we should keep in mind that continuum Hamiltonians are generally only meaningful after regularization, for example, on a lattice \cite{Creutz83}.

To simulate lattice gauge models on quantum computers, some truncation scheme is needed to render the number of degrees of freedom finite.
In light of this, another possibility is to look at quantum link models, which have finite dimensional gauge degrees of freedom on each link \cite{Wiese13}.  For these models there is an exact gauge symmetry.  Incidentally, it is known that the link degrees of freedom could be represented by fermions \cite{Wiese13}.  So it seems that it ought to be straightforward to view these models as local models in the qubit picture also.

This idea of fermions plus gauge interactions leading to local qubit models seems intimately connected to the ideas of \cite{LW05a,LW05}, where the starting point is local bosonic models.  The aim is then to recover fermions interacting with gauge fields from this.  We could think of the results we have seen here as going in the opposite direction, in that they verify the idea that fermions and gauge fields can lead to local qubit models.  It would be good to pursue this connection further.

\section{Quantum Computation and Relativistic Physics}
The primary justification for studying quantum computing is the idea that it should allow us to efficiently solve some problems that classical computing cannot.  Simulating physics is one example of such a problem.  But it is plausible that quantum computers could make it possible to deal with other obstacles involved in simulation.  One obstacle that we have already seen is the problem of fermion doubling.  While there are many proposed solutions, each has different drawbacks.  As these solutions are geared towards classical simulations, they involve local Hamiltonians, a requirement that is important from a practical point of view but also a physical one.  The possibilities change if we look at solutions geared towards quantum simulations of physics.  These need not involve local Hamiltonians.  For example, in a discrete time picture it may be possible to simulate chiral fermions without doubling, while also maintaining strict causality.  We should ask the question of whether there are any other advantages to quantum simulations of high energy physics aside from a potential speedup.

In \cite{JLP12,JLP14} two quantum field theories were shown to be efficiently simulable on a quantum computer.  One aspect of this that made things difficult was the calculation of coupling constants.  Figuring out how these change with the lattice spacing is crucial to doing a simulation.  With this in mind, it is natural to ask ourselves whether quantum computation can help with this.  Does it offer any potential benefits towards calculating how couplings scale and how to better understand renormalization?  In principle, we could represent the coupling constant by a quantum system so that superpositions of systems evolving with different coupling constants could be created.  Whether or not this would be of any use is not clear.  To make any headway towards applying quantum computing to understand renormalization better, we would need to formulate a more concrete computational problem.

More generally, quantum information may provide further tools to improve our understanding of high energy physics.  Already, ideas from quantum information have been applied to study the renormalization group from an information theoretic perspective \cite{BO14}.  So we end part one with the broader question:\ can quantum computation, and more generally quantum information, help us to do more than simulate physics?  What else can we learn about nature?

\part{Quantum Information and Statistical Physics}
\chapter{Introduction}
\begin{center}
{\em Those are my principles, and if you don't like them... well, I have others.}\\
- unknown
\end{center}
Statistical physics is a vast and beautiful branch of physics.  It provides us with a framework for understanding physical processes involving far too many degrees of freedom to simply write down an equation of motion and even approximately solve it.  The basis of statistical physics is the dynamical laws of the constituent particles (including whether they are classical or quantum, fermions or bosons), as well as some additional general principles.  It is these principles that will concern us in this, the second part of this thesis.  We will ask whether we can do better:\ can we justify these principles or even remove them entirely?  It would be fascinating if we could derive these principles as consequences of the framework, which is quantum theory in our case.

One of these principles is the assumption of equal a priori probability.  The content of the assumption is to take every accessible microstate of a system in equilibrium to be equally likely.  Its main use is in deriving a simple form for the equilibrium state of a subsystem in contact with an environment.  A reason we may be unsatisfied with this principle is that it relies on subjective ignorance to derive physical results.  Fortunately, in the quantum case, there are good reasons to believe that this assumption can be circumvented \cite{GLTZ06,PSW06}.  In both of these papers it was proved that, provided the environment is sufficiently big, for most pure states of the subsystem and environment, the state of the subsystem is very close to the state resulting from the assumption of equal a priori probability.  This is a result of entanglement between the subsystem and the environment.  Therefore, any uncertainty in the subsystem's state is objective uncertainty.  This result requires statements about most states with respect to a particular measure on states,\footnote{This is called the Haar measure.  In the following chapter we will define this and other terms, as well as providing a physical definition of equilibration.} so it is not clear whether this is the whole story, but it does provide cause for optimism.  Here, the fact that nature is quantum appears to offer more than in the classical case, as such a result is not possible classically.  This prompts the question of whether quantum theory tells us more about the other postulates.

Over the past ten years, advances have been made towards justifying another of the foundational postulates of statistical physics:\ the assumption that systems equilibrate.  Unlike the example of equally a priori probability, which is an intermediate step in deriving the equilibrium state, the assumption that systems equilibrate is backed up by experience.  Still, it would be interesting to try to understand why this occurs and to see how generic the behaviour is.  It is just as interesting to ask if there are situations when equilibration does not occur.  In fact, engineering systems in a laboratory that do not equilibrate over reasonable timescales could have practical uses.

A lot of progress has been made towards proving that equilibration occurs for quantum systems (including results reported in this thesis), often by applying methods from quantum information \cite{Tasaki98,GMM04,CDEO08,Reimann08,LPSW09,Short10,LPSW10,Reimann10,CE10,RGE11,SF12,MGLFS14}.  It was shown in \cite{Reimann08,LPSW09,Short10,Reimann10,LPSW10} that, in a few different settings but with great generality, quantum systems reach equilibrium.  In these cases equilibration was shown to occur over infinite timescales.  As interesting as this is, for results about equilibration to have any physical importance, it is vital that the timescale over which it occurs is realistic.  Stronger results have been obtained for free bosons on a lattice, where equilibration over finite intervals was proved in \cite{CDEO08,CE10}.  This included calculations of both the time taken to reach equilibrium and the equilibrium state itself.  The approach used Lieb-Robinson bounds, which effectively bound the rate of propagation of information on a lattice.  

Equilibration of quantum systems has also been studied by averaging over initial states, Hamiltonians or measurements.  Historically, one of the first analyses of equilibration of quantum systems was done in \cite{BL59}, where the initial state was averaged over.  The main result was that most of the time (over an infinite time interval) most states have expectation values for coarse-grained measurements that are close to those resulting from the assumption of equal a priori probability.

The idea of averaging over Hamiltonians to make statements about equilibration for most Hamiltonians was employed in \cite{BCHHKM12,VZ12,Cramer12,MRA13}.  Furthermore, in \cite{MGLFS14}, it was shown that fast equilibration occurs for most measurements.  There and in \cite{GHT13} it was also shown that one can construct examples where the equilibration time is incredibly long.  We will not look at averaging over Hamiltonians, states or measurements here, rather we will assume a fixed but very general Hamiltonian, measurement set and initial state.  This is the main strength of our approach.  It is also the main weakness, as the bounds on the equilibration time are quite big.

Once equilibration has occurred, the next issue to deal with is the form of the equilibrium state.  The majority of statistical physics is built upon the idea that the subsystem's equilibrium state is a Gibbs state.  The usual justification for this requires the assumption of equal a priori probability, weakness of the interaction between subsystem and environment, and exponential growth of the density of energy levels of the environment.  To date, very little has been justified rigorously regarding this, with the exception of \cite{RGE11}, which requires very weak interactions.  Other interesting results regarding thermalization, which is the process by which a subsystem reaches the Gibbs state, were obtained in \cite{MAMW13}.  There it was shown that, with translationally invariant dynamics and some other caveats, subsystems equilibrate to the reduced state of the Gibbs state of a larger subsystem.  These results ought to be particularly useful in the regime of strong interactions.

A nice, succinct review of equilibration and thermalization, focusing on many-body systems, is given in \cite{EFG14}, while \cite{Gogolin14} provides a more detailed and general discussion.

The breakdown of this part of the thesis is as follows.  In chapter \ref{chap:Background}, we will see the necessary concepts from quantum information and statistical physics that will play a role in discussing equilibration and thermalization.  The main results will be presented in chapter \ref{chap:Equilibration}, where we discuss equilibration over finite intervals and bounds on the time taken for equilibration to occur.  We also look at the nature of the equilibrium state and discuss the related question of when a subsystem retains information about its initial state, something that precludes thermalization.  We end with a discussion of the results, as well as open questions, in chapter \ref{chap:Conclusions2}.

\chapter{Background}
\label{chap:Background}
\section{Quantum Information}
Inspired by the knowledge that our world is fundamentally quantum, the subject of quantum information is dedicated to understanding how to process and manipulate information as stored by quantum systems.  It incorporates a large body of ideas and tools, though we will only need a few of them directly here.  But as the connection between quantum information and many-body physics grows, it is likely that further insight into equilibration and thermalization could be found by applying other ideas from quantum information.

In the previous chapter, we had cause to mention a few tools from quantum information.  Introducing these properly is a bit of a digression, but it is useful to understand the main ideas behind them.

Many recent results in the foundations of quantum statistical physics make statements that are of a statistical nature about a system's properties.  Often these involve averaging over states, as in \cite{BL59}.  And in quantum theory there is a mathematically natural measure on the set of pure states, called the Haar measure.  To be completely accurate, what happens is we use the Haar measure on the unitary group $\textrm{U}(d)$ to define a measure on the state space.  Let us see how this works.  There is a measure $\mu$ on the group with the property that subsets of the group that are related by the group action have same size.  In other words, given a subset of the group $G\subseteq \textrm{U}(d)$ and an element of the group $U$, we have $\mu(UG)=\mu(GU)=\mu(G)$.  The corresponding measure on states is constructed by applying unitaries from the group onto states and then averaging over the group:\ 
\begin{equation}
 \Big\langle f(\ket{\psi})\Big\rangle_{\ket{\psi}}=\int_{\textrm{U}(d)}\!\textrm{d}\mu(U)\, f(U\ket{\psi}),
\end{equation}
where $f$ is some function of vectors in the Hilbert space.  A useful example is to take $f(\ket{\psi})=\ket{\psi}\bra{\psi}$, which gives us the average state of a system,
\begin{equation}
 \int_{\textrm{U}(d)}\textrm{d}\mu(U)\, U\ket{\psi}\bra{\psi}U^{\dagger}=\frac{\openone}{d}.
\end{equation}
To see that this must be the case, we use the invariance of the measure under the group action.  This means that whatever state we get must be invariant under transformations of the form $\rho\rightarrow V\rho V^{\dagger}$ for any unitary $V$.  The maximally mixed state is the only one with this property.  To ensure that the result was a state, we also needed to use the fact that the measure is normalized, meaning
\begin{equation}
 \int_{\textrm{U}(d)}\textrm{d}\mu(U)=1.
\end{equation}
From the point of view of equilibration or thermalization, it often helps to use the Haar measure to prove results because of its nice mathematical properties.

Another useful set of results are Lieb-Robinson bounds \cite{LR72}.  These bound how far information can travel on a lattice in a given time under evolution via a local Hamiltonian.  The physical intuition is that there is a speed of sound, which is a property of the type of material and is determined by the Hamiltonian.  The most basic form of the bound is the following.  Given two operators $A$ and $B$ localized on regions $X$ and $Y$ of a quantum lattice system, then, with $A(t)$ the evolved operator in the Heisenberg picture,
\begin{equation}
 \|[A(t),B]\|\leq c\, e^{-b\big(d(X,Y)-v|t|\big)},
\end{equation}
where $b$, $c$ and $v$ are constants, and $d(X,Y)$ is the distance between the regions $X$ and $Y$.  The constant $v$ depends on the Hamiltonian, and $c$ is a function of $\|A\|$, $\|B\|$, the sizes of the regions $X$ and $Y$ and the Hamiltonian.  This inequality is useful for approximating evolved observables by observables on finite regions of space.  This works because, once the commutator of $A(t)$ with all other operators on a region is small, then the action of $A(t)$ on that region must be close to the identity.  When trying to understand the evolution of an observable over time, these inequalities allow us to effectively restrict our attention to finite regions.

What will be most useful from quantum information for our purposes will be distance measures between states, which we will discuss in section \ref{sec:Distinguishing States}.  Before that, in section \ref{sec:Time Averages and Energy Filtering} it will be useful to look at averaging expectation values and other quantities over time.

As is often the case in quantum information, we will focus on proving results, such as bounds on equilibration times, for finite dimensional quantum systems.  That said, it is reassuring that this assumption is not always necessary:\ the main results in section \ref{sec:Timescales for Equilibration of Expectation Values} were extended to infinite dimensional Hilbert spaces in \cite{RK12}.  To be precise, the results extend with the condition that the Hilbert space be separable, meaning there is a countable basis, and that the Hamiltonian has pure point spectrum, meaning that it has a discrete spectrum and the eigenvectors are elements of the Hilbert space.  For us, it will be enough to take our quantum system to be $d$-dimensional.

\subsection{Time Averages and Energy Filtering}
\label{sec:Time Averages and Energy Filtering}
To find bounds on the equilibration time, we will need to take time averages of quantities, like expectation values.  A prime example is the average state of a system over time.  We expect that, if a quantum system equilibrates, it will equilibrate to its time average state.  We will define what we mean by equilibration in section \ref{sec:Equilibration}.  Here, we will look at time averages over intervals.  We denote the time average over the interval $[0,T]$ by $\langle \cdot \rangle_{T}$.  This is defined to be
\begin{equation}
\langle a(t)\rangle_{T} = \frac{1}{T}\int_{0}^{T}\! a(t) \textrm{d}t,
\end{equation}
where $a(t)$ is some function of $t$.  A special case that is often relevant is the average over the interval $[0,\infty)$, which is defined by taking the limit:\ 
\begin{equation}
\langle a(t)\rangle_{\infty} = \lim_{T\rightarrow \infty}\frac{1}{T}\int_{0}^{T}\! a(t) \textrm{d}t.
\end{equation}
For a closed system evolving unitarily via a time independent Hamiltonian with energies $E_i$, we have
\begin{equation}
 \rho(t) = e^{-iHt} \rho(0) e^{iHt} = \sum_{i,j}e^{-i(E_{i}-E_{j})t}\rho_{ij}|i\rangle\langle j|.
\end{equation}
So the time average state is
\begin{equation}
\av{\rho(t)}_{T} = \sum_{i,j}\av{e^{-i(E_{i}-E_{j})t}}_{T}\rho_{ij}|i\rangle\langle j|.
\end{equation}
If the eigenvalues are non degenerate, the infinite time average is simply
\begin{equation}
\omega=\langle \rho(t)\rangle_{\infty} = \sum_{i}\rho_{ii}|i\rangle\langle i|,
\end{equation}
where we use $\omega$ to denote the infinite time average state.  This is a special case of the dephasing channel.  A quantum channel is just a map that takes quantum states to quantum states.\footnote{To be completely accurate, it is a completely positive trace preserving map.  Complete positivity ensures that, acting on a subsystem of a larger system, the map takes the quantum state of the entire system to another legitimate quantum state.}  The behaviour of the state over time is clearly determined by the energy gaps,\footnote{The gaps we are discussing are the gaps between energy levels of a finite dimensional quantum system.  This is not to be confused with the question of whether a Hamiltonian is gapped or gapless, which requires going to the thermodynamic limit.} which is something we will look at more in section \ref{sec:Distribution of Energy Gaps}.

\begin{compactwrapfigure}{r}{0.47\textwidth}
\centering
\begin{minipage}[r]{0.43\columnwidth}%
\centering
    \resizebox{6.5cm}{!}{\includegraphics{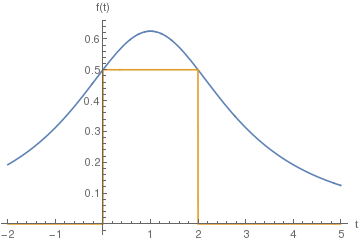}}
    \footnotesize{\caption[Plot of a Lorentzian upper bounding a top hat function]{Plot of a Lorentzian (blue) and a top hat function (yellow).  The Lorentzian is chosen to be greater than or equal to the top hat function at every point.}}
    \label{fig:Lorentzian}
\end{minipage}
\end{compactwrapfigure}

Generally, our goal will be to upper bound time averaged quantities, like $\langle\tr{\rho(t)A}\rangle_T$.  Sometimes, however, it is easier to get upper bounds for averages of positive functions by weighting the time average differently.  This works as follows.  The time average over $[0,T]$ corresponds to integrating over all time with the weight $h(t)$, given by
\begin{equation}
\label{eq:tophat}
h(t) =\begin{cases} \frac{1}{T}\ \ \textrm{if}\ t\in[0,T]\\
0\ \ \textrm{otherwise.}
\end{cases}
\end{equation}
Suppose we have a function $f(t)$ with $f(t)\geq h(t)$ for all $t$.  Then we define the weighted average
\begin{equation}
\langle a(t)\rangle_{f} = \int_{\mathbb{R}}f(t)a(t) \textrm{d}t,
\end{equation}
where $a(t)\geq 0$.  It follows that $\langle a(t)\rangle_{f}\geq \langle a(t)\rangle_{T}$.  There are some useful possibilities.  For example, we could choose $f(t)$ to be a Gaussian or Lorentzian function.  The latter of these, which we will use later, is defined to be
\begin{equation}
\label{eq:Lorentzian}
 f(t)=\frac{\f{5}{4}T}{T^2+(t-\f{T}{2})^2},
\end{equation}
where the constants are chosen to ensure that this is greater than the top hat function defined in equation (\ref{eq:tophat}).

In some scenarios, we can think of these weighted averages in a different way:\ they are sometimes equivalent to energy filtering, which is used in \cite{FWBSE14}, for example.  The idea there is to average a local observable in the energy basis weighted in a similar way. 
\begin{equation}
 I^H_f(A)=\int_{\mathbb{R}} f(t) A(t)\textrm{d}t,
\end{equation}
where $A(t)=e^{iHt}Ae^{-iHt}$ is the evolved observable in the Heisenberg picture, and $f(t)$ is a positive function.  The idea is to choose $f(t)$ in such a way as to optimize the trade-off between spatial locality and narrowness in the energy basis.  For example, if we use the top hat function $h(t)$ for some very large $T$, then
\begin{equation}
 I^H_h(A)=\frac{1}{T}\! \int_{0}^{T}\! A(t)\textrm{d}t \simeq \sum_{i}A_{ii}\ket{i}\bra{i},
\end{equation}
where we assumed that the energies were not degenerate.  This is very narrow in the energy basis.  But it is averaged over a very long time, so we would expect that it is not a local observable in general.

\subsection{Distinguishing States}
\label{sec:Distinguishing States}
What we are looking for here is a notion of distance between quantum states that has a good physical interpretation.  For pure states, it makes sense to use the metric induced by the inner product on the Hilbert space, $\|\cdot\|_2$.

To get a physical distance measure, we look at metrics arising in quantum information theory.  A natural choice then is the trace distance between two states.  This is defined as follows.  Given two quantum states $\rho$ and $\sigma$, the trace distance between them is
\begin{equation}
 D(\rho,\sigma) = \f{1}{2}\|\rho-\sigma\|_1=\f{1}{2}\textrm{tr}|\rho-\sigma|,
\end{equation}
where $\|\cdot\|_1$ is the trace norm.  This gives the maximum difference in probability arising from doing a measurement on the two states.\footnote{For pure states, $\ket{\psi}$ and $\ket{\phi}$, the trace distance is equal to $\sqrt{1-|\langle \psi|\phi\rangle|^2}$.}  In other words, this formula is equivalent to
\begin{equation}
 D(\rho,\sigma) = \max_{P}\tr{P(\rho-\sigma)},
\end{equation}
where $P$ is a projector, and the right hand side is maximized over all projectors.  There is a nice discussion of distance measures, including the trace distance in \cite{NC00}.  More generally, the trace distance allows us to bound the difference in expectation values on two states via
\begin{equation}
|\tr{\rho A}-\tr{\sigma A}| \leq \left(a_{\textrm{max}}-a_{\textrm{min}}\right)D(\rho,\sigma),
\end{equation}
where $a_{\textrm{max}}$ and $a_{\textrm{min}}$ are the maximum and minimum eigenvalues of $A$.

A final point about the trace distance is that an equivalent definition is
\begin{equation}
 D(\rho,\sigma)=\max_{\{P_i\}}\frac{1}{2}\sum_{i}|\tr{P_i\rho}-\tr{P_i\sigma}|.
\end{equation}
The maximum is over all complete sets of projectors, meaning they sum to the identity.

It is good to keep in mind that we want to look at systems with many degrees of freedom or many particles.  So a valid objection to raise here is that we may not be able to do any measurement we want on the system.  After all, the idea that we could do {\em any} measurement we want on $10^{23}$ particles is fanciful, to put it mildly.  Of course, in some situations we could be studying the state of a very small system; maybe we are interested in a few spins on a long spin chain.  In this case, the trace distance is a sensible measure of distance between states.

Nevertheless, this line of reasoning leads us to a more reasonable notion of distance between states.  By simply restricting the set of measurements that we maximize over to a set of measurements we can do in practice, we arrive at what we will call the distinguishability \cite{Short10}, defined by
\begin{equation}
D_{\mathcal{M}}(\rho,\sigma)=\max_{\{P_i\}\in\mathcal{M}}\frac{1}{2}\sum_{i}|\tr{P_i\rho}-\tr{P_i\sigma}|.
\end{equation}
Notice that, if the set of measurements $\mathcal{M}$ includes all projective measurements, then we recover the trace distance.

As an aside, here and in the definition of the trace distance we could make a generalization to include all POVM (Positive Operator Valued Measure) measurements.  In that case, each measurement outcome $i$ has a corresponding positive operator $M_i$.  Then the probability of getting that outcome if the system is in the state $\rho$ is $\tr{\rho M_i}$.  For the probabilities to add to one, we require $\textstyle\sum_i M_i = \openone$.  So a complete set of projectors is an example of a POVM.

The distinguishability between two states has a nice interpretation in terms of a game where we try to distinguish between two states $\rho$ and $\sigma$.  The best strategy is to do a measurement and guess that we have whichever state is more likely to give the observed outcome.  Then the probability of success is
\begin{equation}
p_{\rm success} = \frac{1}{2} +\frac{1}{2}D_{\mathcal{M}}(\rho,\sigma).
\end{equation}

Some further comments on the distinguishability may be useful.  The first is that, if $\rho=\sigma$, then it is zero.  Also, it is symmetric and obeys the triangle inequality.  But there may exist states $\rho$ and $\sigma$ for which the distinguishability is small despite the two states being very different, meaning the trace distance between them is large.  This is a consequence of the fact that we observed before:\ it may not be possible to do any measurement on a system.  As a result of this last property, in general the distinguishability is not a metric.  This is because a metric $d(x,y)$ is zero if and only if $x=y$.

\section{Statistical Physics}
Our aims in deriving results about equilibration are quite general, which will turn out to be both a strength and a weakness.  Because of this, we do not need to be very specific about the nature of the quantum systems we will study.  That said, there are a few very weak conditions on the systems that we need for equilibration, which we will introduce in the following sections.  Also, there are a few basic physical principles that ought to play a role in future work.

Our main interest is in systems that are macroscopic, much like those typically encountered in statistical physics, where there is a relatively small subsystem in contact with a much larger environment.  Still, the formulas and bounds can be applied to small systems.  This opens up the possibility of testing ideas out on relatively small systems in the lab, like trapped ions, for example.

\begin{compactwrapfigure}{r}{0.41\textwidth}
\centering
\begin{minipage}[r]{0.39\columnwidth}%
\centering
    \resizebox{4.5cm}{!}{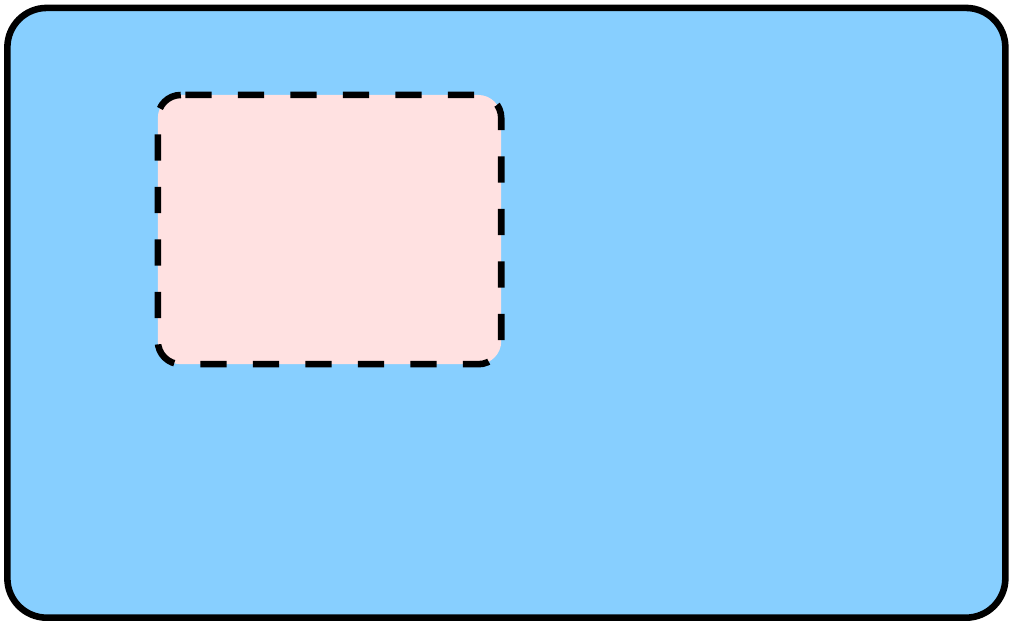}
    \footnotesize{\caption[A subsystem and its environment]{Subsystem and environment.  The boundary of the subsystem $\partial A$ is marked by a dashed line.  With local interactions, the interaction Hamiltonian couples sites across the boundary.\ \\ \ }}
    \label{fig:Subsystem_and_environment}
\end{minipage}
\end{compactwrapfigure}

For equilibration to occur when the system comprises a subsystem and an environment, we expect that there must be some interaction between them.  Indeed, some interaction should be necessary to result in equilibration to a thermal state, as the subsystem must forget its initial state.  We will assume, however, that the system (meaning the subsystem plus environment) itself is isolated and evolves via a time independent Hamiltonian $H=\sum_{n} E_n P_n$, where $E_n$ are energies, and $P_n$ is the projector onto the eigenspace of $H$ corresponding to energy $E_n$.  We denote the number of distinct energies by $d_E$.

In being as general as possible, we will not specify the type of system, but we will require that its Hamiltonian obeys the very weak requirement that there are few energy gaps that have the same size.  Another requirement that is crucial is that the system has a high effective dimension, which means that the state of the system is well spread over many energy eigenstates.  We will discuss these conditions further in the following sections.

In our approach, we consider the unitary evolution of a closed system.  There is another option, where the subsystem is still in contact with an environment but we do not work with the full Hamiltonian.  The subsystem Hamiltonian, together with some assumptions about the interaction with the environment, are used to derive an equation describing the evolution of the subsystem, called a Lindblad equation \cite{NC00}.  We will not consider this approximation here.  Because the memory effects of the environment are not taken into account, this approach, like any approximate method, becomes more inaccurate as the timescale grows.  As we deal with the exact Hamiltonian, this is not a problem here.

It is possible to use the Haar measure to average over Hamiltonians, leading to results about equilibration for most Hamiltonians \cite{BCHHKM12,VZ12,Cramer12}.  This works by fixing a Hamiltonian $H$, calculating quantities (such as the equilibration time) in terms of the Hamiltonian $UHU^{\dagger}$, and then averaging over $U$ via the Haar measure.  While the Haar measure makes calculations feasible, we should not forget that we expect physical Hamiltonians to have properties like locality and maybe translational invariance.  It seems likely that most of the Hamiltonians resulting from the Haar measure average include interaction terms coupling many sites over large distances.  So, naively, we would expect that physical Hamiltonians make up only a small subset of all Hamiltonians according to the Haar measure.  It would be more natural but almost certainly more complicated to work with a measure on local Hamiltonians.

Along similar lines, one can say something about equilibration times with respect to most measurements \cite{MGLFS14}.  The situation is analogous, but now we fix a projector $P$ and average quantities evaluated in terms of $UPU^{\dagger}$ over $U$.  It was shown in \cite{MGLFS14}, that most measurements lead to very fast equilibration times even when the initial state is far from equilibrium.  Again, this raises a similar point:\ physical observables, like $\sigma_x$ on one spin, would seem to constitute a very small portion of the set of all observables.

The breakdown of this section is as follows.  We will give a sensible definition of equilibration in section \ref{sec:Equilibration}.  In section \ref{sec:Macroscopic Systems}, we define the effective dimension, particularly in light of the systems we consider being macroscopic.  And, finally, in section \ref{sec:Distribution of Energy Gaps}, we discuss the non degenerate energy gaps condition, as well as a weaker condition required in later proofs.  We also look at how the energy gaps of physical systems are distributed.  This is useful because fluctuations from the equilibrium state are controlled by the energy gaps.

\subsection{Equilibration}
\label{sec:Equilibration}
All the systems we will look at are closed and evolve via time independent Hamiltonians.  A consequence of this is that it is not possible for a system to evolve to a steady state if it did not start in one.  Furthermore, it is known for systems that have discrete spectra that there are recurrences \cite{Bocchieri57,Schulman78}.  This means that, if we wait sufficiently long, the system will return arbitrarily close to its initial state.  So the naive notion of equilibration as relaxation to a steady state is not appropriate in our setting.

We need a more realistic definition of equilibration.  We say that a system equilibrates over some interval $[0,T]$ if, for most times in that interval, the system is close to a fixed state.  An interesting aspect of this definition is the idea that equilibration may depend on the interval of time we are looking at.  This is not too surprising.  For example, a fresh cup of tea will cool to room temperature in about half an hour.  So over the course of a few hours we would say that the equilibrium state is cold tea.  Left in the corner for two or three months, the tea will evaporate.\footnote{This experiment was accidentally performed by the author.}  So over this longer period the equilibrium state is evaporated tea.  The idea that the equilibrium state can depend on the timescales involved is something we will return to later in chapter \ref{chap:Conclusions2}.

How close two states are depends on what measurements we can do, which is something we discussed in \ref{sec:Distinguishing States}.  And quantifying closeness by the distinguishability allows us to define equilibration.  Because the distinguishability is positive, if its average over some interval is small, then it must be small most of the time.  This leads us to the following definition.
\begin{definition}
\label{def:eq}
A system in state $\rho(t)$ equilibrates to the state $\sigma$ over the interval $[0,T]$ with respect to the set of measurements $\mathcal{M}$ if and only if
\begin{equation}
 \left<D_{\mathcal{M}}\big(\rho(t),\sigma\big)\right>_{T} \leq \epsilon,
\end{equation}
where $\epsilon$ is much less than one.
\end{definition}

In other words, if $\left<D_{\mathcal{M}}\big(\rho(t),\sigma\big)\right>_{T}\leq \epsilon$, then we cannot distinguish $\rho(t)$ from $\sigma$ with the measurements in $\mathcal{M}$ for most times in $[0,T]$.

We have been a little vague about how small $\epsilon$ ought to be.  One option is to choose it in an operational sense.  We can choose some $\delta>0$ such that, if $D_{\mathcal{M}}\big(\rho(t),\sigma\big)\leq \delta$, then we say the states are indistinguishable.  Suppose that, averaged over $[0,T]$, the distinguishability is $\epsilon$.  Then the maximum amount of time that the distinguishability can be greater than $\delta$ is given by $\Delta T =\f{\epsilon}{\delta}T$.  For equilibration to occur, we want the system to spend most of its time close to $\sigma$.  So we want $\f{\Delta T}{T}$ to be sufficiently small, which gives us $\epsilon$.

\begin{compactwrapfigure}{r}{0.45\textwidth}
\centering
\begin{minipage}[r]{0.44\columnwidth}%
\centering
    \resizebox{6.5cm}{!}{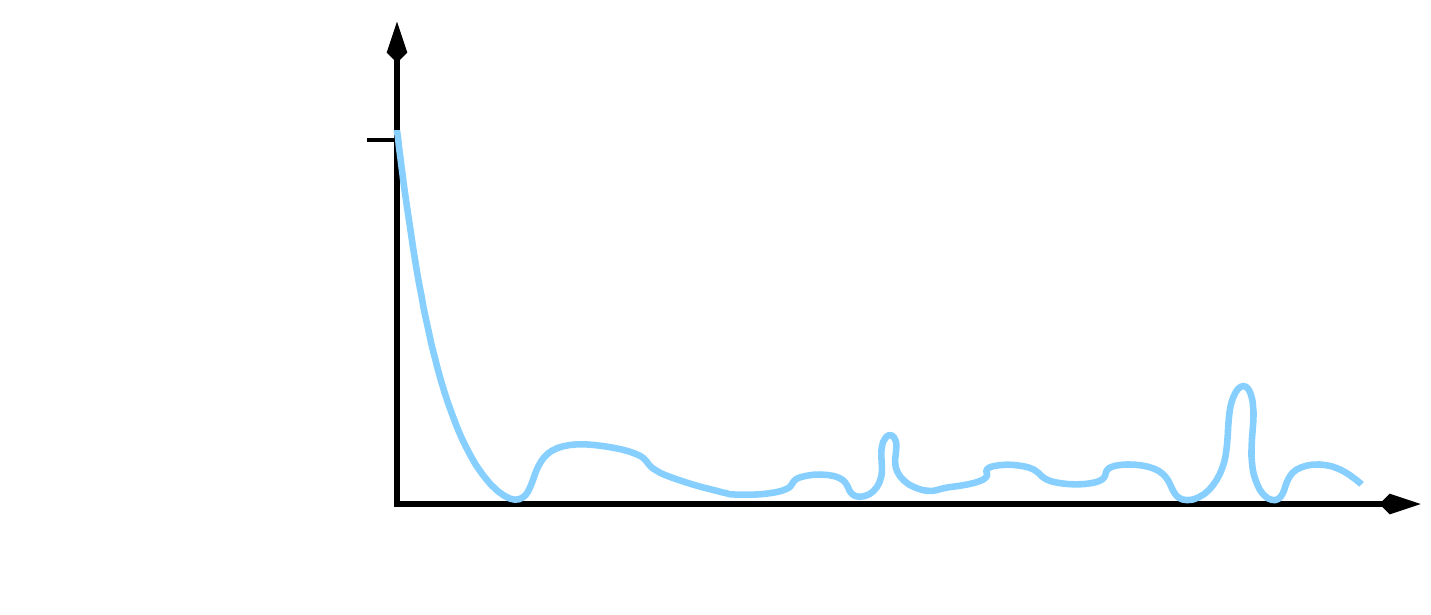}
    \footnotesize{\caption[Equilibration with respect to a set of measurements]{Equilibration with respect to a set of measurements.  There will always be fluctuations and even recurrences, but these will be rare.\ \\ \ }}
    \label{fig:Equilibration}
\end{minipage}
\end{compactwrapfigure}

In a sense, we have just replaced one problem with another:\ we need to fix a $\delta$ and a $\Delta T$.  Luckily, the bounds we will encounter for $\left<D_{\mathcal{M}}\big(\rho(t),\sigma\big)\right>_{T}$ decrease with the size of the interval $T$.  Indeed, this is how we will define the equilibration time.  We choose it to be a $T$ that is sufficiently big that the time average of the distinguishability is less than $\epsilon$.  So, for whatever $\epsilon$ we choose, we can always achieve this by picking a large enough $T$.  Now, the crucial point is that we really care about how the equilibration time grows with the system size.  Because of this, the actual value of $\epsilon$ is immaterial as long as it is constant.

Additionally, in some of the infinite time bounds, $\av{D_{\mathcal{M}}\big(\rho(t),\sigma\big)}_{\infty}$ is bounded above by something exponentially small in the system size, so there is no need to worry about choosing an $\epsilon$.

Another definition of equilibration was used in \cite{Reimann08}.  This is essentially equilibration of expectation values.
\begin{definition}
We say that the expectation value of $A$ equilibrates if
\begin{equation}
\frac{\av{|\tr{\rho(t)A}-\tr{\omega A}|^2}_{T}}{\|A\|^2}\leq \epsilon_A^2,
\end{equation}
where  $\epsilon_A$ is much less than one.
\end{definition}
This definition requires that the average fluctuations of the expectation value over time are small relative to their maximum possible value, $\|A\|$.  Again, the value of the right hand side, $\epsilon_A$, is something that is operationally motivated.  In \cite{Reimann08}, the discussion is framed in terms of the resolution of the experimental apparatus.  The idea is that, if $\|A\|\, \epsilon_A$ is sufficiently small compared to the resolution, we cannot detect these fluctuations most of the time.

That said, we do not measure expectation values.  And even if the expectation value of a particular observable equilibrates, this does not necessarily mean that it is impossible to distinguish $\rho(t)$ from its time average by measuring that observable \cite{Short10}.  It also does not tell us that the system has equilibrated.  For an example in terms of equilibration over the interval $[0,\infty)$, suppose the initial state is
\begin{equation}
 \ket{\psi(0)}=\frac{1}{\sqrt{d}}\sum_i\ket{i},
\end{equation}
where $\ket{i}$ are energy eigenstates.  For any time $t$, there exists a measurement that can distinguish between the system state $\rho(t)$ and its time average with very high probability.  In fact, the trace distance is
\begin{equation}
D(\rho(t),\omega) = \f{1}{2}\textrm{tr}|\rho(t)-\omega|=1-\frac{1}{d}.
\end{equation}
To see this, use $\omega=\openone/d$, assuming the energies are not degenerate.  Because the Hilbert space dimension $d$ is typically very large, this tells us that the state of the system is always far from the time average state.  In spite of this, the expectation values of observables do equilibrate in this situation.  This is a consequence of $\tr{\omega^2}$ being very small, as we will see in chapter \ref{chap:Equilibration}.  Incidentally, a two-outcome measurement that allows us to distinguish these states well has projectors $\rho(t)$ and $\openone-\rho(t)$.

Results about equilibration of expectation values will be useful as a stepping stone to prove results about equilibration defined in terms of the distinguishability, as in definition \ref{def:eq}.

\subsection{Macroscopic Systems}
\label{sec:Macroscopic Systems}
We will define macroscopic systems to be systems with particle numbers of the order of $10^{23}$ or greater.  While the results in this part hold for smaller systems, it is good to keep in mind that the main goal is to understand the foundations of statistical physics, so the formulas we obtain need to be applicable to macroscopic systems.

The number of energy levels increases exponentially with the number of particles (or degrees of freedom).\footnote{We are assuming that the Hamiltonian does not have too many degeneracies.  This would rule out something like $H=\openone$, which only has one energy level.}  This is because the Hilbert space dimension increases exponentially with the number of subsystems.  And because the energy range only increases linearly with the number of particles, this means that the number of energy levels in a fixed energy window will increase exponentially.

For equilibration to occur, we expect that recurrences will need to be rare.  So we need an assumption about how much of the Hilbert space the system explores over time.  We can quantify this by the effective dimension, which is defined to be
\begin{equation}
d_{\eff}=\frac{1}{\displaystyle\sum_{n}\left(\textrm{tr}\left[P_{n}\rho(0)\right]\right)^2}.
\end{equation}
This is an indicator of how well a system's state is spread over energy eigenstates.  For example, an energy eigenstate has $d_{\eff}=1$.  Also, a system with equal probability to be in any energy eigenstate has $d_{\eff}=d_E$.

For a pure state, the effective dimension tells us how much the system changes over time.  This is because the state is a superposition of energy eigenstates, and the phase of each of these terms oscillates at a different rate, depending on the energy.

The main point, which holds for pure or mixed states, is that a large effective dimension means that the system's state is spread over many energy eigenstates.  And it is practically impossible to prepare a macroscopic system in a state with a low effective dimension \cite{Reimann08}.  This is because the exponential energy level density implies that the effective dimension should increase exponentially with the system size.  But macroscopic systems have $O(10^{23})$ particles, which makes the number of energy levels in a fixed window incredibly dense.  So we should expect the effective dimension to be extremely large.

\subsection{Distribution of Energy Gaps}
\label{sec:Distribution of Energy Gaps}
Unsurprisingly, the bounds on equilibration times we will see in chapter \ref{chap:Equilibration} depend strongly on the energy gaps.  Things like how dense the spectrum is and how many very small gaps there are play a huge role, as we will see.  This seems sensible.  After all, the elements of the density matrix in the energy basis are
\begin{equation}
\label{eq:matel}
 \rho_{ij}e^{-i(E_i-E_j)t},
\end{equation}
which oscillate with frequency determined by the energy gaps.  So answers to questions regarding how the state changes over time will need to involve the energy gaps.  Because of (\ref{eq:matel}), we should expect that the more small gaps there are, the slower the state changes over time.  This is something we will see more concretely in the following chapter:\ very small gaps imply very long equilibration times.

It will be useful to choose some notation for the gaps.  Given the energy gap $E_i-E_j$, with $i\neq j$, we define $G_{\alpha}=E_i-E_j$, where $\alpha=(i,j)$.  We also define
\begin{equation}
\label{eq:Nlevel}
N(\varepsilon) = \max_E |\{ \alpha | G_\alpha \in [E,E+\varepsilon) \}|.
\end{equation}

\begin{compactwrapfigure}{r}{0.37\textwidth}
\centering
\begin{minipage}[r]{0.35\columnwidth}%
\centering
    \resizebox{2.5cm}{!}{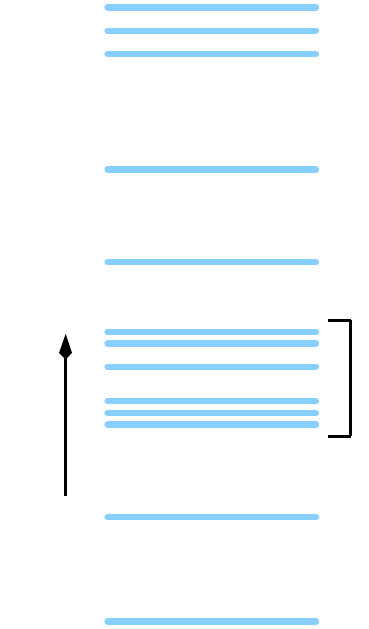}
    \footnotesize{\caption[$N(\varepsilon)$, a measure of the density of energy gaps]{$N(\varepsilon)$, a measure of the density of energy gaps.\ \\ \ }}
    \label{fig:Energy_gap_density}
\end{minipage}
\end{compactwrapfigure}

This counts the maximum number of energy gaps in a window of size $\varepsilon$.  So $N(\varepsilon)$ is a measure of how dense the spectrum is.  The maximum degeneracy of any energy gap is determined by
\begin{equation}
 D_G = \lim_{\varepsilon \rightarrow 0^+} N(\varepsilon).
\end{equation}
In \cite{Reimann08,LPSW09}, it was necessary to assume that $D_G=1$ in order to prove equilibration.  This is the assumption of non degenerate energy gaps.  Put another way, it means that
\begin{equation}
 E_i-E_j=E_k-E_l
\end{equation}
is never satisfied, unless $i=j$ and $k=l$ or $i=k$ and $j=l$.  The results we will see in the next chapter weaken this assumption.  It will be sufficient for us to assume that $D_G$ is small.

There is some intuition behind why we need the degeneracy of the most degenerate energy gap $D_G$ to be small.  Take the expectation value of the operator $A$ in the state $\rho(t)$,
\begin{equation}
\tr{\rho(t)A}=\displaystyle\sum_{ij}\rho_{ij}A_{ji}e^{-i(E_i-E_j)t}.
\end{equation}
\begin{compactwrapfigure}{r}{0.49\textwidth}
\centering
\begin{minipage}[r]{0.47\columnwidth}%
\centering
    \resizebox{7.0cm}{!}{\includegraphics{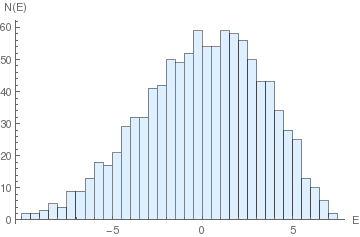}}
    \footnotesize{\caption[Energy distribution of the Heisenberg Hamiltonian with random couplings]{Histogram showing the energy level distribution of the Heisenberg Hamiltonian with random couplings restricted to $[0,1]$.  Even for $10$ spins, the approximately Gaussian shape is noticeable.}}
    \label{fig:Random_Heisenberg_Eigenvalues}
\end{minipage}
\end{compactwrapfigure}
If there are many degenerate energy gaps, then there are many terms that oscillate with the same frequency.  So it is conceivable that, with the right choice of projector $P$,
\begin{equation}
|\tr{P\rho(t)}-\tr{P\av{\rho(t)}_T}|
\end{equation}
could be large for a lot of the time in $[0,T]$, meaning equilibration cannot occur.

The requirement that there be few degenerate energy gaps has another justification, which is that it rules out non interacting systems.  Imagine we have a subsystem and environment that do not interact.  Their energies would look like
\begin{equation}
E_{ij} = E_i^S + E_j^B,
\end{equation}
where $E_i^S$ and $E_j^B$ are the subsystem and environment energies respectively.  But then there are many degenerate energy gaps.  In fact, $D_G$ would be greater than the dimension of the environment's state space, assuming the spectra are non degenerate.

There are some things we can say about the distribution of energy levels and energy gaps.  For a start, it is generally believed, and can be proved in some cases \cite{KLW14}, that the distribution of energies in quantum systems composed of many small subsystems is a Gaussian.  So we have that the probability (according to the continuum distribution) of finding an energy level between $E$ and $E+\textrm{d}E$ is
\begin{equation}
 p(E)=\frac{1}{\sqrt{2\pi}\sigma}e^{-\frac{E^2}{2\sigma^2}},
\end{equation}
where $\sigma$ is the width.  Here, a constant has been subtracted from the Hamiltonian so that the average energy is zero.

\begin{compactwrapfigure}{r}{0.49\textwidth}
\centering
\begin{minipage}[r]{0.47\columnwidth}%
\centering
    \resizebox{7.0cm}{!}{\includegraphics{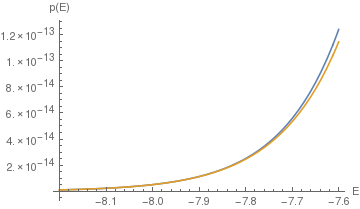}}
    \footnotesize{\caption[Gaussian distribution restricted to a small window]{Gaussian energy distribution restricted to a small window (blue) compared to the exponential approximation (yellow).  Note that we are far from the centre of the Gaussian.}}
    \label{fig:Zoom_into_Gaussian}
\end{minipage}
\end{compactwrapfigure}

On physical grounds, however, we know that states are effectively restricted to some relatively small energy window.  This results in the distribution of energy levels for physical systems effectively being exponential.  To see this is straightforward.  Suppose the system is in a small energy window, centred on $-\overline{E}$, where $\overline{E}$ is positive.  We assume this window is far from the peak of the Gaussian.  Far means that $|\overline{E}|$ is large compared to the width $\sigma$.  Then we have that
\begin{equation}
\begin{split}
 p(E) & =\frac{1}{\sqrt{2\pi}\sigma}e^{-\frac{(E-\overline{E})^2}{2\sigma^2}}\\
  & =\frac{1}{\sqrt{2\pi}\sigma}e^{-\frac{\overline{E}^2}{2\sigma^2}}e^{\frac{2\overline{E}E}{2\sigma^2}}e^{-\frac{E^2}{2\sigma^2}}.
 \end{split}
\end{equation}

Now, as $E$ is small compared to $|\overline{E}|$ and hence $\sigma$, the last term in the product on the last line is approximately one.  Therefore,
\begin{equation}
 p(E)\simeq \alpha e^{\beta E},
\end{equation}
where $\alpha$ and $\beta$ are constants that depend on $\overline{E}$ and $\sigma$.  Since $\alpha$ can be fixed by normalization, the only free constant is $\beta=\frac{1}{k_B\mathbf{T}}$, where $k_B$ is Boltzmann's constant and $\bf{T}$ is the temperature.

We can always subtract a constant from the Hamiltonian, so that the lowest energy is $0$, and we can write the energy range as $[0,\Delta]$.  It follows that $\alpha=\beta/(e^{\beta\Delta}-1)$.

Since gaps are so important, it is interesting to look at the properties of the distribution of gap sizes.  First, let us calculate the average gap size $\av{|E_i-E_j|}$.
\begin{equation}
\begin{split}
 \av{|E_i-E_j|} & = \sum_{ij}p_ip_j|E_i-E_j| = 2\! \sum_{E_i\geq E_j}p_ip_j(E_i-E_j)\\
& \simeq 2\int_{0}^{\Delta}\!\textrm{d}E\int_{0}^{E}\!\textrm{d}E^{\prime}p(E)p(E^{\prime})(E-E^{\prime})\\
&= 2\alpha^2\int_{0}^{\Delta}\!\textrm{d}E\,e^{\beta E}\int_{0}^{E}\!\textrm{d}E^{\prime}e^{\beta E^{\prime}}(E-E^{\prime})\\
& = 2\alpha^2\int_{0}^{\Delta}\!\textrm{d}E\,e^{\beta E}\left(\frac{e^{\beta E}}{\beta^2}-\frac{E}{\beta}-\frac{1}{\beta^2}\right)= 2\frac{\alpha^2}{\beta^2}\left(\frac{e^{2\beta\Delta}}{2\beta}-\Delta e^{\beta \Delta}-\frac{1}{2\beta}\right)\\
& = 2\frac{e^{2\beta\Delta}-2\beta\Delta e^{\beta \Delta}-1}{2\beta(e^{\beta\Delta}-1)^2} \simeq\frac{1}{\beta}.
\end{split}
\end{equation}
In the last line, we assumed that $\beta\Delta=\textstyle\frac{\Delta}{k_B\mathbf{T}}\gg 1$.  This is reasonable for macroscopic systems since $k_B \bf{T}$ is roughly the energy of a single particle, whereas $\Delta$ is the spread of energies the system can occupy.  So the average size of an energy gap is $1/\beta=k_B\mathbf{T}$.  Next, we can find the second moments.
\begin{equation}
 \begin{split}
  \av{|E_i-E_j|^2} & =\sum_{ij}p_ip_j(E_i-E_j)^2\\
 &=2\sum_{i}p_iE_i^2-2\left(\sum_ip_iE_i\right)^2\\
 &=2\sigma_E^2,
 \end{split}
\end{equation}
where $\sigma_E$ is the standard deviation of the energy (in the exponential distribution).  So we have $\av{|E_i-E_j|^2}=2\sigma_E^2$.  Now it is not hard to show that the width of the distribution of energy levels is $\sigma_E\simeq 1/\beta$.  This allows us to calculate the width of the gap distribution $\sigma_{\textrm{gap}}$, which turns out to be
\begin{equation}
 \sigma_{\rm{gap}}= \sqrt{\av{|E_i-E_j|^2}-\av{|E_i-E_j|}^2} \simeq \frac{1}{\beta}.
\end{equation}
By Chebyshev's inequality we know that most of the gaps are comparable to $1/\beta$.  Now, $k_B\bf{T}$ may seem small, but it is effectively constant.  In contrast to this, we would expect the {\it smallest} gap to be exponentially small in the number of particles.

\chapter{Equilibration}
\label{chap:Equilibration}
\section{Introduction}
Equilibration is ubiquitous in nature, yet the mechanism behind it is not fully understood.  In both quantum and classical dynamics, the fundamental equations of motion are time reversible,\footnote{Incidentally, equilibration in classical and quantum systems has been compared in \cite{RE13,MS14}.} and, given unlimited measurement capabilities, it is clear that we would never see equilibration occur in either case.  If we could resolve individual particles in a classical gas, we would be able to track density fluctuations and see that, at that level, the state of the gas is constantly changing.  In this part of the thesis we will look at equilibration with physical restrictions on the measurements we can do, which we discussed in the previous chapter.  This is also typically the case for studies of equilibration in classical systems, where it is necessary to coarse-grain observables in order to prove that equilibration occurs.

It is natural to wonder whether equilibration relies upon specific properties of physical systems:\ is translational invariance of Hamiltonians, for example, necessary for equilibration to occur?  As we will see, if we are not concerned about the timescale, then almost no assumptions are necessary to see that equilibration occurs.  But if we do care about the timescale, then the question is open.  It is not known in general whether fast equilibration follows naturally from some reasonable restrictions on the dynamics.  It is known to be true in some specific cases, such as free bosons in \cite{CDEO08,CE10}.

In this chapter, we will study equilibration and the corresponding timescales in a very general setting.  We will prove our results while assuming as little as possible about the dynamics.  A strength of such a general approach is that the formulas and bounds are applicable to all finite dimensional systems.  A weakness is that, as we will see, the upper bounds on the equilibration time are far too large for systems with physical spectra.  A sensible aim for future work is to add further assumptions to see if stronger results can be found.  Obtaining realistic estimates for equilibration times will certainly require further assumptions, as it is possible to construct example systems that equilibrate very slowly, something we will see in section \ref{sec:Slow Equilibration}.

As far as thermalization goes, proving that equilibration occurs is only half the battle.  For example, there is still the issue of independence of the equilibrium state of the subsystem on its initial state.  To see that this is necessary, suppose a small subsystem in contact with a large environment thermalizes, then information about the initial state of the subsystem must be lost.\footnote{Of course, the information is not really lost, as the dynamics are unitary.  Instead, information about the subsystem's initial state is disseminated throughout the environment.  For the example of a spin chain, we would expect information to spread out and travel along the chain.  There are recurrences, but typically these are incredibly rare.}  It is crucial, therefore, in understanding thermalization that we understand when the information about the subsystem's initial state is lost in the environment.

The breakdown of this chapter is as follows.  In section \ref{sec:Equilibration Times}, we look at equilibration over finite times and use this to upper bound the equilibration time.  In the process, we show equilibration of expectation values of observables over finite timescales.  Then we use this as a tool to prove that equilibration occurs in finite time with respect to two important examples of measurement sets:\ the set of all measurements on a subsystem (assumed to be small compared to the environment) and a finite set of measurements.  We also look at an example of a system that equilibrates very slowly, which implies that our upper bounds on the equilibration time cannot be made much better without incorporating further assumptions.  In section \ref{sec:The Equilibrium state}, we look at the nature of the equilibrium state itself.  In this case, we focus on the subsystem's equilibrium state and the question of its independence on its initial state.

\section{Equilibration Times}
\label{sec:Equilibration Times}
Our first goal is to prove that equilibration happens over finite intervals.  Afterwards, we will use this to calculate upper bounds on the equilibration time.

We defined equilibration in section \ref{sec:Equilibration} to account for restrictions on the set of measurements that we can do.  One special case that is particularly important for us is the restriction to all measurements on a subsystem.  The distinguishability with respect to these measurements is equivalent to the trace distance on the subsystem.  Equilibration in this sense over infinite timescales was proved in \cite{LPSW09}.  The other setting that we will look at considers measurements restricted to some finite set, which takes account of the fact that measurement precision is bounded.  That systems equilibrate over infinite timescales in this sense was proved in \cite{Short10}.

\begin{compactwrapfigure}{r}{0.52\textwidth}
\centering
\begin{minipage}[r]{0.50\columnwidth}%
\centering
    \resizebox{6.5cm}{!}{\includegraphics{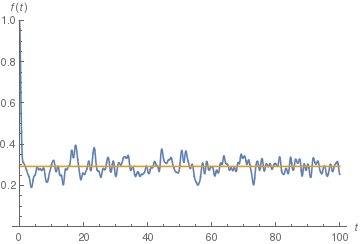}}
    \footnotesize{\caption[Equilibration on a spin chain]{Equilibration on a spin chain with seven sites evolving via the Heisenberg Hamiltonian with random couplings.  The observable is $P=\ket{\uparrow_x}\bra{\uparrow_x\!}\otimes \ket{\uparrow_x}\bra{\uparrow_x\!}\otimes \openone$.  The initial state is $P$ normalized to have trace one.  The blue curve is $\tr{P\rho(t)}$ and the yellow line is the time average over $[0,100]$.}}
    \label{fig:009}
\end{minipage}
\end{compactwrapfigure}

In \cite{Short10} it was also proved that equilibration over infinite timescales in both senses follows from the equilibration of expectation values, which is the form of equilibration analysed in \cite{Reimann08}.  Equilibration of expectation values over infinite intervals can be seen from the following formula \cite{Reimann08}.  Given any operator $A$ and the assumption of non-degenerate energy gaps,
\begin{equation}
\label{eq:Rei}
\frac{\left<|\textrm{tr}[\rho(t)A]-\textrm{tr}[\omega A]|^2\right>_{\infty}}{\|A\|^2}\leq \frac{1}{d_{\eff}}.
\end{equation}
As we saw in section \ref{sec:Macroscopic Systems}, we would expect $d_{\eff}$ to be exponentially large in the system size. Therefore, for most of the time in $[0,\infty)$, the expectation value $\textrm{tr}[\rho(t)A]$ will be close to its time average $\textrm{tr}[\omega A]$ relative to the scale set by $\|A\|$.

By extending this bound to finite intervals, we will be able to prove equilibration over finite intervals with respect to restricted measurement sets and measurements on subsystems.  This will in turn lead us to a bound on the equilibration time, which is our main goal.

\subsection{Finite Time Equilibration of Expectation Values}
\label{sec:Timescales for Equilibration of Expectation Values}
Our aim in this section is to generalize equation (\ref{eq:Rei}) for equilibration of expectation values to arbitrary time intervals, which will eventually allow us to bound the equilibration time.  In the process, we will be able to replace the assumption of nondegenerate energy gaps by the requirement that the degeneracy of the most degenerate gap is not too big.  The main content of this section is the following theorem \cite{SF12,MGLFS14}.  This result is essentially the same as the bound in \cite{MGLFS14}, which improved on the bound in \cite{SF12}, but with slightly better constants.

\begin{theorem}
\label{th:exptval}
Take a system evolving via a time independent Hamiltonian in the state $\rho(t)$ with $d_E$ energy levels.  For any
operator $A$, any $\varepsilon >0$ and time $T >0 $
\begin{equation}
\label{eq:eqofexpct}
\frac{\av{ |\tr{\rho(t)A} - \tr{ \omega A }|^2 }_T}{\norm{A}^2}
\leq  \frac{N(\varepsilon)}{d_{\eff}}\frac{5\pi}{2}\bigg[\frac{3}{4} + \frac{1}{\varepsilon T} \bigg].
\end{equation}
Recall that $N(\varepsilon)$ was defined in equation (\ref{eq:Nlevel}).
\end{theorem}

\begin{proof}
We start by assuming the system is in a pure state $\ket{\psi(t)}$, and at the end we will extend the result to mixed states by purification.

To deal with any degenerate energy levels, we choose an eigenbasis of $H$ so that $\ket{\psi(t)}$ only overlaps one eigenstate $\ket{n}$ for each distinct energy.  Then the state at time $t$ is
\begin{equation}
\ket{\psi(t)} = \sum_{n} c_n e^{-i E_n t} \ket{n},
\end{equation}
where $c_n = \braket{n}{\! \psi(0)}$.  Because $\ket{\psi(t)}$ lives in the subspace spanned by the states $\ket{n}$, its evolution is equivalent to evolution via the non-degenerate Hamiltonian $H'=\sum_n E_n \proj{n}$.  Therefore, the equilibrium state is $\omega = \sum_n |c_n|^2 \proj{n}$ and the effective dimension is given by $d_{\eff} = \frac{1}{\sum_n |c_n|^4} = \frac{1}{\tr{\omega^2}}$.

Writing the matrix elements of $A$ as $A_{ij}= \bra{i} A \ket{j}$,
\begin{align}
\label{eq:eqexpt1}
\langle |&\tr{\rho(t)A}  -\tr{\omega A}|^2 \rangle_T  = \Big\langle \Big|\sum_{i \neq j}  (c_j^*A_{ji} c_i) e^{-i (E_i-E_j)t} \Big|^2\Big\rangle_T \nonumber \\
 & \leq  \sum_{\scriptsize \begin{array}{c} i \neq j \\ k\neq l \end{array}} \!\!\! \! (c_j^* A_{ji} c_i)  (c_l^* A_{lk} c_k)^* \av{e^{i[(E_k-E_l) - (E_i - E_j)]t} }_T.
\end{align}
It will make the following formulas more succinct if we rewrite this expression in terms of energy gaps $G_{\beta}=E_i-E_j$, with $\beta = (i,j)$ and $\alpha = (k,l)$.  We also define the vector
\begin{equation}
v_{\beta}=v_{(i,j)} =  c_j^* A_{ji} c_i
\end{equation}
and the Hermitian matrix
 \begin{equation}
 M_{\alpha\beta} = \av{ e^{i(G_{\alpha} - G_{\beta})t} }_T.
\end{equation}
Equation (\ref{eq:eqexpt1}) becomes
\begin{align}
\av{ |\tr{\rho(t)A} - \tr{\omega A}|^2 }_T &=  \sum_{\alpha, \beta} v_{\alpha}^* M_{\alpha\beta}  v_{\beta}.
\end{align}
Then we have
\begin{equation}
 \begin{split}
 \av{ |\tr{\rho(t)A} - \tr{\omega A}|^2 }_T & \leq  \norm{M} \sum_{\alpha}  |v_{\alpha}|^2   \\
 & =\norm{M} \sum_{i\neq j}  |c_i|^2 |c_j|^2 |A_{ji}|^2  \\
 & \leq \norm{M} \sum_{i, j}  |c_i|^2 |c_j|^2 |A_{ji}|^2 \\
 & = \norm{M} \tr{A \omega A^{\dag} \omega }.
 \end{split}
 \end{equation}
Now we use the Cauchy-Schwartz inequality for operators with the inner product $\tr{A^{\dag} B}$, as well as the inequality for positive operators $P, Q$, $\tr{PQ} \leq \|P\| \tr{Q} $.  Then we get
\begin{equation}
 \begin{split}
 \av{ |\tr{\rho(t)A} - \tr{\omega A}|^2 }_T & \leq  \norm{M}\sqrt{\tr{A^{\dag}\!A\, \omega^2} \tr {A A^{\dag} \omega^2}} \\
 &\leq \norm{M}\norm{A}^2 \tr{\omega^2} \\
&=\frac{\norm{M}\norm{A}^2}{d_{\eff}}.
 \end{split}
\end{equation}
In the special case where the Hamiltonian has no degenerate energy gaps, we can take the infinite-time limit $T \rightarrow \infty$ so that $M$ becomes the identity matrix, and hence $\norm{M}=1$. Then we recover the previous bound given by equation (\ref{eq:Rei}).

Our interest is in the general case, with $T$ finite, and we will not assume that the energy gaps are not degenerate. Because $M$ is Hermitian, we can use the matrix norm bound
\begin{equation}
\label{eqn:norm}
\norm{M} \leq \max_{\beta} \sum_{\alpha} |  M_{\alpha\beta}|.
\end{equation}
This follows from $\norm{M}^2 \leq \vertiii{M}_1 \vertiii{M}_{\infty}$, where $\vertiii{M}_{\infty}$ and $\vertiii{M}_{1}$ are the row and column matrix norms \cite{Horn85}.  Since $M$ is hermitian, these are the same so that
\begin{equation}
  \norm{M}\leq \sqrt{\vertiii{M}_{\infty}\vertiii{M}_{1}} =\max_{\beta} \sum_{\alpha} |  M_{\alpha\beta}|.
\end{equation}
For more details, see \cite{Horn85}.

Instead of averaging over the interval $[0,T]$, we can upper bound equation (\ref{eq:eqexpt1}) by using a weighted average, as we saw in section \ref{sec:Time Averages and Energy Filtering}.  If we choose the weight to be the Lorentzian given in equation (\ref{eq:Lorentzian}) in section \ref{sec:Time Averages and Energy Filtering}, then the matrix elements of M are
\begin{equation}
\label{eqM}
M_{\alpha\beta} = \frac{5}{4}\int_{\mathbb{R}} \frac{T}{T^2+(t-\f{T}{2})^2}\, e^{i(G_{\alpha} - G_{\beta})t}  \textrm{d}t.
\end{equation}
Now we just need to use the formula for the Fourier transform of a Lorentzian, given by
\begin{equation}
 \int_{\mathbb{R}} \frac{T}{T^2+t^2}\, e^{ipt}  \textrm{d}t = \pi e^{-|p|T}.
\end{equation}
So we get
\begin{equation}
\label{eqM1}
M_{\alpha\beta}= \frac{5\pi}{4} e^{-|G_{\alpha} - G_{\beta}|T}e^{-i(G_{\alpha} - G_{\beta})T/2}.
\end{equation}

We want to bound the sum in equation (\ref{eqn:norm}).  To do this, we break it up into intervals of width $\varepsilon$, centered on a given energy gap $G_\beta$. There can be at most $N(\varepsilon)$ gaps $G_{\alpha}$ that satisfy $(k+\frac{1}{2}) \varepsilon >  G_{\alpha} - G_{\beta} \geq (k-\frac{1}{2}) \varepsilon$ for each $k$.  In the $k=0$ interval, we take $|M_{\alpha
\beta}|\leq 1$.  When $k$ is not zero, $|G_{\alpha} - G_{\beta}| \geq (|k|-\frac{1}{2}) \varepsilon$, and so from equation (\ref{eqM1}),
\begin{equation}
|M_{\alpha \beta}|\leq \frac{5\pi}{4} e^{-(|k|-1/2)\varepsilon T}.
\end{equation}

The sum $\sum_{\alpha} |M_{\alpha\beta}|$ contains $d_E(d_E-1)$ terms, and it is maximized by having as many terms with
small values of $|k|$ as possible.  Then we get
\begin{equation}
\sum_{\alpha} |  M_{\alpha\beta}| \leq \frac{5\pi}{4}N(\varepsilon)\left(1  +  2e^{\varepsilon/2 T}\! \sum_{k=1}^{\frac{d_E(d_E-1)}{2}} e^{-k\varepsilon T} \right).
\end{equation}
The first term comes from the $k=0$ interval, and the second term comes from the intervals with non zero $k$.  The next step is to bound the sum by the same sum but with $k$ running from one to infinity, which leads to a geometric series.  Then
\begin{equation}
\begin{split}
\sum_{\alpha} |  M_{\alpha\beta}| & \leq \frac{5\pi}{4}N(\varepsilon)\left(1  +  \frac{2e^{\varepsilon/2 T}}{e^{\varepsilon T}-1} \right)\\
& = \frac{5\pi}{4}N(\varepsilon)\left(1  +  \frac{2e^{-\varepsilon/2 T}}{1-e^{-\varepsilon T}} \right).
\end{split}
\end{equation}
To simplify this, we use
\begin{equation}
\begin{split}
 \frac{2e^{-x/2}}{1-e^{-x}} & =\frac{1}{1-e^{-x/2}}-\frac{1}{1+e^{-x/2}}\\
 & \leq \frac{1}{1-e^{-x/2}} -\frac{1}{2} \leq \frac{1}{2}+\frac{2}{x},
 \end{split}
\end{equation}
where we used $\f{1}{1-e^{-z}}\leq 1+\f{1}{z}$ in the last step.  This leads to the bound
\begin{equation}
\sum_{\alpha} |  M_{\alpha\beta}|  \leq \frac{5\pi}{2}N(\varepsilon)\left[\frac{3}{4}  +  \frac{1}{\varepsilon T} \right].
\end{equation}
Putting this all together, we get that
\begin{equation}
\label{eq:alpha}
 \frac{\langle|\tr{\rho(t)A}- \tr{\omega A}|^2\rangle_T}{\|A\|^2}\leq \frac{5\pi}{2}\frac{N(\varepsilon)}{d_{\eff}} \left[\frac{3}{4}  +  \frac{1}{\varepsilon T} \right].
\end{equation}
This is an improvement over the bound in \cite{SF12}, which was
\begin{equation}
\frac{\langle|\tr{\rho(t)A}- \tr{\omega A}|^2\rangle_T}{\norm{A}^2} \leq \frac{N(\varepsilon)}{d_{\eff}} \bigg[1 + \frac{8 \log_2 d_E}{\varepsilon T} \bigg].
\end{equation}
So the improvement is $O(\log_2 d_E)$.  Actually, equation (\ref{eq:alpha}) is also a slight improvement over the bound in \cite{MGLFS14}, in which the factor on the right was $(\f{3}{2}+\f{1}{\varepsilon T})$.

So much for pure states. The final step is to extend the result to mixed states by doing a purification, as in \cite{Short10}.  For any initial state $\rho(0)$ on $\mH$, we define a pure state $\ket{\psi(0)}$ on $\mH \otimes \mH$ such that the reduced state on the first system is $\rho(0)$.  We recover the original evolution $\rho(t)$ of the first system by evolving $\ket{\psi(t)}$ under the joint Hamiltonian $H \otimes I$.  The expectation value of an operator $A$ in the state $\rho(t)$ is the same as the expectation value of
$A \otimes I$ on the joint system.  Furthermore, $\|A\|=\|A\otimes I\|$.  We also have that, crucially, $N(\varepsilon)$ is the same for $H\otimes I$ as it is for $H$.  Also, the effective dimension of the purified system is the same as the effective dimension of the original system, which follows from $\textrm{tr}[P_E\rho(0)]=\textrm{tr}[P_E\otimes I \proj{\psi(0)}]$.
\end{proof}

\subsection{Finite Time Equilibration for Systems and Subsystems}
\label{sec:Finite Time Equilibration for Systems and Subsystems}
\begin{compactwrapfigure}{r}{0.4\textwidth}
\centering
\begin{minipage}[r]{0.39\columnwidth}%
\centering
    \resizebox{5.5cm}{!}{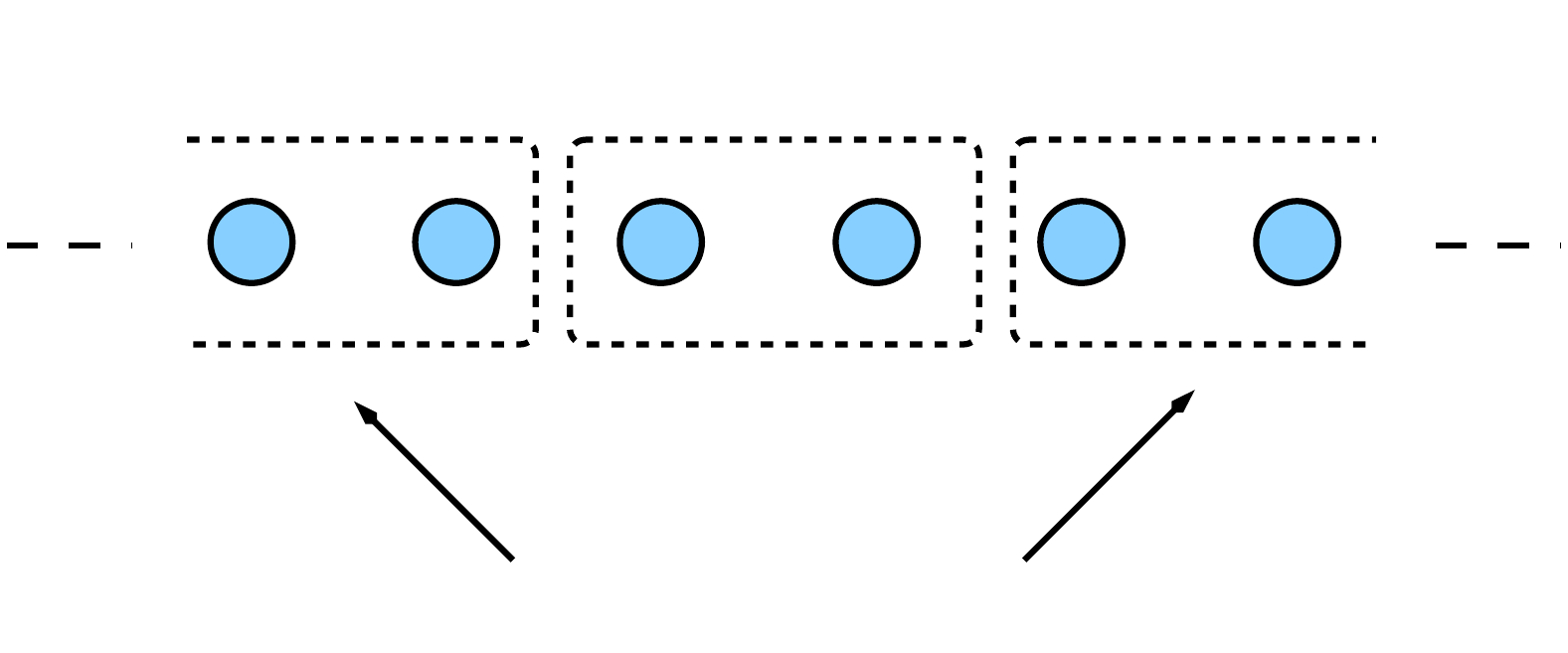}
    \footnotesize{\caption[Subsystem of a spin chain]{A subsystem of a spin chain consisting of the two sites in the middle.  The rest of the chain forms the environment.\ \\ \ }}
    \label{fig:Spin_Chain_Subsystem}
\end{minipage}
\end{compactwrapfigure}

We would like to use the result in the previous section, which proves equilibration of expectation values, to say something about equilibration with respect to a set of measurements.  The two cases that will concern us are when the measurement set consists of all measurements on a subsystem or a finite set of possible measurements on the whole system.

First, following \cite{Short10}, we will use the equilibration of expectation values to prove equilibration of a closed system, when there are realistic constraints on the set of possible measurements $\cM$.  The infinite time result from \cite{Short10}, with the assumption of non degenerate energy gaps, is
\begin{equation}
\label{eq:E1}
\left<D_{\mathcal{M}}(\rho(t),\omega)\right>_{\infty}\leq \frac{\mathcal{S}(\mathcal{M})}{4\sqrt{d_{\eff}}},
\end{equation}
where $\mathcal{S}(\mathcal{M})$ denotes the number of {\it outcomes} of all the measurements that we can do.  At first glance, this does not look useful, as $\mathcal{S}(\mathcal{M})$ is extremely large.  Typically, however, we expect it to be insignificant compared to $\sqrt{d_{\eff}}$ \cite{Short10}.  Remember that we expect $d_{\eff}$ to be exponentially large with the system size, whereas $\mathcal{S}(\mathcal{M})$ is determined by the measurements we can do in practice.  For macroscopic systems, $d_{\eff}$ is exponential in $O(10^{23})$.  On the other hand, measurement settings and results have to be recorded; even with a computer that has a petabyte of storage space on the hard disk, we only have about $10^{16}$ bits of data.\footnote{For comparison, the large hadron collider in CERN generates about $30$ petabytes of data each year \cite{CERN14}.}  At best this would allow us to record $10^{16}$ possible measurement outcomes.

Therefore, with the assumption of non-degenerate energy gaps, we see equilibration for realistic measurements on large systems over the interval $[0,\infty)$.

We can extend this to finite time intervals and Hamiltonians that may have degenerate energy gaps, provided there are not too many.  To prove this, we need to incorporate the results of theorem \ref{th:exptval}, as done in \cite{SF12,MGLFS14}.

\begin{theorem}
\label{th:syseqmeas}
Take a system in the state state $\rho(t)$ that evolves via a Hamiltonian with $d_E$ distinct energies.  For any $\varepsilon >0$ and time $T >0 $, the average distinguishability of $\rho(t)$ from $\omega$ over the interval $[0,T]$ using measurements in the set $\cM$ is bounded by
\begin{equation}
\label{eq:C1}
\left<D_{\mathcal{M}}(\rho(t),\omega)\right>_{T}\leq \frac{\mathcal{S}(\mathcal{M})}{4\sqrt{d_{\eff}}}\sqrt{\frac{5\pi}{2}N(\varepsilon) \left[\frac{3}{4}  +  \frac{1}{\varepsilon T} \right]},
\end{equation}
where $\mathcal{S}(\mathcal{M})$ is the number of measurement outcomes in $\cM$.
\end{theorem}
\begin{proof}
This proof follows similar steps to a proof in \cite{Short10}, but incorporates the improved bound from theorem \ref{th:exptval}.

First, 
\begin{equation}
\begin{split}
\av{D_\cM (\rho(t), \omega)}_T &=  \av{ \max_{M(t) \in \cM} D_{M(t)} (\rho(t), \omega)}_T\\
&\leq  \sum_{M \in \cM} \av{ D_M (\rho(t), \omega)}_T\\
&=\frac{1}{2}  \sum_{M \in \cM}  \sum_a \av{ | \tr{M_a \rho(t) }  - \tr{M_a \omega} | }_T\\
&\leq \!\frac{1}{2} \! \sum_{M \in \cM}\! \sum_a\! \sqrt{ \av{ \big( \tr{M_a \rho(t) }  \!-\! \tr{M_a \omega} \big)^2 } }_T.
\end{split}
\end{equation}
Next, define $\tilde{M}_a= M_a - \frac{1}{2}I$ for all POVM elements $M_a$ so that $\|\tilde{M}_a\| \leq \frac{1}{2}$.  Then
\begin{equation}
 \begin{split}
\av{D_\cM (\rho(t), \omega)}_T &\leq \frac{1}{2}\!  \sum_{M \in \cM}\! \sum_a\! \sqrt{ \av{ \big( \tr{\tilde{M}_a \rho(t)}  \!-\! \tr{\tilde{M}_a \omega} \big)^2 } }_T\\
&\leq  \frac{1}{2}\sum_{M \in \cM} \sum_a \sqrt{\frac{5\pi}{2}\frac{ \|
\tilde{M}_a \|^2 N(\varepsilon)}{d_{\eff}}\left[\frac{3}{4}  +  \frac{1}{\varepsilon T}\right]}\\
&\leq  \frac{{\cal S}(\mM)}{4\sqrt{d_{\eff}}}\sqrt{\frac{5\pi N(\varepsilon)}{2}\left[\frac{3}{4}+\frac{1}{\varepsilon T}\right]},
\end{split}
\end{equation}
where ${\cal S}(\mM)$ is the sum of the number of outcomes for all measurements in $\mM$.
\end{proof}

\begin{compactwrapfigure}{r}{0.49\textwidth}
\centering
\begin{minipage}[r]{0.47\columnwidth}%
\centering
    \resizebox{6.5cm}{!}{\includegraphics{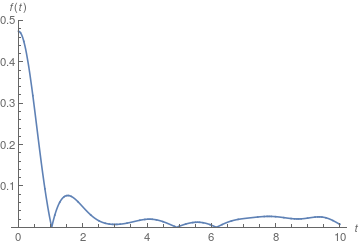}}
    \footnotesize{\caption[Subsystem equilibration on a spin chain]{Equilibration on a spin chain with seven sites evolving via the Heisenberg Hamiltonian with random couplings.  The initial state is $\ket{\uparrow_x}\bra{\uparrow_x\!}\otimes \ket{\uparrow_x}\bra{\uparrow_x\!}\otimes \openone$ normalized to have trace one.  The curve shows the trace distance between the state of the first spin at time $t$ and its time average.}}
    \label{fig:0090}
\end{minipage}
\end{compactwrapfigure}

Similarly, we can employ the bound in theorem \ref{th:exptval} to derive a finite time analogue of the main theorem in \cite{LPSW09}, which shows that over infinite timescales subsystems equilibrate, provided they are small compared to the environment.  The main analysis in \cite{LPSW09} culminates in the formula
\begin{equation}
\label{eq:LPSW}
\left<D\left(\rho_{S}(t),\omega_S\right)\right>_{\infty}\leq\frac{1}{2}\sqrt{\frac{d_{S}^2}{d_{\eff}}},
\end{equation}
where $d_S$ denotes the subsystem dimension, $\rho_S(t)$ is the subsystem state at time $t$ and $\omega_S$ is its average over $[0,\infty)$.  This bound tells us that subsystems equilibrate with respect to all measurements over the interval $[0,\infty)$.  Notice that there is no assumption on the number of measurements we can do; the result is that, allowing any measurement on the subsystem, its state is indistinguishable from $\omega_S$ for most of $[0,\infty)$.  For this bound to be useful, we need the subsystem dimension to be small compared to $\sqrt{d_{\eff}}$, which requires that the dimension of the environment's state space $d_B$ be much larger than the subsystem's.  This is necessary because $d_{\eff}\leq d_S d_B$.  The assumption of nondegenerate energy gaps is also necessary.

Now let us extend this theorem to finite times and replace the assumption of nondegenerate energy gaps by the requirement that the degeneracy of the most degenerate gap is not too big.  This is based on a theorem in \cite{SF12}.
\begin{theorem}
\label{th:subeq}
Take a system evolving via a Hamiltonian that has $d_E$ distinct energies.  For any $\varepsilon >0$ and time $T >0 $,
the trace distance between the state of the subsystem at time $t$, given by $\rho_S(t)$, and its infinite-time average $\omega_S = \av{\rho_S(t)}_{\infty}$ averaged over the interval $[0,T]$ is bounded above by
\begin{equation} \label{eq:C2}
\left<D\left(\rho_S(t),\omega_S\right)\right>_{T}
\leq \frac{1}{2}\sqrt{\frac{d_{S}^2}{d_{\eff}}\frac{5\pi N(\varepsilon)}{2}\left[\frac{3}{4}+\frac{1}{\varepsilon T}\right]},
\end{equation}
where $d_S$ is the dimension of the subsystem.
\end{theorem}
\begin{proof}
Again, this proof uses the basic argument from a proof in \cite{Short10}, but takes into account the improved results of theorem \ref{th:exptval}.

There is a particularly useful orthonormal basis\footnote{By saying that these are orthonormal, we mean with respect to the inner product $(A,B)=\tr{A^{\dagger}B}$.} of operators on the subsystem's Hilbert space from \cite{Schwinger60}.  These are the $d_S^2$ operators
\begin{equation}
F_{(k,l)} =\frac{1}{\sqrt{d_S}} \sum_{n} e^{\frac{2 \pi i n k}{d_S}} \ket{(n+l)\, \textrm{mod}\, d_S} \bra{n},
\end{equation}
where $l,k \in \{0,1,\ldots d_S-1 \}$, and the states $\ket{n}$ are an arbitrary orthonormal basis on the subsystem. 

We can expand the subsystem states in this operator basis so that $(\rho_S(t) - \omega_S) = \sum_{k,l} c_{(k,l)}(t) F_{(k,l)}$.  Now recall that the trace distance is $D(\rho, \sigma) = \f{1}{2}\textrm{tr} | \rho - \sigma |$.  Then
\begin{eqnarray}
\av{D(\rho_S(t), \omega_S)}_T &=& \frac{1}{2} \av{\textrm{tr} \big|\sum_{k,l} c_{(k,l)}(t) F_{(k,l)} \big|}_T  \nonumber \\
&\leq& \frac{1}{2} \av{ \sqrt{ d_S\,  \textrm{tr} \big|\sum_{k,l} c_{(k,l)}(t) F_{(k,l)} \big|^2}}_T  \nonumber \\
&\leq& \frac{1}{2}  \sqrt{ d_S \sum_{k_i,l_i}  \av{c_{(k_2,l_2)}(t) c^*_{(k_1,l_1)}(t) }_T \tr{F_{(k_1,l_1)}^{\dag} F_{(k_2,l_2)} }}  \nonumber \\
&=& \frac{1}{2} \sqrt{ d_S \sum_{k,l} \av{|c_{(k,l)}(t)|^2}_T} \nonumber \\
&=& \frac{1}{2} \sqrt{ d_S \sum_{k,l} \av{|\tr{(\rho_S(t) - \omega_S) F^{\dag}_{(k,l)}}|^2}_T} \nonumber \\
\end{eqnarray}
The second line followed by using the inequality for matrix norms $\|B\|_1\leq \sqrt{N}\|B\|_2$, where $B$ is an $N\times N$ matrix, $\|B\|_1=\textrm{tr}|B|$ and $\|B\|_2=\sqrt{\tr{B^{\dagger}B}}$ \cite{Horn85}.  By applying theorem \ref{th:exptval} we get
\begin{eqnarray}
\av{D(\rho_S(t), \omega_S)}_T &\leq& \frac{1}{2} \sqrt{ d_S \sum_k \|F^{\dag}_k \|^2\frac{5\pi N(\varepsilon)}{2 d_{\eff}}\left[\frac{3}{4}+\frac{1}{\varepsilon T}\right]} \nonumber \\
&\leq& \frac{1}{2} \sqrt{\frac{d_S^2}{d_{\eff}}\frac{5\pi N(\varepsilon)}{2}\left[\frac{3}{4}+\frac{1}{\varepsilon T}\right]}.
\end{eqnarray}
\end{proof}

We can obtain a comparison to the previous infinite time results of equations (\ref{eq:E1}) and (\ref{eq:LPSW}) by choosing $\varepsilon$ equal to the minimum spacing between (non degenerate) energy gaps.  This means
\begin{equation}
\varepsilon = \varepsilon_{\min} = \min_{\alpha, \beta} \{ |G_{\alpha} - G_{\beta}| : G_{\alpha} \neq  G_{\beta} \}.
\end{equation}
Setting $\varepsilon= \varepsilon_{\min}$, it follows that $N(\varepsilon_{\min})$ is equal to the degeneracy of the most degenerate energy gap $D_G$.  Substituting this into equation (\ref{eq:C2}),
\begin{equation}
\left<D\left(\rho_S(t),\omega_S\right)\right>_{T} \leq \frac{1}{2}\sqrt{\frac{d_{S}^2}{d_{\eff}}\frac{5\pi D_G}{2}\left[\frac{3}{4}+\frac{1}{\varepsilon_{\min} T}\right]}.
\end{equation}
Taking the limit as $T \rightarrow \infty$, we recover a version of equation (\ref{eq:LPSW}) that applies to any Hamiltonian.
\begin{equation}
\left<D\left(\rho_S(t),\omega_S\right)\right>_{\infty} \leq \frac{1}{2}\sqrt{\frac{d_{S}^2}{d_{\eff}}\frac{5\pi D_G}{2}}.
\end{equation}
This is more general than equation (\ref{eq:LPSW}), in that it allows some degenerate energy gaps, provided $D_G$ is not too big.  Still, if the non degenerate gaps condition is satisfied, then $D_G=1$, and we get something weaker by a constant factor than the previous result in equation (\ref{eq:LPSW}).  Had we used a straightforward time average in the proof of theorem \ref{th:exptval}, instead of a weighted average, we would have recovered (\ref{eq:LPSW}) exactly as $T\rightarrow \infty$.  The disadvantage of this is that, for finite $T$, the time dependent term would be much larger (by a factor of $\log_2 d_E$).

So the end result is a bound showing that subsystems equilibrate over finite intervals even with degenerate energy gaps, provided there are not too many.  We know that $D_G < d_E$, with the maximal value $D_G=d_E-1$ occurring when consecutive energy levels are all equally spaced, which would be the case for a harmonic oscillator with an energy cut-off.  But it is more than reasonable to expect that almost any nontrivial Hamiltonian will have a much smaller value of $D_G$.

\subsection{Equilibration Times with Physical Spectra}
\label{sec:Equilibration Times with Physical Spectra}
The next step is to look at equilibration times.  Equilibration has occurred when
\begin{equation}
\label{eq:eqtime32}
 \frac{N(\varepsilon)}{d_{\eff}}\frac{5\pi}{2}\bigg[\frac{3}{4} + \frac{1}{\varepsilon T} \bigg]\ll 1.
\end{equation}
For equilibration of subsystems, as in theorem \ref{th:subeq}, we need the product of this with $d_S^2$ to be small.  For equilibration of systems with a finite number of measurements, as in theorem \ref{th:syseqmeas}, we need the product of this with $\m{S}(\m{M})^2$ to be small.

Let us assume that both $\m{S}(\m{M})^2$ and $d_S^2$ are fixed, as our measurement capabilities will not grow with the system size for large systems.  Because these are fixed, we need not worry about how they scale with the system size.  Other quantities will scale, however.  We expect the effective dimension to grow exponentially with system size, which would make the bounds better, but we also expect $\varepsilon_{\min}$ to become exponentially smaller as the system size grows.  Additionally, for fixed $\varepsilon$, $N(\varepsilon)$ should grow exponentially with increasing system size.  To obtain the optimal bound, some balance must be found when choosing $\varepsilon$.

And to prove equilibration occurs at all, then the left hand term in equation (\ref{eq:eqtime32}),
\begin{equation}
\label{eq:inftime}
 \frac{N(\varepsilon)}{d_{\eff}}\frac{15\pi}{8},
\end{equation}
must be small.  Then the equilibration time goes like
\begin{equation}
 T\simeq \frac{1}{\varepsilon}.
\end{equation}
It is hard to be to precise here, but we have some freedom, because of how large $d_{\eff}$ is, to choose $\varepsilon$ bigger to make the timescale bound shorter, while keeping (\ref{eq:inftime}) small enough to ensure equilibration occurs.

Let us estimate $N(\varepsilon)$ for a system with an exponential distribution of energies.  In this case, the probability of having a gap of size $G$ is
\begin{equation}
\begin{split}
 P(G)=\int_G^{\Delta} \textrm{d}E\, p(E)p(E-G) & =\alpha^2 \int_G^{\Delta} \textrm{d}E\, e^{\beta E}e^{\beta (E-G)}\\
 & = \alpha^2\frac{e^{2\beta\Delta-\beta G}-e^{\beta G}}{2\beta}\\
 & \simeq \frac{\beta}{2}e^{-\beta G},
 \end{split}
\end{equation}
where the last step follows because $\alpha= \f{\beta}{e^{\beta\Delta}-1}$ and $e^{\beta\Delta}\gg 1$.  Then, because there are $d_E(d_E-1)$ gaps in total, $N(\varepsilon)$ is given by
\begin{equation}
\begin{split}
\label{eq:Nepa}
 N(\varepsilon) & \simeq d_E(d_E-1)\int_{-\varepsilon/2}^{\varepsilon/2}\textrm{d}G\, \frac{\beta e^{-\beta G}}{2}\\
 & \simeq d_E^2 \frac{e^{\beta \varepsilon/2}-e^{-\beta \varepsilon/2}}{2}\simeq d_E^2\frac{\beta\varepsilon}{2},
 \end{split}
\end{equation}
where the last step follows by assuming $\beta\varepsilon\ll 1$.  For equilibration to occur, we need (\ref{eq:inftime}) to be small, but the effective dimension can be at most $d_E$.  This means that we need to ensure $N(\varepsilon)$ is smaller than $d_E$ by choosing $\varepsilon$ to be sufficiently small.  But according to equation (\ref{eq:Nepa}), this is only possible if $\varepsilon\leq \f{1}{d_E}$.  This is disastrous because this is exponentially small in the system size.  The problem is that, even though most gaps are not exponentially small, there are still exponentially many of them that are.  Therefore, the equilibration time goes like $\f{1}{d_E}$ at best, which is extraordinarily long.

We should not be too pessimistic, however.  Obviously, the more general the hypotheses, the bigger the timescale bounds are going to be.  And, as far as scaling goes, these bounds are about as good as we can find without extra assumptions.  The example in the following section justifies this.

\subsection{Slow Equilibration}
\label{sec:Slow Equilibration}
In this section, we will reproduce an example of a system that equilibrates very slowly from \cite{MGLFS14}.  This demonstrates that quantum systems that equilibrate do not all do so quickly, so further restrictions will be needed to arrive at more realistic bounds on the equilibration time for physical systems.
\begin{theorem}
We can construct a pure state and a measurement consisting of two projectors $P$ and $\openone -P$ such that the system equilibrates extremely slowly.  Quantitatively, for any positive integer $K$ and $\epsilon>0$ and all times $t$ satisfying
 \begin{equation}
  t\in[0,\frac{2\epsilon K}{\sigma_E}],
 \end{equation}
the system is far from its infinite-time average state, meaning
 \begin{equation}
 D_{\m{M}}\left(\rho(t),\omega\right) \geq 1-\epsilon^2-\sqrt{\frac{K}{d_{\eff}}},
 \end{equation}
 where $\sigma_E$ is the standard deviation of the energy in the state $\rho(t)$ and $\m{M}$ denotes the measurement consisting of $P$ and $\openone-P$.  But the system still equilibrates over a long enough timescale, meaning
\begin{equation}
 \av{D_{\m{M}}\left(\rho(t),\omega\right)}_{\infty} \leq 2\sqrt{\frac{K}{d_{\eff}}}.
\end{equation}
\end{theorem}
\begin{proof}
Since the initial state is pure, we can choose the energy eigenbasis such that the state only has overlap with one eigenstate corresponding to a given energy.  The state can then be written as
 \begin{equation}
  \ket{\psi(t)}=\sum_n c_n e^{-iE_nt}\ket{n}.
 \end{equation}
Now, for some $\delta t$, which may be negative, we have
\begin{equation}
 \begin{split}
  |\langle\psi(t) |\psi(t+\delta t)\rangle|^2 & =\abs{\sum_{n}^{d_E}|c_n|^2e^{-iE_n\delta t}}^2\\
  & =\sum_{n,m}^{d_E}|c_n|^2|c_m|^2\cos([E_n-E_m]\delta t)\\
  & \geq 1-\frac{\delta t^2}{2}\sum_{n,m}^{d_E}|c_n|^2|c_m|^2(E_n-E_m)^2\\
  & = 1-\delta t^2\left(\av{E^2}-\av{E}^2\right)=1-\sigma_E^2\delta t^2.
 \end{split}
\end{equation}
To get the third line, we needed the inequality $\cos(x)\geq 1-\f{x^2}{2}$.  The result is that, for any $\epsilon$,
\begin{equation}
|\langle\psi(t) |\psi(t+\delta t)\rangle|^2\geq 1-\epsilon^2, 
\end{equation}
provided $|\delta t|\leq\f{\epsilon}{\sigma_E}=\tau$.  The trick now is to define the subspace
\begin{equation}
 \mathcal{H}_{T}=\textrm{span}\left\{\ket{\psi\big([2j+1]\tau\big)}|j\in\{0,..,K-1 \}\right\}.
\end{equation}
This is the span of $K$ snapshots of the state at different times.  Roughly, the idea is that we are tracking the state as it evolves over time.  For any time $t\in[0,2K\tau]$, the state $\ket{\psi(t)}$ is close to a state in this subspace.  So for each $j\in\{0,..,K-1 \}$, we can write the projector onto this subspace as
\begin{equation}
 P=\ket{\psi([2j+1]\tau)}\bra{\psi([2j+1]\tau)}+P^{\prime},
\end{equation}
where $P^{\prime}$ is the projector onto the orthogonal complement of $\ket{\psi([2j+1]\tau)}$ in $\mathcal{H}_{T}$.  Then, for $j$ such that $2j\tau\leq t\leq 2(j+1)\tau$,
\begin{equation}
\begin{split}
 \tr{P\rho(t)} & =\bra{\psi(t)}\left(\ket{\psi([2j+1]\tau)}\bra{\psi([2j+1]\tau)}+P^{\prime}\right)\ket{\psi(t)}\\
 & \geq \abs{\langle\psi(t)|\psi([2j+1]\tau)\rangle}^2\geq 1-\epsilon^2.
 \end{split}
\end{equation}
Then it follows that, with the measurement set $\m{M}=\{P,\openone-P\}$,
\begin{equation}
\begin{split}
 D_{\m{M}}\left(\rho(t),\omega\right) & =\abs{\tr{P(\rho(t)-\omega)}}\\
 & \geq \tr{P\rho(t)}-\tr{P\omega}\\
 & \geq \tr{P\rho(t)}-\sqrt{\tr{P}\tr{\omega^2}}\\
 & \geq 1-\epsilon^2-\sqrt{\frac{K}{d_{\eff}}}.
 \end{split}
\end{equation}
The third line follows from the Cauchy-Schwarz inequality for the Hilbert-Schmidt inner product.  The last line follows because $\tr{\omega^2}\leq\f{1}{d_{\eff}}$ and the rank of $P$ is at most $K$ since $\m{H}_T$ is the span of $K$ vectors.

Furthermore,
\begin{equation}
\begin{split}
 \av{D_{\m{M}}\left(\rho(t),\omega\right)}_{\infty} & =\av{\abs{\tr{P(\rho(t)-\omega)}}}_{\infty}\\
 & \leq \av{\tr{P\rho(t)}+\tr{P\omega}}_{\infty}\\
 & = 2\tr{P\omega}\leq 2\sqrt{\frac{K}{d_{\eff}}}.
 \end{split}
\end{equation}
\end{proof}
We can choose $K=\epsilon^4d_{\eff}$, with $\epsilon$ chosen to be small but such that $K$ is an integer.  This gives a more readable form of the theorem:\ for all times $t$ satisfying
 \begin{equation}
\label{eq:longti}
  t\in[0,\frac{c d_{\eff}}{\sigma_E}],
 \end{equation}
 where $c$ is a constant, the system is far from its infinite-time average state, meaning
 \begin{equation}
 D_{\m{M}}\left(\rho(t),\omega\right) \geq 1-2\epsilon^2,
 \end{equation}
but it still equilibrates over long times since
\begin{equation}
 \av{D_{\m{M}}\left(\rho(t),\omega\right)}_{\infty} \leq 2\epsilon^2.
\end{equation}

So, since the timescale in equation (\ref{eq:longti}) is dominated by $d_{\eff}$, equilibration to $\omega$ takes time exponential in the system size.  So general bounds on the equilibration time will have to take this into account.  Whether or not there is a local Hamiltonian and local measurement that can achieve timescales like this is not clear.

Another way to interpret this theorem is that there are systems that do not equilibrate to their infinite-time average states over physically reasonable timescales.  It is clearly in equilibrium over timescales that are small compared to $\f{c d_{\eff}}{\sigma_E}$ with respect to the measurement in the theorem.  We will revisit this type of temporary equilibration in chapter \ref{chap:Conclusions2}.

\section{The Equilibrium state}
\label{sec:The Equilibrium state}
Even if realistic estimates on the equilibration time are achieved, there is still the question of what the equilibrium state actually is.  Ideally, we would like to understand the conditions under which thermalization occurs, meaning the subsystem has equilibrated to a Gibbs state.  The Gibbs state is
\begin{equation}
 \frac{e^{-\beta H_S}}{Z},
\end{equation}
where $H_S$ is the subsystem Hamiltonian, $\beta$ is the inverse temperature and $Z=\tr{e^{-\beta H_S}}$ normalizes the state.  The importance of the Gibbs state in statistical physics cannot be understated.

The standard setting for studying thermalization consists of a subsystem and environment that have been brought into contact at $t=0$.  So the overall initial state is a product state.  These interact over time and the subsystem is presumed to equilibrate.  There has been some success towards rigorously proving that the equilibrium state in this scenario is indeed the Gibbs state.  In \cite{RGE11}, this was proved, provided the interaction strength satisfied
\begin{equation}
\label{eq:py}
 \|V\|\ll 1/\beta=k_B\bf{T},
\end{equation}
and some other conditions on the initial state were met.\footnote{Incidentally, the proof uses theorem \ref{th:Bhatia} from chapter \ref{chap:Quantum Cellular Automata and Field Theory} to analyze the effects of weak interactions on projectors onto energy subspaces.}  We expect $k_B\bf{T}$ to be approximately the energy per particle.

The condition in equation (\ref{eq:py}) only holds when the interaction strength does not grow with the size of the boundary between the subsystem and the environment, which for local Hamiltonians, is where the interaction term lives.  This suggests that the result of \cite{RGE11} is most useful for one dimensional systems, as in this case the boundary between the subsystem and environment is constant size.  Then it is reasonable for the interaction strength to be bounded above by something that does not scale with the size of the boundary.

There are two important necessary requirements for the Gibbs state to be reached.  Both of these involve some form of initial state independence of the equilibrium state.  The first, which will not concern us here but was studied in \cite{LPSW09}, is independence of the subsystem's equilibrium state on the precise details of the environment's initial state.  The idea is that only coarse grained variables like the temperature of the environment should play a role.  The second requirement, which will concern us here, is independence of the equilibrium state of the subsystem on its initial state.  For the rest of this section, we will just refer to this as initial state independence.

A concept that appears often in discussions of thermalization is integrability.  This is a property of the system's Hamiltonian, though there is no consensus on a precise definition.  There is a clear discussion of the various definitions in \cite{GME11}.  The rough idea, however, is that the model is simple in some sense:\ one definition has that there are locally conserved quantities (arising from local observables that commute with the Hamiltonian).  Another definition sometimes used is that the model is solvable in some sense.  For our purposes, however, it will be enough that all non interacting models are certainly integrable.

The main result of \cite{GME11} is that models that are far from being integrable may not have initial state independence.  In a similar vein, it is interesting to ask whether models close to integrable necessarily always retain some memory of their initial conditions.  Here we will see a simple example where an arbitrarily small interaction removes all memory effects of the subsystem's initial state.

Suppose we have a qubit with energies $0$ and $\nu$, corresponding to the states $\ket{0}$ and $\ket{\nu}$ respectively.  Suppose also that the environment is a harmonic oscillator with energies $E_n=n\nu$.  The non-interacting Hamiltonian is
\begin{equation}
\begin{split}
H_0  =\nu\ket{\nu}\bra{\nu}\otimes\openone_B+\openone_S\otimes\sum_{n=0}^{\infty}n\nu\ket{n}\bra{n}.
\end{split}
\end{equation}
We can rewrite this in a useful way:\ 
\begin{equation}
\begin{split}
 H_0 &  =\ket{0}\bra{0}\otimes\sum_{n=0}^{\infty}n\nu\ket{n}\bra{n}+\ket{\nu}\bra{\nu}\otimes\sum_{n=0}^{\infty}(n+1)\nu\ket{n}\bra{n}\\
 & =\sum_{n=0}^{\infty}(n+1)\nu\bigg(\ket{0}\bra{0}\otimes\ket{n+1}\bra{n+1}+\ket{\nu}\bra{\nu}\otimes\ket{n}\bra{n}\bigg).
\end{split}
\end{equation}
Consider the interaction
\begin{equation}
 V=\lambda\sum_{n=0}^{\infty}\bigg(\ket{0}\bra{\nu}\otimes\ket{n+1}\bra{n}+\ket{\nu}\bra{0}\otimes\ket{n}\bra{n+1}\bigg).
\end{equation}

\begin{compactwrapfigure}{r}{0.33\textwidth}
\centering
\begin{minipage}[r]{0.32\columnwidth}%
\centering
    \resizebox{4cm}{!}{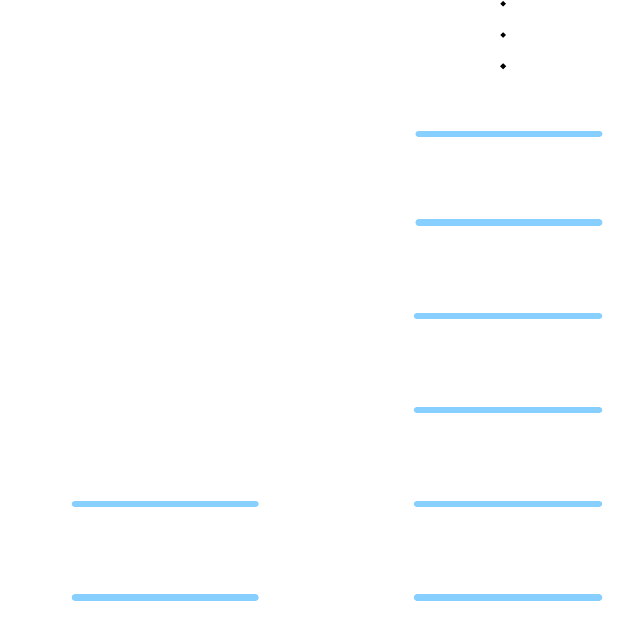}
    \footnotesize{\caption[Energy levels of a qubit and simple harmonic oscillator]{Energy levels of a qubit on the left and a simple harmonic oscillator on the right.}}
    \label{fig:Qubit_Oscillator}
\end{minipage}
\end{compactwrapfigure}

So the new Hamiltonian is block diagonal, with blocks like
\begin{equation}
 \begin{pmatrix}
            (n+1)\nu & \lambda\\
            \lambda & (n+1)\nu\\
           \end{pmatrix},
\end{equation}
when written in the $\ket{0}\ket{n+1}$ and $\ket{\nu}\ket{n}$ basis.  Each term in the sum in $V$ only couples $\ket{0}\ket{n+1}$ and $\ket{\nu}\ket{n}$, so the eigenvectors of $H_0+V$ are superpositions of the eigenvectors of $H_0$ with the same energy.  The new eigenvectors are $\textstyle{\frac{1}{\sqrt{2}}}\big(\ket{0}\ket{n+1}\pm\ket{\nu}\ket{n}\big)$, which have eigenvalues $(n+1)\nu\pm\lambda$.  So now all energy eigenstates are maximally entangled, except the ground state $\ket{0}\ket{0}$.  Therefore, the time average state of the subsystem over $[0,\infty)$ is maximally mixed, unless the initial state has overlap with the ground state.  But this is all independent of $\|V\|=|\lambda|$, provided it is non-zero.  Thus, an arbitrarily small interaction strength leads to initial state independence.

Regardless of how we define integrability, the non-interacting model here is definitely integrable.  And it would be reasonable to guess that, if integrable models do not typically have initial state independence, then neither would models that are close to integrable.  Here, however, we can tune $\lambda$ to make this model arbitrarily close to being integrable, yet the subsystem's equilibrium state is independent of its initial state.  That said, it does not thermalize in general.

There are a few issues about this model that we should take note of.  One is that, while $\|V\|$ does not affect subsystem state independence, it ought to affect the equilibration {\it time} because it affects the gaps between energy levels.

Also, the Hamiltonian of the total subsystem plus environment, even after adding the interaction, has many degenerate energy gaps.  We could get around this by having $\lambda$ depend on $n$.  Then the whole system equilibrates, and the results about initial state independence still hold.

This toy model is a little unnatural.  It would be interesting to see if similar behaviour occurs by locally perturbing spin models.  This would be particularly useful as integrability is generally framed in terms of these models.

\chapter{Conclusions and Open Problems}
\label{chap:Conclusions2}
\begin{center}
{\em I don’t understand the question, and I won’t respond to it.}\\
- Lucille Bluth, Arrested Development
\end{center}
To sum up, by using tools from quantum information, including purification and various distance measures between states, we were able to bound the equilibration time with great generality.  One significant aspect of these bounds was that they were dominated by the small energy gaps.  In fact, a feature of the bounds was that gaps that were exponentially small in the system size dominated the formulas, despite the fact that most gaps are not so small, as we saw in chapter \ref{chap:Background}.  Nevertheless, while the estimates of the equilibration time for systems with physical spectra were long, this is still a valuable first step.  It is clear that the next step should be to examine these bounds with more assumptions on the measurement set and Hamiltonian. 

So at the heart of the issue is the following question:\ what do local observables look like in the energy eigenbasis of a local Hamiltonian?  All the proofs in the previous chapter required us to work in the energy eigenbasis.  Indeed, it is hard to see how time averaging could be done in a local basis, without an entirely different approach as in \cite{CDEO08,CE10}.  That said, understanding the properties of eigenstates of local Hamiltonians is an incredibly difficult question.  One possibility is to use the results of \cite{KLW14}.  There it is proved that almost all eigenstates of local spin Hamiltonians are maximally mixed on small subsystems in the limit as the number of spins goes to infinity.  In this context, however, it is important to bear in mind that physical states are not spread out over all energies.  Instead, they have energy restricted to some small window.  So it is not yet clear how useful the results of \cite{KLW14} are from our point of view.

Understanding how local observables look in the energy eigenbasis of local Hamiltonians should also be useful in order to see when the equilibrium state is the Gibbs state.  After all, the issue here is finding out what the equilibrium state of the subsystem is with respect to measurements on the subsystem itself.

Another possibility with potentially general applicability is to look for a good physically motivated measure on states, Hamiltonians or measurements.  In \cite{MGLFS14}, it was shown that equilibration occurs quickly for a two outcome measurement corresponding to the projectors $P$ and $\openone-P$ when $P$ has low rank compared to $d_{\eff}$.  And one of the take-home messages from this paper was that these fast equilibration results, while interesting, do not appear to apply to physically useful observables:\ for example, a projector onto spin up along the $z$-axis on the first spin on a spin chain has high rank because it acts like the identity on the rest of the system.  So a measure that takes the locality of physical observables into account, if mathematically convenient, should lead to useful results.

\section{Equilibration}
In order to obtain stronger results regarding equilibration, it might also be useful to take more of the practical details of experiments into account.  For one thing, measurements take a finite amount of time, so it seems natural that quickly oscillating terms in the density matrix will average to zero over the course of a measurement.  Of course, it is not these terms that cause the problems in our bounds.  It is the small gaps, which correspond to slowly oscillating terms in the density matrix.  With these in mind, it may help to take the finite duration of experiments into account.

For example, when we look at expectation values, the terms corresponding to the small gaps will change very slowly over the timescales of realistic experiments.\footnote{This really defines what it means for a gap to be small here.}  Let us look at the average of $\exp(-i(E_n-E_m)t)=\exp(-iG_{nm}t)$ over the interval $[0,\tau]$ corresponding to the duration of the experiment.  The time $\tau$ is fixed but need not be the same as the equilibration time.
\begin{equation}
  \frac{1}{\tau}\int_0^{\tau}\!\textrm{d}t\,e^{-iG_{nm}t}=\frac{e^{-iG_{nm}\tau}-1}{-iG_{nm}\tau}
\end{equation}
For all gaps where $|G_{nm}|\tau\ll 1$, this will be very close to one since
\begin{equation}
 \frac{e^{-iG_{nm}\tau}-1}{-iG_{nm}\tau}=1+O(|G_{nm}|\tau).
\end{equation}
Now, it is sufficient to show that the system equilibrates to {\it some} state; it does not have to be the infinite time average state $\omega$.  In fact, over a finite interval $[0,\tau]$ one would expect that, if the system equilibrates, it will do so to the state resulting from averaging over that interval $\omega_{\tau}$.  So it is natural to look at the distance from this state to see equilibration.  This means that we are concerned with the quantity
\begin{equation}
\label{eq:u}
 \tr{\rho(t)A}-\tr{\omega_{\tau}A}=\sum_{n,m}a_{mn}\rho_{nm}\left(e^{-iG_{nm}t}-\frac{e^{-iG_{nm}\tau}-1}{-iG_{nm}\tau}\right).
\end{equation}
Here, because we are only concerned with times during the experiment, $t\in[0,\tau]$, we know that for the small gaps $e^{-iG_{nm}t}=1+O(|G_{nm}|t)$, with $t\leq\tau$.  Then
\begin{equation}
e^{-iG_{nm}t}-\frac{e^{-iG_{nm}\tau}-1}{-iG_{nm}\tau}=O(|G_{nm}\tau|).
\end{equation}
So {\it individual} terms corresponding to small gaps in equation (\ref{eq:u}) are small.  It is certainly true that this idea would lead to better bounds on the equilibration time.  The big question is whether it could lead to bounds that are not exponential in the system size for physical systems.

A related issue is the nature of the equilibrium state over different timescales.  In the previous chapter, our focus was on equilibration to the infinite-time average state $\omega$.  In reality, however, systems may evolve through a series of temporary equilibrium states over time.  Recall the example of the cup of tea from chapter \ref{chap:Background}.

Defining the time average of the state $\rho(t)$ over the interval $[0,T]$ by $\omega_T = \av{\rho(t)}_T$, the timescales over which equilibration occurs to temporary equilibrium states will correspond to values of $T$ with
\begin{equation}
 \frac{\av{ |\tr{\rho(t)A} - \tr{ \omega_T A }|^2 }_T}{\norm{A}^2}\ll 1.
\end{equation}
Now suppose a system has a set of different time intervals $[0,T_i]$ over which it equilibrates to distinguishable equilibrium states $\omega_{T_i}$.  Then the timescale for each equilibration must be much longer than the previous one. This follows because a state cannot be close to $\omega_{T_i}$ for most of $[0,T_i]$ and $\omega_{T_{i+1}}$ for most of $[0,T_{i+1}]$ unless $T_{i+1}$ is much greater than $T_i$.  In particular, this means that, if a system equilibrates to its infinite-time average state, then the timescale for any temporary equilibration to a very different state must be much less than the timescale for equilibration to the infinite-time average state.

In \cite{Kastner11} a lattice model is presented that has an equilibration time that diverges with the number of sites.  The main result is summed up by a theorem.
\begin{theorem}
\label{th:Kastner}
 Take a Hamiltonian on $N$ spins with on-site $Z_n$ terms corresponding to an external magnetic field and interaction terms $Z_nZ_m$ with coupling coefficients decaying like $|n-m|^{-\alpha}$, where $0\leq \alpha < 1$.  Let $A$ be an observable of the form $A=\sum_n a_n X_n$ and let the initial state $\rho(0)$ be diagonal in the local $X_n$ eigenbasis.  It follows that, for any $T$ and small $\delta$, there exists an $N_T$ such that
 \begin{equation}
  |\tr{\rho(t) A}-\tr{\rho(0)A}|<\delta
 \end{equation}
 for all $t\leq T$ and $N>N_T$.  It is important to mention that the $Z_nZ_m$ terms in the Hamiltonian are rescaled for each $N$ in order to make the energy per spin finite.  This rescaling plays a crucial role in the resulting equilibration timescale.
\end{theorem}
In words, for any $T$ we can choose a sufficiently big system such that the state stays close to its initial state over the interval $[0,T]$.  In \cite{Kastner11}, this is interpreted to mean that the equilibration time diverges in the thermodynamic limit (as $N$ tends to infinity).  And it is suggested that, because of this, equilibration would not occur in practice if a big enough system with this class of Hamiltonian could be prepared in the lab.\footnote{It is good to point out that, if the system stays close to its initial state for the duration of an experiment, then it {\it is} in equilibrium.}  We can think about this in a different way:\ it is just an excellent example of a system that has different equilibration states over different timescales.  The system is actually in equilibrium over the timescale defined in the theorem, and we can have equilibration to a different state over longer timescales.

This poses the question of whether there are cases where the physically observed process of equilibration is not to the infinite-time average state but to some short term equilibrium state.  The example with the cup of tea suggests that this is possible.  But if that is so, does that mean that attempts to find reasonable upper bounds on the equilibration time are doomed to failure if we only consider equilibration to the infinite-time average state?  Working with the infinite-time average state is the mathematically convenient option, but it may not always be the right one.

It would also be interesting to compare the systems of theorem \ref{th:Kastner} with the slowly equilibrating system in section \ref{sec:Slow Equilibration}.  In the latter, the equilibration time grows exponentially with the system size.  It is not clear, however, whether such behaviour is possible with a local Hamiltonian and local measurements.

This brings us to the question of the role played by locality in equilibration.  For example, is it true that interacting local Hamiltonians generally equilibrate quickly?  With this in mind, it would be interesting to average over Hamiltonians with local interaction terms.  Perhaps, instead of using local Haar averages (which may be feasible too), it could be an option to look at Hamiltonians with $X_nX_{n+1}$, $Y_nY_{n+1}$ and $Z_nZ_{n+1}$ interactions terms between sites with couplings that are independent and identically distributed random variables.  These Hamiltonians do not generally satisfy translational invariance, but, while locality is physically indispensable, translational invariance can be broken easily.  For example, if the interaction strength depended on distance in a lattice of trapped ions, then variations in the distance would lead to inhomogeneous interaction strengths.

\section{Thermalization}
The arguments of \cite{RGE11}, proving that in some cases the equilibrium state is a Gibbs state, go via the trace distance.  Their main theorem results in a bound on the trace distance between the equilibrium state and the Gibbs state.  A step in the argument bounds this from above by an expression involving $\|P -P_0\|_1$, where $P$ and $P_0$ denote the projectors onto the energy eigenspaces corresponding to the energy range $[E,E+\Delta]$ in the interacting and non interacting cases respectively.  A problem with this is that this distance $\|P -P_0\|_1$ is sensitive to global changes in energy projectors, whereas in practice we only have access to some restricted measurement set, which may just correspond to measurements on the subsystem in question.  So it may be possible to extend the results of \cite{RGE11} by using a distance measure other than the trace distance.  A good next step would then be to look at the distinguishability with a restricted measurement set, possibly with some assumptions on how well the measurements could distinguish energy projectors corresponding to the interacting and non interacting Hamiltonians.  Hopefully, with some sensible assumptions, the results of \cite{RGE11} could be extended to systems where the interaction strength is extensive.

An option that we have ignored so far is the eigenstate thermalization hypothesis \cite{RS12}, which, roughly speaking, postulates that the energy eigenstates of the interacting Hamiltonian look like a Gibbs state on the subsystem.  This essentially transfers all the mathematical difficulty of proving thermalization onto proving that the energy eigenstates have this form.  The jury is still out on this, but it does emphasize the point that more information is needed regarding the local properties of energy eigenstates of physical Hamiltonians.

\bookmarksetup{startatroot}
\addtocontents{toc}{\bigskip}

\markboth{}{}
\cleardoublepage
\phantomsection
\thispagestyle{plain}
\addcontentsline{toc}{chapter}{Epilogue}
\printindex

\begin{center}
    
    \vspace{9cm}
    \Large
    \textbf{Epilogue}
\end{center}
In the first part of this thesis, we explored the possibility of using discrete-time models from quantum computation as discretized models of relativistic quantum systems.  Much of what we saw gave us cause for optimism that this idea could lead to new insights into relativistic physics or even offer new possibilities for quantum simulation.  Even more, it prompted the question of whether nature could be a quantum cellular automaton.  There was some evidence in this direction, which provides compelling fodder for speculation.  After all, if we can simulate physical systems arbitrarily well by discretized systems, then it is natural to ask whether the laws of nature are indeed discrete.

In the second part of this thesis, we studied fundamental questions in statistical physics in an abstract manner.  The motivation was that, by making the fewest possible assumptions, we could make the most general statements about things like the equilibration time of a quantum system.  Perhaps one of the most interesting offshoots of this approach is the idea that whether or not equilibration occurs, as well as the form of the equilibrium state, depends crucially on the timescale.

Summed up in a sentence, the central theme in this thesis was the application of ideas from quantum information to fundamental physics.  To be fair, this idea is not a new one.  As we saw in the introductions to both parts one and two, quantum information has been applied with great success to high energy physics \cite{BHVVV14,TCL14,CMP10,CLEGRGS11,CMLS12} and the foundations of statistical physics \cite{GLTZ06,PSW06,Tasaki98,GMM04,CDEO08,Reimann08,LPSW09,Short10,LPSW10,Reimann10,CE10,RGE11}.  Nevertheless, if there is a message to take home here, it is that there is a great deal more to learn by applying quantum information to problems in physics in general.  It is exciting to wonder what new discoveries lie around the corner.

\cleardoublepage
\phantomsection
\thispagestyle{plain}
\addcontentsline{toc}{chapter}{Bibliography}
\printindex

\bibliographystyle{unsrt}

\markboth{}{}
\cleardoublepage
\phantomsection
\thispagestyle{plain}
\addcontentsline{toc}{chapter}{Appendix:\ Dynamics on the b.c.c.\ Sublattice}
\printindex
\begin{center}
    
    \vspace{9cm}
    \Large
    \textbf{Appendix:\ Dynamics on the b.c.c.\ Sublattice}
\end{center}
\label{sec:Dynamics on the b.c.c. Sublattice}
\hypertarget{QQAppendix}{As} we saw in section \ref{sec:Fermion Doubling in Two Dimensions} in chapter \ref{chap:Quantum Walks and Relativistic Particles}, we may want to restrict our attention to particles that live on a b.c.c.\ sublattice.  Here we will show that, even if the system just has fermionic modes on the b.c.c.\ sublattice there is still a local decomposition of the dynamics.  It is far more convenient to work in the single particle quantum walk picture, which extends to the fermionic picture by second quantization.  In the quantum walk picture, an example of a local unitary is the swap that, for {\it fixed} coordinates $n$ and $m$, implements
\begin{equation}
 \ket{n,m}\leftrightarrow \ket{n+1,m+1}
\end{equation}
and acts like the identity on all other states.  After second quantization, the corresponding unitary would swap the fermion modes at sites $(n,m)$ and $(n+1,m+1)$, while leaving all other modes unaffected.

So let us look at the dynamics restricted to a b.c.c.\ sublattice for the two dimensional massless Dirac quantum walk.  Before restricting to the sublattice, the evolution operator is the product of two conditional shifts:
\begin{equation}
 U =e^{-iP_x\sigma_xa}e^{-iP_z\sigma_za}.
\end{equation}
Notice that instead of taking the second conditional shift to be $e^{-iP_y\sigma_ya}$, as it was in chapter \ref{chap:Quantum Walks and Relativistic Particles}, we take take it to be $e^{-iP_z\sigma_za}$.  This just makes the formulas simpler as complex numbers do not appear.

Let us look at the b.c.c.\ sublattice that includes the point $(0,0)$.  A particle on that sublattice will remain on it under the evolution above.  Restricted to that sublattice, this defines a quantum walk with unitary $V$.  Though it is not immediately clear that $V$ can be implemented by products of local unitaries, we will now show that this is the case.  The notation is easier to read if we relabel
\begin{equation}
 \begin{split}
  \ket{z} & =\ket{\uparrow_z}\ \textrm{and}\ \ket{z^{\perp}} =\ket{\downarrow_z}\\
  \ket{x} & =\ket{\uparrow_x}\ \textrm{and}\  \ket{x^{\perp}}  =\ket{\downarrow_x}.
 \end{split}
\end{equation}

First, we apply the conditional shift
\begin{equation}
 \begin{split}
  R_1\ket{z}\ket{n_x,n_z} & =\ket{z}\ket{n_x+1,n_z+1}\\
  R_1\ket{z^{\perp}}\ket{n_x,n_z} & =\ket{z^{\perp}}\ket{n_x-1,n_z-1}
 \end{split}
\end{equation}
for all $(n_x,n_z)$ on the b.c.c.\ sublattice.
These can be implemented by products of local unitaries.  This works in the following way.  We apply the swaps
\begin{equation}
 \ket{z}\ket{n_x,n_z}\leftrightarrow \ket{z^{\perp}}\ket{n_x+1,n_z+1}
\end{equation}
for each $(n_x,n_z)$ on the b.c.c.\ sublattice, and then we apply the $X$ operator to the coin degree of freedom.  Next, for each $(n_x,n_z)$, we apply a local unitary that superposes the states $R_1\ket{z}\ket{n_x,n_z}$ and $R_1\ket{z^{\perp}}\ket{n_x,n_z+2}$.  Denoting the product of these local unitaries by $R_2$, we get
\begin{equation}
 \begin{split}
  R_2R_1\ket{z}\ket{n_x,n_z} & = \frac{1}{\sqrt{2}}\left(\ket{z}\ket{n_x+1,n_z+1} +\ket{z^{\perp}}\ket{{n_x-1},n_z+1} \right)\\
  R_2R_1\ket{z^{\perp}}\ket{n_x,n_z+2} & =\frac{1}{\sqrt{2}}\left(\ket{z}\ket{n_x+1,n_z+1} -\ket{z^{\perp}}\ket{{n_x-1},n_z+1}\right).
 \end{split}
\end{equation}
Finally, we apply the Hadamard unitary $H$, which takes $\ket{z}$ to $\ket{x}$ and $\ket{z^{\perp}}$ to $\ket{x^{\perp}}$.  This results in
\begin{equation}
 \begin{split}
  HR_2R_1\ket{z}\ket{n_x,n_z} & = \frac{1}{\sqrt{2}}\left(\ket{x}\ket{n_x+1,n_z+1} +\ket{x^{\perp}}\ket{{n_x-1},n_z+1} \right)\\
  HR_2R_1\ket{z^{\perp}}\ket{n_x,n_z+2} & =\frac{1}{\sqrt{2}}\left(\ket{x}\ket{n_x+1,n_z+1} -\ket{x^{\perp}}\ket{{n_x-1},n_z+1}\right),
 \end{split}
\end{equation}
which is exactly what we would get by applying the original unitary $V$.  Therefore, $V$ is locally implementable.  An analogous but more complicated process works for the Weyl quantum walk in three dimensional space.

\end{document}